\documentclass[11pt,letter]{article}

\pdfoutput=1 

\usepackage[english]{babel}
\usepackage{amsmath,amssymb,graphicx,enumerate,latexsym,calc,capt-of,xcolor}
\usepackage{epigraph}

\usepackage{dsdshorthand}
\usepackage{xcolor}
\usepackage{booktabs}
\usepackage[utf8]{inputenc}

\usepackage{cite}

\usepackage{amsthm}
\newtheorem  {theorem} {Theorem} [section]
\newtheorem  {lemma}       [theorem] {Lemma}

\newtheoremstyle{myremark}
{3pt}
{3pt}
{}
{}
{\bf}
{.}
{.5em}
{}
\theoremstyle{myremark}
\newtheorem{remark}{Remark}[section]

\definecolor{darkblue}{rgb}{0.1,0.1,.7}
\usepackage[colorlinks, linkcolor=darkblue, citecolor=darkblue, urlcolor=darkblue, linktocpage,backref=page]{hyperref} 
\renewcommand*{\backref}[1]{}
\renewcommand*{\backrefalt}[4]{%
	\ifcase #1 (Not cited.)%
	\or        (Cited on p.~#2.)%
	\else      (Cited on pp.~#2.)%
	\fi}

\usepackage[]{latexsym}
\usepackage{geometry}
\usepackage{marginnote}

\newcommand{\assign}{:=}
\newcommand{\asterisk}{\mathord{*}}
\newcommand{\cdummy}{\cdot}
\newcommand{\comma}{{,}}
\newcommand{\nin}{\not\in}
\newcommand{\nobracket}{}

\newcommand{\tmem}[1]{{\em #1\/}}
\newcommand{\tmmathbf}[1]{\ensuremath{\boldsymbol{#1}}}
\newcommand{\tmop}[1]{\ensuremath{\operatorname{#1}}}
\newcommand{\tmtextbf}[1]{{\bfseries{#1}}}
\newcommand{\tmtextit}[1]{{\itshape{#1}}}

\newcommand{\infixor}{\mathrm{or}}
\newcommand{\infixand}{\mathrm{and}}

\numberwithin{equation}{section}

\usepackage{chngcntr}
\counterwithin{figure}{section}

\usepackage{etoolbox}
\apptocmd{\thebibliography}{\setlength{\itemsep}{0em}}{}{}

\begin{document}
\vspace*{-.6in} \thispagestyle{empty}
\begin{flushright}
\end{flushright}
\vspace{1cm} 
{\Large
	\begin{center}
		{\bf Distributions in CFT II. Minkowski Space}
	\end{center}
}
\vspace{1cm}
\begin{center}
	{\bf Petr Kravchuk$^a$, Jiaxin Qiao$^{b,c}$,  Slava Rychkov$^{c,b}$}\\[2cm] 
	{
		$^{a}$ Institute for Advanced Study, Princeton, New Jersey 08540, U.S.A. \\
		$^b$  
		Laboratoire de Physique de l'Ecole normale sup\'erieure, ENS,\\ 
		{\small Universit\'e PSL, CNRS, Sorbonne Universit\'e, Universit\'e de Paris,} F-75005 Paris, France
		\\
		$^c$  Institut des Hautes \'Etudes Scientifiques, 91440 Bures-sur-Yvette, France\\
	}
	\vspace{1cm}
\end{center}

\setlength\epigraphrule{0pt}
\setlength\epigraphwidth{.4\textwidth}
\epigraph{\emph{I wanted to learn about elementary particles\\ by studying boiling water.}}{Alexander Polyakov \cite{Polyakov}}

\vspace{4mm}
\begin{abstract}
  CFTs in Euclidean signature satisfy well-accepted rules, such as the
  convergent Euclidean OPE. It is nowadays common to assume that CFT
  correlators exist and have various properties also in Lorentzian signature.
  Some of these properties may represent extra assumptions, and it is an open
  question if they hold for familiar statistical-physics CFTs such as the
  critical 3d Ising model. Here we consider Wightman 4-point functions of
  scalar primaries in Lorentzian signature. We derive a minimal set of their
  properties solely from the Euclidean unitary CFT axioms, without using extra
  assumptions. We establish all Wightman axioms (temperedness, spectral
  property, local commutativity, clustering), Lorentzian conformal invariance,
  and distributional convergence of the s-channel Lorentzian OPE. This is done
  constructively, by analytically continuing the 4-point functions using
  the s-channel OPE expansion in the radial cross-ratios $\rho, \bar{\rho}$.
  We prove a key fact that $| \rho |, | \bar{\rho} | < 1$ inside the forward
  tube, and set bounds on how fast $| \rho |, | \bar{\rho} |$ may tend to 1
  when approaching the Minkowski space.
  
  We also provide a guide to the axiomatic QFT literature for the modern CFT
  audience. We review the Wightman and Osterwalder-Schrader (OS) axioms for
  Lorentzian and Euclidean QFTs, and the celebrated OS theorem connecting
  them. We also review a classic result of Mack about the distributional OPE
  convergence. Some of the classic arguments turn out useful in our setup.
  Others fall short of our needs due to Lorentzian assumptions (Mack) or
  unverifiable Euclidean assumptions (OS theorem). 
\end{abstract}

\vspace{.2in}
\vspace{.3in}
\hspace{0.2cm} April 2021

\newpage

{
	\tableofcontents
}

\section{Introduction}

Quantum Field Theory (QFT) can be studied via constructive or axiomatic
approaches. The former use microscopic formulations, while the latter rely on
axioms. There are many constructive approaches, e.g.\ using Hamiltonian, path
integral, lattice, etc. There are also many axiomatic approaches, corresponding
to various sets of axioms (Wightman {\cite{Streater:1989vi}},
Osterwalder-Schrader {\cite{osterwalder1973,osterwalder1975}}, Haag-Kastler
{\cite{Haag:1963dh,Haag:1992hx}}, etc.). Historically, axiomatic approaches
played an important role in clarifying general QFT properties, but they did
not have a tremendous success in making predictions about concrete theories in $d>2$ dimensions.\footnote{In $d=2$ significant progress has been achieved axiomatically for massive integrable models using the S-matrix bootstrap \cite{Zamolodchikov:1978xm} as well as for rational CFTs \cite{Belavin:1984vu}.}
This started to change recently, with the revival of the bootstrap philosophy
{\cite{Rattazzi:2008pe}}. Our focus here will be on conformal field theories
(CFTs) in dimension $d \geqslant 2$, i.e.\ QFTs invariant under the action of
conformal group, which are nowadays studied via the \tmtextit{conformal
bootstrap}. This axiomatic approach led to precise determinations of many
experimentally measurable quantities, such as the critical exponents of the 3d
Ising
{\cite{ElShowk:2012ht,El-Showk:2014dwa,Kos:2014bka,Simmons-Duffin:2015qma,Kos:2016ysd}},
$O (N)$ {\cite{Kos:2013tga,Kos:2015mba,Kos:2016ysd,Chester:2019ifh,Chester:2020iyt}} and other
critical points (see review~{\cite{RMP}}).\footnote{There is also an ongoing
revival of the S-matrix bootstrap applicable to nonintegrable massive QFTs in $d\ge 2$
{\cite{Paulos:2016fap,Paulos:2016but,Paulos:2017fhb,Cordova:2018uop,Guerrieri:2019rwp,EliasMiro:2019kyf,Cordova:2019lot,Karateev:2019ymz,Correia:2020xtr,Guerrieri:2020bto,Hebbar:2020ukp,Tourkine:2021fqh,Sinha:2020win,Haldar:2021rri,He:2021eqn}}.}

The numerical conformal bootstrap relies on the ``Euclidean CFT
axioms'',\footnote{The term ``Euclidean bootstrap axioms'' is also sometimes
used.} which specify properties of correlation functions in any unitary CFT in
$\mathbb{R}^d$ via a set of simple and commonly accepted rules, such as the
unitarity bounds on primary operator dimensions, conformally invariant and
convergent Operator Product Expansion (OPE), and reality constraints on OPE
coefficients.

On the other hand, correlation functions in a general unitary QFT (and in
particular in a CFT) should satisfy Osterwalder-Schrader (OS) and Wightman
axioms. It is then interesting and important to inquire what is the relation
of Euclidean CFT axioms to these other sets of axioms.\footnote{Clarifying the
relation to the Haag-Kastler axioms appears more challenging as those axioms
do not deal with correlation functions but with operator algebras.} To carry
out this analysis will be the main goal of our paper. Our main conclusion will
be that the Euclidean CFT axioms imply OS axioms and Wightman axioms for $2,
3$ and $4$-point functions. In this paper we only consider bosonic operators.

The relation of Euclidean CFT and OS axioms is perhaps not so surprising since
they both deal with the Euclidean correlation functions. It is more
interesting that we are able to construct Minkowski $n$-point functions (for
$n = 2, 3, 4$), and show that they satisfy Wightman axioms, such as
temperedness, spectral condition, and unitarity. Temperedness (being a
tempered distribution) is a crucial property of Minkowski correlation
functions: it shows that in a certain averaged sense they are meaningful
everywhere including the light-cone and double light-cone singularities. One
might be tempted to think that in CFT this question is relatively trivial due
to OPE. However, as discussed in {\cite{paper1,Qiao:2020bcs}}, already for
4-point functions there exist causal configurations of points in Minkowski
space, away from the null cones, for which no OPE channel is convergent in the
conventional sense. We briefly discuss one such example in the conclusions
(Sec.\ \ref{conclusions}).

A theorem of Osterwalder and Schrader
{\cite{osterwalder1973,osterwalder1975}} says that, under some extra
assumption, OS axioms imply Wightman axioms. Unfortunately this extra
assumption, the so called ``linear growth condition'', which involves the
Euclidean $n$-point functions with arbitrarily high $n$, appears impossible to
verify from the Euclidean CFT axioms (see Sec.\ \ref{OS}). For this reason
we cannot appeal to the OS theorem, and we will give an independent derivation
of Wightman axioms for CFT correlators.

The study of distributional properties of CFT correlators started in our
recent work {\cite{paper1}}. There, we considered expansions of the CFT 4-point
function $g (\rho, \bar{\rho})$ in terms of conformally invariant cross-ratios
$\rho$, $\bar{\rho}$. While such expansions converge in the usual sense for $|
\rho |, | \bar{\rho} | < 1$, in {\cite{paper1}} we showed that they also
converge for $| \rho |, | \bar{\rho} | = 1$ in the sense of distributions. As
explained in {\cite{paper1}}, results of this sort follow naturally from the
theory of functions of several complex variables (namely Vladimirov's
theorem), given some apriori information about the growth of the analytically
continued correlator. That key insight of {\cite{paper1}}, \emph{``Look for a
powerlaw bound!''}, will be transported here from the cross-ratio to the
position space.

The readers interested in our main technical result---analytic continuation of
a scalar Euclidean CFT 4-point function to the forward tube and showing that the
Minkowski 4-point function is a tempered distribution---may follow the
\tmtextbf{fast track:} start with the executive summary in Sec.\ \ref{executive}, proceed to Secs.\ \ref{strategy} and \ref{23-point} (skipping
\ref{MinkFromEucl} and \ref{comparison}), then continue with Secs.\ \ref{sec:informal}-\ref{power4-point} (optionally including Sec.\ \ref{secondpass}) and finish with the discussion in Sec.\ \ref{conclusions}.
This is only about 20-25 pages.

On the other hand, we made an effort to make the exposition self-contained
and to review main ideas and results of the axiomatic quantum field
literature, directly or tangentially related to our discussion. This explains
the great total length of our work. The reader will find here:
\begin{itemize}
  \item A review of classic QFT axioms: Wightman (Sec.\ \ref{Waxioms}), OS
  (Sec.\ \ref{OSaxioms}). A review of main implications among these axioms:
  how OS reflection positivity robustly implies Wightman positivity (Sec.\ \ref{sec:Wpos}). A review of the Osterwalder-Schrader theorem about how OS
  axioms imply Wightman axioms under the additional assumptions of the linear
  growth condition (Sec.\ \ref{OS}).
  
  \item A \ formulation of `Euclidean CFT axioms' for unitary CFT in
  Euclidean space $\mathbb{R}^d$ (Sec.\ \ref{ECFTax}). We consider bosonic
  fields in arbitrary tensor representations. Our axioms encode in a
  consistent and non-redundant manner the main properties used in the
  numerical conformal bootstrap calculations.\footnote{See also
  {\cite{Rychkov:2020rcd}} for a recent informal exposition of Euclidean CFT
  axioms (incomplete as it omits tensor fields) for mathematical physics
  audience. Ref.\ {\cite{Schwarz:2015fva}} attempted the axiomatization of
  Euclidean CFT in $d > 2$ dimensions similar to Segal's axioms in $d = 2$
  {\cite{tillmann2004}}. It is not immediately obvious if the axioms of Ref.\ {\cite{Schwarz:2015fva}} are equivalent to ours, or how to connect them to
  practical CFT calculations.} They are applicable to any globally conformally
  invariant theory in $d \geqslant 2$. We do not include the axioms involving
  the local stress tensor and the conserved currents. In particular our axioms
  would be too weak (but valid) when applied to local 2d CFTs, as they know
  nothing about the Virasoro algebra.\footnote{Recall that while in $d=2$
  assuming the existence of a local stress tensor immediately implies Virasoro
 symmetry, no such dramatic statements are currently known in $d>2$.} See Remark \ref{ECFTax-comp} for a
  comparison between our axioms and the CFT rules gathered in the conformal
  bootstrap review {\cite{RMP}}.
  
  \item A derivation of OS axioms from Euclidean CFT axioms for 4-point function
  (Sec.\ \ref{CFTtoOS}). A notable result is a rigorous proof that the state
  produced by two operators in lower Euclidean half-space belongs to the CFT
  Hilbert space generated by single-point operator insertions. The higher-point
  case is discussed in App.\ \ref{OShigher}, where we need a somewhat
  stronger form of the OPE axiom than in Sec.\ \ref{ECFTax}.
  
  \item A derivation of Wightman axioms from Euclidean CFT axioms for scalar
  4-point functions (Sec.\ \ref{sec:4-point}). As mentioned above, this is the main
  technical result of the paper. The key observation is that the analytic
  continuation from Euclidean to Minkowski can be done in a way which keeps
  the s-channel $\rho, \bar{\rho}$ less than 1 in absolute value along the
  continuation contour. When we take the Minkowski limit $| \rho |, |
  \bar{\rho} |$ stay less than 1 for some causal orderings and approaches 1
  for others (see {\cite{Qiao:2020bcs}} for a classification) but even if $|
  \rho |, | \bar{\rho} | \rightarrow 1$ they approach this limit sufficiently
  slowly (``powerlaw bound'') which guarantees that the Minkowski 4-point
  function is a tempered distribution everywhere. E.g.\ using this argument we
  can show for the first time that the CFT 4-point function is a tempered
  distribution on the double light-cone singularity.
  
  \item We include also a derivation of other expected properties of Minkowski
  4-point functions, such as conformal invariance, unitarity, clustering, and
  local commutativity (Secs.\ \ref{ConfMink}-\ref{local-comm}). The reader
  may find it curious how some of the steps do not use conformal invariance at
  all but follow simply from analyticity and/or OS positivity.
  
  \item Sec.\ \ref{secondpass} proves a curious geometric ``Cauchy-Schwarz''
  inequality for $\rho, \bar{\rho}$ variables which provides an alternative
  way of understanding why $| \rho |, | \bar{\rho} | < 1$ in the forward tube.
  It bounds $\rho, \bar{\rho}$ for a generic configuration by $\rho,
  \bar{\rho}$ for reflection-symmetric configurations. It would be interesting
  to find an elementary proof of this inequality (our proof uses some facts
  about conformal blocks).
  
  \item Sec.\ \ref{OPEconvMink} shows that the (s-channel) conformal block
  expansion of 4-point Wightman functions converges in the sense of
  distributions for all configurations of points in Minkowski space. It is
  also shown that the OPE for the state-valued distributions $| \mathcal{O}
  (x_1) \mathcal{O} (x_2) \rangle$ with $x_1, x_2 \in \mathbb{R}^{d - 1, 1}$
  converges in the sense of distributions. We discuss the relationship of
  these results to the classic work of Mack {\cite{Mack:1976pa}} and prove
  estimates for the convergence rates of these expansions.
  
  \item Sec.\ \ref{OS} contains a review of the papers
  {\cite{osterwalder1973,osterwalder1975}} by Osterwalder and Schrader. In
  particular, we discuss the gap in the arguments of {\cite{osterwalder1973}}
  which precludes the derivation of Wightman axioms from the OS axioms of
  {\cite{osterwalder1973}}, and explain in detail how this gap is filled in
  {\cite{osterwalder1975}} with the addition of new axioms. 
  
  \item {App.~\ref{literature} is a guide to the modern Lorentzian
  CFT literature: conformal collider bounds, light-cone bootstrap, causality constraints, the Lorentzian OPE inversion formula, light-ray operators, etc. Our results help put some of these considerations on a firmer footing. We indicate the most critical remaining steps, which still wait to be rigorously derived from the Euclidean CFT axioms.}
  
\end{itemize}
We conclude in Sec.~\ref{conclusions}. Some additional technical details are given in
Apps.~\ref{OShigher}-\ref{IntLem1}. 

\subsection{Executive summary of results for CFT experts}\label{executive}

This paper is rather lengthy as a result of our attempt to make it
self-contained. In this section we give a brief summary of the main technical
results, aimed at the more expert readers who may not wish to read the
expository parts of this paper. Note, however, that here we omit many
secondary results, some of which are mentioned above.

The basic question we address in this paper is the question of the
distributional properties of Wightman 4-point functions in CFTs. As is
well-known, Wightman $n$-point distributions are recovered from the boundary
values of functions holomorphic in the forward tube $\mathcal{T}_n$. {For an $n$-point
function 
\be
	\<0|\cO_1(x_1)\cdots \cO_n(x_n)|0\>
\ee
the forward tube is defined as the set of $x_i\in \C^{1,d-1}$ subject to 
\be
\mathrm{Im}\, x_1\prec \mathrm{Im}\, x_2\prec\cdots \prec \mathrm{Im}\, x_n,
\label{forward-tube-Mink}
\ee
where $\prec$ denotes the causal ordering in $\R^{1,d-1}$. Analyticity in $\cT_n$ and existence of
the boundary value as $\mathrm{Im}\, x_i\to 0$ is usually derived 
from Wightman or OS axioms (with extra assumptions in the latter case). 
In this paper we want to understand this question from the point of view
of CFT axioms.} 

With the
cases $n = 2, 3$ being relatively trivial in a CFT, our main observation is
that a particular OPE channel for 4-point functions converges everywhere in
the interior of $\mathcal{T}_4$. Specifically, we take the OPE $\mathcal{O}
(x_1) \times \mathcal{O} (x_2)$ in the Wightman function of identical scalar operators
\begin{equation}
  \langle 0 | \mathcal{O} (x_1) \mathcal{O} (x_2) \mathcal{O} (x_3)
  \mathcal{O} (x_4) | 0 \rangle . \label{introcrr}
\end{equation}
This OPE is expected to converge, at least distributionally for real $x_i$,
from the results of Mack {\cite{Mack:1976pa}}. However, his work assumes
Wightman axioms from the beginning, and our goal here is to clarify the
implications of Euclidean CFT axioms, which only assume convergence of the
Euclidean OPE.

To see that this OPE channel converges, we show in Lemma \ref{bound} that for
any configuration of $x_i$ in $\mathcal{T}_4$ the radial cross-ratios $\rho$
and $\bar{\rho}$  for this OPE belong to the open unit disk (for the definition of radial cross-ratios, see Sec.\ \ref{Eucl4-point}):
\begin{equation}
  | \rho |, | \bar{\rho} | < 1.
\end{equation}
This implies convergence of the conformal block expansion in $\cO (x_1)
\times \cO (x_2)$ channel in the interior of $\mathcal{T}_4$. A technical
way to see this is to note that the expansion in descendants
\begin{equation}\label{eq:introexpansion}
  g (\rho, \bar{\rho}) = \sum_{h, \bar{h} > 0} p_{h, \bar{h}} \rho^h
  \bar{\rho}^{\bar{h}},
\end{equation}
where $g$ is the conformally-invariant part of the 4-point function and
$p_{h, \bar{h}} > 0$, can be bounded term-by-term by
\begin{equation}
  | p_{h, \bar{h}} \rho^h \bar{\rho}^{\bar{h}} | \leqslant p_{h, \bar{h}} r^h
  r^{\bar{h}}, \label{introtermbyterm}
\end{equation}
where $r = \max (| \rho |, | \bar{\rho} |) < 1$. The right-hand side of this
inequality is a term in the expansion of $g (r, r)$, a Euclidean configuration
in which the OPE is known to converge, so~\eqref{eq:introexpansion} is dominated
by a convergent series. Therefore, \eqref{eq:introexpansion} is convergent for $r<1$, and moreover
uniformly so on compact subsets, since each term $p_{h, \bar{h}} r^h
r^{\bar{h}}$ is monotonic. We can then conclude that the sum $g(\rho,\bar \rho)$ is a holomorphic function.

This reasoning also gives us the inequality
\begin{equation}
  | g (\rho, \bar{\rho}) | \leqslant g (r, r) .
\end{equation}
So, we find that the correlator can be recovered inside of $\mathcal{T}_4$ from
the $\cO (x_1) \times \cO (x_2)$ OPE, is analytic there, and is bounded
by a Euclidean configuration.

In Sec.\ \ref{section:1-r} we establish a stronger form of Lemma
\ref{bound}, schematically,
\begin{equation}
  1 - r (c) \geqslant \tmop{dist} (c, \partial \mathcal{T}_4)^a \label{rbound}
\end{equation}
for some $a > 0$, where $c \in \mathcal{T}_4$ is a configuration of 4 points
in $\mathcal{T}_4$. (The more precise form also bounds $1 - r (c)$ as $c$ goes
to infinity.) This immediately implies a powerlaw bound on $g (r, r)$ near the
boundary of $\mathcal{T}_4$. Indeed, near $r \rightarrow 1$ the correlator is
dominated by the identity operator in the crossed channel, and so
\begin{equation}
  g (r (c), r (c)) \leqslant C (1 - r (c))^{- 4 \Delta_{\varphi}},
\end{equation}
and thus
\begin{equation}
  | g (c) | \leqslant C \tmop{dist} (c, \partial \mathcal{T}_4)^b
\end{equation}
for some real $b$. This allows us to use Vladimirov's Theorem \ref{ThVlad},
which implies that the boundary limit (as $x_i$ approach real Minkowski
values) of {\eqref{introcrr}} exists in the space of tempered distribution.
(We establish a more refined bound for $x_i \rightarrow \infty$ to claim
temperedness.)

The above bounds hold just as well for the truncated conformal block expansion
as for the full correlator. A variant of Theorem \ref{ThVlad} then allows us
to conclude that the conformal block expansion, while converging in the sense
of functions in the interior of $\mathcal{T}_4$, converges in the space of
tempered distributions on the Minkowski boundary.

We extend the above results to correlators of non-identical scalars by
replacing the term-by-term bound {\eqref{introtermbyterm}} with a standard
Cauchy-Schwarz argument, bounding the correlator in terms of a product of two
reflection-symmetric Euclidean correlators. While it is intuitively obvious
that similar arguments should also work for operators with spin, we found that
the extension to spinning operators, due to the complexity of tensor
structures, requires enough additional work to warrant a separate paper
{\cite{paper2a}}.

Finally, in Sec.\ \ref{secondpass} we prove Theorem \ref{boundThm}, which
gives an optimal bound of the form {\eqref{rbound}}. Specifically, it is
\begin{equation}
  r (c)^2 \leqslant r (c_{12}) r (c_{34}), \label{introoptimal}
\end{equation}
together with a bound for the right-hand side. Here, if $c = (x_1, x_2, x_3,
x_4)$ (where $x_i$ are real in Minkowski space), then $c_{12} \equiv (x_1,
x_2, x_2^{\ast}, x_1^{\ast})$ and $c_{34} \equiv (x_4^{\ast}, x_3^{\ast}, x_3,
x_4)$. The bound for the right-hand side is easier to obtain because the
configurations $c_{12}$ and $c_{34}$ are reflection symmetric. This is done in
Sec.\ \ref{rhsbound}. The bound {\eqref{introoptimal}} looks like a
Cauchy-Schwarz-type inequality, and is indeed derived from the
Cauchy-Schwarz inequality for unitary conformal blocks {\eqref{CB-CS}}. The
latter is true because of the unitarity of conformal representations
corresponding to these blocks. In the limit $\Delta + \ell \rightarrow
\infty,$ $\Delta - \ell$ fixed, conformal blocks are dominated by $r^{(\Delta
+ \ell) / 2}$, which reduces the conformal block Cauchy-Schwarz inequality to
{\eqref{introoptimal}}.

\section{Axioms}

\subsection{Wightman axioms}\label{Waxioms}

In this section we will state the properties of Wightman correlation functions in
a unitary QFT, to which we will refer here as ``Wightman axioms.'' These
axioms appear as as ``properties of
the vacuum expectation values'' in {\cite{Streater:1989vi}}, Sec.\ 3-3, and as (W1)-(W5) in \cite{simon1974}, Theorem II.1. Refs.\ {\cite{Streater:1989vi,simon1974}} give in
addition another set of axioms (called G{\aa}rding-Wightman axioms in \cite{simon1974}) saying that fields are operator-valued
distributions in the Hilbert space on which the Lorentz group acts, etc. This
other set of axioms will not be used in this work. In any case, the Wightman
reconstruction theorem {\cite{Streater:1989vi}} says that the two sets of
axioms are equivalent.

A unitary QFT in Minkowski space studies $n$-point correlators
\begin{equation}
  \langle \varphi_1 (x_1) \ldots \varphi_n (x_n) \rangle, \label{nptMink}
\end{equation}
(Wightman functions) of local operators $\varphi_i (x)$, $x \in
\mathbb{R}^{1, d - 1}$.\footnote{{In this paper, we use the term ``operator" according to CFT terminology, where the term ``field" is used in Wightman and Osterwalder-Schrader terminology. For example, in Refs.\ \cite{Streater:1989vi,simon1974}, the term ``operator" is reserved for the actual unbounded Hilbert space operator $\phi(f)$ defined on a dense domain of the Hilbert space (which justifies the terminology ``operator-valued distribution" $\phi$ in the G{\aa}rding-Wightman formulation).}} For simplicity in this paper we will only consider
bosonic operators, although more generally one should allow fermionic
operators and spinor representations. Wightman functions are translation and
$\tmop{SO} (1, d - 1)$ invariant.\footnote{Here and everywhere in this paper, $\tmop{SO}(1,d-1)$ means the connected component of the group. E.g. ${\rm diag}(-1,-1,1,\dots,1)$ is not included. Of course, it is possible for the theory to also possess discrete space-time symmetries such as time inversion or space reflection. The consequences of such symmetries are straightforward to figure out and are not discussed in this paper.} We will choose a basis of local operators
$\mathcal{O}_i$ transforming in irreducible $\tmop{SO} (1, d - 1)$
representations $\rho_i$. Then, Wightman functions remain invariant when
\begin{equation}
  \mathcal{O}_i^{\alpha} (x) \rightarrow {\rho_i (g)^{\alpha}_{\ \beta}}\,
  \mathcal{O}_i^{\beta} (g^{- 1} x), \label{transform}
\end{equation}
where $g \in \tmop{SO} (1, d - 1)$, and $\rho_i (g)$ are finite-dimensional
matrices of the representation $\rho_i$ $(\alpha, \beta = 1 \ldots \dim
\rho_i)$. Let $\mathcal{C}$ be the complex vector space whose elements are
arbitrary components of $\mathcal{O}_i$'s, and their \tmtextit{finite} linear
combinations with constant complex coefficients. Operators $\varphi
\mathcal{}_i$ in {\eqref{nptMink}} can be arbitrary elements of $\mathcal{C}$,
and Wightman function {\eqref{nptMink}} is multi-linear in $\varphi
\mathcal{}_i$. Note that in this and the next section derivatives of local
operators (of any order) are counted as independent operators, while in the
CFT Sec.\ \ref{ECFTax} we will start making distinction between primaries
and their derivatives.

Wightman functions {\eqref{nptMink}} are required to be tempered
distributions, i.e.\ can be paired with Schwartz class test functions $f (x_1,
\ldots, x_n)$. For this reason they are sometimes referred to as ``Wightman
distributions''. Note that the test functions $f (x_1, \ldots, x_n)$ with
which Wightman functions are paired do not have to vanish at coincident points
(unlike for the Schwinger functions discussion in the next section). This
means that, in a distributional sense, Wightman functions have meaning for all
configurations, including coincident points and light-cone singularities.
Translation and Lorentz invariance of Wightman functions are also understood
not pointwise but in the sense of distributions (i.e.\ that the pairing should
remain invariant if the test function is transformed in the dual
way).\footnote{Although Wightman functions can be shown to be real-analytic at
some totally spacelike-separated configurations (Jost points), in general they
may be singular even away from light cones (in particular when there are
timelike separations).}

We will not consider here other Minkowski correlators, such as retarded,
advanced, or time ordered, which are obtained from Wightman functions
multiplying by theta-functions of time coordinate differences, and whose
distributional properties require a separate discussion.

Limiting to the bosonic case as we are, \tmtextbf{local commutativity} (also
called microcausality) morally says that operators commute at spacelike
separation. Wightman axioms impose this as a constraint on Wightman functions:
\begin{equation}
  \langle \varphi_1 (x_1) \ldots \mathcal{} \varphi_p (x_p) \mathcal{}
  \varphi_{p + 1} (x_{p + 1}) \ldots \varphi_n (x_n) \rangle = \langle
  \varphi_1 (x_1) \ldots \mathcal{} \varphi_{p + 1} (x_{p + 1}) \mathcal{}
  \varphi_p (x_p) \ldots \varphi_n (x_n) \rangle \label{Wightman:causality}
\end{equation}
whenever $x_p - x_{p + 1}$ is spacelike: $(x_p - x_{p + 1})^2 > 0$ (in this paper we use the $-+\ldots+$ convention). Since we
are talking about distributions, this constraint means that
{\eqref{Wightman:causality}} holds when paired with any test function whose
support is contained in $(x_p - x_{p + 1})^2 > 0$.

\tmtextbf{Clustering} says that correlators should factorize if two groups of
points are far separated in a spacelike direction:
\begin{equation}
  \langle \varphi_1 (x_1) \ldots \mathcal{} \varphi_p (x_p) \mathcal{}
  \varphi_{p + 1} (x_{p + 1} + \lambda a) \ldots \varphi_n (x_n + \lambda a)
  \rangle \rightarrow \langle \varphi_1 (x_1) \ldots \mathcal{} \varphi_p
  (x_p) \rangle \langle \varphi_{p + 1} (x_{p + 1}) \ldots \varphi_n (x_n)
  \rangle \label{Wightman:cluster}
\end{equation}
as $\lambda \rightarrow \infty$ for any spacelike vector $a$, limit
understood in the sense of distributions.

We next discuss the spectral condition. By translation invariance we can
write
\begin{equation}
  \langle \varphi_1 (x_1) \ldots \varphi_n (x_n) \rangle = W (\xi_1,
  \ldots, \xi_{n - 1}), \hspace{3em} \xi_k = x_k - x_{k + 1},
\end{equation}
where $W$ is a tempered distribution in one less variable. Consider its
Fourier transform:
\begin{eqnarray}
  \hat{W} (q_1, \ldots, q_{n - 1}) & = & \int W (\xi_1, \ldots,
  \xi_{n - 1}) e^{i \underset{k = 1}{\overset{n - 1}{\sum}} q_k \cdummy \xi_k}\,
  d \xi_1 \ldots d \xi_n, 
\end{eqnarray}
where $q_k = (E_k, \mathbf{q}_k)$, $\xi_k = (t_k, \mathbf{\xi}_k)$, $q_k
\cdummy \xi_k = - E_k t_k +\mathbf{q}_k \cdummy \mathbf{\xi}_k$. Since $W$
is a tempered distribution, the Fourier transform $\hat{W}$ is well
defined and is also a tempered distribution. The \tmtextbf{spectral condition}
then says that $\hat{W}$ must be supported in the product of closed
forward light cones, i.e.\ in the region
\begin{equation}
  E_k \geqslant \mathbf{q}_k, \hspace{3em} k = 1, 2, \ldots, n - 1.
  \label{Wightman:spectral}
\end{equation}
For the two remaining conditions we need to discuss conjugation. Physically,
each operator $\varphi$ should have a conjugate $\varphi^{\dagger}$. In the
discussed framework we cannot define $\varphi^{\dagger}$ as an adjoint of an
operator acting on a Hilbert space, since we do not have a Hilbert space.
Instead, we will simply assume that there is a rule which associates
$\varphi^{\dagger}$ to $\varphi$, and impose the expected relations at the
level of correlation functions (Eq.\ {\eqref{Hermiticity}} below). This rule,
conjugation map $\dagger : \mathcal{C} \rightarrow \mathcal{C}$, associates to
each independent component $\mathcal{O}_i^{\alpha}  (\alpha = 1 \ldots \dim
\rho_i)$ of the above-mentioned basis of $\mathcal{C}$ a conjugate operator
$(\mathcal{O}_i^{\alpha})^{\dagger}$. This map is required to be an
involution, i.e.\ $\dagger \dagger = 1$. Furthermore, it is extended to
the whole of $\mathcal{C}$ by anti-linearity, i.e.\ $(c_1 \varphi_1 + c_2
\varphi_2 \mathcal{})^{\dagger} = c_1^{\ast} \varphi_1^{\dagger} + c_2^{\ast}
\varphi_2^{\dagger}$.\footnote{The $\dagger$ operation is denoted by $*$ in \cite{Streater:1989vi}.}  

Let us group operators $(\mathcal{O}_i^{\alpha})^{\dagger}$ in a multiplet
which we denote by $\mathcal{O}_i^{\dagger}$, i.e.\ $(\mathcal{O}_i^{\dagger})^{\alpha} =(\mathcal{O}_i^{\alpha})^{\dagger}$.\quad We will see below that
$(\mathcal{O}_i^{\dagger})^{\alpha}$ transform under $g \in \tmop{SO} (1, d -
1)$ with matrices complex-conjugate to those of $\mathcal{O}_i^{\alpha}$:
\begin{equation}
  \mathcal{O}_i^{\alpha} \rightarrow \rho_i (g)^{\alpha}_{\ \beta} \,
  \mathcal{O}_i^{\beta} \quad \Rightarrow \quad
  (\mathcal{O}^{\dagger}_i)^{\alpha} \rightarrow \overline{\rho_i
  (g)^{\alpha}_{\ \beta}} \, (\mathcal{O}^{\dagger}_i)^{\beta} .
  \label{R*}
\end{equation}
In other words, $\mathcal{O}^{\dagger}_i$ transforms in the conjugate
representation $\overline{\rho_i}$.

Since we are considering only bosonic operators, the relevant representations
$\rho_i$ are tensors $T^{\mu_1 \ldots \mu_l}$, on which $g \in \tmop{SO} (1, d
- 1)$ act as:
\begin{equation}
  \begin{array}{lll}
    T^{\mu_1 \ldots \mu_l} & \rightarrow & (\rho_i (g) T)^{\mu_1 \ldots
    \mu_l} = {g^{\mu_1}_{\ \nu_1} \ldots g^{\mu_l}_{\ \nu_l}}
    T^{\nu_1 \ldots \nu_l} .
  \end{array} \label{rhoi}
\end{equation}
Depending on $\rho_i$, these tensors have some fixed rank and mixed symmetry
properties. In addition, in even $d$, for tensors with $d / 2$ antisymmetric
indices, (anti-)chirality\footnote{Chiral and anti-chiral representations are sometimes also called ``self-dual'' and ``anti-self-dual''. We use
	``chiral'' and ``anti-chiral'' to avoid the clash with ``dual representation'' in mathematician's sense.} constraints must be imposed. All tensor
representations of $\tmop{SO} (1, d - 1)$ are real (i.e.\ matrices $\rho_i
(g)^{\alpha}_{\ \beta}$ in {\eqref{R*}} can be chosen real), except for
(anti-)chiral representations in $d = 0 \tmop{mod} 4$ which are
complex-conjugate to each other. For operators in real representations we can
choose a basis such that $\mathcal{O}_i =\mathcal{O}^{\dagger}_i$.

After this intermezzo we are ready to formulate hermiticity and positivity
conditions. \tmtextbf{Hermiticity} says that complex conjugate correlators
equal correlators of conjugated operators in inverted order:
\begin{equation}
  \overline{\langle \varphi \nobracket_1 (x_1) \ldots \varphi_n (x_n)
  \rangle} = \langle \varphi_n^{\dagger} (x_n) \ldots \varphi_1^{\dagger}
  (x_1) \rangle . \label{Hermiticity}
\end{equation}
This would be true of course if $\varphi$'s were operators acting on a Hilbert
space, with $\varphi^{\dagger}$'s their adjoints. In the present framework
without Hilbert space it is imposed as an axiom. This axiom implies in
particular {\eqref{R*}}, i.e.\ that $\mathcal{O}^{\dagger}_i$ transforms in
the conjugate irrep $\overline{\rho_i}$.

The last Wightman axiom, \tmtextbf{positivity}, is most conveniently written
down using the language of states. One considers basic ket states $| \psi (f,
\varphi_1, \ldots, \varphi_n) \rangle$, associated with $n$ local operators
$\varphi_1, \ldots, \varphi_n \in V$ and a complex Schwartz test function of
$n$ variables $f$. One defines the inner product on basic ket states by
\begin{eqnarray}
  \langle \psi (g, \chi_1, \ldots, \chi_m) | \psi (f, \varphi_1, \ldots,
  \varphi_n) \rangle & \assign & \int d x\, d y\, \overline{g (x_1, \ldots,
  x_m)} f (y_1, \ldots, y_n) \\
  &  & \qquad \times \langle \chi^{\dagger}_m (x_m) \ldots
  \chi_1^{\dagger} (x_1) \varphi_1 (y_1) \ldots \varphi_n (y_n) \rangle .
  \nonumber
\end{eqnarray}
The vector space of ket states $\mathcal{H}_0$ consists of \tmtextit{finite}
linear combinations $| \Psi \rangle$ of basic ket states, with the inner
product extended to it by (anti)linearity. \tmtextbf{Positivity} then says
that the so defined inner product is positive semidefinite:
\begin{equation}
  \langle \Psi | \Psi \rangle \geqslant 0 \quad \forall | \Psi \rangle \in
  \mathcal{H}_0 . \label{Wightman:positivity}
\end{equation}
\begin{remark}
  \label{rem:states}A comment is in order concerning the meaning of these
  states. They may be seen as just a convenient notation, since Eq.
  {\eqref{Wightman:positivity}} can be rewritten without ever using the word
  ``state'' (see {\cite{Streater:1989vi}}, Eq.\ (3-35)). But they are more than
  that: \ the vector space of states $\mathcal{H}_0$ is ``almost'' the Hilbert
  space $\mathcal{H}$ of our QFT. The only difference between $\mathcal{H}_0$
  and $\mathcal{H}$ is that $\mathcal{H}_0$ is not complete and may contain
  some states of zero norm. However, since $\mathcal{H}_0$ has a positive
  semidefinite inner product, as expressed by Eq.
  {\eqref{Wightman:positivity}}, we can obtain from it a Hilbert space
  $\mathcal{H}$ via a standard procedure of completion and
  modding out by states of zero norm. This is the first step of the Wightman
  reconstruction theorem {\cite{Streater:1989vi}}, and the resulting Hilbert
  space $\mathcal{H}$ turns out to be (possibly a superselection sector of)
  \tmtextit{the} Hilbert space of the QFT, on which fields can then be
  realized as operator valued distributions.
\end{remark}

\begin{remark}
  \label{hermfrompos}Although we included hermiticity as a separate axiom
  because of its suggestive form, it can be derived from positivity,
  considering the states of the form $| \Psi \rangle = | \psi (f_0, 1) \rangle
  \nobracket + | \psi (f, \varphi_1, \ldots, \varphi_n) \rangle$ where $f_0
  \in \mathbb{C}$ and $1$ is the unit operator.
\end{remark}

\begin{remark}
  Another interesting positivity property of Wightman functions is called Rindler Reflection positivity,
  or Wedge Reflection positivity~\cite{Casini:2010bf}. A restricted version of this property (with wedge-ordered points)
  can be derived from Wightman axioms, while a stronger version (no wedge-ordering)
  follows from Tomita-Takesaki theory which relies on Haag-Kastler axioms~\cite{Casini:2010bf}. In CFT context this property has been discussed, e.g.,
  in~\cite{Hartman:2016lgu}. We will not discuss these properties in this paper. However, it would be interesting to
  see whether the stronger form of Rindler positivity (including distributional information) can be derived from CFT axioms without the appeal to Tomita-Takesaki theory (the weaker version following from our results on Wightman axioms and~\cite{Casini:2010bf}). {We believe this can be done, and it could be a nice exercise for someone wishing to master our techniques.} 
\end{remark}

\subsection{Osterwalder-Schrader axioms}\label{OSaxioms}

We next describe a version of the Osterwalder and Schrader axioms
{\cite{osterwalder1973,osterwalder1975}} of Euclidean unitary QFT (see the end
of the section about the relation to the original OS axioms). The setup is
similar to Wightman axioms with $\tmop{SO} (d)$ replacing $\tmop{SO} (1, d -
1)$. We consider a basis of local bosonic operators $\mathcal{O}_i$
transforming\footnote{In the sense of Eq.\ {\eqref{transform}} where now $g \in
\tmop{SO} (d)$.} in $\tmop{SO} (d)$ irreps $\rho_i$, counting derivatives as
independent operators. Finite linear combinations of their components span a
complex vector space $\mathcal{C}$ of local operators. The axioms specify
properties of translation and $\tmop{SO} (d)$ invariant $n$-point correlators {(often called Schwinger functions)}
\begin{equation}
  \langle \varphi_1 (x_1) \ldots \varphi_n (x_n) \rangle, \quad \varphi_i \in
  \mathcal{C}, \quad x_i \in \mathbb{R}^d \label{nptEucl} .
\end{equation}
These correlators are defined away from coincident points (i.e.\ whenever $x_i
\neq x_j$ for each $i, j$). We will assume
that\footnote{\label{real-anal}Recall that a $C^{\infty}$ function of $m$ real
variables is called real-analytic in a domain $D \subset \mathbb{R}^m$ if it
has a convergent Taylor series expansion in a small ball around every point of
this domain. Equivalently, such a function has an analytic extension to a
small open neighborhood of this domain inside $\mathbb{C}^m$.}
\begin{equation}
  \text{correlators are real-analytic,} \label{real-anal1}
\end{equation}
and grow not faster than some power when some points approach each other or go
to infinity, i.e.
\begin{equation}
  | \langle \varphi_1 (x_1) \ldots \varphi_n (x_n) \rangle | \leqslant C
  \left( 1 + \max_{i \neq j} \left( \frac{1}{| x_i - x_j |}, | x_i - x_j |
  \right) \right)^p \label{OSmod}
\end{equation}
with some correlator-dependent positive constants $C, p$. Unlike Wightman
axioms, OS axioms do not bother what happens \tmtextit{precisely} at
coincident points (not even in the sense of distributions).

As we are limiting to the bosonic case, correlators remain invariant when
operators are permuted:\footnote{In particular one can sort all operators so
that the Euclidean time coordinates are ordered $x^0_1 \geqslant x_2^0
\geqslant \cdots \geqslant x_n^0$, and Euclidean correlator for any other
ordering can be obtained by trivially permuting field labels.}
\begin{equation}
  \langle \varphi_1 (x_1) \ldots \varphi_n (x_n) \rangle = \langle
  \varphi_{\pi (1)} (x_{\pi (1)}) \ldots \varphi_{\pi (n)} (x_{\pi (n)})
  \rangle . \label{perminv}
\end{equation}
To formulate the Euclidean version of hermiticity and positivity, we will need
some simple facts about $\tmop{SO} (d)$ representations. Abstractly, for any
irrep $\rho$ acting $T^{\alpha} \rightarrow \rho (g)^{\alpha}_{\ \beta}
T^{\beta}$, the conjugate representation $\bar{\rho}$ acts with complex
conjugate matrices $\overline{\rho (g)^{\alpha}_{\ \beta}}$. Since
$\tmop{SO} (d)$ is compact, we have $\bar{\rho} \simeq \rho^{\ast}$, the dual
representation. The $\tmop{SO} (d)$ irreps $\rho$ are again tensors
$T^{\mu_1 \ldots \mu_l}$ like in {\eqref{rhoi}}, of in general mixed symmetry,
and with (anti)-chirality constraints if having $d / 2$ antisymmetric
indices in even $d$. All of them are real, except for chiral
representations in $d = 2 \tmop{mod} 4$ which are complex-conjugate to the
anti-chiral ones.\footnote{This well-known shift from $d = 0 \tmop{mod} 4$
for $\tmop{SO} (1, d - 1)$ is induced by raising the indices of the
$\varepsilon$-tensor. E.g.\ $\varepsilon^{01} \varepsilon_{10} = - 1$ for
$\tmop{SO} (2)$, while it is 1 for $\tmop{SO} (1, 1)$.}

We will also need the reflected representation $\rho^R$ with matrices $\rho^R
(g) = \rho (g^R)$, where $g \rightarrow g^R = \Theta g \Theta,$ $\Theta =
\tmop{diag} (- 1, 1, \ldots, 1),$ is an automorphism of $\tmop{SO} (d)$. For
tensor representations, we can consider the map
\begin{equation}
  T^{\mu_1 \ldots \mu_l} \rightarrow {\Theta^{\mu_1}_{\ \nu_1} \ldots
  	\Theta^{\mu_l}_{\ \nu_l}} T^{\nu_1 \ldots \nu_l},
  \label{intertwiner:theta}
\end{equation}
which preserves rank and mixed symmetry properties. It also maps chiral to
anti-chiral tensors in any even $d$. Whenever the representation space is
preserved, this map serves as an intertwiner between $\rho^R$ and $\rho$. This
means that $\rho^R \simeq \rho$ for all tensor representations without
chirality constraints, while this operation interchanges chiral and
antichiral irreps in any even $d$.\footnote{In odd $d$, $\Theta$ is a
product of $- 1$ and an $\tmop{SO} (d)$ matrix, so that $g \rightarrow g^R$ is
an inner automorphism. This provides another argument why $\rho^R \simeq \rho$
for all irreducible $\tmop{SO} (d)$ representations in odd $d$.}

Applying both conjugation and reflection we get the conjugate reflected
representation $\bar{\rho}^R$ (isomorphic to dual reflected). From the above
it follows that $\bar{\rho}^R \simeq \rho$ for all $\tmop{SO} (d)$ irrreps,
except for (anti-)chiral tensors in $d = 0 \tmop{mod} 4$ which are
interchanged.

Just as for Wightman axioms, we will need a conjugation operation $\dagger :
\mathcal{C} \rightarrow \mathcal{C}$ on the vector space of local operators,
which is involutive, anti-linear, and associates to each independent component
$\mathcal{O}_i^{\alpha}  (\alpha = 1 \ldots \dim \rho_i)$ a conjugate operator
$(\mathcal{O}_i^{\dagger})^{\alpha} \assign
(\mathcal{O}_i^{\alpha})^{\dagger}$. Then the \tmtextbf{hermiticity} axiom
takes the form\footnote{Although we write the operators in the r.h.s.\ in the
inverted order like in {\eqref{Hermiticity}}, permutation invariance renders
this detail unimportant for the OS axioms.}
\begin{equation}
  \overline{\langle \varphi \nobracket_1 (x_1) \ldots \varphi_n (x_n)
  \rangle} = \langle \varphi_n^{\dagger} (x^{\theta}_n) \ldots
  \varphi_1^{\dagger} (x^{\theta}_1) \rangle, \label{HermiticityOS}
\end{equation}
similar to the Minkowski counterpart {\eqref{Hermiticity}} but with an
important difference that the operators in the r.h.s.\ are put at reflected
positions
\begin{equation}
  x^{\theta} \assign \Theta x.
\end{equation}
This change has a consequence that $\mathcal{O}^{\dagger}_i$ transforms in the
conjugate \tmtextit{reflected} representation $\overline{\rho_i}^R$,
explaining why we introduced this concept in the first place.\footnote{Indeed
we have ${\langle (\mathcal{O}^{\dagger}_i)^{\alpha} (x) \ldots \rangle =
	\overline{\langle \mathcal{O}^{\alpha}_i (x^{\theta}) \ldots \rangle}
	=\overline{\rho (g)^{\alpha}_{\ \beta}} \langle
	\mathcal{O}^{\beta}_i (g^{- 1} x^{\theta}) \ldots \rangle = \overline{\rho
		(g)^{\alpha}_{\ \beta}} \langle (\mathcal{O}^{\dagger}_i)^{\beta}
	((g^R)^{- 1} x) \ldots \rangle}$.} For self-conjugate-reflected
representations we may choose a basis such that
\begin{equation}
  (\mathcal{O}^{\dagger}_i)^{(\mu)} = \Theta^{(\mu)}_{(\nu)}
  \mathcal{O}_i^{(\nu)}, \label{realOS}
\end{equation}
where $\Theta^{(\mu)}_{(\nu)} : = {\Theta^{\mu_1}_{\ \nu_1} \ldots
	\Theta^{\mu_l}_{\ \nu_l}}$ is the intertwiner {\eqref{intertwiner:theta}}.

To write positivity, basic ket states $| \psi (f, \varphi_1, \ldots,
\varphi_n) \rangle$ are associated with $n$ local operators $\varphi_1,
\ldots, \varphi_n \in \mathcal{C}$ and a complex compactly supported
Schwartz test function of $n$ variables $f (x_1, \ldots, x_n)$ which
vanishes unless all points are in the lower half space and have time variables
ordered: $0 > x_1^0 > x_2^0 > \cdots > x_n^0$. These support requirements
were absent in the Wightman case. The inner product on the basic ket states is
defined by
\begin{eqnarray}
  \langle \psi (g, \chi_1, \ldots, \chi_m) | \psi (f, \varphi_1, \ldots,
  \varphi_n) \rangle & \assign & \int d x\, d y\, \overline{g (y_1^{\theta},
  \ldots, y_m^{\theta})} f (x_1, \ldots, x_n) \nonumber\\
  &  & \times \langle \chi^{\dagger}_m (y_m) \ldots \chi_1^{\dagger} (y_1)
  \varphi_1 (x_1) \ldots \varphi_n (x_n) \rangle,  \label{inner:OS}
\end{eqnarray}
and is extended by (anti)linearity to the vector space
$\mathcal{H}_0^{\tmop{OS}}$ of \tmtextit{finite} linear combinations $| \Psi
\rangle$ of basic ket states. In this notation, \tmtextbf{positivity} takes
the same form as {\eqref{Wightman:positivity}}, i.e.\ that the so defined inner
product must be positive semidefinite:
\begin{equation}
  \langle \Psi | \Psi \rangle \geqslant 0 \quad \forall | \Psi \rangle \in
  \mathcal{H}_0^{\tmop{OS}} . \label{OS:positivity}
\end{equation}
This is referred to as ``OS reflection positivity'' because of the reflected
$g$ arguments in {\eqref{inner:OS}}, differently from the Wightman case.
Because of this reflection and the above test function support requirements,
all operators in {\eqref{inner:OS}} sit at separated positions. This is one
reason why the OS axioms involve ordinary functions, without worrying about
coincident points. In contrast, Wightman positivity integrates operator
insertions over coincident points and makes sense only for distributions.

Just as in the Wightman case (Remark \ref{rem:states}), we can complete the
vector space $\mathcal{H}_0^{\tmop{OS}}$, mod out by states of zero norm, and
obtain a Hilbert space $\mathcal{H}^{\tmop{OS}}$ of the Euclidean theory.

Although we included hermiticity as an independing axiom, it can be derived
from positivity, just as in Remark \ref{hermfrompos} in the Wightman case.

One simple consequence of OS reflection positivity is pointwise positivity of
$2 n$-point functions at reflection invariant configurations of points:
\begin{equation}
  \langle \varphi^{\dagger}_n (x^{\theta}_n) \ldots \varphi_1^{\dagger}
  (x^{\theta}_1) \varphi_1 (x_1) \ldots \varphi_n (x_n) \rangle \geqslant 0
  \label{OSnaive1}
\end{equation}
for any $x_1, \ldots, x_n$ in the lower half space.\footnote{For tensor
operator in self-conjugate-reflected representations, choosing the real basis
{\eqref{realOS}}, this becomes $\langle \ldots \Theta^{(\mu)}_{(\nu)}
\mathcal{O}^{(\nu)} (x^{\theta}) \mathcal{O}^{(\mu)} (x) \ldots
\rangle \geqslant 0$ (no sum on $\mu$), i.e.\ tensor indices are also
reflected.} This follows from {\eqref{OS:positivity}} by taking $| \Psi
\rangle = | \psi (f, \varphi_1, \ldots, \varphi_n) \rangle$ and localizing $f$
near one configuration of points. In general, imposing {\eqref{OSnaive1}} for
all $\varphi$'s and $x$'s would be clearly weaker than full OS reflection
positivity. E.g.\ {\eqref{OS:positivity}}, but not {\eqref{OSnaive1}}, can be
used to bound 3-point functions in terms of 2- and 4-point functions, or
non-reflection-invariant 4-point functions by reflection-invariant ones.
However for CFTs we will see below that OS reflection positivity can be
reduced to a form of {\eqref{OSnaive1}} for 2-point functions plus a form of
{\eqref{HermiticityOS}} for 3-point functions.

Finally, the OS \tmtextbf{clustering} asserts that
\begin{eqnarray}
  & \lim_{\lambda \rightarrow \infty} \int d x\, d y\, \overline{g
  (y_1^{\theta}, \ldots, y_m^{\theta})} f (x_1, \ldots, x_n) \langle
  \chi^{\dagger}_m (y_m) \ldots \chi_1^{\dagger} (y_1) \varphi_1 (x_1 +
  \lambda a) \ldots \varphi_n (x_n + \lambda a) \rangle &  \nonumber\\
  & = \int d x\, d y\, \overline{g (y_1^{\theta}, \ldots, y_m^{\theta})} f
  (x_1, \ldots, x_n) \langle \chi^{\dagger}_m (y_m) \ldots \chi_1^{\dagger}
  \nobracket (y_1) \rangle \langle \nobracket \varphi_1 (x_1) \ldots \varphi_n
  (x_n) \rangle &  \label{clusterint}
\end{eqnarray}
for any Schwartz test functions $f (x_1, \ldots, x_n)$ and $g (y_1, \ldots,
y_m)$ supported for $0 > x_1^0 > x_2^0 > \ldots > x_n^0$ and $0 > y_1^0 >
y_2^0 > \ldots > y_n^0$, for any local fields $\varphi_1, \ldots, \varphi_n$
and $\chi_1, \ldots, \chi_m$, and for any $a \in \mathbb{R}^d$ which is
parallel to the $x^0$ plane ($a^0 = 0$, called ``purely spatial" elsewhere). The latter requirement is somewhat
analogous to having the Wightman cluster property {\eqref{Wightman:cluster}}
to be satisfied only for spacelike $a$.\footnote{This is axiom E4 in
{\cite{osterwalder1973}}. Ref.\ {\cite{osterwalder1973}} also mentions a
stronger axiom E4', but we will be content here with checking the easier axiom
E4.}

Note that the Minkowski operators can be mapped to Euclidean operators. In
particular any $\tmop{SO} (1, d - 1)$ irrep can be mapped to an $\tmop{SO}
(d)$ irrep. This map of irreps originates from the map between the two Lie
algebras which have the same complexification. It can then be shown that a
pair of conjugate $\tmop{SO} (1, d - 1)$ irreps is mapped to a pair of
$\tmop{SO} (d)$ irreps which are conjugate-reflected to each other. This gives
another rationale for the appearance of reflected irreps in the OS axioms.

\begin{remark}
  \label{OSnewVSold}The stated version of OS axioms includes the assumption of
  real analyticity {\eqref{real-anal1}} and the bound {\eqref{OSmod}}. These
  assumptions are natural from physics perspective; they also easily follow
  from Wightman axioms. The original OS axioms did not include
  {\eqref{real-anal1}} nor {\eqref{OSmod}}, but included instead a differently
  stated assumption:
  \begin{equation}
    \text{correlators are distributions on $^0 \mathcal{S}$,} \label{OSorig}
  \end{equation}
  where $^0 \mathcal{S}$ is the space of Schwartz test function vanishing at
  coincident points with all their \ derivatives.
  
  We would like to discuss here the relation between
  {\eqref{real-anal1}}+{\eqref{OSmod}} and {\eqref{OSorig}}. In one direction
  this is easy: clearly {\eqref{OSmod}} implies {\eqref{OSorig}}. In the other
  direction it can be shown that {\eqref{OSorig}} and other OS axioms (in
  particular OS positivity and rotation invariance) imply real analyticity
  {\eqref{real-anal1}}. This is a result of \
  {\cite{osterwalder1973,osterwalder1975}} and {\cite{Glaser1974}}. It is also
  possible to derive {\eqref{OSmod}} from {\eqref{OSorig}} and other OS axioms
  {\cite{osterwalder1975}}. These issues will be reviewed further in Sec.\ \ref{OS}.
\end{remark}

\begin{remark}
	{In Ref.\ \cite{osterwalder1975}, Osterwalder and Schrader introduced an extra assumption on Euclidean correlators, the \textit{linear growth condition}. We will come back to this in Sec.\ \ref{OS}. Here we would like to stress that this condition is not needed for the implications discussed in the previous paragraph. On the other hand, the linear growth condition was used in Ref.\ \cite{osterwalder1975} when showing that the Euclidean correlators can be Wick-rotated to Lorentzian signature, resulting in Wightman distributions.}
\end{remark}

\subsection{Euclidean CFT axioms}\label{ECFTax}

Wightman and OS axioms stated in the previous two sections are standard. We
took care to present them for general operator representations and in general
$d$. We will now present axioms for Euclidean unitary CFT. Just as OS axioms,
these concern correlators in Euclidean signature, but there is an extra
assumption of conformal invariance. Another feature of the CFT axioms is that
assumptions are imposed on simple building blocks (2- and 3-point functions)
from which more complicated correlators can be constructed. Properties of
these complicated correlators then follow. The point of our paper is how one
can recover OS axioms and (after Wick rotation) Wightman axioms in this setup.

A Euclidean unitary CFT in $\mathbb{R}^d$ ($d \geqslant 2$) deals with local
\tmtextit{primary} operators $\mathcal{O}_i (x)$ and with their $n$-point
correlation functions $\langle \mathcal{O}_{i_1} (x_1) \ldots
\mathcal{O}_{i_n} (x_n) \rangle$. Correlators are real-analytic functions
defined away from coincident points, which are permutation-invariant as in
{\eqref{perminv}}. Each primary is characterized by its scaling dimension
$\Delta_i$ and is an $\tmop{SO} (d)$ tensor transforming in an irreducible
representation $\rho_i$.\footnote{Operators can also be grouped into
multiplets of the global symmetry group $G$ which a CFT might have, but we
will not discuss global symmetry here. For simplicity we will only consider
bosonic operators. More generally one should allow fermionic operators and
spinor representations.} The scaling dimensions are real and nonnegative, with
the unit operator having dimension zero. The set of scaling dimensions
(``spectrum'') is assumed to be discrete, by which we mean that there are finitely
many $\Delta_i$'s in any finite interval $[a, b] \subset
\mathbb{R}$.\footnote{There exist 2d unitary CFTs, such as the Liouville
theory, with a continuous spectrum of scaling dimensions. In this case axioms
need to be modified. All known unitary CFTs in $d \geqslant 3$ have a discrete
spectrum.}

The set of all local operators of a CFT consists of primaries $\mathcal{O}_i
(x)$ and their space-time derivatives $\partial_{\mu_1} \ldots
\partial_{\mu_n} \mathcal{O}_i (x)$, often referred to as descendants. The
correlation functions of the descendant operators are simply the derivatives
of the correlation functions of primary operators. They are well-defined since
the correlators of primaries are assumed to be real-analytic.

Parameters $\Delta_i$ and $\rho_i$ determine transformation properties of
$\mathcal{O}_i (x)$ under the conformal group $\tmop{SO} (d + 1, 1)$, and
correlators remain invariant under these standard transformations which we
will not write down. These constraints determine the functional form of
1,2,3-point functions. In particular, the unit operator is the only one with a
nonzero 1-point function. {See, e.g.,~\cite{Poland:2018epd} for a review of these
facts.} 

{
An important fact that follows from the conformal invariance of correlation functions
is that one is allowed to insert an operator at spatial infinity. The primary operator at spatial infinity is defined as
\be\label{eq:inftydefn}
	\<\cO_i(\oo)\cdots\>\equiv\lim_{L\to +\oo}L^{2\De_i}\<\cO_i(L\hat e_0)\cdots\>.
\ee
To see that this limit exists one can use a conformal map that takes
$\oo$ to a finite point and moves no other operators to infinity. After applying this 
map the limit~\eqref{eq:inftydefn} turns into a limit in 
which all points approach finite values. We conclude that~\eqref{eq:inftydefn}
exists, and is then of course independent of the concrete conformal map that we chose.
In the definition~\eqref{eq:inftydefn} we have chosen a particular direction ($\hat e_0$) for
the limit. Using conformal symmetry it is easy to show that~\eqref{eq:inftydefn} is independent
of this direction, up to a rotation on the indices of $\cO_i$. In what follows we will
always allow for Euclidean CFT correlators to have one of the operators to be at $\oo$.
}

{A Euclidean unitary CFT comes equipped with a conjugation operation $\dagger$, an involutive antilinear operator on the vector space of local operators (including both primaries and descendants), similarly to the OS axioms. Every primary $\mathcal{O}_i$ is mapped by $\dagger$ to a conjugate primary $\mathcal{O}^{\dagger}_i$ such that the 2-point function $\langle \mathcal{O}_i^{\dagger} \mathcal{O}_i \rangle$ does not vanish. }The $\mathcal{O}_i$ and $\mathcal{O}^{\dagger}_i$ have equal scaling dimensions, and transform in the conjugate-reflected irreps. Recall that in Sec.\ \ref{OSaxioms} we saw that most $\tmop{SO} (d)$ irreps are self-conjugate-reflected, $\rho_i \simeq \bar{\rho}_i^R$, the only exception being (anti-)chiral tensors in $d = 0 \tmop{mod} 4$ which are exchanged by this operation. For operators in self-conjugate-reflected irreps we may choose operator basis such that Eq.\ {\eqref{realOS}} holds, which we copy here:
\begin{equation}
  (\mathcal{O}^{\dagger}_i)^{(\mu)} = \Theta^{(\mu)}_{(\nu)}
  \mathcal{O}_i^{(\nu)} . \label{realOS1}
\end{equation}
The functional form of the $\langle \mathcal{O}_i^{\dagger} (x) \mathcal{O}_i
(y) \rangle$ 2-point function is fixed by conformal symmetry:
\begin{equation}
  \langle (\mathcal{O}^{\dagger}_i)^{(\mu)} (x) \mathcal{O}_i^{(\nu)} (y)
  \rangle = \cN_i I^{(\mu), (\nu)} (x - y), \label{2-pointfixed}
\end{equation}
where $(\mu), (\nu)$ are collections of tensor indices (of equal length),
$I^{(\mu), (\nu)} (x - y)$ is a tensor function depending only on $\Delta_i,
\rho_i$, {and $\cN_i$ is a constant. 
	
	{ We are free to rescale the basis of primaries, multiplying them by some positive constants which should be the same for both $\mathcal{O}_i$ and $\mathcal{O}_i^\dagger$. This transformation clearly preserves the above conditions on $\dagger$, as well as further conditions which will be discussed below, notably \eqref{CFThermiticity}. Using this freedom we can rescale $\cN_i$ by a positive real number and fix $|\cN_i|$ in some arbitrary unimportant way, e.g.\ so that some component of the 2-point function
	is of absolute value one at unit separation. On the other hand the phase of $\cN_i$ cannot be changed in this way. Instead, it is uniquely determined by the positivity condition discussed below.}

Positivity is imposed in Euclidean CFT axioms only on 2-point functions. We
write it again using the language of states. Basic ket states are $|
\partial^{(\beta)} \mathcal{O}_i^{(\nu)} \rangle \nobracket$ where
$\mathcal{O}_i^{(\nu)}$ is a primary component and $\partial^{(\beta)}$ an
arbitrary derivative. The inner product is defined as
\begin{equation}
  \langle \partial^{(\alpha)} \mathcal{O}_i^{(\mu)} | \partial^{(\beta)}
  \mathcal{O}_i^{(\nu)} \rangle = \Theta^{(\alpha)}_{(\alpha')} \langle
  \partial^{(\alpha')} (\mathcal{O}^{\dagger}_i) \nobracket^{(\mu)} (x_N)
  \partial^{(\beta)} \mathcal{O}_i^{(\nu)} (x_S) \rangle, \label{innerCFT}
\end{equation}
i.e.\ as the value of the shown 2-point function inserting the operators at
$x_S = (- 1, 0, \ldots, 0)$ and $x_N = (1, 0, \ldots, 0) = (x_S)^{\theta}$ \
(where N,S stands for north, south). For ket states with $i \neq j$ the inner
product vanishes since the 2-point function is zero. This inner product is
extended by (anti)linearity to the vector space $\mathcal{H}_0^{\tmop{CFT}}$
of finite complex linear combinations of basic ket states. In this language,
Euclidean CFT \tmtextbf{positivity} reads exactly as the Wightman and OS
positivity: $\langle \Psi | \Psi \rangle \geqslant 0$ for all states of this
restricted form. More prosaically, this can also be stated that the infinite
matrices $M^{(\alpha) (\mu), (\beta) (\nu)}_i$ built out of 2-point functions
in the r.h.s.\ of {\eqref{innerCFT}} are all positive semidefinite when
restricted to finite subspaces. 

{CFT positivity can be analyzed primary by primary, and it depends only on
the primary 2-point function, Eq.\ {\eqref{2-pointfixed}} which determines the full matrix $M^{(\alpha) (\mu), (\beta)
(\nu)}_i$. Clearly, only one phase of the normalization constant $\cN_i$ in Eq.\ {\eqref{2-pointfixed}} can give rise to a positive definite matrix, so that phase is uniquely fixed. Once the phase of $\cN_i$ is fixed, positivity for a given primary depends only on its $\Delta$, $\rho$.} It then can be shown that CFT positivity holds if and only if every $\Delta$, in addition to being real and non-negative, lies above a certain
minimal $\rho$-dependent value (``unitarity bound''):
\begin{equation}
  \Delta \geqslant \Delta_{\min} (\rho) .
\end{equation}
These unitarity bounds are documented in the literature, e.g.\ we have
$\Delta_{\min} = d / 2 - 1$ for scalars, and $d + \ell - 2$ for spin-$\ell$,
$\ell \geqslant 1$. For arbitrary $\tmop{SO} (d)$ representations see
{\cite{Minwalla:1997ka}}.\footnote{We chose to express CFT positivity
inserting operators at the points $(\pm 1, 0, \ldots, 0)$ which corresponds to
the N-S quantization (see {\cite{EPFL}}) and will facilitate the comparison
with the Osterwalder-Schrader reflection positivity. Equivalently, one could
go via a conformal transformation to the more familiar radial quantization corresponding to inserting the
operators at $0$ and $\infty$. CFT positivity is then equivalent to radial
quantization states having positive norm on every level, which is how the
unitarity bounds are usually worked out in Euclidean CFTs
{\cite{Minwalla:1997ka}}. In mathematical language, this latter condition
corresponds to having a positive-definite Shapovalov form on the parabolic
Verma module. Recent work {\cite{Yamazaki:2016vqi,Penedones:2015aga}}
explained how the determinant formulas by Jantzen
{\cite{jantzen_kontravariante_1977}} provide a rigorous justification of the
Euclidean unitarity bounds (both in the necessary and sufficient directions).}

For future uses, we wish to define the CFT Hilbert space
$\mathcal{H}^{\tmop{CFT}}$ via completion of $\mathcal{H}_0^{\tmop{CFT}}$,
after modding out by zero norm states (for operators saturating the unitarity
bounds, some descendants have zero norm). This can be done abstractly, or
explicitly using a basis as we now describe. Throwing out zero-norm
descendants, the remaining states can be organized choosing an orthonormal
basis. We may choose such a basis independently among descendants of each
primary, and then combine all these bases, e.g.\ in the order of non-decreasing
scaling dimensions. \ The elements of $\mathcal{H}^{\tmop{CFT}}$ are then
formal linear combinations $\sum_n c_n | n \rangle,$where $| n \rangle$ are
orthonormal basis elements, and $c_n$ is an arbitrary complex $\ell_2$
sequence. The norm on $\mathcal{H}^{\tmop{CFT}}$ is the $\ell_2$ norm of
the sequence $c_n$. Restricting to sequences $c_n$ which have only a finite
number of nonzero elements, we get elements of $\mathcal{H}_0^{\tmop{CFT}}$
(modulo the zero-norm states).

Let us continue with the axioms. CFT \tmtextbf{hermiticity} condition is
imposed only on the 2-point and 3-point functions, namely:
\begin{equation}
  \langle (\mathcal{O}^{\dagger}_i)^{(\mu)} (x_1) \mathcal{O}_i^{(\nu)}
  (x_2) \rangle =  \overline{\langle \mathcal{O}_i^{(\mu)} (x_1^{\theta})
  (\mathcal{O}_i^{\dagger})^{(\nu)} (x_2^{\theta}) \rangle}, 
  \label{hermicity2-point}
\end{equation}
{which is also a consequence CFT
positivity and in particular fixed the phase of $\cN_i$ up to a sign},\footnote{For the special case $x_1 = x_N$, $x_2 = x_S$, Eq.\ {\eqref{hermicity2-point}} is nothing but hermiticity of the matrix $M_i^{(\mu),
(\nu)}$, a consequence of positive-semidefiniteness. The general case reduces
to the special one mapping $x_1, x_2$, to $x_N$, $x_S$ by a conformal
transformation (both sides of {\eqref{hermicity2-point}} have the same conformal
transformation properties).} and
\begin{eqnarray}
  \langle \mathcal{O}_i^{(\mu)} (x_1) \mathcal{O}_j^{(\nu)} (x_2)
  \mathcal{O}_k^{(\lambda)} (x_3) \rangle & = & \overline{\langle
  (\mathcal{O}_i^{\dagger})^{(\mu)} (x^{\theta}_1)
  (\mathcal{O}_j^{\dagger})^{(\nu)} (x^{\theta}_2)
  (\mathcal{O}^{\dagger}_k)^{(\lambda)} (x^{\theta}_3) \rangle} . 
  \label{CFThermiticity}
\end{eqnarray}
for any 3 primaries $\mathcal{O}_i$, $\mathcal{O}_j$,
$\mathcal{O}_k$.\footnote{Note that since $\dagger$ is involutive, this
covers the case when an operator and its conjugate are interchanged between
the two sides of this equation.} Similarly to the 2-point function case, this
condition can be simplified using the conformally invariant tensor structures,
with an important difference that the normalization of operators has already
been fixed. Conformal invariance constrains the 3-point functions to take the
form
\begin{equation}
  \langle \mathcal{O}_i^{(\mu)} (x_1) \mathcal{O}_j^{(\nu)} (x_2)
  \mathcal{O}_k^{(\lambda)} (x_3) \rangle = \sum_{a = 1}^{N_{i j k}} f^a_{i j
  k}  \langle \mathcal{O}_i^{(\mu)} (x_1) \mathcal{O}_j^{(\nu)} (x_2)
  \mathcal{O}_k^{(\lambda)} (x_3) \rangle_a, \label{3-pointGeneral}
\end{equation}
where $\langle \mathcal{O}_i^{(\mu)} (x_1) \mathcal{O}_j^{(\nu)} (x_2)
\mathcal{O}_k^{(\lambda)} (x_3) \rangle_a$ span the finite-dimensional space
(of dimension $N_{i j k}$) of solutions of conformal invariance constraints on
the 3-point functions of the operators with given $\Delta_s, \rho_s$ ($s = i,
j, k$). On the other hand the coefficients $f_{i j k}^a \in \mathbb{C}$ are
not fixed by conformal symmetry (no sum on $i, j, k$ in the r.h.s.\ of
{\eqref{3-pointGeneral}}). We often refer to $f^a_{i j k}$ as the ``OPE
coefficients.'' It is always possible to choose the basis structures $\langle
\mathcal{O}_i^{(\mu)} (x_1) \mathcal{O}_j^{(\nu)} (x_2)
\mathcal{O}_k^{(\lambda)} (x_3) \rangle_a$ to satisfy the hermiticity
constraint {\eqref{CFThermiticity}} individually, in which case the OPE
coefficients must satisfy
\begin{equation}
  (f^a_{i j k})^{\ast} = f_{\bar{i}  \bar{j}  \bar{k} }^a,
  \label{OPEcoeff:conjugation}
\end{equation}
where the barred indices refer to the conjugate operators
$\mathcal{O}_i^{\dagger}, \mathcal{O}_j^{\dagger}, \mathcal{O}^{\dagger}_k .$
In particular, when all three operators are self-conjugate-reflected i.e.\ satisfy {\eqref{realOS1}}, the OPE coefficients $f^a_{i j k}$ must be real.

Finally, unitary CFTs enjoy a \tmtextbf{convergent operator product
expansion} (OPE). This means that any correlation function\footnote{Importantly, we allow here for one of the operators to
be inserted at spatial infinity.} satisfies
\begin{equation}
  \langle \mathcal{O}_i^{(\mu)} (x_1) \mathcal{O}_j^{(\nu)} (x_2)
  \mathcal{O}_m^{(\rho)} (x_3) \cdots \rangle = \sum_k \sum_{a = 1}^{N_{i j
  k}} f_{i j k}^a C_{a, (\lambda)}^{(\mu) (\nu)} (x_1, x_2, x_0, \partial_0)
  \langle (\mathcal{O}^{\dagger}_k)^{(\lambda)} (x_0)
  \mathcal{O}_m^{(\rho)} (x_3) \cdots \rangle, \label{OPEgeneral}
\end{equation}
where the first sum runs over all primary operators $\mathcal{O}_k$ in the
theory, and $C_{a, (\lambda)}^{(\mu) (\nu)} (x_1, x_2, x_{0,} \partial_0)$ is
a formal sum of the form
\begin{equation}
  C_{a, (\lambda)}^{(\mu) (\nu)} (x_1, x_2, x_0, \partial_0) =
  \sum_{\alpha} C_{a, (\lambda), \alpha}^{(\mu) (\nu)} (x_1, x_2, x_0)
  (\partial / \partial x_0)^{\alpha} . \label{Cexp}
\end{equation}
This differential operator is determined by conformal symmetry\footnote{There is an ambiguity when $\cO_k$ is in a short conformal
representation (in unitary theories this happens only if $\cO_k$ is a conserved current or a free field). This subtlety will not play
any role in this paper.}
and depends only on
\ $\Delta_s, \rho_s$ ($s = i, j, k$). Here $(\partial / \partial
x)^{\alpha} = (\partial / \partial x^0)^{\alpha_0} \ldots (\partial /
\partial x^{d - 1})^{\alpha_{d - 1}}$ with $\alpha = (\alpha_0, \ldots,
\alpha_{d - 1}) \in (\mathbb{Z}_{\geqslant 0})^d$ a multiindex. Convergence of
OPE means that the sum {\eqref{OPEgeneral}}, with $C$ expanded as in
{\eqref{Cexp}}, converges whenever \
\begin{equation}
  x_1, x_2 \in B (x_0, R), \quad R = \min (| x_3 - x_0 |, | x_4 - x_0 |,
  \ldots) \label{ball}
\end{equation}
where $B (x_0, R)$ is an open ball centered at $x_0$ and of radius $R$. In
other words, OPE converges whenever $x_1$ and $x_2$ are the two closest
operator insertions to $x_0$ (in Euclidean distance). Convergence should be
understood carefully as follows. For each $\mathcal{O}^{\dagger}_k$ in the
r.h.s.\ of {\eqref{OPEgeneral}}, and for each $n \in \mathbb{Z}_+$, we perform
finite summation over $a, \lambda$, and all multiindices $\alpha$ with $|
\alpha | = n$. We are left then with the doubly infinite sum
\begin{equation}
  \sum_k \sum_{n = 0}^{\infty} g_{k, n} (\{ x_i \}) \label{doublyinf} .
\end{equation}
This doubly infinite sum has to converge absolutely for every $x_1, x_2$ as in
{\eqref{ball}}.\footnote{The requirement of absolute convergence can be
somewhat relaxed, see Sec.\ \ref{OSfromCFT}.}

That the same coefficients $f_{i j k}^a$ appear in the OPE
{\eqref{OPEgeneral}} and the 3-point function {\eqref{3-pointGeneral}} follows
immediately by using the OPE inside the latter 3-point function.

Local Euclidean CFTs contain the conserved stress tensor operator $T_{\mu
\nu}$ of dimension $d$, and in case of continuous global symmetry, conserved
global symmetry currents $J_{\mu}$ of dimensions $d - 1$. We will not discuss
here additional axioms involving 3-point functions and OPE coefficients of
these operators, related to their conservation and Ward identities, see e.g.\ {\cite{Dymarsky:2017xzb,Dymarsky:2017yzx}}.

\begin{remark}
  \label{ECFTax-comp}The just given Euclidean CFT axioms are more careful in
  what concerns reality constraints than the set of CFT rules gathered in the
  conformal bootstrap review {\cite{RMP}}. They are also more economical: e.g.\ Ref.\ {\cite{RMP}} assumed OS reflection positivity and clustering for
  $n$-point functions, which for us will be theorems to prove, not
  assumptions. Ref.\ {\cite{RMP}} also included some constraints on the CFT
  data which emerge when considering CFT in Minkowski signature, most notably
  the Averaged Null Energy Condition (ANEC). In this paper we will establish
  all Wightman axioms for scalar Minkowski CFT 4-point functions from the
  Euclidean CFT axioms, but we will not discuss ANEC. A proof of ANEC
  {\cite{Faulkner:2016mzt}} has been given using the Haag-Kastler axioms for
  general QFT. CFT arguments have also been given in
  {\cite{Hartman:2016lgu,Kravchuk:2018htv}}, but they rely on some assumptions
  which have not been rigorously proven from axioms. It would be interesting
  to fill these gaps and establish ANEC as a theorem from Euclidean CFT
  axioms.\footnote{\label{caveats} {The argument in
  {\cite{Hartman:2016lgu}} uses an OPE asymptotic expansion on the second
  sheet, outside of the range of convergence of the OPE rigorously implied by
  the Euclidean CFT axioms (see App.\ \ref{Tom} where we review this method
  going back to {\cite{Hartman:2015lfa}}).} In {\cite{Kravchuk:2018htv}} ANEC
  is derived using manipulations with a generalization of the Lorentzian
  inversion formula of {\cite{Caron-Huot:2017vep}}, of which some have not
  been rigorously justified. For example, the derivation starts with the
  Euclidean inversion formula, which is readily justified from Harmonic
  analysis only for external scaling dimensions on principal series and
  square-integrable correlators, none of which is generically the case in the
  required setup.}
\end{remark}

\section{Euclidean CFT $\Rightarrow$ Osterwalder-Schrader}\label{CFTtoOS}
In this section we will discuss some simple consequences of CFT axioms and in
particular will show that they imply OS axioms for 4-point functions (the case of
higher-point functions is more subtle and is discussed in App.\ \ref{OShigher}). Here we will prove only the OS reflection-positivity and the
cluster property. The ``Euclidean temperedness'' bound {\eqref{OSmod}} will
follow from our arguments in the following sections, where we establish
power-law bounds on CFT correlation functions. The remaining OS axioms are a
subset of the CFT axioms.

\subsection{OS reflection positivity}\label{OSfromCFT}

{In this section we prove OS positivity
for compactly-supported test functions. The extension to Schwartz functions is easy once 
we establish~\eqref{OSmod} for CFT correlators, see Remark \ref{lastOS}.}

To someone familiar with OPE in CFTs, the goal of this section may seem straightforward to
achieve: in usual CFT treatments, one can express any state as a sum of ``single operator states'',
created by action of a single insertion of primaries and their descendants on the vacuum, and our two-point
positivity axiom makes sure that the inner products of such states are positive-semidefinite.
The part of this wisdom that doesn't work immediately in our setup is being able to express
states in the OS Hilbert space as infinite OPE sums. Our OPE axiom is much weaker than this statement,
in particular the truncation order for the OPE needed to achieve a given precision $\epsilon$ in
\eqref{OPEgeneral} may a priori depend on everything in the left-hand side, including coordinates $x_i$ and
the ``spectator'' operators $\mathcal{O}_m,\ldots$, in an arbitrary fashion. (On the contrary, for a Hilbert space
OPE statement, the truncation order is good for all choices of spectator operators as long as an 
appropriate norm remains bounded by a fixed constant.)  In spite of this difficulty, in this and the next section we will be able to recover reflection positivity for $n$-point functions up to $n\le 4$ and Hilbert space OPE convergence for states created by up to $2$ operator insertions, but doing so requires some care. }

First let us slightly reformulate the OPE convergence property. Consider an
$n$-point correlation function with operators inserted at $x_1 \ldots x_n .$
Let $S$ be a hyperplane and $x_0$ be a point such that $x_1, x_2, x_0$ are
on one side of $S$ while all the other points $x_i$, $i>2$, are on the other side.
Using a conformal transformation we can map $S$ to a sphere $S'$ so that $x_0$
is mapped to the center of $S'$ which we denote by $x_0'$. Let $x_i'$ denote
the positions of all the other points $x_i$ after this map. We can then use
the OPE {\eqref{OPEgeneral}} for the correlation function evaluated at $x_i'$,
\begin{equation}
  \langle \mathcal{O}_i^{(\mu)} (x'_1) \mathcal{O}_j^{(\nu)} (x'_2)
  \mathcal{O}_m^{(\rho)} (x'_3) \cdots \rangle = \sum_k \sum_{a = 1}^{N_{i j
  k}} f_{i j k}^a C_{a, (\lambda)}^{(\mu) (\nu)} (x'_1, x'_2, x_0',
  \partial_{0'}) \langle (\mathcal{O}^{\dagger}_k)^{(\lambda)} (x'_0)
  \mathcal{O}_m^{(\rho)} (x'_3) \cdots \rangle .
\end{equation}
Transforming this expansion term-by term to the original coordinates $x_i$ we
find the convergent expansion (with convergence understood in the same sense
as in the previous section\footnote{Careful reading of the argument below
shows that, in the 4-point case, the requirement of absolute convergence of
{\eqref{doublyinf}} could be replaced by a weaker requirement that we can find
any subsequence of partial sums of {\eqref{doublyinf}} which approximates the
correlator pointwise.})
\begin{equation}
  \langle \mathcal{O}_i^{(\mu)} (x_1) \mathcal{O}_j^{(\nu)} (x_2)
  \mathcal{O}_m^{(\rho)} (x_3) \cdots \rangle = \sum_k \sum_{a = 1}^{N_{i j
  k}} f_{i j k}^a \frak{C}_{a, (\lambda)}^{(\mu) (\nu)} (x_1, x_2, x_0,
  \mathcal{D}_0) \langle (\mathcal{O}^{\dagger}_k)^{(\lambda)} (x_0)
  \mathcal{O}_m^{(\rho)} (x_3) \cdots \rangle . \label{OPEplanar}
\end{equation}
Here the differential operators $\mathcal{D}^{(\alpha)}$ are simply the
derivatives $\left( \frac{\partial}{\partial x'} \right)^{(\alpha)}$ expressed
in the original coordinates $x$, and conjugated by the conformal
transformation factor of $\mathcal{O}_k^{\dagger}$. The functions
$\frak{C}_{a, (\lambda)}^{(\mu) (\nu)}$ are obtained from $C_{a,
(\lambda)}^{(\mu) (\nu)}$ by our conformal transformation. The important point
is that truncation of $C_{a, (\lambda)}^{(\mu) (\nu)}$ in order of derivatives
$\partial^{(\alpha)}$ corresponds to truncation of $\frak{C}_{a,
(\lambda)}^{(\mu) (\nu)}$ in order of operators $\mathcal{D}^{(\alpha)}$.

We now specialize to $S$ being the $x^0 = 0$ plane, $x_0 = x_S = (- 1, 0,
\ldots),$ and take $S'$ to be the unit sphere with the center $x_0' = 0$. Then
the derivatives $\left( \frac{\partial}{\partial x'} \right)^{(\alpha)}
\mathcal{O}_k^{(\mu)} (x_0')$ are eigenstates of the standard dilatation
generator $D$ with eigenvalues $\Delta_k + | \alpha |$. Note that $D$ has two
fixed points: $x_0' = 0$ and infinity. Applying our conformal map, we find
that the derivatives $\mathcal{D}^{(\alpha)} \mathcal{O}_k^{(\mu)} (x_S)$ are
in turn eigenstates, with the same eigenvalues, of the conformal generator $D'
= (K^0 - P^0) / 2$ that preserves $x_S$ and $x_N = x_S^{\theta}$ (which is the
image of infinity under our conformal map) and acts by dilatations near these
two points. This, together with the conformal invariance and diagonality of
2-point functions, implies
\begin{equation}
  \langle (\mathcal{D}^{(\alpha)})^{\theta} \mathcal{O}_j^{\dagger (\mu)}
  (x_N) \mathcal{D}^{(\beta)} \mathcal{O}_k^{(\nu)} (x_S) \rangle \propto
  \delta_{| \alpha |, | \beta |} \delta_{j, k}, \label{descOrth}
\end{equation}
where $(\mathcal{D}^{(\alpha)})^{\theta}$ is obtained from
$\mathcal{D}^{(\alpha)}$ by replacing $x \rightarrow x^{\theta}$.

The OPE {\eqref{OPEplanar}} gives an expansion for ket states $| \Psi \rangle
\in \mathcal{H}_0^{\tmop{OS}}$ created by two local operators in terms of ket
states created by a single operator, which are elements of
$\mathcal{H}_0^{\tmop{CFT}}$. We would like to have a dual version of this
expansion for bra states $\langle \Psi |$. For this we need to understand how
the OPE transforms under the conjugation. Note that the formal differential
operators $\frak{C}_{a, (\lambda)}^{(\mu) (\nu)} (x_1, x_2, x_0,
\mathcal{D}_0)$ can be uniquely determined by the equation
\begin{equation}
  \langle \mathcal{O}_1^{(\mu)} (x_1) \mathcal{O}_2^{(\nu)} (x_2)
  \mathcal{O}_3^{(\rho)} (x_3) \rangle_a = \frak{C}_{a, (\lambda)}^{(\mu)
  (\nu)} (x_1, x_2, x_0, \mathcal{D}_0) \langle
  (\mathcal{O}^{\dagger}_3)^{(\lambda)} (x_0) \mathcal{O}_3^{(\rho)} (x_3)
  \rangle,
\end{equation}
where it is understood that the points are arranged as above, so that the
formal sum defined by $\frak{C}_{a, (\lambda)}^{(\mu) (\nu)} (x_1, x_2, x_0,
\mathcal{D}_0)$ actually converges. By applying complex conjugation on both
sides and using the 2-point and 3-point hermiticity constraints
{\eqref{hermicity2-point}} and {\eqref{CFThermiticity}} we find
\begin{equation}
  \langle (\mathcal{O}^{\dag}_1)^{(\mu)} (x_1^{\theta})
  (\mathcal{O}^{\dag}_2)^{(\nu)} (x_2^{\theta})
  (\mathcal{O}^{\dag}_3)^{(\rho)} (x_3^{\theta}) \rangle_a = [\frak{C}_{a,
  (\lambda)}^{(\mu) (\nu)} (x_1, x_2, x_0, \mathcal{D}_0)]^{\asterisk}
  \langle \mathcal{O}_3^{(\lambda)} (x_0^{\theta})
  (\mathcal{O}^{\dag}_3)^{(\rho)} (x_3^{\theta}) \rangle,
\end{equation}
which implies that
\begin{equation}
  [\frak{C}_{a, (\lambda)}^{(\mu) (\nu)} (x_1^{\theta}, x_2^{\theta},
  x_0^{\theta}, \mathcal{D}_0^{\theta})]^{\asterisk} =
  \widetilde{\frak{C}}_{a, (\lambda)}^{(\mu) (\nu)} (x_1, x_2, x_0,
  \mathcal{D}_0), \label{OPEkernelDagger}
\end{equation}
where $\widetilde{\frak{C}}_{a, (\lambda)}^{(\mu) (\nu)} (x_1, x_2, x_0,
\mathcal{D}_0)$ is the formal sum that appears in the OPE for operators
with conjugate-reflected quantum numbers.

We are now ready to prove OS positivity for 4-point functions. Let $|
\Psi_0 \rangle$ be an OS ket state involving at most two local operators, i.e.\footnote{In the notation of section \ref{OSaxioms} this could be written as $\sum_{i,k,\alpha,\beta}| \psi(f_{i,j,(\alpha),(\beta)},\,\mathcal{O}_i^{(\alpha)},\,\mathcal{O}_j^{(\beta)})\rangle$.}
{\begin{equation}
		| \Psi_0 \rangle =\sum_{i, j, \alpha, \beta} \int dx_1\, dx_2\, f_{i, j, (\alpha)
			(\beta)} (x_1, x_2)| \mathcal{O}_i^{(\alpha)} (x_1) \mathcal{O}_j^{(\beta)} 
		(x_2) \rangle, \label{H2}
\end{equation}}
where $f (x_{1,} x_2)$ is a compactly supported test function vanishing unless
$0 > x_1^0 > x_2^0$. (Terms with one or no operators are realized by setting
one or both operators to the identity.) Since by Euclidean CFT axioms the
correlation functions are real-analytic, the integrals that appear in the
expression for $\langle \Psi_0 | \nobracket \Psi_0 \rangle$ can be
approximated by finite Riemann sums, reflection-symmetric if necessary. This
implies that, for any $\varepsilon > 0$, we can pass from $| \Psi_0 \rangle$
to a ket state $| \Psi \rangle$ which is created by a \tmtextit{finite}
linear combination of insertions of up to two local operators with $x_1^0,
x_2^0 < 0$: 
{
	\begin{equation}
  | \Psi \rangle = \sum_{i, j,  \alpha, \beta, x_1, x_2} c_{i, j,  (\alpha)
  (\beta), x_1,x_2} |\mathcal{O}_i^{(\alpha)} (x_1) \mathcal{O}_j^{(\beta)} (x_2) \rangle,
  \label{H2fin}
\end{equation}
}
and has the property that
\begin{equation}
  | \langle \Psi_0 | \Psi_0 \rangle - \nobracket \langle \Psi_0 | \Psi
  \rangle \nobracket | < \varepsilon, \qquad | \langle \Psi_0 | \Psi
  \rangle - \nobracket \langle \Psi | \Psi \rangle \nobracket | <
  \varepsilon, \label{PsiPsi0prop}
\end{equation}
so that as a result
\begin{equation}
  | \langle \Psi_0 | \Psi_0 \rangle - \nobracket \langle \Psi | \Psi
  \rangle \nobracket | < 2 \varepsilon \label{stp}
\end{equation}
We are therefore reduced to proving the nonnegativity of $\langle \Psi |
\Psi \rangle \nobracket$.

Now, the OPE convergence axiom implies that, for any $\varepsilon > 0$,
starting from $| \Psi \rangle$ and using the OPE {\eqref{OPEplanar}} in the
half-space $x^0 < 0$, we can construct a state $| \psi \rangle \nobracket = |
\psi_{ \Lambda} \rangle \in \mathcal{H}_0^{\tmop{CFT}}$ created by a finite
linear combination of local operators at $x_S$ such that
\begin{equation}
  | \langle \Psi | \psi_{} \rangle - \langle \Psi | \Psi \rangle |
  < \varepsilon . \label{appr1}
\end{equation}
Here $\Lambda$ is an OPE truncation cutoff which we need to increase
appropriately as $\varepsilon$ gets smaller. Namely, we will obtain $|
\psi_{} \rangle$ by keeping in the OPE all terms with $\Delta_k + | \alpha
| < \Lambda$, where $\Delta_k$ is the dimension of a primary
$\mathcal{O}_k$ appearing in the OPE, and $\alpha$ is the order of the
descendant $\mathcal{D}^{(\alpha)} \mathcal{O}_k$.

We can then repeat this procedure in the upper half-plane $x^0 > 0$ and
construct a state $\langle \psi' | = \langle \psi_{\Lambda'} |
\nobracket$ of local operators inserted at $x_N$ such that
\begin{equation}
  | \langle \Psi | \psi \rangle - \nobracket \langle \psi' | \psi
  \rangle \nobracket | < \varepsilon . \label{psiprime}
\end{equation}
Eq.\ {\eqref{OPEkernelDagger}} and the reality constraint
{\eqref{OPEcoeff:conjugation}} for the OPE coefficients imply that the state
$\langle \psi' |$ differs from $\langle \psi |$ at most by where the OPE
expansion was truncated. Furthermore, we can always assume that $\langle
\psi' |$ contains at least all the terms that $\langle \psi |$ does (i.e.\ $\Lambda' \geqslant \Lambda$), since adding more OPE terms to $\langle \psi'
|$ can only improve {\eqref{psiprime}}.

Eq.\ {\eqref{descOrth}} then implies that $\langle \psi' | \psi \rangle
\nobracket = \langle \psi | \psi \rangle \nobracket$ and therefore
\begin{equation}
  | \langle \Psi | \psi \rangle - \nobracket \langle \psi | \psi \rangle
  \nobracket | < \varepsilon . \label{appr2}
\end{equation}
Combining this with {\eqref{appr1}} we conclude:
\begin{equation}
  | \langle \Psi | \Psi \rangle - \nobracket \langle \psi | \psi \rangle
  \nobracket | < 2 \varepsilon, \label{PsiPsi}
\end{equation}
Since by the CFT positivity axiom $\langle \psi | \psi \rangle \nobracket$ is
non-negative, we conclude that $\langle \Psi | \Psi \rangle \nobracket$ is
also non-negative. This completes the proof of OS positivity for states
created by up to two operator insertions.

\subsection{Denseness and Hilbert space implications}\label{Hilbert}

Here we will describe some useful byproducts of the just given argument. Note
that {\eqref{appr1}} and {\eqref{appr2}}, in addition to {\eqref{PsiPsi}},
also implies
\begin{equation}
  \| \Psi - \psi_{\Lambda} \| \equiv \langle \Psi - \psi_{\Lambda} | \Psi -
  \psi_{\Lambda} \rangle = \langle \Psi | \Psi \rangle - \langle \Psi |
  \psi_{\Lambda} \rangle + \langle \psi_{\Lambda} | \psi_{\Lambda} \nobracket
  \rangle - \langle \psi_{\Lambda} | \Psi \rangle < 2 \varepsilon .
  \label{PsiPsiDense}
\end{equation}
This means that any $| \Psi \rangle$ can be approximated arbitrarily well
by a $| \psi_{\Lambda} \rangle \nobracket \in \mathcal{H}_0^{\tmop{CFT}}$. In
other words, $\mathcal{H}_0^{\tmop{CFT}}$ is a dense subspace of the vector
space of $\Psi$'s.

This fact has a simple but quite powerful consequence involving the CFT
Hilbert space $\mathcal{H}^{\tmop{CFT}}$, defined in Sec.\ \ref{ECFTax}
as the completion of $\mathcal{H}_0^{\tmop{CFT}}$. Eq.\ {\eqref{PsiPsiDense}}
implies, using the triangle inequality $\| \psi_{\Lambda_1} - \psi_{\Lambda_2}
\| \leqslant \| \Psi - \psi_{\Lambda_1} \| + \| \Psi - \psi_{\Lambda_2} \|$,
that the states $| \psi_{\Lambda} \rangle \nobracket$ corresponding to smaller
and smaller $\varepsilon$ form a Cauchy sequence. Therefore, these states have
a limit in $\mathcal{H}^{\tmop{CFT}}$ as $\Lambda \rightarrow \infty$,
which we call $| \psi_{\infty} \rangle \nobracket$. This $\psi_{\infty}$ is
nothing but the full, untruncated, OPE expansion of the state $\Psi$. We claim
that the map mapping $\Psi$'s to the corresponding $\psi_{\infty}$'s is
isometric, i.e.\ it preserves the inner products:
\begin{equation}
  \langle \Phi | \Psi \rangle = \langle \varphi_{\infty} | \psi_{\infty}
  \rangle . \label{OPEhilbert}
\end{equation}
Here the inner product on the l.h.s.\ is the OS inner product, computed using
CFT 4-point functions with operators inserted in the lower and upper
half-spaces, while the inner product in the r.h.s.\ is the \
$\mathcal{H}^{\tmop{CFT}}$ inner product, defined as the limit of
$\mathcal{H}_0^{\tmop{CFT}}$ inner product removing the cutoff:
\begin{equation}
  \langle \varphi_{\infty} | \psi_{\infty} \rangle \assign \lim_{\Lambda
  \rightarrow \infty} \langle \varphi_{\Lambda} | \psi_{\Lambda} \rangle .
  \label{HCFTinner}
\end{equation}
The proof of {\eqref{OPEhilbert}} is straightforward. We write:
\begin{eqnarray}
  & \langle \varphi_{\Lambda} | \psi_{\Lambda} \rangle = \langle \Phi +
  (\varphi_{\Lambda} - \Phi) | \Psi + (\psi_{\Lambda} - \Psi) \rangle =
  \langle \Phi | \Psi \rangle + \tmop{err} (\Lambda), & \\
  & \tmop{err} (\Lambda) = \langle \varphi_{\Lambda} - \Phi | \Psi \rangle +
  \langle \Phi | \psi_{\Lambda} - \Psi \rangle + \langle \varphi_{\Lambda} -
  \Phi | \psi_{\Lambda} - \Psi \rangle . & 
\end{eqnarray}
By Eq.\ {\eqref{PsiPsiDense}} we know that $\| \Psi - \psi_{\Lambda} \|$, $\|
\Phi - \varphi_{\Lambda} \|$ go to zero as $\Lambda \rightarrow \infty$.
Hence, $\tmop{err} (\Lambda) \rightarrow 0$ and {\eqref{OPEhilbert}} is
proved.

Eqs.\ {\eqref{OPEhilbert}} and {\eqref{HCFTinner}} mean that OPE converges in
the sense of the CFT Hilbert space. This property is often used in the CFT
literature (see Sec.\ \ref{Eucl4-point}). Note that CFT axioms in Sec.\ \ref{ECFTax} only assume pointwise OPE convergence, which is a weaker
statement. Curiously, by the given arguments, pointwise OPE convergence plus
CFT positivity imply Hilbert space convergence, at least for the 4-point
functions.

In the above argument we used the Hilbert space $\mathcal{H}^{\tmop{CFT}}$,
the completion of $\mathcal{H}_0^{\tmop{CFT}}$. We may introduce a second
Hilbert space as the completion of the space of $\Psi$'s, call it
$\mathcal{H}^{(2)}$. This Hilbert space contains e.g.\ all $\Psi_0$ states \
{\eqref{H2}}. (Similarly to {\eqref{PsiPsiDense}}, Eqs.\ {\eqref{PsiPsi0prop}}
and {\eqref{stp}} imply that $\Psi$ states are dense in the $\Psi_0$ states.)
Although $\mathcal{H}^{(2)}$ may look like a ``bigger'' space than
$\mathcal{H}^{\tmop{CFT}}$, actually it's not. Indeed, the map from $\Psi$
to $\psi_{\infty}$ extends to an isometric map from $\mathcal{H}^{(2)}$
to $\mathcal{H}^{\tmop{CFT}}$. In other words, Eq.\ {\eqref{OPEhilbert}}
remains true for any $\Phi, \Psi \in \mathcal{H}^{(2)}$. Since we can view $\cH^\text{CFT}_0$ as a
subspace of the space of $\Psi$'s, we also have that this map is surjective. The Hilbert spaces
$\mathcal{H}^{(2)}$ and $\mathcal{H}^{\tmop{CFT}}$ are thus unitarily
equivalent and may be identified. 

\subsection{OS clustering}\label{OSclustering}

Here we will derive the OS clustering {\eqref{clusterint}} from CFT axioms. We
will consider only $m + n = 4$, i.e.\ when the left-hand side of
{\eqref{clusterint}} can be written in terms of a 4-point function (this
also covers $m + n < 4$ since we can choose some of $\varphi$'s or $\chi$'s to
be the trivial identity field). We can assume that all $\chi$'s and
$\varphi$'s in {\eqref{clusterint}} are primary fields, since any derivatives
can be integrated by parts.

First as a general remark, assuming OS positivity, clustering (for any $m, n$)
only needs to be established point-wise, i.e.
\begin{equation}
  \lim_{\lambda \rightarrow \infty} \langle \chi^{\dagger}_m (y_m) \ldots
  \chi_1^{\dagger} (y_1) \varphi_1 (x_1 + \lambda a) \ldots \varphi_n (x_n +
  \lambda a) \rangle = \langle \chi^{\dagger}_m (y_m) \ldots
  \chi_1^{\dagger} \nobracket (y_1) \rangle \langle \nobracket \varphi_1 (x_1)
  \ldots \varphi_n (x_n) \rangle .
\end{equation}
This follows from the dominated convergence theorem. Indeed, OS positivity and
translation invariance (recall that we only consider $a^0 = 0$!) implies a
uniform in $\lambda$ bound\footnote{This follows, similarly to
{\eqref{OSnaive1}}, by applying OS positivity to the state $| \Psi \rangle = |
\psi (F, \varphi_1 \ldots \varphi_n) \rangle + e^{i \alpha} | \psi (G, \chi_1
\ldots \chi_n) \rangle \nobracket$ where $F, G$ tend to delta functions
localizing the operators at points $x_1 + \lambda a, \ldots, x_n + \lambda a$
and $y_1^{\theta}, \ldots, y_n^{\theta}$ respectively, and choosing the phase
$\alpha$ appropriately.}
\begin{eqnarray}
  & | \langle \chi^{\dagger}_m (y_m) \ldots \chi_1^{\dagger} (y_1)
  \varphi_1 (x_1 + \lambda a) \ldots \varphi_n (x_n + \lambda a) \rangle |^2
  \leqslant & \\
  &  \langle \chi^{\dagger}_m (y_m) \ldots \chi_1^{\dagger} (y_1) \chi_1
  (y_1^{\theta}) \ldots \chi_n (y_n^{\theta}) \rangle \times \langle
  \varphi^{\dagger}_n (x_n^{\theta}) \ldots \varphi_1^{\dagger}
  (x_1^{\theta}) \varphi_1 (x_1) \ldots \varphi_n (x_n) \rangle . &  \nonumber
\end{eqnarray}
It then follows that the integrand in {\eqref{clusterint}} is bounded by a
$\lambda$-independent integrable function, and the dominated convergence theorem is
applicable.

Going back to the 4-point function case which is our focus in this section,
let us start with $m = n = 2$. Since we already proved OS positivity for
states created by at most two operators (Sec.\ \ref{OSfromCFT}), in this
case we can rely on the above observation and we only need to check the
point-wise limit:
\begin{equation}
  \lim_{\lambda \rightarrow \infty} \langle \chi^{\dagger}_2 (y_2)
  \chi_1^{\dagger} (y_1) \varphi_1 (x_1 + \lambda a) \varphi_2 (x_2 + \lambda
  a) \rangle = \langle \chi^{\dagger}_2 (y_2) \chi_1^{\dagger} \nobracket
  (y_1) \rangle \langle \nobracket \varphi_1 (x_1) \varphi_2 (x_2) \rangle .
\end{equation}
To see this, we apply the OPE {\eqref{OPEplanar}} to $\chi^{\dagger}_2
(y_2) \chi_1^{\dagger} (y_1)$ in the left-hand side. The results of the
previous section imply that this OPE can be interpreted as expanding the state
in the Hilbert space $\mathcal{H}$ created by $\chi^{\dagger}_2 (y_2)
\chi_1^{\dagger} (y_1)$ in terms of eigenstates of $(K^0 - P^0) / 2$. This
implies that the OPE converges uniformly in $\lambda$ since the norm of $|
\varphi_1 (x_1 + \lambda a) \varphi_2 (x_2 + \lambda a) \rangle$ is
independent of $\lambda$ due to translation invariance ($a^0 = 0$!). We can
thus use the OPE to approximate $\langle \chi^{\dagger}_2 (y_2)
\chi_1^{\dagger} (y_1) \varphi_1 (x_1 + \lambda a) \varphi_2 (x_2 + \lambda a)
\rangle$ for any $\lambda$ to within any $\varepsilon > 0$ by a finite sum of
3-point functions of the form $\langle (\mathcal{D}^{(\alpha)})^{\theta}
\mathcal{O}_i^{(\nu)} (x_N) \varphi_1 (x_1 + \lambda a) \varphi_2 (x_2 +
\lambda a) \rangle$ times some $\lambda$-independent coefficients. Of these,
only the term corresponding to the identity operator, i.e.\ the one with
$(\mathcal{D}^{(\alpha)})^{\theta} \mathcal{O}_i^{(\nu)} = 1$, does not decay
with $\lambda$. It is easily verified that the contribution of this term is
precisely equal to $\langle \chi^{\dagger}_2 (y_2) \chi_1^{\dagger}
\nobracket (y_1) \rangle \langle \nobracket \varphi_1 (x_1) \varphi_2 (x_2)
\rangle$. This finishes the proof of clustering for $m = n = 2.$

In the remaining case $m = 3, n = 1$, we will consider the limit for the
integral (since we have not yet proved OS positivity for states involving 3
operators). Note that $\langle \varphi_1 (x) \rangle = 0$ unless $\varphi_1
\propto 1$, in which case the cluster property becomes trivial. This means
that we only need to prove
\begin{equation}
    \lim_{\lambda \rightarrow \infty} \int d x\, d y\, \overline{g
    (y_1^{\theta}, y_2^{\theta}, y_3^{\theta})} f (x_1) \langle
    \chi_3^{\dagger} (y_3) \chi_2^{\dagger} (y_2) \chi_1^{\dagger} (y_1)
    \varphi_1 (x_1 + \lambda a) \rangle = 0 \label{m3n1}
\end{equation}
for $\varphi_1 \neq 1.$ This in turn is a very simple consequence of conformal
invariance, and of the fact that $\Delta_{\varphi} > 0$ for all operators but
the identity. We will be somewhat schematic. The main point is that the
configuration $(y_3, y_2, y_1, \infty)$ is nonsingular from the conformal
kinematics point of view (for which the conformal compactification $S_d$ of
the Euclidean space $\mathbb{R}^d$ is the appropriate arena). One way to see
it is that the cross ratios are finite in this limit. Thinking in a pedestrian
way, we can find a conformal transformation $g_{\lambda}$ which will move
points $(y_3, y_2, y_1, x_1 + \lambda a)$ to some points which have finite
limits as $\lambda \rightarrow \infty$. Transforming the integral
{\eqref{m3n1}} to this coordinate system, the only singular behavior at large
$\lambda$ comes from the Weyl transformation factor as the operator
$\varphi_1$ is moved from near infinity to a finite position. This factor
implies that the integral {\eqref{m3n1}} will go to zero as $\lambda^{- 2
\Delta_{\varphi}}$, proving clustering in this particular case. See Sec.\ \ref{3+1MinkCluster} for additional details. More generally, the same argument
will also work for arbitrary $m$ as long as $n = 1$.

\section{Euclidean CFT $\Rightarrow$ Wightman: Basic
strategy}\label{strategy}

We will now pass to the main task of our paper: given a Euclidean unitary CFT,
recover Minkowski correlators and show that they satisfy Wightman axioms.

Let us first discuss this problem without assuming conformal invariance.
Suppose we know correlators $G^E_n (x_1, \ldots, x_n)$ of a scalar field in a
Euclidean QFT which is translationally and rotationally invariant, but not
necessarily conformally invariant. We are assuming, as discussed above, that
the correlators $G^E_n$ are defined and real-analytic (see footnote
\ref{real-anal}) for non-coincident Euclidean points ($x_k \in \mathbb{R}^d,
x_i \neq x_j$).

We would like to recover correlators in Minkowski signature. We are only
interested here in Wightman correlation functions, where the operator ordering
is fixed while the Minkowski time coordinates vary independently. We will call
them simply ``Minkowski correlators''. Starting from this section we will
focus on correlators of scalar primaries; correlators of fields in general
$\tmop{SO} (d)$ representations will be considered in our future publication
{\cite{paper2a}}.

To understand the equations below, it helps to keep in mind the basic
heuristic. If we had a Hilbert space, field operators $\phi$, and a
Hamiltonian $H$, then the Minkowski correlators would be given by
\begin{equation}
G_n^M (x^M_1, \ldots, x^M_n) = \langle 0 | \phi \left( 0,
   \mathbf{x}_1 \right) e^{- i H (t_1 - t_2)} \phi \left( 0, \mathbf{x}_2
   \right) e^{- i H (t_2 - t_3)} \ldots | 0 \rangle, \qquad x_k^M
   = \left( t_k, \mathbf{x}_k \right) . 
 \end{equation}
while the Euclidean correlators by
\begin{equation}
G^E_n (x_1, \ldots, x_n) = \langle 0 | \phi \left( 0,
   \mathbf{x}_1 \right) e^{- H (\epsilon_1 - \epsilon_2)} \phi \left( 0,
   \mathbf{x}_2 \right) e^{- H (\epsilon_2 - \epsilon_3)} \ldots | 0
   \rangle, \qquad x_k = \left( \epsilon_k, \mathbf{x}_k \right),
   \quad \epsilon_k > \epsilon_{k + 1}, 
  \end{equation}
We stress that the r.h.s.\ of these two equations will never be used in this
paper. We just use them to illustrate the intuitive property that $G_n^M$ can
be recovered from $G^E_n$ by analytic continuation $\epsilon_k \rightarrow
\epsilon_k + i t_k$ and sending $\epsilon_k \rightarrow 0$ while respecting
$\epsilon_k > \epsilon_{k + 1}$. In other words, there is a holomorphic function
$G_n$ which reduces to $G_n^M$ in one limit and to $G^E_n$ in another. The
precise domain of analyticity of $G_n$ can be clarified from the Wightman
axioms {\cite{Streater:1989vi}}. Their basic consequence is that Minkowski
correlators can be analytically continued to the ``forward tube'' (see below),
which contains {the configuration space of Euclidean time-ordered points as a section.} In this paper we will derive
Wightman axioms, rather than assume them. In particular, we will carry out
analytic continuation to the forward tube just from the properties of the
Euclidean correlators.

Let us put these observations into a definition of what it means to recover
$G_n^M$ from $G^E_n$. We consider $n$-point configurations with
complexified coordinates:
\begin{equation}
  c = (x_1, \ldots, x_n), \qquad x_k = \left( x_k^0, \mathbf{x}_k \right) \in
  \mathbb{C}^d . \label{ccompl}
\end{equation}
The ``forward tube'' $\mathcal{T}_n$ is defined as the set of all such
configurations for which the differences $y_k = x_k - x_{k + 1} = (y_k^0,
\mathbf{y}_k) \in \mathbb{C}^d$ satisfy the constraint:
\begin{equation}
  \tmop{Re} y_k^0 > | \tmop{Im} \mathbf{y}_k | \label{forward}, \qquad k = 1,
  \ldots, n - 1.
\end{equation}
Equivalently, this means that vectors $\tmop{Im} (i y_k^0, \mathbf{y}_k)$
belong to the open forward light cone of $\mathbb{R}^{1, d - 1}$, explaining
the name ``forward tube''.\footnote{The just given definition of the forward tube is adapted to the Euclidean coordinates. In Sec.\ \ref{executive}, Eq.\ \eqref{forward-tube-Mink}, we wrote the same definition in terms of Minkowski coordinates $x_k^M=(-i x_k^0, \mathbf{x}_k)$.}

Let $\mathcal{D}_n$ be the subset of the forward tube consisting of the
configurations with real spatial parts $\mathbf{x}_k$. Equivalently, we have:
\begin{equation}
  \label{def:Dn} \mathcal{D}_n = \{\, c\ |\ x_k^0 = \epsilon_k + i
  t_k,\ \mathbf{x}_k \in \mathbb{R}^{d - 1},\ \epsilon_1 > \epsilon_2 > \ldots >
  \epsilon_n \} .
\end{equation}
Finally, we denote by $\mathcal{D}^E_{n}$ the Euclidean part of
$\mathcal{D}_n$ obtained by setting all $t_k = 0$.

Minkowski correlators are then defined by the following two-step procedure:

\tmtextbf{Step 1.} One finds an extension $G_n^E$ from $\mathcal{D}^E_n$ to a
function $G_n (x_1, \ldots, x_n)$ such that one of the two conditions is
satisfied:
\begin{equation}
  \text{$G_n$ is defined on $\mathcal{T}_n$, and holomorphic in all variables
  $x_k^0$, $\mathbf{x}_k$,} \label{Gass1}
\end{equation}
or
\begin{equation}
  \text{$G_n$ is defined on $\mathcal{D}_n$, is holomorphic in variables $x_k^0$
  and is real-analytic in $\mathbf{x}_k$.} \label{Gass}
\end{equation}
Real analyticity in $\mathbf{x}_k$ means that $G_n$ can be extended from
$\mathcal{D}_n$ to a holomorphic function defined on a neighborhood of
$\mathcal{D}_n$ which allows small imaginary parts for $\mathbf{x}_k$. This
neighborhood can be arbitrarily small. Condition {\eqref{Gass}} is thus weaker
than {\eqref{Gass1}} and may be easier to check, although Theorem \ref{ThVlad}
below shows that the two conditions are equivalent under the ``powerlaw
bound'' assumption.

\tmtextbf{Step 2.} Minkowski correlators are defined as the limits of $G_n$
from inside $\mathcal{D}_n$ by sending $\epsilon_i \rightarrow 0$:\footnote{We will see in Theorem \ref{ThVlad} that this limit has to be taken
along a fixed direction and is independent of direction. If the stronger
condition {\eqref{Gass1}} holds, the limit can in fact be taken along any
direction in the forward null cone.}
\begin{equation}
  G_n^M (x_1^M, \ldots, x_n^M) = \lim_{\epsilon_i \rightarrow 0} G_n (x_1,
  \ldots, x_n), \qquad x_k^M = \left( t_k, \mathbf{x}_k \right), \quad k = 1
  \ldots n \label{limit} .
\end{equation}
As mentioned several times, Minkowski correlators are expected to be tempered
distributions, and therefore this limit has to be understood in the
distributional sense. To show that the limit exists and has properties
required by Wightman axioms, one relies on the following powerful theorem of
several complex variables: 

\begin{theorem}[Vladimirov's theorem]
  \label{ThVlad}Suppose that the function $G_n$ is translation- and rotation-invariant, satisfies
  {\tmem{{\eqref{Gass}}}} and in addition satisfies everywhere on
  $\mathcal{D}_n$ the following `powerlaw bound' with some positive constants
  $C_n, A_n, B_n$:
  \begin{gather}
    |G_n (x_1, \ldots, x_n)| \leqslant C_{n}  \left( 1 + \max_k 
    \dfrac{1}{\epsilon_k - \epsilon_{k + 1}} \right)^{A_n}  (1 + \max_i  | x_i
    - x_{i + 1} |)^{B_n},  \label{powerlawbound}\\
    |x_i - x_j |^2 \equiv  | \epsilon_i + it_i - \epsilon_j - it_j |^2 +
| \mathbf{x}_i - \mathbf{x}_j|^2.  \label{absdef} 
  \end{gather}
  Then:
  \begin{enumerate}
    \item Limit {\tmem{{\eqref{limit}}}} exists in the sense of tempered
    distributions. The limiting value $G_n^M$ is a tempered distribution.\footnote{This part of the theorem does not need translation- and rotation-invariance of $G_n$.}
    
    \item The distribution $G_n^M$ is Poincar\'e-invariant and satisfies the
    Wightman spectral condition. I.e.\ its Fourier transform $W (p_1, \ldots,
    p_{n - 1})$ with respect to the differences $x_k^M - x_{k + 1}^M$ has
    support in the product of closed forward light cones, which is the region
    $E_k \geqslant | \mathbf{p}_k |$, $p_k = (E_k, \mathbf{p}_k)$.
    
    \item The function $G_n$ can be extended from a
    holomorphic function on the whole forward tube $\mathcal{T}_n$. The limit
    {\eqref{limit}} exists also from the forward tube, i.e.\ when $\tmop{Re}
    y_k^0 \rightarrow 0, | \tmop{Im} \mathbf{y}_k | \rightarrow 0$, satisfying
    {\eqref{forward}}.
  \end{enumerate}
\end{theorem}

See App.\ \ref{Vlad} for the proof of Vladimirov's theorem and a reminder
of what the limit in the sense of distributions means. In the process of the
proof, it will be established that the holomorphic function $G_n$ on
$\mathcal{T}_n$ can be written as a ``Fourier-Laplace'' transform
\begin{equation}
  G_n (x_1, \ldots, x_n) = \int d p_1 \ldots d p_{n - 1}\, W (p_1, \ldots, p_{n
  - 1}) e^{\underset{k = 1}{\overset{n - 1}{\sum}} (- E_k (x^0_k - x^0_{k +
  1}) + i\mathbf{p}_k \cdummy (\mathbf{x}_k -\mathbf{x}_{k + 1}))},
  \label{FL1}
\end{equation}
where $W$ is a tempered distribution which is the Fourier transform of the
tempered distribution $G_n^M$, mentioned in Part 2 of the theorem.

To use Theorem \ref{ThVlad}, one needs to verify the powerlaw bound
{\eqref{powerlawbound}}. This strategy was first developed by Osterwalder and
Schrader (OS) {\cite{osterwalder1973,osterwalder1975}}.\footnote{OS used a
slightly stronger version of Theorem \ref{ThVlad} with real analyticity in
$\mathbf{x}_k$ replaced by the weaker assumption of continuity in these
variables, but this difference is not essential.} Their full list of
assumptions included, in addition to reflection positivity and other OS axioms
listed in Sec.\ \ref{OSaxioms}, a less widely known \tmtextit{linear growth
condition}, which roughly says that $G_n^E$ (appropriately integrated) grows
with $n$ not faster than a power of $n!$ and the degree of its singularities
grows not faster than linearly in $n$. The proof of the powerlaw bound was the
most technical part of the OS construction, and it crucially relied on the
linear growth condition. See App.\ \ref{OS} for the review.

In this paper we aim to define Minkowski correlators of a \tmtextit{conformal}
field theory, given Euclidean correlators satisfying the CFT axioms of Sec.\ \ref{ECFTax}. As seen in Sec.\ \ref{OSfromCFT}, reflection positivity is
robustly encoded in CFT via the positivity requirements for 2-point functions
and reality constraints on the OPE coefficients. On the other hand, not much
is known about how CFT $n$-point functions grow with $n$. In particular, we
are unable to justify the OS growth condition in our setup, hence we cannot
appeal to the OS theorem.

In this paper we will be able to circumvent this difficulty, by giving an
alternative proof of the powerlaw bound for the most important in applications
cases of 2, 3 and 4-point functions. Then, by \ Theorem \ref{ThVlad},
these correlators exist in Minkowski space and are tempered distributions. Our
proof of the powerlaw bound uses only the Euclidean CFT axioms. In fact, the
two- and 3-point function case is almost trivial, these correlators being
fixed by conformal invariance. The 4-point function case is much deeper and
is one of our main results. Remaining Wightman axioms not mentioned in Theorem
\ref{ThVlad} (positivity and clustering) will also be shown to hold.

\begin{remark}
  In practice, to compute the Minkowski correlator function one may connect a
  Minkowski configuration to a Euclidean configuration by a curve $c (s), 0
  \leqslant s \leqslant 1$, where $c (0)$ is Euclidean, $c (1)$ Minkowskian,
  and $c (s)$, $0 < s < 1$, belong to the forward tube. In general, the curve
  should remain in the forward tube except for the endpoint $c (1)$. This
  means that we must have strict inequalities:
  \begin{equation}
    \tmop{Re} x_1^0 (s) > \tmop{Re} x_2^0 (s) > \cdots > \tmop{Re} x_n^0 (s)
    \label{ftineq}
  \end{equation}
  everywhere along the analytic continuation contour, except for $s = 1$. See
  Fig.\ \ref{illustr}.
  
  In the literature, one sometimes encounters a different prescription for
  computing the Minkowski correlators (see e.g.\ {\cite{Hartman:2015lfa}},
  Sec.\ 3.1), where one puts all points but one at Minkowski positions, and
  considers correlators as a holomorphic function of the complexified coordinate
  of the remaining point. One then imagines that Wightman functions are
  holomorphic functions branching at light-cone separation, and that one can
  access different operator orderings by going around branch points. We would
  like to warn the reader that this prescription has to be taken with a grain
  of salt. To our knowledge there is no general result that the only Wightman
  functions singularities are branch cuts on the light cones. This is known to
  be true only in some special cases, e.g.\ for CFT 2-point and 3-point functions, as
  well as for CFT 4-point functions in $d = 2$ {\cite{Maldacena:2015iua}}. While
  some analytic continuation beyond the forward tube can be done in a general
  QFT (to the so called permuted extended tube), it does not suffice to
  justify the analytic continuation prescription of {\cite{Hartman:2015lfa}}
  in a general QFT. {In CFTs in higher dimensions, the
  prescription of Ref.\ {\cite{Hartman:2015lfa}} has some applicability, with
  the understanding that the correlator is analytic along the continuation
  contour but may stop being analytic at the endpoint (see App.\ \ref{Tom}).}
\end{remark}

\begin{figure}[h]\centering
  {\includegraphics[width=0.33\textwidth]{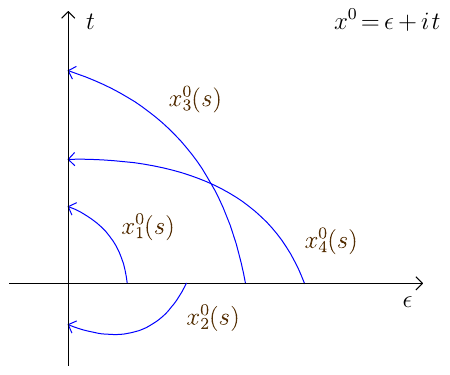}}
  \caption{Inequalities {\eqref{ftineq}} should be satisfied along the
  analytic continuation contour.\label{illustr}}
\end{figure}

\begin{remark}
	\label{lastOS}
	\newcommand{\eps}{\epsilon}A powerlaw bound in the forward tube \eqref{powerlawbound} of course implies a powerlaw bound for the Euclidean 4pt function itself. Together with rotation invariance, this will imply the remaining OS axiom, the ``Euclidean temperedness bound'' \eqref{OSmod}. Indeed, by rotation invariance, we can choose the direction of the $x^0$ axis before applying the Euclidean powerlaw bound. Let us choose the $x^0$ direction so that, after ordering the operators according to $\eps_1>\eps_2>\eps_3>\eps_4$, we have $\eps_k-\eps_{k+1}\geqslant \alpha |x_k-x_{k+1}|$ for each $k$. Such a direction exists for a sufficiently small positive $\alpha$, depending on $d$ and the number of points but not on $x_i$.\footnote{For each pair of points $(x_i,x_j)$ we consider the set of direction $\hat{e}_0$ such that $|(x_i-x_j)\cdot\hat{e}_0|\leqslant|x_i-x_j|\sin \delta$. This gives a subset $\mathcal{U}_\delta$ of the sphere $S^{d-1}$ with ${\rm Vol}(\mathcal{U}_\delta)\leqslant2\delta\, {\rm Vol}(S^{d-2})$. If we choose 
		$\delta_*=\frac{{\rm Vol}(S^{d-1})}{2 n(n-1)\,\text{Vol}(S^{d-2})}$, then the total volume excluded by considering all possible $(x_i,x_j)$ pairs is less than ${\rm Vol}(S^{d-1})/2$. Therefore, we can find a direction $\hat{e}_0$ such that the opposite inequality $|(x_i-x_j)\cdot\hat{e}_0|\geqslant  |x_i-x_j| \sin \delta_*$ holds for all pairs. Then, renumbering the points in the order of decreasing $x_i^0$, we obtain $x^0_k-x^0_{k+1}\geqslant  |x_k-x_{k+1}| \sin \delta_*$.} Applying the Euclidean case of the powerlaw bound \eqref{powerlawbound} in this frame we obtain \eqref{OSmod}.
\end{remark}

\subsection{Recovering Minkowski averages from Euclidean
averages}\label{MinkFromEucl}

Minkowski correlators provided by Theorem \ref{ThVlad}, being tempered distributions, can be paired with a
Schwartz test function $F$:
\begin{equation}
(G_n^M, F) = \int d x\, G_n^M (x_1 \ldots x_n) F (x_1 \ldots x_n) .
\label{GMav}
\end{equation}
Here we will discuss how these pairings can be computed given the Euclidean correlators (compare {\cite{osterwalder1973}}, Sec.\ 4.3). This discussion will be needed in Sec.\ \ref{sec:Wpos} below and may be skipped on the first reading. 

Eq.\ \eqref{GMav} can equivalently be expressed via the Fourier transform $W$ of $G_n^M$
with respect to $x_k - x_{k + 1}$:
\begin{equation}
  (G_n^M, F) = \int d p\, W (p_1, \ldots, p_{n - 1}) f (p_1 \ldots p_{n - 1}),
  \label{GMavFT}
\end{equation}
where $f (p_1 \ldots p_{n - 1}) = \hat{F} (- p_1, p_1 - p_2, \ldots, p_{n - 2}
- p_{n - 1}, p_{n - 1})$.
Natural pairings for Euclidean
correlators are
\begin{equation}
  (G_n^E, \varphi) = \int d x\, G_n^E (x_1 \ldots x_n) \Phi (x_1 \ldots x_n) .
  \label{GEav}
\end{equation}
where $\Phi$ is a $C^{\infty}$ test function compactly supported in $x_1^0 >
\cdots > x_n^0$. We wish to discuss how pairings {\eqref{GMav}} or
{\eqref{GMavFT}} can be found given {\eqref{GEav}}. 

By the Fourier-Laplace
representation {\eqref{FL1}}, we can write {\eqref{GEav}} as
\begin{equation}
  (G_n^E, \varphi) = \int d p\, W (p_1, \ldots, p_{n - 1}) g (p_1, \ldots, p_{n
  - 1}), \label{GEavFT}
\end{equation}
where $g$ is any Schwartz class function which agrees inside the forward light
cones with $\tilde{\varphi}$, the Fourier-Laplace transform of $\varphi (y_1,
\ldots, y_{n - 1}) = \int d x_n\, \Phi \left( x_n + \sum_{i = 1}^{n - 1} y_i,
x_n + \sum_{i = 2}^{n - 1} y_i, \ldots, x_n \right)$:
\begin{gather}
g \in \mathcal{S}, \qquad g (p) = \tilde{\varphi} (p) \qquad (p_k \in
\overline{V_+}) .  \label{gphi}\\
  \tilde{\varphi} (p_1, \ldots, p_{n - 1}) =  \int d y\, \varphi (y_1 \ldots
  y_{n - 1}) e^{\underset{k = 1}{\overset{n - 1}{\sum}} (- p^0_k y^0 +
  i\mathbf{p}_k \cdummy \mathbf{y}_k)}, \nn
\end{gather}
Note that we cannot just put $g = \tilde{\varphi}$ because $\tilde{\varphi}$
is by itself not a Schwartz function (it may grow exponentially in the
negative $p^0_{k}$ directions, although it will decrease exponentially in
the positive one, since $\varphi$ is supported at $y_k^0 > 0$). On the other
hand the values of $g$ outside the light cones, where $W$ is supported, are
unimportant. We can for example take
\begin{equation}
  g (p_1, \ldots, p_{n - 1}) = \chi (p^0_1) \ldots \chi (p^0_{n - 1})
  \tilde{\varphi} (p_1, \ldots, p_{n - 1}), \label{gex}
\end{equation}
where $\chi (s)$ is a $C^{\infty}$ function which equals identically 1 for $s
\geqslant 0$ and 0 for $s < - 1$.

Suppose then that we find a sequence of $C^{\infty}$ functions $\{ \varphi_r
\}_{r = 1}^{\infty}$ compactly supported at $y^0_k > 0$, the corresponding
functions $g_r \in \mathcal{S}$ such that $g_r = \tilde{\varphi}_r$ inside the
light cones, and in addition that $g_r \rightarrow f$ in the sense of the
Schwartz space (i.e.\ that all Schwartz space seminorms of the difference go to
zero), where $f$ is the function in {\eqref{GMavFT}}. Let us put $\Phi_r (x_1,
\ldots, x_n) = \varphi (x_1 - x_2, x_2 - x_3, \ldots, x_{n - 1} - x_n) \omega
(x_n)$ where $\omega$ is any $C_0^{\infty}$ function of integral one. Then we will
have
\begin{equation}
  (G_n^E, \Phi_r) = (W, g_r) \longrightarrow (W, f) = (G_n^M, F) \qquad (r
  \rightarrow \infty),
\end{equation}
and so we will solve the problem of computing Minkowski averages given
Euclidean averages. The following lemma, loosely
related to Lemma 8.2 in {\cite{osterwalder1973}}, shows that it is indeed possible to find such sequences
$\varphi_r$ and $g_r$ for any Schwartz class $f$. 

\begin{lemma}
  \label{lemma:fcheckdense}The set of functions $g \in \mathcal{S}
  (\mathbb{R}^{d (n - 1)})$ which satisfy {\eqref{gphi}} for some $\varphi$ a
  $C^{\infty}$ test function compactly supported in $y_k^0 > 0$ is dense in
  the Schwartz space.
\end{lemma}

We will give a formal proof; see App.\ \ref{IntLem1} for some intuition. We
will consider the case $n = 2$ as $n > 2$ is no more complicated. For each
$\varphi$ we consider the set $A_{\varphi}$ of Schwartz functions $g$ which satisfy
{\eqref{gphi}}:
\begin{equation}
  A_{\varphi} \assign \left\{ g \in \mathcal{S} (\mathbb{R}^d) \middle| 
  \nobracket g |_{\overline{V_+}} = \tilde{\varphi} \right\} .
\end{equation}
We know that $A_{\varphi}$ is non-empty, e.g.\ we can take $g$ from
{\eqref{gex}} (it is not hard to show that this is a Schwartz function). Our
lemma says that $A \equiv$ ``the union of $A_{\varphi}$ over all $\varphi$''
is dense in $\mathcal{S} (\mathbb{R}^d)$. The proof will be by contradiction.
Note that $A$ is a linear subspace of $\mathcal{S} (\mathbb{R}^d)$. If
$\overline{A} \neq \mathcal{S} (\mathbb{R}^{d})$, then there exists a
tempered distribution $T \in \mathcal{S}' (\mathbb{R}^d)$ such that $T$
vanishes on all test functions from $A$ but does not vanish
identically.\footnote{The corresponding statement for normed spaces is
standard, being a well-known consequence of the Hahn-Banach theorem (see e.g.\ {\cite{brezis}}, Corollary 1.8). For the Schwartz space, we can first find a
Schwartz norm $| \cdot |_n$, such that $\overline{A}$ is not everywhere dense
with respect to this norm, and then apply the standard statement with the norm
$| \cdot |_n$. This gives a linear functional $T$ on $\mathcal{S}
(\mathbb{R}^d)$ continuous with respect to $| \cdot |_n$, hence $T \in
\mathcal{S}' (\mathbb{R}^d)$.}

So $T$ in particular vanishes on $A_0 = \left\{ g \in \mathcal{S}
(\mathbb{R}^{d}) \middle|  \nobracket g |_{\overline{V_+}} = 0 \right\}$
(take $\varphi = 0$). This means that the support of the distribution $T$ is
contained inside $\overline{V_+}$. Consider the Fourier transform of $T$,
\begin{equation}
  \hat{T} (x) \assign \int \frac{d^d p}{(2 \pi)^{d}}\, T (p) e^{i p^0 x^0 -
  i\mathbf{p} \cdummy \mathbf{x}} .
\end{equation}
We can consider $\hat{T} (x)$ for real $x$ where it is a distribution. Since
$\tmop{supp} (T) \subseteq \overline{V_+}$ it is also natural to consider
$\hat{T} (\xi + i \eta)$ where $\xi, \eta$ are real and $\eta$ is in the
forward cone. We know that $\hat{T} (x)$ is a holomorphic function for such $x
= \xi + i \eta$. We also know that the distribution $\hat{T} (x)$ for real $x$
can be obtained as a limit of the holomorphic function $\hat{T} (\xi + i
\eta)$ as $\eta \rightarrow 0$.

Let us now come back to the assumption that $(T, g) = 0$ for any $g \in
A_{\varphi}$. We will apply this to a function $g$ of the form $g = X (p) \tilde{\varphi}$ where $X
(p)$ is a $C^{\infty}$ function identically 1 on the forward
light cone and such that $X (p) e^{- p^0 x^0 + i\mathbf{p} \cdot
\mathbf{x}}$ is in Schwartz class for any $x^0 > 0$. It is easy to see that
such functions $X(p)$ exist. Writing $(T, g)$ in full we get:
\begin{eqnarray}
  0 = (T, g) & = & \int d p\, T (p) X (p) \int d x\, e^{- p^0 x^0 +
  i\mathbf{p} \cdummy \mathbf{x}} \varphi (x) \nonumber\\
  & = & \int d x\, \varphi (x) \int d p\, T (p) X (p) e^{- p^0 x^0 +
  i\mathbf{p} \cdummy \mathbf{x}} \nonumber\\
  & = & \int d x\, \varphi (x) \hat{T} (i x^0, \mathbf{x}) . 
  \label{int:Tf}
\end{eqnarray}
The swap of the order of integration between the first and the second line can
be justified as follows. Since $T (p)$ is a tempered distribution, we can
write it as a finite sum of derivatives of continuous functions of power
growth: $T (p) = \underset{}{} \sum_{\alpha} \partial_p^{\alpha} F_{\alpha}
(p)$. Using distributional integration by parts, we can the rewrite the first
line of {\eqref{int:Tf}} as a sum of ordinary integrals, apply Fubini's
theorem to swap the integration order, and integrate by parts back to express
the answer in terms of $T (p)$.

Because $\varphi (x)$ has compact support in the region $x^0 > 0$, the
argument of $\hat{T}$ in the last line of (\ref{int:Tf}) is of the form $\xi +
i \eta$ with $\eta = (x^0, \tmmathbf{0})$ in the forward light cone, where
we know $\hat{T}$ is analytic. So, from the fact that the last line of
(\ref{int:Tf}) vanishes for any $\varphi$ we conclude that
\begin{equation}
  \hat{T} (\xi + i \eta) = 0, \qquad \xi = (0, \mathbf{x}), \quad \eta = (i
  x^0, \tmmathbf{0}), \quad x^0 > 0.
\end{equation}
The set of these points is a totally real submanifold, and so by analyticity
we conclude that $\hat{T} (\xi + i \eta)$ is identically zero for any $\xi
\in \mathbb{R}^d$ and any $\eta$ in the forward cone. Furthermore, as
mentioned above, $\hat{T} (x)$ for real $x$ is a boundary value of $\hat{T}
(\xi + i \eta)$. Therefore, $\hat{T} = 0$ in the sense of distributions.
However we assumed above that $T$ was not identically zero. The reached
contradiction shows that $A$ is dense in $\mathcal{S}$.

\section{Two- and 3-point functions}\label{23-point}

Let us see how the strategy from Sec.\ \ref{strategy} works for the CFT 2-point
and 3-point functions. The Euclidean 2-point and 3-point correlators of scalar primaries
are given by [$x_{i j}^2 = (x_i - x_j)^2$]
\begin{eqnarray}
  G^E_2 (x_1, x_2) & = & \frac{1}{(x^2_{12})^{\De}},  \label{G2E}\\
  G^E_3 (x_1, x_2, x_3) & = & \frac{c_{123}}{(x^2_{12})^{h_{123}}
  (x^2_{13})^{h_{132}} (x_{23}^2)^{h_{23}}}, \qquad h_{i j k} = (\Delta_i +
  \Delta_j - \Delta_k) / 2.  \label{G3E}
\end{eqnarray}

In this case, the standard way to obtain the Wightman functions is to write
these Euclidean correlators in terms of $x^2_{i j}$ with $i < j$ (as we did).
Substituting the analytic continuation of $x_{i j}^2$,
\begin{equation}
  x^2_{i j} = (x_i - x_j)^{\mu} (x_i - x_j)_{\mu} = (x^0_i -
  x^0_j)^2 + (\mathbf{x}_i - \mathbf{x}_j)^2, \qquad x_i = \left( x_i^0,
  \mathbf{x}_i \right) \in \mathbb{C}^d \label{xij-cont} .
\end{equation}
into the Euclidean 2-point and 3-point functions expressions, we obtain their analytic
continuations. Suppose further that $x^2_{i j} \neq 0$, $i < j$, in the
forward tube (this will be shown below). Then the functions
\begin{equation}
  c \mapsto x_{i j}^2 \quad (i < j) \label{cxij2}
\end{equation}
are holomorphic functions from the forward tube to $\widetilde{\mathbb{C}
\backslash \{ 0 \}}$, the universal covering of the complex plane minus the
origin. On the other hand $z \mapsto z^h$ is holomorphic from this universal
covering to $\mathbb{C}$. Composing these two holomorphic functions, we conclude
that $(x_{i j}^2)^h$, $i < j$, are holomorphic on the forward tube. Hence this
procedure analytically extends the Euclidean 2-point and 3-point functions to the
whole forward tube $\mathcal{T}_n$ $(n = 2, 3)$.

We will now give a simple lemma which proves that indeed $x^2_{i j} \neq 0$,
$i < j$, in the forward tube. Actually the lemma says that a bit more is true,
namely $x^2_{i j} \in \mathbb{C} \backslash (- \infty, 0]$. This has the
following practical consequence. In general, to compute the analytic
continuation of $(F (c))^h$, where $F (c)$ is a nonzero holomorphic function, we
need to know the phase of $F (c)$, i.e.\ to which sheet of the Riemann surface
$\widetilde{\mathbb{C} \backslash \{ 0 \}}$ it belongs. To compute the phase
we need to connect $c$ to some $c_E$ by a curve and analytically continue along
this curve, following the phase. However, this is unnecessary for $(x_{i
j}^2)^h$. Indeed, by the lemma below $x_{i j}^2$ always belongs to the
principal sheet. So there is no need to use a curve to compute the phase: it
can be computed unambiguously just by plugging the coordinates into
{\eqref{xij-cont}}.

\begin{lemma}
  \label{xij2h}Let $y = (y^0, \mathbf{y}) \in \mathbb{C}^d$ satisfy $\tmop{Re}
  y^0 > | \tmop{Im} \mathbf{y} |$. Then
  \begin{equation}
    y^2 \equiv (y^0)^2 +\mathbf{y}^2 \in \mathbb{C} \backslash (- \infty, 0] .
    \text{} \label{xij2}
  \end{equation}
\end{lemma}

\begin{proof}
  We will denote by Greek letters $\xi, \eta$, etc., vectors of Minkowski space
  $\mathbb{R}^{1, d - 1}$ with the Minkowski inner product $\xi^2 = -
  (\xi^0)^2 + \tmmathbf{\xi}^2$.\footnote{\label{signature}We remind the readers that everywhere in this
  paper we are using $-, +, \ldots, +$ Minkowski signature.} Decomposing the
  vector $(i y^0, \mathbf{y})$ into its real and imaginary parts:
  \begin{equation}
    (i y^0, \mathbf{y}) = \xi + i \eta,
  \end{equation}
  condition $\tmop{Re} y^0 > | \tmop{Im} \mathbf{y} |$ means that $\eta^0 > 0$
  and $- \eta^2 > 0$, i.e.\ $\eta$ is in the open forward light cone, which we
  will denote by $\eta \succ 0$. In this notation, we have to prove that
  \begin{equation}
    (\xi + i \eta)^2 = \xi^2 - \eta^2 + 2 i (\xi \eta) \nin \mathbb{C}
    \backslash (- \infty, 0] . \label{toprove}
  \end{equation}
  where by our conventions all inner products involving $\xi, \eta$ are
  Minkowski. Suppose this is violated, i.e.
  \begin{equation}
    (\xi \eta) = 0, \quad \xi^2 - \eta^2 < 0, \label{incomp}
  \end{equation}
  for some $\xi, \eta$. Since $\eta$ is timelike, $(\xi \eta) = 0$ implies
  that $\xi$ is spacelike. But then $\xi^2 - \eta^2 = \xi^2 + (- \eta^2) > 0$.
  Thus the two conditions in {\eqref{incomp}} cannot both be true, and
  {\eqref{toprove}} is proved.
\end{proof}

As the next step of implementing the strategy from Sec.\ \ref{strategy}, we
need to check that the constructed analytic continuations satisfy a powerlaw
bound so that we can apply Theorem \ref{ThVlad}. Although we already
constructed analytic continuation to the whole $\mathcal{T}_n$, we only need
to check the bound on $\mathcal{D}_n$ which is somewhat easier. The powerlaw
bound follows from the following lemma.

\begin{lemma}
  \label{x2bnd}
  (a) Let $y = (\varepsilon + i s, \mathbf{y})$, $\varepsilon, s
  \in \mathbb{R}$, $\mathbf{y} \in \mathbb{R}^{d - 1}$. Then $y^2$ is bounded
  above and below in the absolute value, as follows:
  \begin{equation}
    \varepsilon^2 \leqslant | y^2 | \leqslant | y |^2 \equiv | \varepsilon + i
    s |^2 +\mathbf{y}^2 ;
  \end{equation}
  (b) On $\mathcal{D}_n$ ($n = 2, 3$), each $1 / (x^2_{i j})^h$ factor in
  {\eqref{G2E}}, {\eqref{G3E}} ($h \in \mathbb{R}$) satisfies a powerlaw
  bound:
  \begin{equation}
    \Bigl| \frac{1}{(x^2_{i j})^h} \Bigr| \leqslant \frac{| x_i - x_j
    |^B}{(\epsilon_i - \epsilon_j)^A} \qquad (i > j),
  \end{equation}
  where $A = 2h$, $B = 0$ for $h$ positive and $A = 0$, $B = - 2h$ for $h$
  negative.
\end{lemma}

\begin{proof}
  (a) The upper bound is obvious. Let us show the lower bound by an explicit
  computation (see Lemma \ref{zeta2bnd} below for an alternative proof). We
  have:
  \begin{equation}
    | y^2 |^2 \equiv | (\varepsilon + i s)^2 +\mathbf{y}^2 |^2 =
    (\varepsilon^2 - s^2 +\mathbf{y}^2)^2 + 4 \varepsilon^2 s^2,
  \end{equation}
  Minimizing this in $\mathbf{y}$, we get
  \[ \min_{\mathbf{y}} | y^2 |^2 = \left\{\begin{array}{l}
       (\varepsilon^2 - s^2)^2 + 4 \varepsilon^2 s^2 = (\varepsilon^2 +
       s^2)^2, \qquad | s | \leqslant \varepsilon,\\
       4 \varepsilon^2 s^2, \qquad | s | \geqslant \varepsilon .
     \end{array}\right. \]
  Minimizing this next in $s$, we find that the absolute minimum is located at
  $\mathbf{y}= 0$,  $s = 0$, and is equal to $\varepsilon^4$. Part (b) follows
  from (a).
\end{proof}

Now that we have the powerlaw bound, we can apply Theorem \ref{ThVlad}. We
conclude that the Minkowski 2-point and 3-point functions, defined as $\epsilon_i
\rightarrow 0$ limits of the analytically continued Euclidean correlators,
exist, are Lorentz-invariant tempered distributions, and satisfy the spectral
condition.\footnote{\label{noteMarc1}Since these are tempered distributions,
their Fourier transforms are well defined. Explicit expressions for these
Fourier transforms are known in many cases. See {\cite{Gillioz:2018mto}} for
$\langle \mathcal{O}_{\Delta, l} (p) \mathcal{O}_{\Delta, l} (- p) \rangle$
and {\cite{Gillioz:2019lgs}} for $\langle \mathcal{O}_{\Delta_1} (p_1)
\mathcal{O}_{\Delta_2} (p_2) \mathcal{O}_{\Delta, l} (p_3) \rangle$.}

\subsection{Comparison with the $i
\varepsilon$-prescription}\label{comparison}

Here we will comment on the ``$i \varepsilon$-prescription'' often used in the
literature to define Minkowski 2-point and 3-point correlators, and how it compares
with our definition. We will focus on the 2-point case for definiteness (same
remarks hold for the 3-point case).

The $i \varepsilon$-prescription defines the Minkowski 2-point correlator $G_2^M
(x_1^M, x_2^M)$ as
\begin{equation}
  \frac{1}{(- (s - i \varepsilon)^2 +\mathbf{y}^2)^{\De}}, \label{GMnaive}
\end{equation}
with $s = t_1 - t_2$, $\mathbf{y}=\mathbf{x}_1 -\mathbf{x}_2$ and taking the
$\varepsilon \rightarrow 0^+$ limit. The precise meaning of the limit is often
left implicit in the physics literature. Away from the light cone the 2-point correlator is
an ordinary function, the limit can be understood pointwise and it agrees with
our definition. Clearly, on the light cone the limit must be understood in
distributional sense, integrating against a test function $f (s, \mathbf{y})$. {That is what we showed above: Vladimirov's theorem guarantees that the limit $\e\to 0$ exists as a tempered distribution and can be therefore integrated against any Schwartz test function. In physics literature, one instead often hears that such integrals should be defined by ``shifting the
integration contour''. Note however that this alternative way of understanding the $\e\to 0$ limit would only work for analytic test functions.} Let us discuss the consequences of this limitation.

It is helpful to recall that the theory of distributions commonly uses three
classes of test functions, denoted $\mathcal{S}, \mathcal{K}, \mathcal{Z}$
{\cite{gelfandshilov}}. Here $\mathcal{S}$ is the space of Schwartz functions,
$\mathcal{K}$ (denoted sometimes by $\mathcal{D}$) is the space of compactly
supported $C^{\infty}$ functions, and $\mathcal{Z}$ consists of entire
holomorphic functions decreasing faster than any power in the real directions
and bounded by some fixed exponential in the imaginary directions. Note that
$\mathcal{K}, \mathcal{Z} \subset \mathcal{S}$. The corresponding distribution
spaces thus satisfy the opposite inclusion: $\mathcal{S}' \subset
\mathcal{K}', \mathcal{Z}'$. The elements of $\mathcal{S}'$ are precisely the
tempered distributions discussed above, $\mathcal{K}'$ are distributions on
the compactly supported test functions\footnote{They are briefly mentioned in
the proof of Theorem \ref{ThVlad}, App.\ \ref{Proof1}, after Eq.
{\eqref{gge}}.}, while $\mathcal{Z}'$ is yet another distributional class.

Importantly, the Fourier transform $\mathcal{F}$ leaves $\mathcal{S}$
invariant. Since the Fourier transform is defined in the theory of
distributions by duality, we also have $\mathcal{F} (\mathcal{S}')
=\mathcal{S}'$: the Fourier transform of a tempered distribution is also a
tempered distribution. On the other hand, one can show (see
{\cite{gelfandshilov}}) that $\mathcal{F} (\mathcal{K}) =\mathcal{Z}$. This is
the rationale behind introducing the space $\mathcal{Z}$, and this also
implies that $\mathcal{F}$ maps $\mathcal{K}'$ to $\mathcal{Z}'$ and vice
versa.

Coming back to {\eqref{GMnaive}}, shifting the integration contour defines
this distribution as an element of $\mathcal{Z}'$. The pairing with a test
function $f \in \mathcal{Z}$ is thus defined by
\begin{equation}
  \int_C d z\, \int d\mathbf{y}\, \frac{1}{(- z^2 +\mathbf{y}^2)^{\De}} f (z,
  \mathbf{y})
\end{equation}
with the contour $C$ running parallel to the real axis in the lower half
plane.\footnote{We can somewhat relax the condition $f \in \mathcal{Z}$. At
the very least, $f$ must be holomorphic in the lower half-plane close to the real
axis and decrease sufficiently fast at infinity for the integral to be
convergent.} \ By the previous paragraph, this is then sufficient to define
the Fourier transform of the 2-point function as an element of $\mathcal{K}'$. By
moving the contour far away from the real axis, one shows that the Fourier
transform vanishes for negative energies, and by Lorentz invariance one
concludes that the support must belong to the forward light cone. These
arguments have parallels in the proof of Part 2 of Theorem \ref{ThVlad} (see
App.\ \ref{Proof1}).

Compared to this simple and almost elementary discussion, Theorem
\ref{ThVlad} proves a stronger statement that the 2-point distribution
{\eqref{GMnaive}} can be extended to test functions of Schwartz class and,
furthermore, to functions which have only a finite number of derivatives as
expressed by Eq.\ {\eqref{GMcont}}. This can be seen as a finer
characterization of the singularity structure at short distances. The Fourier
transform is then also a tempered distribution, thus bounded by some power,
which is a stronger statement than it being an element of $\mathcal{K}'$ since
those can grow arbitrarily fast at infinity.

Since the 2-point and 3-point correlators are known in closed form, one can in
principle verify that their Fourier transform does not grow too fast at
infinity by an explicit computation. This would provide an alternative proof
of temperedness. Our point here is that Theorem \ref{ThVlad} reaches this
conclusion without any computations. For the 4-point correlators considered below,
the Fourier transform cannot be evaluated easily, and Theorem \ref{ThVlad}
appears to be the only realistic way to show temperedness.

It is instructive to discuss why we insist so much on temperedness. In other
words, why Wightman axioms require that the Minkowski $n$-point correlators
must be tempered distributions, and not of some other class? There is a simple
reason why temperedness is a natural requirement, while $\mathcal{K}'$ or
$\mathcal{Z}'$ would not suffice. The point is that Wightman axioms include
\tmtextbf{both} commutativity at spacelike separation and the spectral
condition (the Fourier transform supported in the forward tube). Both these
conditions need compactly supported test functions: the former in position
space, the latter in momentum space. The space $\mathcal{S}$ is large enough
to write both these conditions, while $\mathcal{K}$ or $\mathcal{Z}$ are
inadequate as we would lose one of them.\footnote{For completeness it should
be noted that one can reduce $\mathcal{S}$ a bit and still be able to
formulate both these axioms, as for Jaffe fields {\cite{Jaffe}}, which may
have stronger-than-powerlaw singularities at short distances. For CFTs and for
any theory which asymptotes to a CFT at short distances, there is no reasons
to consider such modifications, and $\mathcal{S}$ remains the natural choice.}

Finally, sometimes by the $i \varepsilon$-prescription one
means the following simplified form of {\eqref{GMnaive}}:
\begin{eqnarray}
  &  & \frac{1}{(- s^2 + i 0^+ \tmop{sign} (s) +\mathbf{y}^2)^{\De}}, 
  \label{Expr1}
\end{eqnarray}
which agrees with {\eqref{GMnaive}} away from the light cone. {By Vladimirov's theorem, this defines a distribution for $s>0$ (including the light cone) and another distribution for $s<0$, but it is not an adequate starting point for defining the distribution around $(s,\mathbf{y})=(0,0)$.} 

\section{Scalar 4-point function}\label{sec:4-point}

This section is the heart of our paper. In it we will show how to define
Minkowski 4-point functions starting from Euclidean 4-point functions of a unitary
CFT. \ We will follow the strategy of Sec.\ \ref{strategy} and in particular
will rely on Theorem \ref{ThVlad}. To avoid inessential details, we will focus
on the case of four identical scalars. Non-identical scalars can be treated by
the same argument (see Sec.\ \ref{nonId}). Additional complications arise
for spinning operators; this case is postponed to a future publication
{\cite{paper2a}}.

So, we consider the Euclidean CFT 4-point function of four identical scalar
Hermitian primaries, which by conformal invariance can be written as:
\begin{equation}
  G^E_4 (c_E) \equiv \langle
  \mathcal{O}(x_1)\mathcal{O}(x_2)\mathcal{O}(x_3)\mathcal{O}(x_4) \rangle =
  \frac{1}{(x_{12}^2 x_{34}^2)^{\Delta_{\mathcal{O}}}} g (c_E) .
  \label{def:Euclidean4-point}
\end{equation}
Here $c_E = (x_1, x_2, x_3, x_4)$ denotes an ordered configuration of four
Euclidean non-coincident points ($x_k \in \mathbb{R}^d, \hspace{1em} x_i \neq
x_j$), $\Delta_{\mathcal{O}}$ is the scaling dimension of $\mathcal{O}$, and
$g (c_E)$ is a real function which depends only on the conformal class of
$c_E$. It can be written as a function of two conformally invariant
cross-ratios $u, v$:
\begin{equation}
  g (c_E) = g (u, v), \qquad u = \dfrac{x_{12}^2 x_{34}^2}{x_{13}^2 x_{24}^2},
  \quad v = \dfrac{x_{14}^2 x_{23}^2}{x_{13}^2 x_{24}^2} . \label{uv}
\end{equation}
Our plan is as follows. After a discussion of the basic issues involved in the
analytic continuation of the 4-point function (Sec.\ \ref{sec:informal}), we
will introduce a representation in terms of the radial coordinates $\rho,
\bar{\rho}$ (Sec.\ \ref{Eucl4-point}), and use it to construct the analytic
continuation to the whole forward tube $\mathcal{T}_4$ (Sec.\ \ref{anal4-point}). This construction works because $\rho, \bar{\rho}$ remain
strictly inside the unit disc everywhere in the forward tube (Lemma
\ref{bound} and Eq.\ {\eqref{3b}}), a fundamental fact proved in Sec.\ \ref{PetrProof}. We then briefly review the well-established powerlaw bound on
$g (\rho, \bar{\rho})$ with respect to $\rho, \bar{\rho}$, and prove a
powerlaw bound on $| \rho (c) |, | \bar{\rho} (c) |$ with respect to $c \in
\mathcal{T}_4$. Combining these powerlaw bounds together, we will get a
powerlaw bound on the analytically continued 4-point function $G_4 (c)$, which by
Theorem \ref{ThVlad} implies (as $c$ approaches the Minkowski region) the
existence of the boundary limit of $G_4 (c)$ as a tempered distribution
(Sec.\ \ref{power4-point}). After establishing temperedness, we will derive the
Minkowski conformal invariance (Sec.\ \ref{ConfMink}), Wightman positivity
(Sec.\ \ref{sec:Wpos}), Wightman clustering (Sec.\ \ref{clusterWightman})
and local commutativity (Sec.\ \ref{local-comm}). Some of them do not rely
on conformal properties: for these we will use the standard arguments given by
Osterwalder and Schrader {\cite{osterwalder1973}}. In Sec.\ \ref{nonId}, we
will generalize the above results to non-identical scalars by using
Cauchy-Schwarz arguments.

\subsection{Informal introduction to basic issues}\label{sec:informal}

Here we wish to outline a few basic difficulties which must be dealt with when
analytically continuing the 4-point function. We will be using $u, v$ coordinates as
an example, although we will see below that other coordinates will be more
suitable for our task. Readers uninterested in philosophical discussions may
skip directly to Sec.\ \ref{Eucl4-point}.

\begin{figure}[h]\centering
\includegraphics[width=0.5\textwidth]{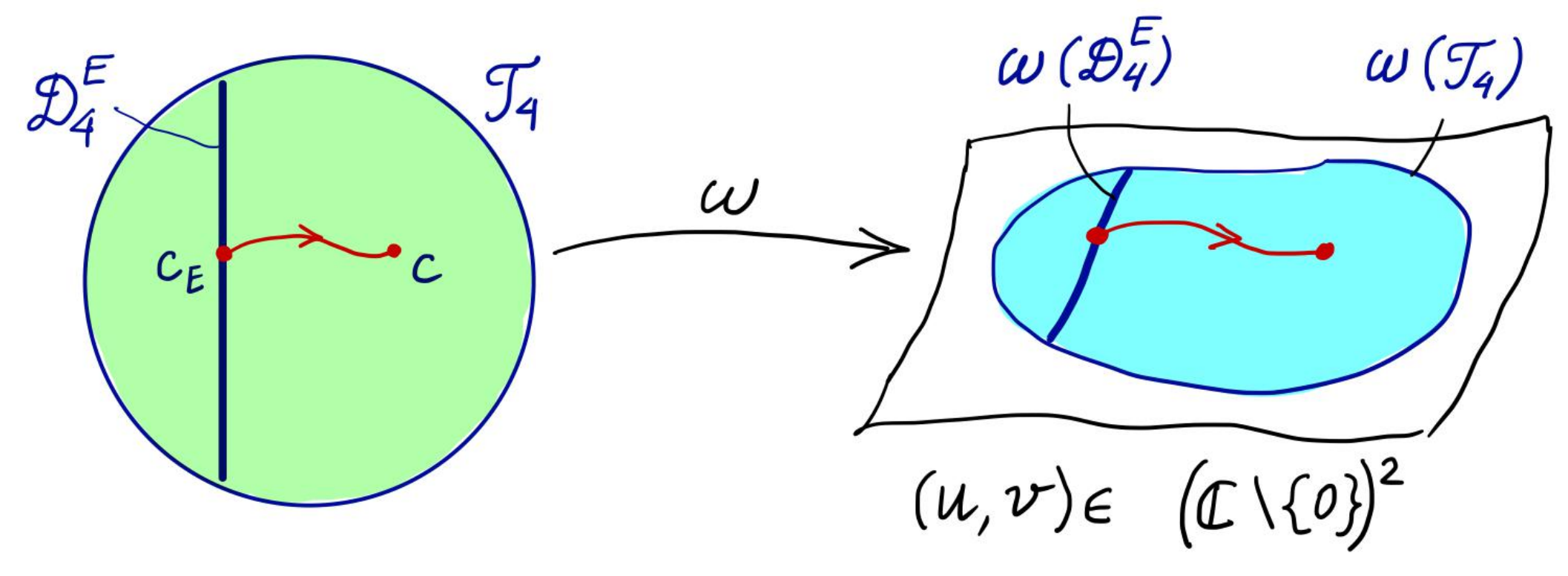}
  \caption{\label{fig:uv}Illustration of the discussion in Sec.\ \ref{sec:informal}.}
\end{figure}

Given any point $c$ of the forward tube, we can connect it to a Euclidean
point $c_E$ by a curve, and analytically continue the 4-point function along the
curve (see Fig.\ \ref{fig:uv}, left). The forward tube being simply connected,
the analytic continuation (if it exists) does not depend on the curve.
Furthermore, let us take into account that our conformal 4-point function
effectively only depends on two variables $u, v$. Applying Lemma \ref{xij2h},
we see that $u, v$ are both nonzero holomorphic functions on the forward tube.
Consider the map:
\begin{equation}
  \omega : c \rightarrow (u, v) \label{omegauv} .
\end{equation}
Since the forward tube is simply connected, we can consider this map as acting
from $\mathcal{T}_4$ to $(\widetilde{\mathbb{C} \backslash \{ 0 \}})^2$, where
tilde denotes the universal cover. Denote by $\omega (\mathcal{T}_4)$ and
$\omega (\mathcal{D}^E_4)$ the images of the forward tube and of its Euclidean
part under this map (see Fig.\ \ref{fig:uv}, right).

Now suppose that we found an analytic continuation of $g (u, v)$ from $\omega
(\mathcal{D}^E_4)$, where it is initially defined, to the whole of $\omega
(\mathcal{T}_4)$. Then we could immediately write down the analytic
continuation of the 4-point function to the forward tube as follows:\footnote{The
prefactor analytically continues just as the 2-point and 3-point functions in Sec.\ \ref{23-point}.}
\begin{equation}
  G_4 (c) = \frac{1}{(x_{12}^2 x_{34}^2)^{\Delta_{\mathcal{O}}}} g (u (c), v
  (c)) . \label{simple}
\end{equation}
This formula defines the analytic continuation to the forward tube as a
composition of two holomorphic functions:
\begin{equation}
  \mathcal{T}_4 \xrightarrow{\omega} \omega (\mathcal{T}_4)
  \xrightarrow{g} \mathbb{C}.
\end{equation}
We would like to use this strategy, but its direct implementation is hindered
by a couple of difficulties:
\begin{itemize}
  \item We don't know much about the shape or even topology of $\omega
  (\mathcal{T}_4)$. E.g.\ we don't know if this set is simply connected. The
  continuous image of a simply connected set, such as the forward tube, does
  not have to be simply connected (Fig.\ \ref{fig:simply}). If $\omega
  (\mathcal{T}_4)$ is not simply connected, there is no guarantee that $g (u,
  v)$ will be single-valued on it. And if $g (u, v)$ has branch cuts, then a
  simple formula like {\eqref{simple}} using only the endpoint values $(u (c),
  v (c))$ will not work; we will need to know in addition ``from which side of
  the cut'' we got to this point along the analytic continuation contour (see
  Fig.\ \ref{fig:simply}).
  
  \item To be sure, we don't know if the above difficulty is actually
  realized. Perhaps the set $\omega (\mathcal{T}_4)$ is, after all, simply
  connected, and $g (u, v)$ has a single-valued analytic continuation to it.
  Even if this is the case, how can we construct this extension starting from
  $g (u, v)$ in the Euclidean region?
\end{itemize}
\begin{figure}[h]\centering
  \includegraphics[width=0.5\textwidth]{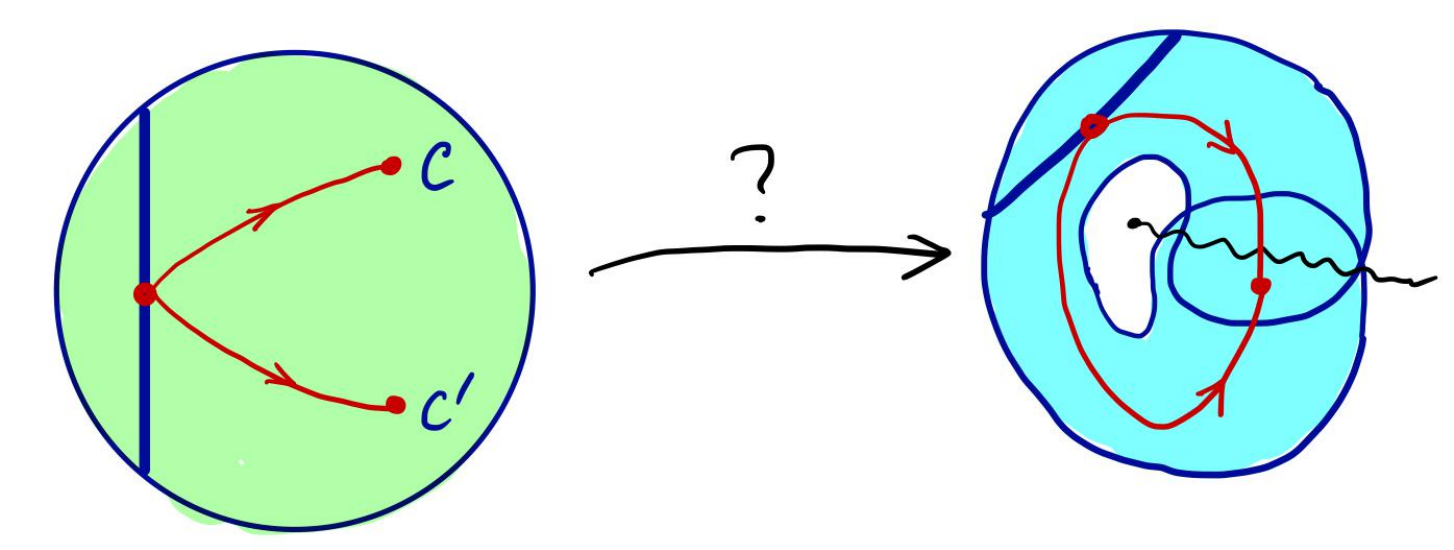}
  \caption{\label{fig:simply}Illustration of a potential difficulty if the set
  $\omega (\mathcal{T}_4)$ were not simply connected (see Sec.\ \ref{sec:informal}).}
\end{figure}

In this paper we will circumvent these difficulties rather than attacking them
head-on. In the Euclidean region, one often uses different variables to
parametrize the cross-ratios $u, v$, such as the Dolan-Osborn variables $z,
\bar{z}$, or the radial variables $\rho, \bar{\rho}$. As one can imagine, a
smart choice of Euclidean variables can greatly simplify the analytic
continuation. We will see that the radial variables are ideally suited for
this task, allowing a natural resolution of the above-mentioned difficulties.

\subsection{Euclidean 4-point function in radial coordinates}\label{Eucl4-point}

We first recall the well-known Dolan-Osborn variables $z, \bar{z}$
{\cite{Dolan:2000ut,Dolan:2003hv}}, which are two complex variables related to
$u, v$ by
\begin{gather}
u = z \bar{z}, \quad v = (1 - z)  (1 - \bar{z}), \qquad\infixor 
  \label{zzbar}\\
z, \bar{z} = \frac{1}{2} \left( 1 + u - v \pm \sqrt{(1 + u - v)^2 - 4
  u} \right) .  \label{zzbarsolved}
\end{gather}
Since in the Euclidean case we only consider non-coincident points, we have $u, v \neq 0$, and hence
$z, \bar{z} \neq 0, 1$. It is possible to fix a Euclidean conformal frame by
setting the four points to positions
\begin{equation}
  x_1 = 0, \quad x_2 = a \hat{e}_0 + b \hat{e}_1, \qquad x_3 = \hat{e}_0,
  \qquad x_4 = \infty \hat{e}_0, \label{frameE}
\end{equation}
where $\hat{e}_{\mu}$ is the standard orthonormal basis of $\mathbb{R}^d$.
Using this frame, we obtain $z, \bar{z} = a \pm i b$. This shows that in the
Euclidean, the variables $z, \bar{z}$ are complex-conjugate $(\bar{z} =
z^{\ast})$.

Euclidean configurations with real $z = \bar{z}$ correspond to four points
lying on a circle, which maps in the frame {\eqref{frameE}} to four points on
a line. The three possibilities $z < 0$, $z \in (0, 1)$, $z \in (1, + \infty)$
are then realized for different cycling orderings. \

The radial variables $\rho, \bar{\rho} \in \mathbb{C}$
{\cite{Pappadopulo:2012jk,Hogervorst:2013sma}} are defined in terms of the
Dolan-Osborn variables by the formula:
\begin{equation}
  \label{def:rho} \rho = f (z), \quad \bar{\rho} = f (\bar{z}), \quad f (w) :
  = \dfrac{w}{(1 + \sqrt{1 - w})^2} .
\end{equation}
The function $f (w)$ in this definition\footnote{The definition assumes the
standard branch of the square root function.} is the uniformization map for
the complex plane minus the cut $(1, + \infty),$ i.e.\ it is a one-to-one map
of $\mathbb{C} \backslash [1, + \infty)$ onto the unit disk. Eq.
{\eqref{def:rho}} thus associates with any Euclidean configuration a pair of
complex conjugate $\rho, \bar{\rho}$ $(\bar{\rho} = \rho^{\ast})$ belonging to
the unit disk: $| \rho | \leqslant 1$. Moreover we have $| \rho | < 1$ except
for the Euclidean configurations with $z = \bar{z} \in (1, + \infty)$. As
explained above, this happens when four points lie on a circle in the cyclic
order $1324$. For such exceptional configurations one may define $\rho,
\bar{\rho}$ by continuity so that $| \rho | = 1$, $\bar{\rho} = \rho^{\ast}$.

The meaning of the coordinate $\rho$ is clarified by mapping the 4-point
configuration to a conformal frame (compare {\eqref{frameE}})
\begin{equation}
  x_1 = - \alpha \hat{e}_0 - \beta \hat{e}_1, \quad x_2 = \alpha \hat{e}_0 +
  \beta \hat{e}_1, \qquad x_3 = \hat{e}_0, \qquad x_4 = - \hat{e}_0,
  \label{frameErho}
\end{equation}
Using this frame, we obtain $\rho, \bar{\rho} = \alpha \pm i \beta$.

There is a small difference between $d = 2$ and $d \geqslant 3$ dimensions.
In $d \geqslant 3$, conformal frames {\eqref{frameE}} and {\eqref{frameErho}}
are unique only up to a sign of $b$ and $\beta$ (flipped rotating by $\pi$ in
the $12$ plane), which implies that pairs $(z, \bar{z})$ and $(\rho,
\bar{\rho})$ are defined only up to permutation. On the other hand in $d = 2$
flipping the sign of $b$ or $\beta$ is a parity transformation, which is not
in the identity component of the conformal group. Hence the conformal frames
are unique and $z, \bar{z}$ as well as $\rho, \bar{\rho}$ are individually
meaningful.

In a unitary Euclidean CFT, the 4-point function admits a power-series
expansion in the $\rho$ coordinate, absolutely convergent whenever $| \rho | <
1$ {\cite{Pappadopulo:2012jk}}. Specifically, \tmtextit{the function $g
(c_E)$ appearing in the 4-point function {\eqref{def:Euclidean4-point}} of four
identical scalar Hermitean primaries has a series expansion of the form}
\begin{equation}
  g (c_E) = \sum_{\delta, m} p_{\delta, m} r^{\delta} e^{i m \theta},
  \label{g:rhoexpansion}
\end{equation}
\tmtextit{where the sum runs over a discrete set of pairs $(\delta, m)$ with
$\delta \geqslant 0$, $m \in 2\mathbb{Z}$, and the variables $r$, $\theta \in
\mathbb{R}$ are the modulus and the phase of $\rho (c_E) = r e^{i \theta}$.
The sum is absolutely convergent when $r = | \rho (c_E) | < 1$. In addition,
we know that $| m | \leqslant \delta$ and $p_{\delta, m} \geqslant 0$ for all
terms in {\eqref{g:rhoexpansion}}. Finally, when $d \geqslant 3$ we have
$p_{\delta, - m} = p_{\delta, m}$, so that the r.h.s.\ of
{\eqref{g:rhoexpansion}} is uniquely defined in spite of $\rho (c_E)$ being
defined only up to complex conjugation.}

The readers familiar with this fact may skip to Sec.\ \ref{anal4-point} where we
will use it to perform analytic continuation. In the rest of this section we
recall how it follows from the CFT axioms
{\cite{Pappadopulo:2012jk,Fitzpatrick:2012yx}}.

We consider the 4-point function in the conformal frame configuration
{\eqref{frameErho}} and write it as the inner product of two states created by
the operators outside and inside a unit sphere $S$ centered at the origin:
\begin{equation}
  \langle \mathcal{O} (1, 0, \tmmathbf{0}) \mathcal{O} (- 1, 0, \tmmathbf{0}
  \tmmathbf{}) | \mathcal{O} (\alpha, \beta, \tmmathbf{0}) \mathcal{O} (-
  \alpha, - \beta, \tmmathbf{0}) \rangle \label{frame0}
\end{equation}
We can find a conformal transformation which maps the sphere $S$ to $x^0 = 0$
plane, its center 0 to $x_S$ and the infinity to $x_N$. This is the setup in
which we developed the CFT Hilbert space picture in Sec.\ \ref{Hilbert}.
Applying the inverse transformation, we are allowed to use the Hilbert space
language in the frame {\eqref{frame0}}, which is the familiar setting of
radial quantization. We decompose the radial quantization Hilbert space,
produced by local operators inserted at the origin, in orthonormalized
eigenstates $| \delta, m \rangle$ of the dilatation $D$ and the planar
rotation $M_{01}$. The ket state is expanded in this basis as
\begin{equation}
  \nobracket | \mathcal{O} (\alpha, \beta, \tmmathbf{0}) \mathcal{O} (-
  \alpha, - \beta, \tmmathbf{0}) \rangle = \sum_{\delta, m} c_{\delta, m}
  r^{\delta - 2 \Delta_{\varphi}} e^{i m \theta} | \delta, m \rangle .
  \label{ketexpr}
\end{equation}
The dependence of the expansion coefficients in this formula on $r$ and
$\theta$ is fixed by knowing how the state in the l.h.s.\ transforms under
rotations and dilatations. The transformation $\theta \rightarrow \theta +
\pi$ swaps the two operators leaving the state invariant for the considered
case of identical operators. Hence the state in the r.h.s.\ also must remain
invariant, proving that $m$ must be even.

Setting $r = 1, \theta = 0$ in {\eqref{ketexpr}}, we get
\begin{equation}
  \nobracket | \mathcal{O} (1, 0, \tmmathbf{0}) \mathcal{O} (- 1, 0,
  \tmmathbf{0}) \rangle = \sum_{\delta, m} c_{\delta, m} | \delta, m \rangle .
  \label{ketexpr1}
\end{equation}
In the considered frame the OS reflection is the inversion with respect to the
sphere $S$: $x^{\mu} \rightarrow x^{\mu} / x^2 .$ In particular, this leaves
$x_3$ and $x_4$ invariant. Applying this transformation to {\eqref{ketexpr1}},
we get
\begin{equation}
  \langle \mathcal{O} (1, 0, \tmmathbf{0}) \mathcal{O} (- 1, 0, \tmmathbf{0}
  \tmmathbf{}) | = \sum_{\delta, m} c^{\ast}_{\delta, m} \langle \delta, m | .
  \label{braexpr}
\end{equation}
Taking the inner product of {\eqref{ketexpr}} and {\eqref{braexpr}}, we get
\begin{equation}
  \langle \mathcal{O} (1, 0, \tmmathbf{0}) \mathcal{O} (- 1, 0, \tmmathbf{0}
  \tmmathbf{}) | \mathcal{O} (\alpha, \beta, \tmmathbf{0}) \mathcal{O} (-
  \alpha, - \beta, \tmmathbf{0}) \rangle = \sum_{\delta, m} | c_{\delta, m}
  |^2 r^{\delta - 2 \Delta_{\varphi}} e^{i m \theta} .
\end{equation}
Comparing this with Eq.\ {\eqref{def:Euclidean4-point}}, and using that $x_{12}^2 =
4 r^2$, $x_{34}^2 = 4$ in the considered conformal frame, we obtain
{\eqref{g:rhoexpansion}} with $p_{\delta, m} = 16^{\Delta_{\mathcal{O}}} |
c_{\delta, m} |^2 \geqslant 0$.

In the above argument we chose for simplicity the sphere of radius 1, but any
sphere of radius $r < r_0 < 1$ would work equally well and give rise to the
same expression. Absolute convergence for $r < 1$ follows, because both the
bra and the ket states are normalizable for such $r_0$ (while for $r_0 = 1$ as
above the bra state $\langle \varphi (x_3) \varphi (x_4 \tmmathbf{}) |$ is not
normalizable).

The restriction $| m | \leqslant \delta$ follows from the 2d unitarity bounds.
The 2d unitarity bound applies, as any $d$-dimensional CFT restricted to a
plane can be seen as a unitary 2d CFT. For 2d primaries of spin $J$ and
dimension $\Delta$, the 2d unitarity bound says $| J | \leqslant \Delta$. The
descendants at level $n \in \mathbb{Z}_{\geqslant 0}$ have $\delta = \Delta +
n,$ $| m - J | \leqslant n$, hence $| m | \leqslant \delta$ follows.

Finally, let us prove that $p_{\delta, m} = p_{\delta, - m}$ in $d \geqslant
3$. We consider Eq.\ {\eqref{ketexpr1}} and perform a $\pi$ rotation in the 12
plane. In the r.h.s. $| \delta, m \rangle \rightarrow | \delta, - m \rangle$
because $M_{01} \rightarrow - M_{01}$ under such a rotation. On the other hand
the l.h.s.\ does not change. This implies that we must have $c_{\delta, m} =
c_{\delta, - m}$, and hence $p_{\delta, m} = p_{\delta, - m}$. (In $d = 2$,
these properties also hold under the additional assumption of parity
invariance.)

\subsection{Analytic continuation}\label{anal4-point}

In this section we will construct the analytic continuation of the Euclidean
4-point function {\eqref{def:Euclidean4-point}} to the forward tube
$\mathcal{T}_4$ (recall the forward tube definition {\eqref{forward}}).
Analytic continuation to $\mathcal{D}_4 \subset \mathcal{T}_4$ has already
been given in {\cite{Qiao:2020bcs}}, Sec.\ 3.4, and we will use a somewhat
streamlined version of that construction. We will analytically continue to the
full forward tube $\mathcal{T}_4$, since this does not lead to additional
complications.

The analytic continuation will be given by the formula
\begin{equation}
  G_4 (c) = \frac{1}{(x_{12}^2 x_{34}^2)^{\Delta_{\mathcal{O}}}} g (c), \qquad
  c \in \mathcal{T}_4 . \label{G4c}
\end{equation}
Here the prefactor trivially analytically continues to $\mathcal{T}_4$
similarly to the 2-point and 3-point functions discussed in Sec.\ \ref{23-point}. We
will construct $g (c)$, analytic continuation of $g (c_E)$, starting from Eq.
{\eqref{g:rhoexpansion}}.

First we have to define the variables $z (c)$, $\bar{z} (c)$ on the forward
tube, which is naturally done as follows. Given a configuration $c \in
\mathcal{T}_4$, we evaluate $u = u (c),$ $v = v (c)$ via {\eqref{uv}}. By
Lemma \ref{xij2h}, $u (c)$ and $v (c)$ are nonzero holomorphic functions on the
forward tube. We then define $z (c)$, $\bar{z} (c)$ via {\eqref{zzbarsolved}}:
\begin{equation}
  z (c), \bar{z} (c) = \frac{1}{2} \left( 1 + u (c) - v (c) \pm \sqrt{[1 + u
  (c) - v (c)]^2 - 4 u (c)} \right) . \label{zuv}
\end{equation}
Unlike for Euclidean configurations, for a general configuration $c \in
\mathcal{T}_4$ these are two complex numbers unrelated by conjugation. Since
$u (c)$ and $v (c)$ are nonzero, Eq.\ {\eqref{zzbar}} implies $z (c), \bar{z}
(c) \in \mathbb{C} \backslash \{ 0, 1 \}$.

Since Eq.\ {\eqref{zuv}} only defines $z (c), \bar{z} (c)$ up to permutation,
we view it as a map from the forward tube to $\mathbb{C}^2 /\mathbb{Z}_2$,
the set of unordered pairs of complex numbers. This map is continuous, and is
analytic everywhere except on $\Gamma \subset \mathcal{T}_4$ where the
expression under the square root vanishes:
\begin{equation}
  \Gamma = \{ c \in \mathcal{T}_4 : [1 + u (c) - v (c)]^2 - 4 u (c) = 0 \} .
  \label{Gamma}
\end{equation}
Actually, it turns out that in $d = 2$ one can resolve the ambiguity inherent
in Eq.\ {\eqref{zuv}} and define $z (c)$, $\bar{z} (c)$ as individually
globally holomorphic functions on $\mathcal{T}_4$. We will bring up this fact
below when we need it. Ref.\ {\cite{Qiao:2020bcs}}, App.\ A, showed that such an
improvement is impossible in $d \geqslant 3$.

The following result is fundamental for our construction. The proof is
elementary but a bit tricky and is postponed to Sec.\ \ref{PetrProof}.

\begin{lemma}
  \label{bound}For any $c \in \mathcal{T}_4$ we have $z (c)$, $\bar{z} (c)
  \nin [1, + \infty)$.
\end{lemma}

We next define $\rho (c)$, $\bar{\rho} (c)$ on $\mathcal{T}_4$, via
\begin{equation}
  \rho (c) = f (z (c)), \quad \bar{\rho} (c) = f (\bar{z} (c)), \label{rc}
\end{equation}
where $f$ is the same function as in {\eqref{def:rho}}, mapping $\mathbb{C}
\backslash [1, + \infty)$ onto the unit disk. By Lemma \ref{bound}, we then
have
\begin{equation}
  0 < | \rho (c) |, | \bar{\rho} (c) | < 1 \text{\qquad for any\qquad} c \in
  \mathcal{T}_4 .\footnote{Note that the converse is not true: the region
  in which $0<|\rho|,|\bar\rho|<1$ is larger than the forward tube. For example, it includes
  the extended forward tube (see Sec.~\ref{localCFT}).} \label{3b}
\end{equation}
{Moreover, $\rho(c)$ and $\bar\rho(c)$ are locally holomorphic away from $\G$.}
Because of this, and since Eq.\ {\eqref{g:rhoexpansion}} for the 4-point
function converges in the Euclidean for any $| \rho | < 1$, we may hope to use
Eq.\ {\eqref{g:rhoexpansion}} to analytically continue $g (c)$ to the whole
forward tube. We will now carry out this strategy. Note that some extra care
is needed, because $\rho (c)$ and $\bar{\rho} (c)$ are, just as $z (c)$ and
$\bar{z} (c)$, not globally holomorphic and are defined only up to permutation
(except in $d = 2$, see below), and because {\eqref{g:rhoexpansion}} contains
in general non-integer powers.

To begin with, we rewrite Eq.\ {\eqref{g:rhoexpansion}} equivalently as
\begin{equation}
  \label{g:rhoexpansion2} g (c_E) = \sum_{\delta, 0 \leqslant m \leqslant
  \delta} (\rho \bar{\rho})^{\delta / 2 - m / 2}  (p_{\delta, m} \rho^m +
  p_{\delta, - m} \bar{\rho}^m) \hspace{0.17em},
\end{equation}
Various pieces of this formula need to be analytically continued to the
forward tube. Consider first
\begin{equation}
  R (c) = \rho (c)  \bar{\rho} (c),
\end{equation}
which is a candidate for the analytic continuation of $\rho \bar{\rho}$ from
the Euclidean region. We can view it as a composition of two functions: $c
\mapsto (\rho (c), \bar{\rho} (c))$ which is a continuous function from
$\mathcal{T}_4$ to $\mathbb{C}^2 /\mathbb{Z}_2$ analytic away from $\Gamma$,
followed by $(\rho, \bar{\rho}) \mapsto \rho \bar{\rho}$ which is a continuous
holomorphic function from $\mathbb{C}^2 /\mathbb{Z}_2$ to $\mathbb{C}$. Hence $R
(c)$ is a continuous function on the forward tube, analytic everywhere except
perhaps on $\Gamma$. However, manifold $\Gamma$ has complex codimension one,
and by an analogue of Riemann's theorem about removable singularities we
conclude that $R (c)$ is in fact analytic also on $\Gamma$, and thus on the
whole $\mathcal{T}_4$.\footnote{The precise argument is as follows. Let us
keep all complex coordinates fixed and vary just one, say $x_{1}^{0}$.
There are two cases: either {\eqref{Gamma}} is identically zero as a function
of $x_{1}^{0}$, or it is a nonzero polynomial of $x_{1}^{0}$. In
the first case $R (c)$ is trivially holomorphic in $x_{1}^{0}$. In the
second case {\eqref{Gamma}} vanishes at most for a few isolated values of
$x_1^0$. We can then apply 1d Riemann's theorem to say that $R (c)$ is also
analytic at those isolated points. By these arguments, we conclude that $R
(c)$ is holomorphic in each variable separately. Finally, a continuous function
of several complex variables holomorphic in each variable separately is jointly
holomorphic {\cite{Osgood}}.}

In addition, $R (c)$ is nonzero in the forward tube. Thus we can lift $R (c)$
to a holomorphic function $\tilde{R} (c)$ from the forward tube to the universal
cover {$\widetilde{\mathbb{C} \backslash \{ 0 \}}$}. Composing this function with $z^h :
\widetilde{\mathbb{C} \backslash \{ 0 \}} \rightarrow \mathbb{C}$, we obtain
an analytic continuation of $(\rho \bar{\rho})^h$ for any $h \in \mathbb{R}$,
which we denote by $R_h (c)$. This discussion mirrors the one around Eq.
{\eqref{cxij2}} in Sec.\ \ref{23-point}. However, unlike $x_{i j}^2$ in that
discussion, it is not true that $\tilde{R} (c)$ always belongs to the
principal sheet of {$\widetilde{\mathbb{C} \backslash \{ 0 \}}$}. So, in general, to
compute the phase of the analytically continued function, one should follow
the phase of $\rho \bar{\rho}$ along a curve joining $c_E$ to $c$.

Following a curve is perfectly fine as a theoretical device. For practical
computations of the phase, one may wish to use instead the following trick
which avoids having to look at the curve. (The reader happy to follow the curve may skip the trick and go directly to Eq.\ {\eqref{Phi}}.) Consider the
identity:
\begin{equation}
  \rho \bar{\rho} = \frac{1}{16} u (1 + \rho)^2 (1 + \bar{\rho})^2 =
  \frac{1}{16} \frac{x^2_{12} x^2_{34}}{x^2_{13} x^2_{24}} Y^2, \quad Y = (1 +
  \rho) (1 + \bar{\rho})^{}, \label{rhotrick}
\end{equation}
which follows by using $z = \frac{4 \rho}{(1 + \rho)^2}$, the inverse of the
relation {\eqref{def:rho}} between $\rho$ and $z$, as well as $u = z \bar{z}$
and the expression for $u$. The function $Y (c) = (1 + \rho (c)) (1 +
\bar{\rho} (c))$ is holomorphic on $\mathcal{T}_4$ by the same ``analyticity
on $\mathcal{T}_4 \backslash \Gamma$ plus Riemann's theorem'' argument as used
above for $\rho \bar{\rho}$. In addition, and this is the key point, because
$| \rho (c) |, | \bar{\rho} (c) | < 1$, we know that $Y (c) \in \mathbb{C}
\backslash (- \infty, 0]$. The upshot of the trick is that Eq.
{\eqref{rhotrick}} expresses $\rho \bar{\rho}$ as a product of factors which
all remain on the principal sheet of $z^h$ upon the analytic continuation.
Hence we can compute the analytic continuation of $(\rho \bar{\rho})^h$ by
\begin{equation}
  R_h (c) = \frac{1}{16^h} \frac{(x^2_{12})^h (x^2_{34})^h}{(x^2_{13})^h
  (x^2_{24})^h} Y (c)^{2 h} \qquad (h \in \mathbb{R}), \label{Rh}
\end{equation}
This determines the phase of $R_h (c)$ unambiguously without having to
look at the curve joining $c_E$ to $c$.

Next, we consider for an integer $m$ a function
\begin{equation}
  \Phi_m (c) = \rho (c)^m + \bar{\rho} (c)^m . \label{Phi}
\end{equation}
Just as $\rho \bar{\rho}$ and $Y$, it is continuous on $\mathcal{T}_4$ and
holomorphic on $\mathcal{T}_4 \backslash \Gamma$, and thus holomorphic on the
whole $\mathcal{T}_4$.

We can now define the analytic continuation of {\eqref{g:rhoexpansion2}}.
Consider first $d \geqslant 3$, when $p_{\delta, - m} = p_{\delta, m}$. In
this case the analytic continuation is given by the formula
\begin{equation}
  g (c) = \sum_{m, \delta, 0 \leqslant m \leqslant \delta} p_{\delta, m}
  R_{\delta / 2 - m / 2} (c) \Phi_m (c) . \label{eq:gtilde}
\end{equation}
This series consists of holomorphic functions, and it reduces to
{\eqref{g:rhoexpansion2}} in the Euclidean region. Furthermore, every term in
the series can be bounded in absolute value by:
\begin{eqnarray}
  | p_{\delta, m} R^{\delta / 2 - m / 2} (c) \Phi_m (c)  | & \leqslant &
  p_{\delta, m}  | \rho (c) \bar{\rho} (c) |^{\delta / 2 - m / 2} (| \rho (c)
  |^m + | \bar{\rho} (c) |^m) \nonumber\\
  & \leqslant & p_{\delta, m} r^{\delta - m} (r^m + r^m),  \label{maj}
\end{eqnarray}
where $r = r (c) = \max (| \rho (c) |, | \bar{\rho} (c) |)$, which is $< 1$ by
Eq.\ {\eqref{3b}}. Here we used $p_{\delta, m} \geqslant 0$ in the first line,
and $\delta - m \geqslant 0$ in the second line. The terms in the r.h.s.\ of
{\eqref{maj}} comprise a positive convergent series whose sum is the Euclidean
4-point function {\eqref{g:rhoexpansion2}} evaluated at $\rho = \bar{\rho} = r
(c)$. This proves that {\eqref{eq:gtilde}} converges uniformly on compact
subsets of $\mathcal{T}_4$, and hence defines a holomorphic function in
$\mathcal{T}_4$.

It remains to consider $d = 2$. As anticipated above, in this case the
functions $z (c), \bar{z} (c)$ are individually globally holomorphic on
$\mathcal{T}_4$. This can be seen introducing coordinates (see
{\cite{Qiao:2020bcs}}, Sec.\ 3.5)
\begin{equation}
  z_k = x_k^0 + i x_k^1, \qquad \bar{z}_k = x_k^0 - i x_k^1, \qquad k = 1, 2,
  3, 4.
\end{equation}
Then the explicit formulas for $z (c), \bar{z} (c)$ are given by:
\begin{equation}
  z (c) = \frac{(z_1 - z_2) (z_3 - z_4)}{(z_1 - z_3) (z_2 - z_4)}, \qquad
  \bar{z} (c) = \frac{(\bar{z}_1 - \bar{z}_2) (\bar{z}_3 -
  \bar{z}_4)}{(\bar{z}_1 - z_3) (\bar{z}_2 - \bar{z}_4)} . \label{zzbarglobal}
\end{equation}
The functions $\rho (c), \bar{\rho} (c)$ defined by {\eqref{rc}} are also
individually globally holomorphic on $\mathcal{T}_4$. As a consequence, the
functions $\rho (c)^m$ and $\bar{\rho} (c)^m$ are individually holomorphic in $d
= 2$, and not just their sum {\eqref{Phi}}. We can therefore define the
analytic continuation of $g (c)$ by the formula (compare {\eqref{eq:gtilde}}):
\begin{equation}
  g (c) = \sum_{m, \delta, 0 \leqslant m \leqslant \delta} R_{\delta / 2 -
  m / 2} (c)  [p_{\delta, m} \rho (c)^m + p_{\delta, - m} \bar{\rho} (c)^m]
  . \label{eq:gtilde2d}
\end{equation}
This formula would be appropriate for non-parity invariant 2d CFTs which may
have $p_{\delta, m} \neq p_{\delta, - m}$. Analyticity follows from the
uniform convergence on compact subsets, by the same argument as for $d
\geqslant 3$.

Finally, we wish to explain how the above construction may be translated into
the language of Sec.\ \ref{sec:informal}, to see how the issues raised there
are resolved. This is instructive but not strictly speaking necessary, so we
will be schematic. In the 2d case, when $\rho (c)$, $\bar{\rho} (c)$ are
individually defined, the translation is in terms of the map
\begin{equation}
  \Omega : c \mapsto (\rho (c), \bar{\rho} (c)) \in (\widetilde{\mathbb{D}
  \backslash \{ 0 \}})^2,
\end{equation}
where $\mathbb{D}$ is the open unit disk, and we lifted each of the maps $\rho
(c)$, $\overline{\rho (c)}$ to the universal cover of $\mathbb{D} \backslash
\{ 0 \}$. This map is the present analogue of $\omega$ in {\eqref{omegauv}}.
The function $g (\rho, \bar{\rho})$ extends analytically to the whole
$(\widetilde{\mathbb{D} \backslash \{ 0 \}})^2$, which makes it unnecessary to
understand the precise shape of $\Omega (\mathcal{T}_4)$.

For $d \geqslant 3$, $\rho (c)$, $\bar{\rho} (c)$ are defined only up to
permutation. Translation can then be done in terms of their symmetric
combinations $\rho \bar{\rho}, \rho + \bar{\rho}$. Any symmetric polynomial in
$\rho$, $\bar{\rho}$, such as the r.h.s.\ of {\eqref{Phi}}, can \ be expressed
as a polynomial in these coordinates. Let then $X$ be the image of
$(\mathbb{D} \backslash \{ 0 \})^2$ under the map $(\rho, \bar{\rho}) \mapsto
(\rho \bar{\rho}, \rho + \bar{\rho})$. The following map is holomorphic on
$\mathcal{T}_4$:
\begin{equation}
  \Omega : c \mapsto (\rho (c) \bar{\rho} (c), \rho (c) + \bar{\rho} (c)) \in
  \tilde{X},
\end{equation}
where we lifted to the universal cover. The above argument can be interpreted
as showing that the function $g (\rho, \bar{\rho})$ extends analytically to
the whole $\tilde{X}$. Understanding the precise shape of $\Omega
(\mathcal{T}_4)$ is once again unnecessary.

\subsection{Proof of $z, \bar{z} \nin [1, + \infty)$}\label{PetrProof}

Here we will prove Lemma \ref{bound} which played such a fundamental role in
the previous section. Just as for Lemma \ref{xij2h}, it will be helpful to use
the Minkowski metric. Thus we pass from Euclidean complex coordinates $x_k \in
\mathbb{C}^d$ to Minkowski complex coordinates $\zeta_k = (i x^0_k,
\mathbf{x}_k) \in \mathbb{C}^{1, d - 1}$. Definitions of $u, v$ are then
rewritten equivalently as
\begin{equation}
  u = \frac{\zeta_{12}^2 \zeta_{34}^2}{\zeta_{13}^2 \zeta_{24}^2}, \quad v =
  \frac{\zeta_{23}^2 \zeta_{14}^2}{\zeta_{13}^2 \zeta_{24}^2},
\end{equation}
where $\zeta_{i j} = \zeta_i - \zeta_j$ and $\zeta^2 = - (\zeta^0)^2
+\tmmathbf{\zeta}^2 .$ We denote
\begin{equation}
  \zeta_k = \xi_k + i \eta_k, \quad \xi_k, \eta_k \in \mathbb{R}^{1, d - 1} .
\end{equation}
We will thus use Minkowski norm for $\xi$'s, $\eta$'s and their differences.
The forward tube condition on $x_k$ is rewritten as $\eta_k - \eta_{k + 1}
\succ 0$ which is the notation for
\begin{equation}
  \eta^0_k - \eta^0_{k + 1} > 0 \quad \infixand \quad - (\eta_k - \eta_{k +
  1})^2 > 0. \label{FTc}
\end{equation}
We will need the following lemma which is related to Lemma \ref{xij2h} (see
the proof at the end of the section).

\begin{lemma}
  \label{Petr}Let $\zeta = \xi + i \eta$ and $\eta^2 < 0$. Then
  
  (a) $\zeta^2 \neq 0$;
  
  (b) Define $\zeta' = \xi' + i \eta'$ by
  \begin{equation}
    \zeta' = \zeta / \zeta^2,
  \end{equation}
  which is finite by Part (a). Then $\eta'$ belongs to the same causal part of
  the light cone (future or past) as $\eta$. I.e.\ $\eta \succ 0 \Rightarrow
  \eta' \succ 0$. Analogously, $\eta \prec 0 \Rightarrow \eta' \prec 0$.
\end{lemma}

Let us start the proof of Lemma \ref{bound}. The $z, \bar{z}$ are defined from
$u, v$ via {\eqref{zzbar}}. It is not hard to see from the first line of
{\eqref{zzbar}} that $z, \bar{z}$ are precisely the two solutions of the
quadratic equation
\begin{equation}
  z^2 - (1 + u - v) z + u = 0. \label{quadeq}
\end{equation}
We thus have to show that, assuming {\eqref{FTc}}, this equation has no
solutions which are real and belong to the interval $[1, + \infty)$.

Without loss of generality, we can assume that $\zeta_3 = 0$.\footnote{It is
important to move $\zeta_3$ (or $\zeta_2$) to zero rather than $\zeta_1$ or
$\zeta_4$, because only then, after applying the inversion, one gets causal
information not only on $\eta'_k$'s but also on some of their differences.}
Then we have $\eta_1, \eta_2 \succ 0$ while $\eta_4 \prec 0$. Then we apply
Lemma \ref{Petr} and map the configuration $(\zeta_1, \zeta_2, 0, \zeta_4)$ to the
configuration $(\zeta_1', \zeta_2', \infty, \zeta_4')$ with $\eta_1', \eta_2'
\succ 0$ while $\eta_4' \prec 0$. These relations imply $\eta'_{14} \succ 0,
\eta'_{24} \succ 0$ which will be used below.

Since $u, v$ are invariant under the inversion, we have {(this can be checked by a direct computation)}
\begin{equation}
  u = \frac{(\zeta'_{12})^2}{(\zeta'_{24})^2}, \quad v =
  \frac{(\zeta'_{14})^2}{(\zeta'_{24})^2},
\end{equation}
and Eq.\ {\eqref{quadeq}} reduces to
\begin{equation}
  (\zeta'_{24})^2 z^2 - [(\zeta'_{24})^2 + (\zeta'_{12})^2 -
  (\zeta'_{14})^2] z + (\zeta'_{12})^2 = 0 .
\end{equation}
Using that $\zeta'_{12} = \zeta_{14}' - \zeta_{24}'$, this equation can be
written equivalently as
\begin{equation}
  (\zeta_{14}' + (z - 1) \zeta_{24}')^2 = 0 . \label{equivquad}
\end{equation}
Now let us suppose that $z \in [1, + \infty)$. Then
\begin{equation}
  \tmop{Im} [\zeta_{14}' + (z - 1) \zeta_{24}'] = \eta_{14}' + (z - 1)
  \eta_{24}' \succ 0 .
\end{equation}
Then Eq.\ {\eqref{equivquad}} is in contradiction with Lemma \ref{xij2h}. Lemma
\ref{bound} is demonstrated.

\tmtextbf{Proof of Lemma \ref{Petr}.} This was shown in
{\cite{Kravchuk:2018htv}}, footnote 74, and we reproduce the argument here for
completeness. Part (a) is a partial case of Lemma \ref{xij2h} (for $\eta \prec
0$ we should apply it to the complex conjugate vector $\zeta^{\ast} = \xi - i
\eta$). Let us show Part (b). To show that $\eta \succ 0 \Rightarrow \eta'
\succ 0$, we write
\begin{equation}
  \zeta' = \frac{\xi + i \eta}{\xi^2 - \eta^2 + 2 i (\xi, \eta)} = \frac{(\xi
  + i \eta)  (\xi^2 - \eta^2 - 2 i (\xi, \eta))}{(\xi^2 - \eta^2)^2 + 4 (\xi,
  \eta)^2}  .
\end{equation}
So, up to a positive factor, $\eta'$ is given by
\begin{equation}
  (\xi^2 - \eta^2) \eta - 2 (\xi, \eta) \xi .
\end{equation}
For $\xi = 0$ this is given by $(- \eta^2) \eta \succ 0$. More generally, this
squares to
\begin{equation}
  (\xi^2 - \eta^2)^2 \eta^2 + 4 (\xi, \eta)^2 \xi^2 - 4 (\xi, \eta)^2  (\xi^2
  - \eta^2) = \eta^2  ((\xi^2 - \eta^2)^2 + 4 (\xi, \eta)^2) < 0 .
  \label{eta1sq}
\end{equation}
Therefore, for all $\xi$, we have that $\eta'$ is timelike. Since we have
shown that $\eta' \succ 0$ for $\xi = 0$, by continuity it follows that $\eta'
\succ 0$ for all $\xi$.

Finally, the implication $\eta \prec 0 \Rightarrow \eta' \prec 0$ follows by
complex conjugation.

\subsection{4-point function powerlaw bound}\label{power4-point}

We wish to show next that the analytically continued 4-point function satisfies a
powerlaw bound, so that we can apply Theorem \ref{ThVlad}. The prefactor in
Eq.\ {\eqref{G4c}} satisfies a powerlaw bound by Lemma \ref{x2bnd}.
Furthermore, Eq.\ {\eqref{maj}} implies that the analytic continuation $g (c)$
constructed in Sec.\ \ref{anal4-point} is bounded by a Euclidean 4-point function,
namely:
\begin{equation}
  | g (c) | \leqslant g_E (c_{\ast}), \label{ggE}
\end{equation}
where $c_{\ast}$ is any Euclidean 4-point function configuration having $\rho
(c_{\ast}) = \bar{\rho} (c_{\ast}) = r = r (c) = \max (| \rho (c) |, |
\bar{\rho} (c) |)$. We choose the conformal frame {\eqref{frameErho}}:
\begin{equation}
  c_{\ast} : \qquad x_1 = - r \hat{e}_0,\ x_2 = r \hat{e}_0,\ x_3 = \hat{e}_0,\
  x_4 = - \hat{e}_0 .
\end{equation}
Using the convergent OPE in the $x_2 \rightarrow x_3$, $x_1 \rightarrow x_4$
channel, we have the asymptotics
\begin{equation}
  G_4^E (c_{\ast}) \sim \frac{1}{(1 - r)^{4 \Delta_{\varphi}}} \qquad (r
  \rightarrow 1) . \label{g4cstar}
\end{equation}
The function $g_E (c_{\ast})$ satisfies the same asymptotics up to a constant,
being related to $G_4^E (c_{\ast})$ via Eq.\ {\eqref{def:Euclidean4-point}} by a
factor which is non-singular in the $r \rightarrow 1$ limit. Since $g_E
(c_{\ast})$ is a positive monotonically increasing function for $0 \leqslant r
< 1$ (see Eq.\ {\eqref{g:rhoexpansion}}), we conclude that it has a bound
\begin{equation}
  g_E (c_{\ast}) \leqslant \frac{\tmop{const} .}{(1 - r (c))^{4
  \Delta_{\varphi}}}, \label{gEcstar}
\end{equation}
and $| g (c) |$ by {\eqref{ggE}} satisfies the same bound.

The upshot of this discussion is that we will have a powerlaw bound on $G_4
(c)$ if we manage to get a powerlaw bound on $\frac{1}{1 - r (c)}$. We will
next state and prove such a bound.

{Before launching into the technical discussion, let us discuss
intuitively why a result like this is expected to be true. We know (Lemma
\ref{bound}) that $| \rho (c) |, | \bar{\rho} (c) | < 1$ and now we wish to
prove that $| \rho (c) |, | \bar{\rho} (c) |$ do not approach 1 too quickly as
$c$ goes to the Minkowski boundary of the forward tube. This may remind the
reader of the Schwarz-Pick lemma, which says that if $f (w)$ is a function
from a unit disk to itself and $f (0) = 0$, then $| f (w) | \leqslant | w |$,
hence providing a bound on how fast $| f (w) |$ can approach 1 as $| w |
\rightarrow 1$. In the 2d case, when $\rho (c)$ and $\bar{\rho} (c)$ are
individually defined holomorphic functions in the forward tube, it is indeed
possible to use the Schwarz-Pick lemma to prove a powerlaw bound on $\max (|
\rho (c) |, | \bar{\rho} (c) |)$ {\cite{lecturesSaclay}}. It should be
possible to generalize the Schwarz-Pick argument to any $d$, although we have
not worked it out in full details.\footnote{For any $d$, the Schwarz-Pick
lemma allows a natural generalization to holomorphic functions in the forward
tube {\cite{KravchukSchwarz-Pick}}.} The proof below will be different and
more direct: it will simply mimic the proof of Lemma \ref{bound}, replacing
all ``$> 0$'' inequalities by ``$\geqslant \varepsilon$'' with an explicit
positive $\varepsilon$.}

\subsubsection{A powerlaw bound on $\frac{1}{1 - r (c)}$}\label{section:1-r}

Let us introduce some notation. We will measure the size of a complex vector
$\zeta \in \mathbb{C}^{1, d - 1}$ by $| \zeta |$,
\begin{equation}
  | \zeta |^2 = | \zeta^0 |^2 + | \zeta^1 |^2 + \cdots + | \zeta^{d -
  1} |^2 .
\end{equation}
Clearly $| (\zeta_1, \zeta_2) | \leqslant | \zeta_1 | | \zeta_2 |$. We also
define for $\zeta = \xi + i \eta$, $\xi, \eta \in \mathbb{R}^{1, d - 1}$, and
$\eta^2 < 0$ (i.e.\ timelike)
\begin{equation}
  S (\zeta) = \max \Bigl( \frac{1}{\sqrt{- \eta^2}}, | \zeta | \Bigr) .
\end{equation}
Thus $S (\zeta)$ is large either if some component of $\zeta$ (real or
imaginary) is large or if $\eta$ approaches the light cone. Note that $S
(\zeta) \geqslant 1$ for any $\zeta$. We will never need $S (\xi + i \eta)$
for spacelike $\eta$.

Finally we consider an analogous function on $\mathcal{T}_4$:
\begin{equation}
  S (c) = \max_{i < j} S (\zeta_{i j}),
\end{equation}
which becomes large if any of $S (\zeta_{i j})$ become large. We claim that
there is the following bound (recall $r (c) = \max (| \rho (c) |, | \bar{\rho}
(c) |)$)
\begin{equation}
  \frac{1}{1 - r (c)} \leqslant 720 S (c)^{12} \qquad (c \in \mathcal{T}_4)
  \label{rhobound} .
\end{equation}
This bound will be shown for any $c$ in the forward tube, which is the natural
setting. When we specify to $c \in \mathcal{D}_4 \subset \mathcal{T}_4$ [see
Eq.\ {\eqref{def:Dn}}], we have
\begin{equation}
  S (c) = \max_{i < j} \max \left\{  \frac{1}{| \epsilon_i - \epsilon_j
  |}, | x_i - x_j | \right\} . \label{Scbound}
\end{equation}
Eq.\ {\eqref{rhobound}} then becomes a powerlaw bound for $\frac{1}{1 - r (c)}$
on $\mathcal{D}_4$ of the form {\eqref{powerlawbound}}, precisely as needed
for applying Theorem \ref{ThVlad}.

The proof of the bound {\eqref{rhobound}} will build upon the proof of $z,
\bar{z} \nin [1, + \infty)$ given in Sec.\ \ref{PetrProof}. There we showed
that $z$ solves Eq.\ {\eqref{equivquad}}, which however is inconsistent
for $z \in [1, + \infty)$ and $c$ in the forward tube. Here we will use the
same Eq.\ {\eqref{equivquad}}, but make the rest of the argument
quantitative, by showing that if $c$ stays away from the boundary or infinity
of the forward tube, so that $S (c)$ is bounded, then both $z (c)$ and
$\bar{z} (c)$ must stay a finite distance away from $[1, + \infty)$, as
measured by an upper bound on $\frac{1}{1 - r (c)}$ expressed by Eq.
{\eqref{rhobound}}. The proof is straightforward but somewhat technical and we
split it into a series of lemmas.

\begin{lemma}
  \label{zeta2bnd}Let $\zeta = \xi + i \eta$, $\eta^2 < 0$. Then for any $\xi$
  \begin{equation}
     | \zeta^2 | \geqslant (- \eta^2) . \label{zeta2bnd0}
  \end{equation}
\end{lemma}

\begin{proof}
  This is a generalization of Lemma \ref{x2bnd}(a) and could be proven
  analogously. We give a slightly different proof for a change. We have
  \begin{equation}
    | \zeta^2 |^2 = (\xi^2 - \eta^2)^2 + 4 (\xi, \eta)^2 = (\xi^2)^2 +
    (\eta^2)^2 + 2 [2 (\xi, \eta)^2 - \xi^2 \eta^2] .
  \end{equation}
  The lemma now follows from the inequality:
  \begin{equation}
    2 (\xi, \eta)^2 - \xi^2 \eta^2 \geqslant 0 . \label{ineqtoprove}
  \end{equation}
  Eq.\ {\eqref{ineqtoprove}} is obvious for $\xi^2 \geqslant 0$, so let us
  consider $\xi^2 < 0$. By Lorentz invariance and homogeneity it's enough to
  consider $\xi = (\pm 1, 0, \ldots, 0)$ in which case the l.h.s.\ of
  {\eqref{ineqtoprove}} becomes $(\eta^0)^2 +\tmmathbf{\eta}^2 \geqslant 0$.
\end{proof}

Then we have the following strengthening of Lemma \ref{Petr}(b):

\begin{lemma}
  \label{Sinv}Let $\zeta = \xi + i \eta$, $\eta^2 < 0$, and $\zeta' = \zeta /
  \zeta^2$. Then
  \begin{equation}
    S (\zeta') \leqslant [S (\zeta)]^3 .
  \end{equation}
\end{lemma}

\begin{proof}
  We have
  \begin{equation}
    | \zeta' | = \frac{| \zeta |}{| \zeta^2 |} \leqslant \text{[by Lemma
    \ref{zeta2bnd}]}  \frac{| \zeta |}{- \eta^2} \leqslant S (\zeta)^3 .
  \end{equation}
  We also have (see the proof of Lemma \ref{Petr}, in particular Eq.
  {\eqref{eta1sq}}) that $\eta^{\prime 2} < 0$ and
  \[ \frac{1}{- \eta^{\prime 2}} = \frac{| \zeta^2 |^2}{- \eta^2} \leqslant
     \text{[by Lemma \ref{zeta2bnd}]} S (\zeta)^6 . \]
\end{proof}

\begin{lemma}
  \label{SSS}Let $\zeta_i \in \mathbb{C}^{1, d - 1}, \eta_i \succ 0$ ($i =
  1, 2$). Then
  \begin{equation}
    S (\zeta_1 + \zeta_2) \leqslant S (\zeta_1) + S (\zeta_2) .
  \end{equation}
\end{lemma}

\begin{proof}
  We have $| \zeta_1 + \zeta_2 | \leqslant | \zeta_1 | + | \zeta_2 |$ and $-
  (\eta_1 + \eta_2)^2 \geqslant - \eta^2_1 - \eta_2^2 $ {(since $\eta_1\cdot\eta_2 <0$).}
\end{proof}

\begin{lemma}
  \label{lemma8}Let $\Upsilon_i = \Phi_i + i \Psi_i \in \mathbb{C}^{1, d - 1},
  \Phi_i, \Psi_i \in \mathbb{R}^{1, d - 1}$, $\Psi_i \succ 0$ (i=1,2), and $z$
  solves the equation
  \begin{equation}
    (\Upsilon_1 + (z - 1) \Upsilon_2)^2 = 0 . \label{Yeq}
  \end{equation}
  Then
  \begin{equation}
    1 - | \rho (z) | \geqslant \delta_0 : = \frac{1}{45 S^4}, \qquad S = \max
    (S (\Upsilon_1), S (\Upsilon_2)) . \label{rhoclose}
  \end{equation}
\end{lemma}

\begin{proof}
  Note that $z = 4 \rho / (1 + \rho)^2$, and so Eq.\ {\eqref{Yeq}} can be
  rewritten as
  \begin{equation}
    ((\rho + 1)^2 \Upsilon_1 - (\rho - 1)^2 \Upsilon_2)^2 = 0. \label{eqrew1}
  \end{equation}
  For $\rho = e^{i \alpha}$, multiplying this equation by $e^{- 2 i \alpha}$,
  it becomes
  \begin{equation}
    (\Upsilon)^2 = 0, \qquad \Upsilon \equiv \left( 2 \cos \frac{\alpha}{2}
    \right)^2 \Upsilon_1 + \left( 2 \sin \frac{\alpha}{2} \right)^2
    \Upsilon_2,
  \end{equation}
  which contradicts Lemma \ref{Petr}(a), since $\tmop{Im} \Upsilon \succ 0$.
  So $\rho$ cannot lie precisely on the unit circle (as we already knew). It
  should then not be surprising that it also cannot get too close to the unit
  circle, which is what {\eqref{rhoclose}} says. This can be shown by a
  straightforward although somewhat technical generalization of the above
  argument.
  
  Denoting $\rho = r e^{i \alpha} = e^{i \alpha} - \delta e^{i \alpha}$,
  $\delta = 1 - r > 0$, and multiplying {\eqref{eqrew1}} by $e^{- 2 i
  \alpha}$, it becomes
  \begin{equation}
    \left( \left( 2 \cos \frac{\alpha}{2} - \delta e^{i \alpha / 2} \right)^2
    \Upsilon_1 + \left( 2 \sin \frac{\alpha}{2} + i \delta e^{i \alpha / 2}
    \right)^2 \Upsilon_2 \right)^2 = 0,
  \end{equation}
  or
  \begin{equation}
    (\Upsilon + \Upsilon')^2 = 0 \label{Y+Y}
  \end{equation}
  with
  \begin{eqnarray}
    & \Upsilon = 4 \cos^2 \frac{\alpha}{2} \Upsilon_1 + 4 \sin^2
    \frac{\alpha}{2} \Upsilon_2, & \\
    & \Upsilon' = \kappa_1 \Upsilon_1 + \kappa_2 \Upsilon_2, & \\
    & \kappa_1 = - 4 \cos \frac{\alpha}{2} \delta e^{i \alpha / 2} + \delta^2
    e^{i \alpha}, \quad \kappa_2 = 4 i \sin \frac{\alpha}{2} \delta e^{i
    \alpha / 2} - \delta^2 e^{i \alpha} . & 
  \end{eqnarray}
  So for $\delta$ small, $\Upsilon'$ is a small correction to $\Upsilon$. We
  write $\tmop{Im} (\Upsilon + \Upsilon') = \Psi + \Psi'$, where
  \[ \Psi = \tmop{Im} \Upsilon = 4 \cos^2 \frac{\alpha}{2} \Psi_1 + 4 \sin^2
     \frac{\alpha}{2} \Psi_2, \quad \Psi' = \tmop{Im} \Upsilon' . \]
  We know that $\Psi \succ 0$. In addition we also have a lower bound on $-
  \Psi^2$:
  \begin{equation}
    - \Psi^2 \geqslant 16 \cos^4 \frac{\alpha}{2} (- \Psi^2_1) + 16 \sin^4
    \frac{\alpha}{2} (- \Psi^2_2) \geqslant \frac{1}{S^2} \times 16 \min
    \left\{ \cos^4 \frac{\alpha}{2}, \sin^4 \frac{\alpha}{2} \right\} =
    \frac{4}{S^2} . \label{lowerbnd}
  \end{equation}
  We will now show that $- (\Psi + \Psi')^2$ remains strictly positive if
  $\delta < \delta_0$. This will imply, by Lemma \ref{Petr}(a), that Eq.
  {\eqref{Y+Y}} cannot hold, and hence we must have $\delta \geqslant
  \delta_0$, i.e.\ Eq.\ {\eqref{rhoclose}}, proving the lemma.
  
  To implement this natural strategy, we will need only crude estimates of
  the size of various terms. Note that $\delta_0 < 1$ since $S \geqslant 1$,
  so in particular we have $\delta^2 \leqslant \delta$. Using this we have the
  bounds $| \kappa_i | \leqslant 5 \delta$, and hence an upper bound
  \begin{equation}
    | \Psi' | \leqslant | \Upsilon' | \leqslant 10 \delta S.
  \end{equation}
  We also have an upper bound $| \Psi | \leqslant 4 S$. Using these,
  {\eqref{lowerbnd}}, and $\delta^2 \leqslant \delta$, we have:
  \begin{eqnarray}
    - (\Psi + \Psi')^2 = - \Psi^2 - 2 (\Psi, \Psi') - (\Psi')^2 & \geqslant &
    \frac{4}{S^2} - 2 | \Psi | | \Psi' | - | \Psi' |^2 \nonumber\\
    & \geqslant & \frac{4}{S^2} - 80 \delta S^2 - 100 \delta^2 S^2
    \nonumber\\
    & \geqslant & \frac{4}{S^2} - 180 \delta S^2 = \frac{4 (1 - \delta /
    \delta_0)}{S^2} \,,
  \end{eqnarray}
  which is strictly positive for $\delta < \delta_0$. As explained above this
  proves the lemma.
\end{proof}

Finally we can prove {\eqref{rhobound}}. We repeat the proof of Lemma
\ref{bound} given in Sec.\ \ref{PetrProof}. As there, we reduce to
configuration having $\zeta_3 = 0$ and obtain that $z$ (as well as $\bar{z}$)
is a solution of Eq.\ {\eqref{equivquad}}, which has the form
{\eqref{Yeq}} with
\begin{equation}
  \Upsilon_1 = \zeta_{14}' = \zeta_1' - \zeta_4', \quad \Upsilon_2 =
  \zeta_{24}' = \zeta_2' - \zeta_4', \quad \zeta_i' = \zeta_i / \zeta^2 \quad
  (i = 1, 2, 4) . \label{def:Upsilon12}
\end{equation}
Let us write $\Upsilon_i = \Phi_i + i \Psi_i \in \mathbb{C}^{1, d - 1},
\Phi_i, \Psi_i \in \mathbb{R}^{1, d - 1}$. As was already pointed out in
Sec.\ \ref{PetrProof}, we have $\Psi_i \succ 0$ ($i = 1, 2$). Furthermore,
by Lemma \ref{Sinv} we know that $S (\zeta_i') \leqslant S (c)^3$, and then
applying Lemma \ref{SSS} that $S (\Upsilon_i) \leqslant 2 S (c)^3$. Thus Lemma
\ref{lemma8} implies {\eqref{rhobound}} (note that $720 = 45 \times 16$).

\begin{remark}
  The bound {\eqref{rhobound}} is not optimal. We will prove a better bound in
  Sec.\ \ref{secondpass}, by a different argument. 
\end{remark}

Let us recap. In Sec.\ \ref{anal4-point} we have analytically continued the
Euclidean 4-point function to the forward tube, and here we showed that this
analytic continuation satisfies a powerlaw bound. Then by Theorem
\ref{ThVlad}, the Minkowski 4-point function defined as the limit {\eqref{limit}}
exists, is a Lorentz-invariant tempered distribution, and satisfies Wightman
spectral condition. In the remainder of this section we will show that this
distribution is also conformally invariant (Sec.~\ref{ConfMink}), that it satisfies the remaining
Wightman axioms (positivity in Sec.~\ref{sec:Wpos}, clustering in Sec.~\ref{clusterWightman}, and local commutativity in Sec.~\ref{local-comm}). {Later in Sec.~\ref{OPEconvMink} we will also show that it
can be computed by a convergent (in the sense of distributions) OPE.}

Now that we know that the Minkowski 4-point function is a distribution everywhere,
one may inquire about the regularity of this distribution. {E.g.\ for some
configurations the 4-point function is actually real-analytic {\cite{Qiao:2020bcs}}. We will come back to
this question in the conclusion section.}

\subsection{Conformal invariance}\label{ConfMink}

Conformal invariance of Euclidean 4-point function {\eqref{def:Euclidean4-point}} can
be described as invariance under finite conformal transformations $x
\rightarrow x' = f (x)$,
\begin{equation}
  \Omega_1 \Omega_2 \Omega_3 \Omega_4 G^E_4 (x'_1, x'_2, x'_3, x'_4)
  = G^E_4 (x_1, x_2, x_3, x_4), \label{EfinInv}
\end{equation}
where $\Omega_i = J (x_i)^{\Delta_{\mathcal{O}}}$ and $J (x) = \det
(\partial f^{\mu} / \partial x^{\nu})^{1 / d}$ is the local scale factor.
Alternatively, and equivalently, this can be expressed as invariance under
infinitesimal conformal transformations, a conformal Ward identity, which says
that the Euclidean correlator is annihilated by a sum of differential
operators, one per point:
\begin{equation}
  \sum_{i = 1}^4 \mathcal{D} (x_i, \partial_{x_i}) G^E_4 (x_1, x_2, x_3, x_4)
  = 0 . \label{Ward}
\end{equation}
There is a differential operator per conformal group generator
($\partial^{\mu}$ for $P_{\mu}$, $x^{\mu} \partial^{\nu} - x^{\nu}
\partial^{\mu}$ for $M_{\mu \nu}$, $x \cdot \partial +
\Delta_{\mathcal{O}}$ for $D$, $x^2 \partial^{\mu} - 2 x^{\mu} (x \cdot
\partial) - 2 x^{\mu} \Delta_{\mathcal{O}}$ for $K_{\mu}$).

Since all these differential operators have polynomial coefficients, Ward
identities {\eqref{Ward}} continue to hold in the forward tube for the
function $G (x_1, x_2, x_3, x_4)$. Taking the limit to the Minkowski boundary,
we obtain that the Minkowski 4-point function satisfies infinitesimal Minkowski
conformal invariance expressed by the Ward identities.

The possibility to take the limit is guaranteed by the standard result that
distributional limits commute with derivatives. Indeed, suppose that we have,
in the sense of distributions, $\lim_{\varepsilon \rightarrow 0}
f_{\varepsilon} = g$. This means that for any test function $\varphi$, we have
$\lim_{\varepsilon \rightarrow 0} (f_{\varepsilon}, \varphi) = (g, \varphi)$.
But then for any derivative $\partial,$
\begin{equation}
  (\partial g, \varphi) = - (g, \partial \varphi) = - \lim_{\varepsilon
  \rightarrow 0} (f_{\varepsilon}, \partial \varphi) = \lim_{\varepsilon
  \rightarrow 0} (\partial f_{\varepsilon}, \varphi), \label{derLimit}
\end{equation}
which implies that $\lim_{\varepsilon \rightarrow 0} \partial f_{\varepsilon}
= \partial g$. A similar argument shows that the limit commutes with
multiplication of distributions by polynomials. All this is analogous to how
we prove Lorentz invariance of the Minkowski correlator in App.\ \ref{Vlad}.

So we have shown that the Minkowski 4-point function satisfies Lorentzian
conformal Ward identities. This means that
\begin{equation}
  \sum_{i = 1}^4 (\mathcal{D}_i G^M_4, \varphi) = 0, \label{MinkConfWard}
\end{equation}
where $\mathcal{D}_i$ are the analytic continuations of the Euclidean
differential operators to Minkowski space, and the pairing with the Schwartz
test functions is defined by integration by parts. Note that the conformal
Ward identities in Minkowski space hold also at coincident points (i.e.\ the
test function $\varphi$ does not have to be zero at coincident points).

Now let us discuss invariance of Minkowski 4-point function under
\tmtextit{finite} Lorentzian conformal transformations. Since $G_4^M$ is a
distribution, the appropriate form of writing is to transform the test
function:
\begin{equation}
  (G_4^M, \varphi) = (G_4^M, \varphi^f), \label{MfinInv}
\end{equation}
where $\varphi^f (x_1, \ldots, x_4) = \varphi (f^{- 1} (x_1), \ldots, f^{- 1}
(x_4)) \prod_{i = 1}^4 J (f^{- 1} (x_i))^{\Delta_{\mathcal{O}} - d}$. However
we have to be careful. This invariance is true not for every test function
$\varphi$ but only for an $f$-dependent subset of test functions.

Let $f_t$ be a smooth family of Lorentzian conformal transformations
connecting $f$ to the identity: $f_0 = \tmop{id}$, $f_1 = f$. Suppose that
\begin{equation}
  \text{$\varphi^{f_t}$ is a Schwartz function for any $f_t$ in the family} .
  \label{req1}
\end{equation}
Then we can integrate infinitesimal conformal invariance and prove that
{\eqref{MfinInv}} is true. For translations, Lorentz transformations and
dilatations, Eq.\ {\eqref{req1}} is clearly satisfied and Eq.\ {\eqref{MfinInv}}
holds for any $\varphi$. However, for general conformal transformations,
{\eqref{req1}} may not necessarily be true. The problems will appear if $f$
is singular on the support of $\varphi$, as $\varphi^f$ may then not be a
Schwartz function. As a concrete example, consider the Lorentzian special
conformal transformation:
\begin{equation}
  f (x) = \frac{x^{\mu} + x^2 b^{\mu}}{1 + 2 x \cdummy b + x^2 b^2} .
  \label{fb}
\end{equation}
The corresponding scale factor is $J (x) = \frac{1}{1 + 2 x \cdummy b + x^2
b^2}$. Take for definiteness spacelike $b = \beta \hat{e}_1$, where $\beta >
0$ and $\hat{e}_1$ is the unit vector in the $x^1$ direction. The
transformation {\eqref{fb}} is then singular for $x^0 = \pm | \mathbf{x}+
\beta^{- 1} \hat{e}_1 |$, where the scale factor blows up, i.e.\ on the
light cone whose vertex is at $x = - \beta^{- 1} \hat{e}_1$.\footnote{Recall
that we are using $-, + \cdots +$ Minkowski signature.} Scaling $\beta$ to zero
we can connect the transformation {\eqref{fb}} to the identity. Under this
scaling the light cone of singularities moves away to infinity along the
negative $x^1$ direction. Requirement {\eqref{req1}}, and hence finite
invariance {\eqref{MfinInv}}, will hold if the light cone of singularities,
while moving away, does not touch the support of $\varphi$ (see Fig.\ \ref{figure:supp0} for the 2d case).

\begin{figure}[h]\centering
\includegraphics[width=8.24452315361406cm,height=4.28230355503083cm]{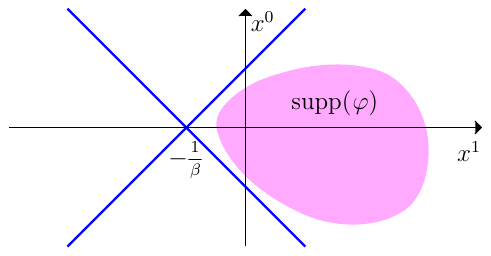}
\caption{\label{figure:supp0}In
the 2d case, the special conformal transformation {\eqref{fb}} is singular on
the blue light cone $x^0 = \pm | x^1 + \beta^{- 1} |$. Suppose $\varphi$ is
supported as shown on the right of the light cone. As $\beta \rightarrow 0$,
the light cones moves towards the left infinity and does not touch $\tmop{supp}
(\varphi)$. Therefore, such a $\varphi$ satisfies the condition for the 
invariance under a finite special conformal transformation {\eqref{fb}}. }
\end{figure}

Note that such a support requirement still leaves the possibility for both
spacelike and timelike separations among the points $x_i$ in the support of
$\varphi$. For $x_i \in \tmop{supp} \varphi$, the points $f (x_i)$ will have
the same causal structure as the points $x_i$, i.e.~$(f (x_i) - f (x_j))^2$
will have the same sign as $(x_i - x_j)^2$. This follows from the fact that $J
(x) > 0$ on $\tmop{supp} (\varphi)$, as guaranteed by being able to
continuously connect to the identity without singularities.

In the early CFT days, it was considered puzzling that Lorentzian special
conformal transformations may change the causal structure of a point
configuration. As we see here, the puzzle can be avoided by either limiting to
infinitesimal conformal invariance, or by restricting the class of test
functions so that the causal structure is preserved. A third way to deal with
the puzzle is to consider the Lorentzian conformal transformations acting on
the Lorentzian cylinder as opposed to the Minkowski space
{\cite{Luscher:1974ez}}. We will revisit the Lorentzian cylinder in our future
publication {\cite{paper3}}.

\begin{remark}
  We would like to contrast the Minkowski conformal Ward identities
  {\eqref{MinkConfWard}} with conformal Ward identities valid for Euclidean
  correlators. Euclidean correlators are real-analytic away from coincident
  points and naturally satisfy conformal Ward identities for such
  configurations. Although in this paper we don't need it, in some questions
  it might be useful to extend Euclidean correlators, in the sense of
  distributions, also to coincident points. One may ask if such an extension
  can be done in a way so that the resulting distributional correlators
  satisfy conformal Ward identities analogous to~\eqref{MinkConfWard}. In general the answer is no, already for
  2-point functions. Namely 2-point functions of primaries of dimension
  $\Delta$ such that $2 \Delta - d \in \mathbb{Z}$ will in general not allow
  even a scale invariant extension at coincident points, let alone
  conformally invariant one. E.g.\ this feature will always be present for the
  stress tensor 2-point function. 
\end{remark}

\subsubsection{Conformal invariance in terms of cross
ratios}\label{conf-cross}

So as we have just seen, Minkowski correlator $G_4^M$ is conformally
invariant. If it were a function, conformal invariance would imply that we
could write it as the usual prefactor times a function of the cross-ratios.
Since it is a distribution, one might hope that it can be written as the
prefactor times a {\tmem{distribution}} of the cross-ratios. We will now
develop this point of view and show that it indeed works, at least 
locally.\footnote{\label{no-global}We note right away that one does not expect a very nice global statement. Indeed, 
the cross-ratio space is morally the moduli space of four points on Minkowski cylinder
$\cM$. (Note that we have not yet constructed the Wightman functions as distributions on Minkowski
cylinder. However, it is not important for the point that we are trying to convey.) This is a quotient space $(\cM^4)/G$
where $G$ is the universal cover of Lorentzian conformal group. This quotient space is rather singular,
which has to do with different configurations in $\cM^4$ having different stability subgroups (light-cones, $z=\bar z$).
The quotient space $(\cM^4)/G$ is not only not smooth, it is not even Hausdorff. So away from some regular regions
of $(\cM^4)/G$ one shouldn't expect a simple statement of the form~\eqref{G4g}, unless one finds a smoother model 
of this moduli space.}

Our goal in this section will be to make sense of the formula:
\begin{equation}
  G_4 (c) = \frac{g (z (c), \bar{z} (c))}{(x_{12}^2
  x_{34}^2)^{\Delta_{\mathcal{O}}}}, \label{G4g}
\end{equation}
where $g (z, \bar{z})$ will be in general a distribution in two variables, and
$g (z (c), \bar{z} (c))$ its pullback to the space $\mathbb{R}^{4 d}$ of
Minkowski 4-point configurations $c$. This equation will be understood in the
sense of integrating both parts with a test function. {Because of the difficulty described in footnote \ref{no-global},} we will only consider
compactly supported $C^{\infty}$ test functions $\varphi (c)$, with the
additional requirement that all $c \in \tmop{supp} \varphi$ have the same
causal ordering. In particular, this implies that $\tmop{supp} \varphi$
contains no $c$'s with lightlike separated pairs. The causal ordering of a
configuration $c = (x_1, x_2, x_3, x_4)$ is encoded by the directed graph with
vertices $1, 2, 3, 4$ and edges $i \rightarrow j$ if $x_j$ belongs to the open
future light cone of $x_i$ (no edge if two points are spacelike).

Since $u, v$ are real in Minkowski space, $z, \bar{z}$ are either both real
(excluding $0, 1$) or complex conjugate. Ref.\ {\cite{Qiao:2020bcs}} divided
all causal orderings into 4 classes according to possible values of $(z,
\bar{z})$:
\begin{itemize}
  \item Class S: $(z, \bar{z}) \in (0, 1) \times (- \infty, 0)$ or the other
  way around
  
  \item Class T: $(z, \bar{z}) \in (0, 1) \times (1, + \infty)$ or the other
  way around
  
  \item Class U: $(z, \bar{z}) \in (- \infty, 0) \times (1, + \infty)$ or the
  other way around
  
  \item Class E causal orderings which contain configurations realizing the
  remaining possibilities:
  \begin{itemize}
    \item E$_{\tmop{su}}$: $(z, \bar{z}) \in (- \infty, 0) \times (- \infty,
    0)$
    
    \item E$_{\tmop{st}}$: $(z, \bar{z}) \in (0, 1) \times (0, 1)$
    
    \item E$_{\tmop{tu}}$: $(z, \bar{z}) \in (1, + \infty) \times (1, +
    \infty)$
    
    \item E$_{\tmop{stu}}$: $z, \bar{z}$ are complex-conjugate and not real
  \end{itemize}
\end{itemize}
Some class E causal orderings realize only one of the four subclasses, while
others contain configurations in each \ subclass. In the latter case different
subclasses are connected along configurations with $z = \bar{z}$ (see Fig.\ 2
in {\cite{Qiao:2020bcs}}).

To simplify the discussion, we will assume that $\tmop{supp} \varphi$ does
not include any configurations with $z = \bar{z}$. In particular, this implies
that all configurations from $\tmop{supp} \varphi$ fall into a single class S,
T, U or a single subclass E$_{\tmop{su}}$, E$_{\tmop{st}}$, E$_{\tmop{tu}}$,
E$_{\tmop{stu}}$. Below we will comment how one can add the $z = \bar{z}$
configurations.

If $\tmop{supp} \varphi$ falls into class S, E$_{\tmop{su}}$, E$_{\tmop{st}}$,
E$_{\tmop{stu}}$, we will have $| \rho |, | \bar{\rho} | < 1$. These cases do
not require special treatment, since the correlator is a function, and Eq.
{\eqref{G4g}} is true in the ordinary sense of functions.

If $\tmop{supp} \varphi$ falls into class T or U, we will have $| \rho | < 1,
| \bar{\rho} | = 1$ or the other way around. Then $g (z, \bar{z})$ will be a
function in $z$ and a distribution in $\bar{z}$.\footnote{For some (but not
all) of these causal orderings, it can be shown using another OPE channel that
$g (z, \bar{z})$ is actually a function of both variables. See Ref.\ {\cite{Qiao:2020bcs}}.} This case can be treated analogously, and simpler,
than the $| \rho |, | \bar{\rho} | = 1$ case discussed below.

Finally, if $\tmop{supp} \varphi$ falls into class E$_{\tmop{tu}}$, we will
have $| \rho |, | \bar{\rho} | = 1$. Then $g (z, \bar{z})$ will generally be a
distribution in two variables. This is the case we will focus on. E.g.\ it is
realized for the causal ordering $1 \rightarrow 3 \rightarrow 2 \rightarrow
4$.

Let us define the distribution $g (z, \bar{z})$ for $z, \bar{z} \in (1, +
\infty)$. We first define the distribution $g (\rho, \bar{\rho})$ with $|
\rho |, | \bar{\rho} | = 1$. This is done using the series in the r.h.s.\ of
Eq.\ {\eqref{g:rhoexpansion2}}, which we now consider as a function of two
independent variables $\rho, \bar{\rho}$. To be precise we consider the
series:
\begin{equation}
  g (\rho, \bar{\rho}) = \sum_{\delta, 0 \leqslant m \leqslant \delta} e^{i
  \Phi (\delta / 2 - m / 2)} (\rho \bar{\rho})^{\delta / 2 - m / 2} 
  (p_{\delta, m} \rho^m + p_{\delta, - m} \bar{\rho}^m) \hspace{0.17em},
  \label{gseries}
\end{equation}
which we view as a holomorphic function on $(\mathbb{D} \backslash (- 1,
0])^2$. Here $e^{i \Phi}$, $\Phi \in \{ 0, \pm 2 \pi, \pm 4 \pi \}$, is the
phase acquired by $\rho (c) \bar{\rho} (c)$ upon analytic continuation from
Euclidean space (as discussed in Sec.\ \ref{anal4-point} this phase is the same
as for $u (c)$). This phase is constant for each causal ordering and it may be
determined by following a path from $c_E$ to $c$ for any particular $c$.
Alternatively, the phase can also be determined from {\eqref{rhotrick}}. E.g.\ the causal ordering $1 \rightarrow 3 \rightarrow 2 \rightarrow 4$ has $\Phi =
0$.\footnote{We have that $x_{12}^2$, $x^2_{34}$, $x^2_{13}$, $x^2_{24}$ all
acquire phase $- \pi$, hence $u = \frac{x_{12}^2 x^2_{34}}{x_{13}^2 x^2_{24}}$
acquires phase 0.}

It's easy to see that function {\eqref{gseries}} satisfies a powerlaw bound as
$| \rho |, | \bar{\rho} | \rightarrow 1$. This is a baby version of the
problems studied in this paper, which was considered in {\cite{paper1}}. The
limit of $g (r e^{i \theta_1}, r e^{i \theta_2})$ as $r \rightarrow 1$
defines a tempered distribution on the boundary of the domain of analyticity,
parametrized by the two angles $\theta_1, \theta_2$. We can write it as $g
(\rho, \bar{\rho})$, with $\rho, \bar{\rho} \in S^1$.\footnote{In {\cite{paper1}} we also discussed a more general distribution
defined on the product of universal covers of two circles. Here Eq.
{\eqref{gseries}} with fixed $\alpha$ will be sufficient for our purposes.}

In fact we are interested only in a part of this distribution, {because $\rho,\bar{\rho}\neq\pm1$ for each fixed causal ordering.} The
points $- 1, 1$ divide the circle into two open arcs, and within $\tmop{supp}
\varphi$, $\rho$ and $\bar{\rho}$ will each live in one or the other arc. Each
arc is mapped smoothly and one-to-one to $(1, + \infty)$ by the $\rho \mapsto
z$ map. Thus we obtain the distribution $g (z, \bar{z})$ defined for $z,
\bar{z} > 1$. Although in general $z (c), \bar{z} (c)$ are defined only up to
permutation, let us define them in the case at hand, with real $z \neq
\bar{z}$, so that $\bar{z} (c) > z (c)$.

Now let us go back to making sense of {\eqref{G4g}}. Suppose first $g (z (c),
\bar{z} (c))$ were a function. Integrating {\eqref{G4g}} against a test
function we have:
\begin{equation}
  \int d^{4 d} c \, G_4 (c) \varphi (c) = \int d^{4d} c\, g (z (c), \bar{z} (c))
  \tilde{\varphi} (c), \qquad \tilde{\varphi} (c) = \frac{\varphi
  (c)}{(x_{12}^2 x_{34}^2)^{\Delta_{\mathcal{O}}}} .
\end{equation}
Note that $\tilde{\varphi} (c)$ is still $C^{\infty}$ since we are away from
light cones. We would like to continue by expressing the r.h.s. of the previous equation as an integral of $g
(z, \bar{z})$ against a two-dimensional test function:
\begin{equation}
\begin{gathered}
    \int d^{4d} c\, g (z (c), \bar{z} (c)) \tilde{\varphi} (c) = \int d z\, d \bar{z}\, g
    (z, \bar{z}) \psi (z, \bar{z}), \\
    \psi (x_1, x_2) = \int d^4 c \,\delta (x_1 - z (c)) \delta (x_2 - \bar{z} (c)) \tilde{\varphi} (c) .
    \label{psidef}
\end{gathered}
\end{equation}
We would like to know if $\psi (x_1, x_2)$ is a smooth function. By our
assumptions, $\tilde{\varphi} (c)$ is supported away from $\bar{z} (c) = z
(c)$. In this region the map $c \rightarrow (z (c), \bar{z} (c))$ is a
submersion, which means that the Jacobian has maximal rank (i.e.\ 2).
Alternatively, this means that the form $d z \wedge d \bar{z}$ does not vanish
anywhere away from $z = \bar{z}$. Showing this is a matter of an easy
computation.\footnote{Start by noting that, away from $z = \bar{z}$, we have
$d z \wedge d \bar{z} \propto d u \wedge d v$ with a nonvanishing prefactor.
We need to understand where $\nabla u$ can become proportional to $\nabla v$.
Using the embedding space formalism {\cite{Costa:2011mg}} we write $u =
\frac{(X_1 X_2) (X_3 X_4)}{(X_1 X_3) (X_2 X_4)}$, $v = \frac{(X_1 X_2) (X_3
X_4)}{(X_1 X_3) (X_2 X_4)}$ where $X_i$ are null cone $d + 2$ dimensional
vectors. For any $X_i, X_j, X_k$ the direction $R_{i, j k} = X_j (X_i X_k) -
X_k (X_i, X_j)$ is tangent to the null cone at $X_i$. Imposing $R_{i, j k}
\cdot \nabla_{X_i}  (u - \alpha v) = 0$ for all unequal $i, j, k$ where
$\alpha$ is a constant, one gets a set of simple algebraic constraints on $u,
v$. These constraints can be easily solved to show that $\alpha = \frac{2 u}{-
1 + u + v}$ while $(1 + u - v)^2 - 4 u = 0$. The latter is precisely the
constraint characterizing $z = \bar{z}$.}

Using the fact that $c \rightarrow (z (c), \bar{z} (c))$ is a submersion,
it's easy to show that $\psi (x_1, x_2)$ is smooth for $\tilde{\varphi} (c)$
supported away from $\bar{z} (c) = z (c)$ (see Chapter III.1 of
{\cite{gelfandshilov}} for such arguments). To summarize, for every smooth
function $\varphi (c)$ compactly supported away from $\bar{z} (c) = z (c)$ and
from the light cones, we constructed a smooth function $\psi (z, \bar{z})$
compactly supported in $1 < z < \bar{z}$ such that
\begin{equation}
  \int d^{4 d} c \, G_4 (c) \varphi (c) = \int d z\, d \bar{z}\, g (z, \bar{z})
  \psi (z, \bar{z}) \label{Gg1}
\end{equation}
holds in case $g (z, \bar{z})$ is a function. We now claim that this equation
continues to hold, with the same $\psi$, in case $g (z, \bar{z})$ is a
distribution. The point is that we can find a sequence of functions $g_n (z,
\bar{z})$ which tend to $g (z, \bar{z})$ in the sense of distributions, so
that the corresponding $\frac{g_n (z (c), \bar{z} (c))}{(x_{12}^2
x_{34}^2)^{\Delta_{\mathcal{O}}}}$ tend to $G_4 (c)$ in the sense of
distributions on $\mathbb{R}^{4 d}$. Since both $\varphi$ and $\psi$ are
smooth, we are allowed to pass to the limit on both sides of the equation,
proving the claim. The functions $g_n (z, \bar{z})$ are given e.g.\ by the
partial sums of the series {\eqref{gseries}}, transformed from the $\rho$ to
the $z$ coordinates.

Let us now discuss how configurations where $z = \bar{z}$ can be included into
this discussion. The basic difficulty is that the map $c \mapsto (z, \bar{z})$
fails to be a submersion near such configurations. So in general the function
$\psi (z, \bar{z})$ will not be smooth. Consider e.g.\ the causal ordering $1
\rightarrow 3 \rightarrow 2 \rightarrow 4$. In this case it's possible to show
(we omit the proof) that the function $\psi (z, \bar{z})$ behaves like
\begin{equation}
  \text{$| z - \bar{z} |^{d - 2}$ times a smooth function near $z = \bar{z}$},
  \label{psizzbar}
\end{equation}
which in general is not smooth unless $d$ is even.

We need to be able to make sense of the r.h.s.\ of {\eqref{Gg1}} for such
non-fully-smooth test functions. This is possible due to the following
observation. Above we explained, following the arguments first presented in
{\cite{paper1}}, that $g (\rho, \bar{\rho})$ is a distribution for $| \rho |,
| \bar{\rho} | = 1$. But in fact it's a bit better than that (the fact not
mentioned in {\cite{paper1}}): it is a distribution in $\rho$ for each fixed
value of $\bar{\rho} / \rho = e^{i \alpha}$! Indeed if we substitute
$\bar{\rho} = e^{i \alpha} \rho$ with a fixed $\alpha$ into {\eqref{gseries}},
we get a holomorphic function in the unit disk of $\rho$, which satisfies a
powerlaw bound, hence its boundary value is a distribution. This can be
generalized to holomorphic maps $\bar{\rho} = f (\rho)$ which maps the unit disk
into itself (or at least a portion of the unit disk near $\rho = \rho_0$ into
the unit disk). Translating to $z, \bar{z}$, this implies in particular that
$g (z, z + t)$ is a distribution for any fixed $t$. In fact, the map $\bar{z}
= z + t$ corresponds to a map $\bar{\rho} = f_t (\rho)$ to which the previous
argument is applicable. So $g (z, \bar{z})$ is by no means the most general
distribution in two variables, as it allows the restriction to the submanifold
$\bar{z} = z + t$ for any $t$. E.g.\ $\delta (z - \bar{z})$ is not allowed by
this property, while $\delta (z + \bar{z})$ is allowed. Following this logic a
bit more carefully, it can be shown (we omit the proof) that $g (z, \bar{z})$
can be paired with test functions $\psi (z, \bar{z})$ which, when expressed in
terms of $s = z + \bar{z}$, $t = \bar{z} - z$, have the following property:
$\psi (s, t)$ is $C^{\infty}$ with respect to $s$ for any fixed $t$, with
bounds on derivatives in the $s$ direction which are integrable in the $t$
direction. Eq.\ {\eqref{psizzbar}} is compatible with this requirement.

A further complications arises near the $z = \bar{z} > 1$ locus for the
causal orderings which include configurations in both E$_{\tmop{tu}}$ and
E$_{\tmop{stu}}$ subclasses. In this case the function $\psi (z, \bar{z})$
defined in {\eqref{psidef}} will consist of two functions $\psi_1 (z,
\bar{z})$ and $\psi_2 (z, \bar{z})$: one defined for real $z, \bar{z}$,
another for complex-conjugate $z, \bar{z}$. The two functions $\psi_i$ will be
glued along the $z = \bar{z} > 1$ line. The resulting glued function will not
in generally be smooth on the $z = \bar{z} > 1$ line (while it will be smooth
away from it). However the directions orthogonal to the line turn out
analogous to the $t$ direction in the previous paragraph, i.e.\ the test
function is actually not required to be smooth in these directions for the
pairing to be defined. This allows to make sense of the formula {\eqref{Gg1}}
also in this case. We omit the details.

\subsubsection{Fixing points}

We would like to put the results of the previous section in the context of a
general question of ``fixing points'' in a distribution. E.g.\ we know that the
Minkowski 4-point function is a translationally invariant distribution. Using
translation invariance we can always fix one of the 4 points to a given
position, e.g.\ zero, and consider it as a distribution with respect to the
remaining 3 positions. One could ask if one can do better than that, i.e.\ to
fix $n$ points to given positions and consider the 4-point function as a
distribution with respect to the remaining $4 - n$ positions. Where the 4-point
function is real-analytic we can of course consider all four points as fixed.

Now, results of Sec.\ \ref{conf-cross} show that, if one excludes lightlike
separations limiting to configurations having some fixed causal ordering, one
can fix a conformal frame, i.e.\ fix three points to some fixed positions, and
the fourth point to a position characterized by two conformal cross ratios,
and consider the distribution as a distribution in only two variables (cross
ratios). It is not clear if results of Sec.\ \ref{conf-cross} can be
generalized to cover lightlike separations.

In some cases it is possible to argue that one can fix more than one point
without using conformal invariance. E.g.\ we may always fix a consecutive pair
of points, i.e.\ $(x_k, x_{k + 1})$, where $k = 1, 2$ or $3$, to
spacelike-separated positions in Minkowski space, while allowing the remaining
two points to approach Minkowski limit from the forward tube. The proof of
Lemma \ref{bound} can be slightly modified to show that $| \rho |, |
\bar{\rho} | < 1$ for such configurations (see Sec.\ \ref{localCFT} below).
Moreover, a powerlaw bound also holds, by a slight modification of the
argument after Eq.\ (\ref{def:Upsilon12}).\footnote{Since $S (c) = \infty$ in
these cases, we cannot rely on (\ref{rhobound}). Instead we directly show
powerlaw bounds on $S (\Upsilon_1), S (\Upsilon_2)$ defined in Eq.
(\ref{def:Upsilon12}). Then the powerlaw bound on $| \rho |, | \bar{\rho} |$
holds by Lemma \ref{lemma8}. For $k = 1$, by fixing $\zeta_3 = 0$, and using
Lemmas \ref{Sinv} and \ref{SSS}, we have $S (\Upsilon_i) \leqslant S
(\zeta_i') + S (- \zeta_4') \leqslant S (x_{i 3})^3 + S (x_{34})^3$ $(i = 1,
2)$. This is the desired powerlaw bound with respect to $x_3$ and $x_4$. For
$k = 2$, $S (\Upsilon_1)$ is bounded as for $k = 1$, while for $S
(\Upsilon_2)$ we argue as follows. Since $x_2$ and $x_3$ are spacelike
separated, after fixing $\zeta_3 = 0$, $\zeta_2$ is a spacelike Minkowski
point, hence so is $\zeta_2'$, i.e.\ $\tmop{Im} (\zeta_4' + \zeta_2') =
\tmop{Im} (\zeta_4')$. Then by Lemma \ref{Sinv}, \ $S (\Upsilon_2) \leqslant S
(\zeta_4') + | \zeta_2' | \leqslant S (\zeta_4)^3 + | \zeta_2' |$, which is
the needed bound. Case $k = 3$ follows by similar arguments or by mapping it
to $k = 1$ via $(x_1, x_2, x_3, x_4) \rightarrow (x_1' = x_4^{\theta}, x_2' =
x_3^{\theta}, x_3' = x_2^{\theta}, x_4' = x_1^{\theta})$ which maps $\rho$ and
$\bar{\rho}$ are to their complex conjugates.} Then our arguments show that
the Minkowski 4-point function is a distribution with respect to the two unfixed
coordinates, which depends analytically on the fixed coordinates. In this case
the unfixed coordinates may have any causal orderings and also lightlike
separation.

One interesting case is that of the double light cone (DLC) singularity, i.e.\ the region close to $x_1 = 0$, $x_3 = \hat{e}_1$, $x_4 = \infty$, while $x_2$
on the light cones of $x_1, x_3$. Our results are the first ones which establish
the existence of the Wightman 4-point function in a neighborhood of DLC.
However, there is a difference between restricting to one causal ordering near
DLC or studying an open neighborhood of DLC which includes several causal
orderings (see Fig.\ \ref{134} in Conclusions). In the former case we can use
directly the results of Sec.\ \ref{conf-cross} and represents the 4-point
function as a distribution in two variables $z, \bar{z}$. In the latter case
we can fix, by the above argument, two successive spacelike points $x_3$ and
$x_4$. We are left with a distribution depending on $x_1, x_2$, i.e.\ $2\times d$
coordinates. This distribution still satisfies conformal invariance Ward
identities w.r.t.\ infinitesimal conformal transformations preserving $x_2$. It
would be interesting to understand how this constrains the distribution at the
DLC.

{Although it is not directly related, }we would also like to
mention here the classic result of Borchers {\cite{Borchers1964}} which says
that it is enough to smear Wightman functions $G_M (x_1, \ldots, x_n = 0)$
with respect to the time variables only, i.e.\ integrating with respect to $h_1
(x_1^0) \ldots h_{n - 1} (x_{n - 1}^0)$ where $h_i \in \mathcal{S}
(\mathbb{R})$, after which they become $C^{\infty}$ functions in the remaining
spatial variables $\mathbf{x}_i$. This result is valid in any QFT satisfying
Wightman axioms. It holds because smearing in time, which acts as an energy
cutoff, is effectively also a momentum cutoff because $| \mathbf{p} |
\leqslant E$.

\subsection{Wightman positivity}\label{sec:Wpos}

Recall that in Sec.\ \ref{OSfromCFT} we showed that CFT axioms imply OS
reflection positivity for 4-point functions. That discussion gives us access to OS
states $| \mathcal{O} (x) \mathcal{O} (y) \rangle \nobracket$ with $0 > x^0
> y^0$, with finite norm, and inner products measured by the Euclidean 4-point
function. We know that these states belong to the CFT Hilbert space, i.e.\ can
be arbitrarily well approximated in norm by states produced by inserting
finite linear combinations of CFT local operators at one point in the
half-space $x^0 < 0$, e.g.\ the south pole $x_S = (- 1, \tmmathbf{0})$.

Now that we analytically continued the 4-point function, we can consider other
states involving operators at complexified coordinates. We wish to prove that
those states belong to the CFT Hilbert space and have a positive definite
inner product. This can be shown by a robust argument, going back to
Osterwalder and Schrader {\cite{osterwalder1973}}, Sec.\ 4.3. The argument
uses only OS positivity and the Fourier-Laplace representation, but not
directly the CFT axioms.

We will consider two new kinds of states. First, states generated by a pair of
Minkowski operators smeared with respect to an arbitrary Schwartz test
function:
\begin{equation}
  | \Psi_M (F) \rangle \nobracket = \int d x\, d y\, F (x, y) | \mathcal{O} (i
  x^0, \mathbf{x}) \mathcal{O} (i y^0, \mathbf{y}) \rangle \nobracket,
  \label{PsiF}
\end{equation}
and second, states generated by a pair of Euclidean operators at complexified
time positions:
\begin{equation}
  \left| \mathcal{O} (x_3) \mathcal{O} (x_4) \rangle, \quad x_i = (\epsilon_i
  + i t_i, \mathbf{x}_i) \right., \quad 0 > \epsilon_3 > \epsilon_4 .
  \label{Ocompl}
\end{equation}
The inner products of states {\eqref{PsiF}} are given by integrals of the
Minkowski 4-point function
\begin{equation}
  \langle \Psi_M (F_1) | \Psi_M (F_2) \rangle \nobracket = \int d x\, G_4^M
  (x_1, x_2, x_3, x_4) \overline{F_1 (x_2, x_1)} F_2 (x_3,
  x_4), \label{PsiF1PsiF2}
\end{equation}
while the natural inner product on the states {\eqref{Ocompl}} is:
\begin{equation}
  \langle \mathcal{O} (x_1) \mathcal{O} (x_2) | \nobracket \mathcal{O} (x_3)
  \mathcal{O} (x_4) \rangle = G_4 (x^{\theta}_2, x^{\theta}_1, x_3, x_4),
  \label{complinner}
\end{equation}
where the OS reflection operation extends to points with complex time
coordinates by:
\begin{equation}
  x = (\varepsilon + i t, \mathbf{x}) \mapsto x^{\theta} = (- \varepsilon + i
  t, \mathbf{x}) . \label{reflCompl}
\end{equation}
The states {\eqref{PsiF}} also have a natural inner product $\langle
\mathcal{O} (x_1) \mathcal{O} (x_2) | \Psi (F) \rangle \nobracket$ with the OS
states.\footnote{We start from the analytically continued Euclidean 4-point
function $G_4 (x_1, x_2, x_3, x_4)$ and take the limit where $x_1, x_2$ are
kept at fixed Euclidean positions, while $x_3, x_4$ approach the Minkowski
space. By Theorem \ref{ThVlad}, the limit is a distribution in $x_3, x_4$, and
the inner product is its pairing with the test function $F$.}

We wish to show that all these new inner products are positive definite and,
moreover, that the new states can be approximated in norm by the smeared OS
2-operator states at Euclidean positions. Note that the positive definiteness
of {\eqref{PsiF1PsiF2}} is precisely Wightman positivity for the 4-point case.

\subsubsection{Wightman states}

Let us start with {\eqref{PsiF1PsiF2}}. Rewriting the inner product in terms
of $W (p_1, p_2, p_3)$, the (distributional) Fourier transform of $G_4^M$
with respect to $y_k = x_k - x_{k + 1}$, we obtain
\begin{equation}
  \langle \Psi_M (F_1) | \Psi_M (F_2) \rangle \nobracket = \int d p\, W (p_1,
  p_2, p_3) [\widehat{F_1} (p_2 - p_1, p_1)]^{\ast}  \hat{F}_2 (p_2 - p_3,
  p_3) \label{pos4-pointFT} .
\end{equation}
We will also need the inner products of the (smeared) OS states
\begin{equation}
  | \nobracket \Psi (H) \rangle = \int d x\, d y\, H (x, y) | \nobracket
  \mathcal{O} (x) \mathcal{O} (y) \rangle \label{PsiH}
\end{equation}
where $H$ is any $C^{\infty}$ function compactly supported at $0 > x^0 >
y^0$. Their inner products are given by
\begin{equation}
  \langle \Psi (H_1) | \Psi (H_2) \rangle \nobracket = \int d x\, G^E_4
  (x_1, x_2, x_3, x_4) \overline{H_1 (x^{\theta}_2, x^{\theta}_1)} H_2
  (x_3, x_4) .
\end{equation}
This can be expressed using the Fourier-Laplace representation {\eqref{FL1}}.
We obtain
\begin{equation}
  \langle \Psi (H_1) | \Psi (H_2) \rangle \nobracket = \int d p\, W
  (p_1, p_2, p_3) \overline{g (H_1) (p_2, p_1)} g (H_2) (p_2, p_3),
  \label{fphipos}
\end{equation}
where $g (H) (p, q)$ is a Schwartz class function related to $H (x, y)$ as
follows. First we form the function $h (y_1, y_2) = H (- y_1, - y_1 - y_2)$
which has support at $y_1^0, y_2^0 > 0$. Next we consider $\tilde{h}$, the
Fourier-Laplace transform of $h (y_1, y_2)$:
\begin{eqnarray}
  & \tilde{h} (p_1, p_2) = \int d y_1\, d y_2\, e^{- p_1^0 y_1^0 + i\mathbf{p}_1
  \cdummy \mathbf{y}_1 - p_2^0 y_2^0 + i\mathbf{p}_2 \cdummy \mathbf{y}_2} h
  (y_1, y_2) . &  \label{fphi}
\end{eqnarray}
Finally, $g (H)$ is an arbitrary Schwartz class function which coincides
with $\tilde{h}$ inside the forward light cones. We also have an analogous
formula for the inner product between states of two types:
\begin{equation}
  \langle \Psi (H) | \Psi_M (F) \rangle \nobracket = \int d p\, W
  (p_1, p_2, p_3) \overline{g (H) (p_2, p_1)} \hat{F}_2 (p_2 - p_3,
  p_3) . \label{PsiHPsiF}
\end{equation}
At this point we recall Lemma \ref{lemma:fcheckdense} from Sec.\ \ref{MinkFromEucl}. That lemma implies that Schwartz functions of the form $g
(H)$ are dense in the Schwartz space. In particular, for any Schwartz $F$, we
can find a sequence of functions $\{ H_r \}_{r = 1}^{\infty}$ such that $g
(H_r) (p_2, p_3) \rightarrow \hat{F} (p_2 - p_3, p_3)$ in the Schwartz space.
Then it follows from {\eqref{pos4-pointFT}}, {\eqref{fphipos}}, {\eqref{PsiHPsiF}}
that
\begin{eqnarray}
  & \langle \Psi (H_r) | \Psi (H_r) \rangle \nobracket \rightarrow
  \langle \Psi_M (F) | \Psi_M (F) \rangle \nobracket, & \\
  & \langle \Psi (H_r) | \Psi_M (F) \rangle \nobracket \rightarrow
  \langle \Psi_M (F) | \Psi_M (F) \rangle \nobracket . &  \nonumber
\end{eqnarray}
From the first equation we conclude that $\langle \Psi_M (F) | \Psi_M
(F) \rangle \nobracket \geqslant 0$, proving Wightman positivity. The two
equations taken together imply that
\begin{equation}
  \langle \Psi (H_r) - \Psi_M (F) | \Psi (H_r) - \Psi_M (F) \rangle
  \nobracket \rightarrow 0,
\end{equation}
i.e.\ OS states can approximate Wightman states in norm.

\subsubsection{OS states for complexified times}\label{CScompl}

Let us discuss next the states {\eqref{Ocompl}} obtained by putting operators
at complexified time positions. In these states we don't take the limit to
Minkowski space, so they are defined without smearing. Using the
Fourier-Laplace representation, their inner product {\eqref{complinner}} is
expressed as
\begin{equation}
  \langle \mathcal{O} (x_1) \mathcal{O} (x_2) | \nobracket \mathcal{O} (x_3)
  \mathcal{O} (x_4) \rangle = G_4 (x^{\theta}_2, x^{\theta}_1, x_3, x_4) =
  \int d p\, W (p_1, p_2, p_3) \overline{f_{x_1, x_2} (p_2, p_1)}
  f_{x_3, x_4} (p_2, p_3),
\end{equation}
where $f_{x, y} (p, q)$, where $0 > \tmop{Re} (x^0) > \tmop{Re} (y^0)$, is
any Schwartz function which agrees with
\begin{equation}
  e^{p^0 x^0 - i\mathbf{p} \cdummy \mathbf{x} - q^0 (x^0 - y^0) +
  i\mathbf{q} \cdummy (\mathbf{x}-\mathbf{y})} .
\end{equation}
for $p, q$ in the forward light cone (where this function is exponentially
decreasing) and extends it somehow outside the light cones (it does not matter
how because $W$ has support in the forward light cones).

Since $f_{x, y}$ is a Schwartz function, it can be approximated by Schwartz
functions of the form $g (H)$ as in the preceding subsection. This implies that non-smeared complexified OS
states can be approximated in norm by Euclidean OS states smeared with
compactly supported test functions. In particular, the inner product
{\eqref{complinner}} is positive definite, providing an extension of pointwise
OS positivity to complexified times:
\begin{equation}
  G_4 (y^{\theta}, x^{\theta}, x, y) \geqslant 0, \qquad (0 > \tmop{Re} x^0
  > \tmop{Re} y^0) .
\end{equation}
As usual, positive-definite inner product implies a Cauchy-Schwarz inequality
for the complexified times:
\begin{equation}
  | G_4 (x_1, x_2, x_3, x_4) |^2 \leqslant G_4 (x_1, x_2,
  x^{\theta}_2, x^{\theta}_1) G_4 (x^{\theta}_4, x^{\theta}_3, x_3, x_4),
  \label{CS4-point}
\end{equation}
valid for $\tmop{Re} x^0_1 > \tmop{Re} x^0_2 > 0 > \tmop{Re} x_3^0 > \tmop{Re}
x^0_4$. The analogues of these properties for conformal blocks will be useful
in Sec.\ \ref{secondpass}.

\begin{remark}
  \label{genrefl}We can extend the reflection operation further for points
  with complexified both time and space coordinates, as
  \begin{equation}
    x = (\varepsilon + i t, \mathbf{x}+ i\mathbf{y}) \mapsto x^{\theta} = (-
    \varepsilon + i t, \mathbf{x}- i\mathbf{y}) . \label{reflCompl1}
  \end{equation}
  With this definition, we can show by the same arguments as above that
  $G_4 (y^{\theta}, x^{\theta}, x, y) \geqslant 0$ (pointwise OS positivity) remains true for $0 \succ (\tmop{Re} x^0, \tmop{Im} \mathbf{x}) \succ
  (\tmop{Re} y^0, \tmop{Im} \mathbf{y})$ where $\eta_1 \succ \eta_2$ means
  $\eta_1 - \eta_2 \succ 0$ (i.e.\ in the forward light cone). We can then show
  that the states $| \nobracket \mathcal{O} (x) \mathcal{O} (y) \rangle$
  make sense for such $x, y$ and can be approximated in norm by integrated
  Euclidean OS states.
\end{remark}

\subsection{Wightman clustering}\label{clusterWightman}

\subsubsection{$2 + 2$ split}

In this section we will derive clustering {\eqref{Wightman:cluster}} for
Wightman 4-point functions (see {\cite{osterwalder1973}}, Sec.\ 4.4). As in
Sec.\ \ref{OSclustering} for the OS case, we will consider 2+2 and 3+1
splits separately. The property we need to prove in the 2+2 case can be
written conveniently in the language of Wightman states $| \nobracket \Psi_M
(F) \rangle$, at our disposal by the discussion in Sec.\ \ref{sec:Wpos}:
\begin{equation}
  \langle \Psi_M (F_1) | \Psi_M \nobracket (U_{\lambda a} F_2) \rangle
  \rightarrow \langle \Psi_M (F_1) | \Omega \nobracket \rangle \langle \Omega
  | \Psi_M \nobracket (F_2) \rangle \label{WCshow}
\end{equation}
as $\lambda \rightarrow \infty$ for any spacelike vector $a$ and any Schwartz
test functions $F_1, F_2$, where $U_{\lambda a}$ is translation: $(U_{\lambda
a} F_2) (x, y) = F_2 (x - \lambda a, y - \lambda a)$, and $\Omega$ is the
vacuum state corresponding to inserting the unit operator. By Lorentz
invariance it's enough to prove this for $a = (0, \mathbf{a})$, purely spatial
vector. In Sec.\ \ref{OSclustering} we showed the OS clustering, which we
can also write using the integrated OS states {\eqref{PsiH}}, as
\begin{equation}
  \langle \Psi (H_1) | \Psi \nobracket (U_{\lambda a} H_2) \rangle
  \rightarrow \langle \Psi (H_1) | \Omega \nobracket \rangle \langle \Omega
  | \Psi \nobracket (H_2) \rangle \label{OSCknow} .
\end{equation}
As explained in Sec.\ \ref{sec:Wpos}, we can find states $| \Psi (H_1)
\rangle \nobracket$ and $| \Psi (H_2) \rangle \nobracket$ which approximate
$| \Psi_M \nobracket (F_1) \rangle$ and $| \Psi_M \nobracket (F_2) \rangle$ in
norm within any $\varepsilon > 0$. Moreover it's obvious from that
construction that the norm is invariant under shifts in purely spatial
direction (i.e.\ the operator $U_{\lambda a}$ is unitary). Hence we have $\|
\Psi (U_{\lambda a} H_2) - \Psi_M (U_{\lambda a} F_2) \| = \| \Psi (H_2)
- \Psi_M (F_2) \| \leqslant \varepsilon$ for any $\lambda$. By these
properties, {\eqref{OSCknow}} implies {\eqref{WCshow}}.\footnote{Indeed we
have $| \langle \Psi_M (F_1) | \Psi_M \nobracket (U_{\lambda a} F_2) \rangle -
\langle \Psi (H_1) | \Psi \nobracket (U_{\lambda a} H_2) \rangle |
\leqslant C \varepsilon$ with some $C$ independent of $\lambda$. Now passing
to the limit $\lambda \rightarrow \infty$ and using {\eqref{OSCknow}} we
obtain $\text{lim sup}_{\lambda \rightarrow \infty} \langle \Psi_M (F_1) |
\Psi_M \nobracket (U_{\lambda a} F_2) \rangle \leqslant \langle \Psi_M (F_1) |
\Omega \nobracket \rangle \langle \Omega | \Psi_M \nobracket (F_2) \rangle +
C' \varepsilon$, and an analogous lower bound on $\text{lim inf}_{\lambda
\rightarrow \infty}$. Since $\varepsilon > 0$ is arbitrary we obtain
{\eqref{WCshow}}.}

\subsubsection{$3 + 1$ split}\label{3+1MinkCluster}

Let us first restate the Euclidean 3+1 clustering argument from Sec.\ \ref{OSclustering} in a somewhat more explicit form, and specializing to
scalars. So let $\varphi (x_1), \chi (x_2, x_3, x_4)$ be two smooth functions
with compact support\footnote{For simplicity, in this section we prove clustering for compactly-supported, as opposed to Schwartz, test functions. We expect that it should be possible to find a proof for Schwartz test functions as well. In any case, the most natural proof would use positivity and the OPE similarly to 2+2 split, provided positivity for higher-point functions is proven (which we don't do in this paper).}
\begin{equation}
  \tmop{supp} (\varphi) \subset \{ x_1^0 > 0 \}, \quad \tmop{supp} (\chi)
  \subset \{ 0 > x_2^0 > x_3^0 > x_4^0 \} . \label{suppreq}
\end{equation}
We would like to show
\begin{equation}
  \lim_{\lambda \rightarrow \infty} (G, \varphi_{\lambda} \otimes \chi) = 0
  \label{toshow31}, \qquad \varphi_{\lambda} \assign \varphi (\cdot - \lambda
  \hat{e}_1)
\end{equation}
where $G = G_4^E$ is the Euclidean 4-point function of four identical scalars, and
$\hat{e}_1$ is the $x^1$ unit vector. The main idea is that we can find a
conformal transformation which moves the point at infinity as well as all the
other points to some finite positions. The suppression of the integral then
comes from the Jacobian of this transformation. Consider a special conformal
transformation $f (x) = \frac{x^{\mu} + x^2 b^{\mu}}{1 + 2 x \cdummy b + x^2
b^2} =\mathcal{J} \circ T_b \circ \mathcal{J}$, where $\mathcal{J}$ is
inversion and $T_b$ is a translation by $b = \hat{e}_1$. We have $f (-
\hat{e}_1) = \infty$, while $f$ is non-singular on $\tmop{supp}
(\varphi_{\lambda})$ and $\tmop{supp} (\chi)$. We also have $f (\infty) =
\hat{e}_1 .$ By conformal invariance we have (compare {\eqref{MfinInv}}) $(G,
\Phi) = (G, \Phi^f)$ where $\Phi^f (x_1, \ldots, x_4) = \Phi (f^{- 1} (x_1),
\ldots, f^{- 1} (x_4)) \prod_{i = 1}^4 J (f^{- 1} (x_i))^{\Delta_{\mathcal{O}}
- d}$, where $J (x) = \frac{1}{1 + 2 x \cdummy b + x^2 b^2}$. We apply this
equation with $\Phi = \varphi_{\lambda} \otimes \chi$. The function $\chi$ is
mapped by this transformation to some smooth function. Suppression of the
integral in the limit $\lambda \rightarrow \infty$ will come from the
transformation of $\varphi_{\lambda}$, which is mapped to
\begin{equation}
  \varphi_{\lambda}^f (x_1) \assign \varphi (f^{- 1} (x_1) - \lambda
  \hat{e}_1) J (f^{- 1} (x_1))^{\Delta_{\mathcal{O}} - d} .
\end{equation}
Namely we have
\begin{equation}
  | (G, \varphi_{\lambda}^f \otimes \chi^f) | \leqslant C (\lambda) I, \quad I
  = \int d x_1\, | \varphi_{\lambda}^f (x_1) |, \quad C (\lambda) = \sup_{x_1
  \in \tmop{supp} \varphi_{\lambda}^f} | (G (x_1, \cdot), \chi^f) | .
\end{equation}
The function $\varphi_{\lambda}^f$ is nonzero for $f^{- 1} (x_1) \in
\tmop{supp} (\varphi) + \lambda \hat{e}_1$, which is a point near infinity for
$\lambda$ large. We conclude that $\varphi_{\lambda}^f$ is supported in a small
neighborhood, order $1 / \lambda$, of $f (\infty) = \hat{e}_1$. Since $G$
is real-analytic at nonzero point separation, this implies that $C
(\lambda)$ is bounded by some constant for $\lambda \geqslant \lambda_0$.
To compute $I$, we do the change of variables $x_1 = f (y)$:
\begin{equation}
  I = \int d y\, | \varphi (y - \lambda \hat{e}_1) | J
  (y)^{\Delta_{\mathcal{O}}} \sim \frac{\tmop{const}}{\lambda^{2
  \Delta_{\mathcal{O}}}} .
\end{equation}
This finishes the proof of Euclidean 3+1 clustering, Eq.\ {\eqref{toshow31}}.

Let us proceed next to show Wightman 3+1 clustering. We will show the same
equation as {\eqref{toshow31}}, namely
\begin{equation}
  \lim_{\lambda \rightarrow + \infty} (G, \varphi_{\lambda} \otimes \chi) = 0,
\end{equation}
where now $G = G_4^M$ is the Minkowski 4-point function, which is a tempered
distribution, and $\varphi (x_1)$ and $\chi (x_2, x_3, x_4)$ are arbitrary
compactly supported test functions (i.e.\ no support requirements analogous to
{\eqref{suppreq}}).\footnote{The method described below cannot be
straightforwardly generalized to the case of Schwartz test functions.} \ The
proof will be based on the same idea of moving the point at infinity to a
finite position, paying attention to $G$ now being distribution, and to the
requirement {\eqref{req1}} on invariance under finite Minkowski conformal
transformations.

We will use the same transformation $f (x) = \frac{x^{\mu} + x^2 b^{\mu}}{1 +
2 x \cdummy b + x^2 b^2}$, $b = \hat{e}_1$. By translation invariance, we may
assume that $\tmop{supp} (\chi)$ lies at larger $x^1$ values than of the
singularity light cone $x^0 = \pm | \mathbf{x}+ \hat{e}_1 |$ of this
transformation (see Sec.\ \ref{ConfMink}). For sufficiently large $\lambda$,
$\tmop{supp} (\varphi_{\lambda})$ will also satisfy this condition. As we
scale $b$ to zero to connect $f$ to the identity, the singularity light cone
moves away to infinity along the negative $x^1$ direction, without touching
$\tmop{supp} (\varphi_{\lambda})$ nor $\tmop{supp} (\chi)$, see Fig.\ \ref{3+1Mink}. Hence requirement {\eqref{req1}} is satisfied and we may apply
invariance {\eqref{MfinInv}}, which says $(G, \varphi_{\lambda} \otimes \chi)
= (G, \varphi_{\lambda}^f \otimes \chi^f)$.

\begin{figure}[h]\centering
  \raisebox{-0.504762041315234\height}{\includegraphics[width=10.1973796405615cm,height=4.57459333595697cm]{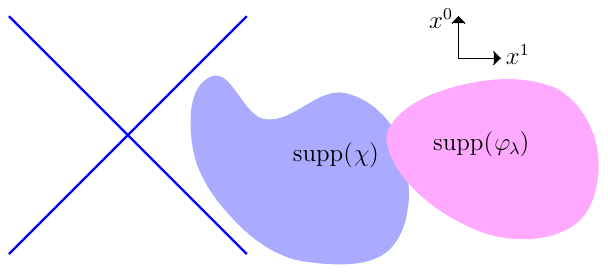}}
  \caption{Location of supports of $\varphi_{\lambda}$ and $\chi$ with respect
  to the singularity light cone of $f$.\label{3+1Mink}}
\end{figure}

Now, using translation invariance of the 4-point function, $G (x_1, x_2, x_3, x_4)
= \tilde{G} (x_2 - x_1, x_3 - x_1, x_4 - x_1)$ we may write
\begin{equation}
  (G, \varphi_{\lambda}^f \otimes \chi^f) = \int d x_1\, \varphi_{\lambda}^f
  (x_1) F (x_1), \label{intMink}
\end{equation}
where $F (x_1) = (\tilde{G}, T_{x_1} \cdot \chi^f)$ and $T_a$ is a
translation. $\tilde{G}$ is a distribution, but since translation is a
continuous operation in the space of test functions, we know that $F (x_1)$ is
a continuous function of $x_1$. When $\lambda$ goes to $+ \infty$, the support
of $\varphi_{\lambda}^f$ shrinks to the point $\hat{e}_1$.\footnote{It is
important for the argument that, as one can easily check, $\tmop{supp}
(\varphi_{\lambda}^f)$ shrinks to a compact set (in fact, a point) and not,
say, spreads out along some light cone.} Hence for $\lambda \geqslant
\lambda_0$ we can bound $| F |$ on $\tmop{supp} (\varphi_{\lambda}^f)$ by a
constant, and estimate {\eqref{intMink}} in absolute value by $\tmop{const}
\times \int d x_1\,  | \varphi_{\lambda}^f (x_1) |$. This remaining integral is
computed via the change of variables as the Euclidean one, and goes to zero as
$\lambda^{- 2 \Delta_{\mathcal{O}}}$, completing the proof.

\subsection{Local commutativity}\label{local-comm}

Let us show that the constructed Minkowski correlators satisfy local
commutativity. This follows by a robust argument which uses only Lorentz
invariance, analyticity in the forward tube, existence of the boundary
distribution, together with real analyticity and permutation symmetry of the Euclidean correlators away from
coincident points (OS {\cite{osterwalder1973}}, Sec.\ 4.5). Here for
completeness we will provide this argument for $n$-point functions which is
its natural setting. In Sec.\ \ref{localCFT} below we will make some remarks
specific to CFT 4-point functions.

So, we start from the Euclidean correlator $G^E (x_1, \ldots, x_n)$ at $x_1^0
> x^0_2 > \cdots > x_n^0$ and its analytic continuation $G (x_1, \ldots,
x_n)$ to the forward tube $\mathcal{T}_n$ which is the set of points $x_k \in
\mathbb{C}^d$ such that their differences $y_k = x_k - x_{k + 1}$ satisfy
$\tmop{Re} y_k^0 > | \tmop{Im} \mathbf{y}_k |$ or equivalently $\eta_k \succ
0$ in terms of $\zeta_k = (i y_k^0, \mathbf{y}_k) = \xi_k + i \eta_k$, $\xi_k,
\eta_k \in \mathbb{R}^{1, d - 1}$. We will write $G$ instead of $G_n$. We know
by Theorem \ref{ThVlad} that this analytic continuation is invariant under
Lorentz transformations $\zeta_k \rightarrow \Lambda \zeta_k$ where $\Lambda
\in L_+^{\uparrow}$, the identity component of the real Lorentz group. 
Since $G$ is translationally invariant, it depends only on
$\zeta_k$, and we will abuse of notation by sometimes writing $G (\zeta_1,
\ldots, \zeta_{n - 1})$ and $(\zeta_1, \ldots, \zeta_{n - 1}) \in
\mathcal{T}_n$ instead of $G (x_1, \ldots, x_n)$ and $(x_1, \ldots, x_n)
\in \mathcal{T}_n$.

\tmtextbf{Step 1.} We will extend domain of analyticity of $G$ using the
complex Lorentz group $L (\mathbb{C})$, defined as the set of complex matrices
$A$ preserving the Minkowski metric, i.e.\ $A^T g A = g$ where $g = \tmop{diag}
(- 1, 1 \ldots 1)$. We will only need the component of $L (\mathbb{C})$
connected to the identity, denoted $L_+ (\mathbb{C})$. For any $\Lambda \in
L_+ (\mathbb{C})$ consider the equation
\begin{equation}
  G (\zeta_1, \ldots, \zeta_{n - 1}) = G (\Lambda^{- 1} \zeta_1, \ldots,
  \Lambda^{- 1} \zeta_{n - 1}) . \label{BHWeq}
\end{equation}
The two sides of this equation coincide for real $\Lambda \in L_+^{\uparrow}$
(by Lorentz invariance of $G_n$), and hence by analyticity in the components
of $\Lambda$ also for complex $\Lambda \in L_+ (\mathbb{C})$, at least for
$\Lambda$ close to 1. In other words, Eq.\ {\eqref{BHWeq}} is just an identity
if $\Lambda \approx 1$ and the arguments of $G_n$ on both sides are in the
forward tube. But a general $\Lambda \in L_+ (\mathbb{C})$ does not preserve
the forward tube. For such $\Lambda$, Eq.\ {\eqref{BHWeq}} extends analytically
$G$ from the forward tube to the set
\begin{equation}
  \mathcal{T}_n' = \bigcup_{\Lambda \in L_+ (\mathbb{C})} \Lambda \cdot
  \mathcal{T}_n,
\end{equation}
called the extended tube. The Bargmann-Hall-Wightman theorem shows that no
further topological obstructions arise in this analytic continuation; see
{\cite{jost1979general}}, p.78 for details. Call this extension $\tilde{G}$.

\tmtextbf{Step 2.} Let us consider $\tilde{G} (x_1, \ldots, x_n)$ for
\begin{equation}
  \epsilon_1 > \ldots > \epsilon_{k - 1} > 0 > \epsilon_{k + 2} > \ldots >
  \epsilon_n,
\end{equation}
while assuming that $\epsilon_k, \epsilon_{k + 1}$ are near zero and much
smaller than other $\epsilon_i$'s, and $| t_k - t_{k + 1} | < |\mathbf{x}_k
-\mathbf{x}_{k + 1} | \nobracket$, $\mathbf{x}_k$, $\mathbf{x}_{k + 1}$ real,
so that $x_k - x_{k + 1}$ approaches a spacelike separation. For $\epsilon_k >
\epsilon_{k + 1}$ this configuration is in the forward tube, so we know
$\tilde{G}$ is analytic there and agrees with $G (x_1, \ldots, x_n)$. Let us
show that the configurations with $\epsilon_k < \epsilon_{k + 1}$ are in the
extended tube. We may set $x_{k + 1} = 0$ for this argument, so that
\begin{equation}
  \zeta_k = (t_k + i \epsilon_k, \mathbf{x}_k) .
\end{equation}
We may assume without loss of generality that $\mathbf{x}_k = (x^1_k, 0,
\ldots 0)$, $x_k^1 > | t_k |$. Then acting on $\zeta_k$ with the complexified
Lorentz transformation
\begin{equation}
  \Lambda_{\theta} = \left(\begin{array}{cc}
    \cosh (i \theta) & \sinh (i \theta)\\
    \sinh (i \theta) & \cosh (i \theta)
  \end{array}\right) \in L_+ (\mathbb{C}),
\end{equation}
with small $\theta$ we get, using $\Lambda_{\theta} \approx
\left(\begin{array}{cc}
  1 & i \theta\\
  i \theta & 1
\end{array}\right)$, $\zeta_k' = \Lambda_{\theta} \zeta_k \approx (t_k, x^1_k)
+ i (\theta x^1_k + \epsilon_k, \theta t_k)$, and thus $\eta_k' \approx
(\theta x^1_k + \epsilon_k, \theta t_k)$. If $\epsilon_k$ is negative but very
small, we can can achieve $\eta_k' \succ 0$ by choosing an appropriate small
$\theta$. We need $\theta$ small so that all the other $\zeta_i'$ remain in
the forward light cone, and this will work because we are assuming that
$\epsilon_k$ is very much smaller than all the other $\epsilon_i$'s.

The bottom line is that the extended tube contains an open set of
configurations as above, with $| t_k - t_{k + 1} | < |\mathbf{x}_k
-\mathbf{x}_{k + 1} | \nobracket$ and $\epsilon_k, \epsilon_{k + 1}$ small,
with $\epsilon_k - \epsilon_{k + 1}$ of any sign. Let us call this set
$\mathcal{Q}_{n, k}$. By restricting this set a bit, we may assume that
$\mathcal{Q}_{n, k}$ is invariant under permutations of $x_k$ and $x_{k + 1}$.
By Step 1 we know that function $\tilde{G}$ is holomorphic in the extended tube
and hence also in $\mathcal{Q}_{n, k}$. In particular, it is analytic if we
set $\epsilon_k = \epsilon_{k + 1} = 0$. This already has an interesting
consequence: The Minkowski correlator is analytic with respect to a pair of
spacelike-separated points (while it remains a distribution with respect to
all the other points). The projection of $\mathcal{Q}_{n, k}$ to the plane
$(\epsilon_k - \epsilon_{k + 1}, t_k - t_{k + 1})$ is shown schematically in
Fig.\ \ref{Qnk}.

\begin{figure}[h]\centering
  \raisebox{-0.496957195215242\height}{\includegraphics[width=5.07528204119113cm,height=4.30852354715991cm]{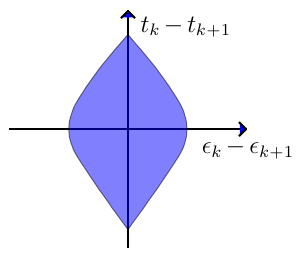}}
  \caption{\label{Qnk}Projection of the set $\mathcal{Q}_{n, k}$, where the
  function $\tilde{G}$ is holomorphic, to the plane $(\epsilon_k - \epsilon_{k +
  1}, t_k - t_{k + 1})$. The vertical extent of this region is determined by
  the condition $| t_k - t_{k + 1} | < | \mathbf{x}_k -\mathbf{x}_{k - 1} |$.
  The horizontal extent is determined, among other things, by the condition
  that $\epsilon_k, \epsilon_{k + 1}$ have to be much smaller that all the
  other $\epsilon_i$'s.}
\end{figure}

The set $\mathcal{Q}_{n, k}$ contains real configurations (horizontal axis in
Fig.\ \ref{Qnk}, setting other $t_i \rightarrow 0$ as well). Restriction of
$\tilde{G}$ to the real part of $\mathcal{Q}_{n, k}$ agrees with the Euclidean
correlator $G^E$. (They agree for $\epsilon_k > \epsilon_{k + 1}$ by
construction and for $\epsilon_k < \epsilon_{k + 1}$ by the uniqueness of
analytic continuation. Recall that the Euclidean correlator $G^E$ is real
analytic everywhere away from coincident points, i.e.\ for $\epsilon_k -
\epsilon_{k + 1}$ of any sign as long as $\mathbf{x}_k \neq \mathbf{x}_{k -
1}$.) One consequence of this fact is that $\tilde{G}$ restricted to the real
part $\mathcal{Q}_{n, k}$ is permutation invariant w.r.t. $x_k \leftrightarrow
x_{k + 1}$:
\begin{equation}
  \tilde{G} (\ldots x_k, x_{k + 1} \ldots) = \tilde{G} (\ldots
  x_{k + 1}, x_k \ldots), \label{Gtperm}
\end{equation}
because the Euclidean correlator has this property. Finally, since
$\mathcal{Q}_{n, k}$ is connected to the real configurations (see Fig.\ \ref{Qnk}), we conclude that permutation invariance {\eqref{Gtperm}} holds
everywhere in $\mathcal{Q}_{n, k}$.\footnote{In fact, $\tilde{G}$ can be
extended to a single-valued holomorphic function on the ``permuted extended
tube'' $\bigcup_{\pi \in S^n} \pi \mathcal{T}_n'$, and satisfied permutation
invariance {\eqref{Gtperm}} on this large set. See {\cite{jost1979general}},
App.\ II, {\cite{Tomozawa}} and {\cite{bogolubov2012general}}, Sec.\ 9.D.
However for our purposes analyticity and permutation invariance on
$\mathcal{Q}_{n, k}$ will suffice.}

We now see the meaning of $\tilde{G}$ for configurations with $\epsilon_k <
\epsilon_{k + 1}$. Via permutation invariance {\eqref{Gtperm}}, such
configurations are mapped to the forward tube and hence can be evaluated as
$G$ for the permuted configurations.

\tmtextbf{Step 3.} We are now ready to show local commutativity. We have to
prove that boundary value limits of two holomorphic functions agree:
\begin{equation}
  \lim_{\epsilon_i \rightarrow 0} G(\ldots x_k, x_{k + 1} \ldots) =
  \lim_{\epsilon_i \rightarrow 0} G(\ldots x_{k + 1}, x_k
  \ldots) \text{} \label{perm1},
\end{equation}
when approaching a Minkowski configuration in which $x_k - x_{k + 1}$ is
spacelike. Note that, by the original definition, the two limits are from
different forward tubes: the first one must respect the condition $\epsilon_k
> \epsilon_{k + 1}$, while the second $\epsilon_{k + 1} > \epsilon_k$. By
Theorem \ref{ThVlad}, Part 3, we can take the limits $\epsilon_i \rightarrow
0$ in any order, so let us send $\epsilon_k, \epsilon_{k + 1} \rightarrow 0$
first, while keeping other $\epsilon_i$ fixed for the moment. For very small
$\epsilon_k, \epsilon_{k + 1}$, the configurations on both sides will be in
$\mathcal{Q}_{n, k}$ where both sides are restrictions of the function
$\tilde{G}$ analytic around $\epsilon_k, \epsilon_{k + 1} = 0$ and
satisfying permutation invariance {\eqref{Gtperm}}. It follows that the two
sides of {\eqref{perm1}} agree in the limit $\epsilon_k, \epsilon_{k + 1}
\rightarrow 0$. Sending the remaining $\epsilon_i \rightarrow 0$ we recover
the local commutativity.

\subsubsection{Local commutativity for CFT 4-point functions}\label{localCFT}

In this paper we analytically continued the CFT 4-point function $\langle
\mathcal{O} (x_1) \mathcal{O} (x_2) \mathcal{O} (x_3) \mathcal{O} (x_4)
\rangle$ to the forward tube using $\rho, \bar{\rho}$ coordinates. We would
like to indicate here that this provides an alternative path to understanding
local commutativity. 
{We have shown previously that $0<|\rho|,|\bar\rho|<1$ in the forward tube.
Since the extended tube is obtained from the forward tube by complexified Lorentz transformations
and $\rho,\bar\rho$ are invariant under such transformations, it follows that $0<|\rho|,|\bar\rho|<1$ also in the extended tube. Below we will show this explicitly for the configurations used
in the proof of local commutativity.}
We consider separately $k = 1$ and $k = 2$ ($k = 3$ being
analogous to $k = 1$).

\tmtextbf{$k = 1$:} Here $x_1, x_2$ approach spacelike-separated Minkowski
points. We know that \ $| \rho |, | \bar{\rho} | < 1$ in $\mathcal{D}_4$,
$\epsilon_1 > \epsilon_2 > \epsilon_3 > \epsilon_4$. Extended tube analyticity
suggests that this must remain true also for $\epsilon_1 = \epsilon_2 >
\epsilon_3 > \epsilon_4$. Indeed, this follows from critical rereading of the
proof of Lemma \ref{bound} (Sec.\ \ref{PetrProof}, Eq.\ {\eqref{quadeq}} and
below). (That proof does not use the condition $\eta_1 \succ \eta_2$ but only
$\eta_1, \eta_2 \succ 0$.\footnote{An alternative argument is as follows. In
Sec.\ \ref{secondpass} we will show the Cauchy-Schwarz inequality for $\rho,
\bar{\rho}$, Theorem \ref{boundThm}, which bounds $\rho, \bar{\rho}$ for any
configuration in the forward tube with $\epsilon_1 > \epsilon_2 > 0 >
\epsilon_3 > \epsilon_4$ in terms of $\rho, \bar{\rho}$ of
``reflection-symmetric'' configurations having $\epsilon_3 = - \epsilon_2$,
$\epsilon_4 = - \epsilon_1$. The proof of Lemma \ref{boundLemma}, Eq.
{\eqref{z12zbar12}} shows that $\rho, \bar{\rho}$ remain less than 1 for the
latter configurations in the limit $\epsilon_1 \rightarrow \epsilon_2$.}) It
is also important for analyticity that $\rho, \bar{\rho}$ not vanish. In the
forward tube $\rho, \bar{\rho}$ do not vanish because $x^2_{i j} \neq 0$, $i <
j$ (Lemma \ref{xij2h}). When $\epsilon_1 = \epsilon_2$ we have $x^2_{12} > 0$
(spacelike separation), hence also nonzero. These observations show that the
CFT 4-point function can be analytically extended, using the $\rho, \bar{\rho}$
expansion, to a neighborhood of points with $\epsilon_1 = \epsilon_2 > 0 >
\epsilon_3 > \epsilon_4$, $x^2_{12} > 0$, in agreement with the general QFT
arguments given above.

Let us now permute the first two points: $(\epsilon_1 + i t_1, \mathbf{x}_1)
\leftrightarrow (\epsilon_2 + i t_2, \mathbf{x}_2)$. In the Euclidean region,
this transformation maps $\rho \rightarrow - \rho, \bar{\rho} \rightarrow -
\bar{\rho}$ and leaves the 4-point function of identical scalars invariant because
the expansion {\eqref{g:rhoexpansion}} contains only even $m$. The same
transformation remains true for complexified times for spacelike separation.
Taking the limit $\epsilon_1, \epsilon_2 \rightarrow 0$, we recover local
commutativity very explicitly.

\tmtextbf{$k = 2$:} Now we are interested in the limit $\epsilon_2 \rightarrow
\epsilon_3$ from inside $\epsilon_1 > \epsilon_2 > \epsilon_3 >
\epsilon_4$.\footnote{\label{noteShock}The discussion on the local
commutativity of this type can also be found in the study of causality in a
shockwave background (see Sec.\ 5 of {\cite{Hartman:2015lfa}}). In
{\cite{Hartman:2015lfa}}, the 2-point function in a shockwave background is
defined by $\langle \mathcal{O} (x) \mathcal{O} (y) \rangle_{\Psi} \assign
\frac{\langle \Psi (i \delta) \mathcal{O} (x) \mathcal{O} (y) \Psi (- i
\delta) \rangle}{\langle \Psi (i \delta) \Psi (- i \delta) \rangle}$, where
$x$, $y$ are Minkowski points and ``$i \delta$'' means the Euclidean point
$(\delta, 0, \ldots, 0)$. In our language it corresponds to the 4-point
function $\langle \Psi (x_1) \mathcal{O} (x_2) \mathcal{O} (x_3) \Psi (x_4)
\rangle$ with $\varepsilon_1 = - \varepsilon_4 = \delta > 0$ and
$\varepsilon_2 = \varepsilon_3 = 0$. We know that the 4-point function is
regular analytic at such configurations. So the commutator $[\mathcal{O} (x),
\mathcal{O} (y)]$ vanishes in the shockwave background when $x$ and $y$ are
spacelike separated.} As for $k = 1$, critical rereading of the proof of
Lemma \ref{bound} shows that $| \rho |, | \bar{\rho} |$ remain less than 1.
(We put in that proof $\zeta_2 = \xi_2 + i \eta_2$, $\xi_2 = (t_2,
\mathbf{x}_2)$ spacelike, and $\eta_2 = (\epsilon_2, \tmmathbf{0})$,
$\epsilon_2 > 0$. The proof does not use the condition $\eta'_2 \succ 0$ but
only $\eta'_{24} \succ 0$. The latter condition remains true for $\epsilon_2
\rightarrow \epsilon_3 = 0$, as $\zeta_2'$ goes to a finite real vector.)
Hence, the CFT 4-point function can be analytically extended, using the $\rho,
\bar{\rho}$ expansion, to a neighborhood of points with $\epsilon_1 >
\epsilon_2 = \epsilon_3 > \epsilon_4$, $x^2_{23} < 0$.

To finish the proof of local commutativity, we fall back on the general
argument, appealing to the permutation invariance of the (real-analytic) CFT
4-point function under $x_2 \leftrightarrow x_3$. (Unlike for $k = 1$, the
s-channel OPE expansion {\eqref{g:rhoexpansion}} cannot be used to make this
step more explicit, as it does not manifestly have this invariance.)

\subsection{Generalization to non-identical scalars}\label{nonId}

In the previous subsections we proved that the 4-point function of
identical scalars has analytic continuation to the forward tube
$\mathcal{T}_4$, and its boundary value in the Minkowski region is a tempered
distribution. Then Minkowski conformal invariance, Wightman positivity,
Wightman clustering and local commutativity follow from their Euclidean
analogues.

In this section we will indicate how to generalize analytic continuation and
temperedness to 4-point functions of non-identical scalars. The proof of
the other properties is the same as in the case of identical scalars.

We consider the 4-point function of scalar primary operators
$\mathcal{O}_i$ with scaling dimensions $\Delta_i$,
\begin{eqnarray}
  G^E_{1234} (c_E) & \assign & \langle \mathcal{O}_1 (x_1) \mathcal{O}_2 (x_2)
  \mathcal{O}_3 (x_3) \mathcal{O}_4 (x_4) \rangle \nonumber\\
  & = & \frac{1}{(x_{12}^2)^{\frac{\Delta_1 + \Delta_2}{2}}
  (x_{34}^2)^{\frac{\Delta_3 + \Delta_4}{2}}} \left( \frac{x_{24}^2}{x_{14}^2}
  \right)^{\frac{\Delta_1 - \Delta_2}{2}} \left( \frac{x_{14}^2}{x_{13}^2}
  \right)^{\frac{\Delta_3 - \Delta_4}{2}} g_{1234} (c_E), 
  \label{def:Eucl4-pointgeneral}
\end{eqnarray}
which reduces to (\ref{def:Euclidean4-point}) when $\Delta_i$'s are identical. The
analytic continuation of the prefactor to the forward tube $\mathcal{T}_4$ is
straightforward. The function $g_{1234} (c_E)$ only depends on the conformal
equivalence class of $c_E$, i.e.\ $g_{1234} (c_E) = g_{1234} (\rho (c_E),
\bar{\rho} (c_E))$. By the similar argument to that in Sec.\ \ref{Eucl4-point},
the function $g_{1234} (c_E)$ has the following series expansion
\begin{equation}
  g_{1234} (c_E) = \left[ \frac{(1 - \rho) (1 - \bar{\rho})}{(1 + \rho) (1 +
  \bar{\rho})} \right]^{\frac{\Delta_1 - \Delta_2 - \Delta_3 + \Delta_4}{2}}
  \underset{\delta, m}{\sum} a_{12} (\delta, m) a_{\bar{4} \bar{3}} (\delta,
  m)^{\ast} r^{\delta} e^{i m \theta}, \qquad \rho (c_E) = r e^{i m \theta},
\end{equation}

where the sum runs over a discrete set of pairs $(\delta, m)$ with $\delta
\geqslant 0$, $m \in \mathbb{Z}$ (not necessarily even for non-identical
scalars), $| m | \leqslant \delta$. Analogously to the case of identical
scalars, the sum is absolutely convergent when $| \rho (c_E) | < 1$ (see
below). Also, when $d \geqslant 3$ we have $p_{\delta, - m} = p_{\delta, m}$,
where $p_{\delta, m} = a_{12} (\delta, m) a_{\bar{4} \bar{3}} (\delta,
m)^{\ast}$. Analogously to Sec.\ \ref{anal4-point}, the analytic continuation of
$g_{1234} (c)$ in $d \geqslant 3$ will be given by the formula (compare
(\ref{eq:gtilde}))
\begin{equation}
  g_{1234} (c) = \left( \frac{x_{14}^2 x_{23}^2}{x_{13}^2 x_{24}^2}
  \right)^{\frac{\Delta_1 - \Delta_2 - \Delta_3 + \Delta_4}{4}} \sum_{m,
  \delta, 0 \leqslant m \leqslant \delta} p_{\delta, m} R_{\delta / 2 - m /
  2} (c) \Phi_m (c) . \label{eq:gtildegeneral}
\end{equation}
In $d = 2$, $p_{\delta, m} \neq p_{\delta, - m}$ but the functions $\rho
(c)^m$ and $\bar{\rho} (c)^m$ are individually holomorphic. In this case the
analytic continuation of $g_{1234} (c)$ is given by the formula (compare
{\eqref{eq:gtilde2d}}):
\begin{equation}
  g_{1234} (c) = \left( \frac{x_{14}^2 x_{23}^2}{x_{13}^2 x_{24}^2}
  \right)^{\frac{\Delta_1 - \Delta_2 - \Delta_3 + \Delta_4}{2}} \sum_{m,
  \delta, 0 \leqslant m \leqslant \delta} R_{\delta / 2 - m / 2} (c) 
  [p_{\delta, m} \rho (c)^m + p_{\delta, - m} \bar{\rho} (c)^m] .
  \label{eq:gtilde2dgeneral}
\end{equation}
We would like to show that

(a) when $r = \max \{ | \rho |, | \bar{\rho} | \} < 1$, the series
\begin{equation}
  \tilde{g}_{1234} (\rho, \bar{\rho}) = \underset{\delta, m}{\sum} a_{12}
  (\delta, m) a_{\bar{4} \bar{3}} (\delta, m)^{\ast} \rho^{(\delta + m) / 2}
  \bar{\rho}^{(\delta - m) / 2} \label{gtildeexp}
\end{equation}
is absolutely convergent;

(b) the remainder $\tilde{g}_{1234} (\rho, \bar{\rho} ; \delta_{\ast}) \assign
\underset{\delta \geqslant \delta_{\ast}, m}{\sum} a_{12} (\delta, m)
a_{\bar{4} \bar{3}} (\delta, m)^{\ast} \rho^{(\delta + m) / 2}
\bar{\rho}^{(\delta - m) / 2}$ has a powerlaw bound, uniform in
$\delta_{\ast}$:
\begin{equation}
  | \tilde{g}_{1234} (\rho, \bar{\rho} ; \delta_{\ast}) | \leqslant C (1 -
  r)^{- \Delta_1 - \Delta_2 - \Delta_3 - \Delta_4} .
  \label{gtildegeneral:bound}
\end{equation}
This is done as follows (compare \ {\cite{paper1}}, Sec.\ 4.2). Consider the
4-point functions $\langle \mathcal{O}_1 \mathcal{O}_2 \mathcal{O}_2^{\dag}
\mathcal{O}_1^{\dag} \rangle$, $\langle \mathcal{O}_4^{\dag}
\mathcal{O}_3^{\dag} \mathcal{O}_3 \mathcal{O}_4 \rangle$, and let
$\tilde{g}_{12 \bar{2} \bar{1}}$, $\tilde{g}_{\bar{4} \bar{3} 34}$ be the
analogues of (\ref{gtildeexp}):
\begin{eqnarray}
  \tilde{g}_{12 \bar{2} \bar{1}} (\rho, \bar{\rho}) & = & \underset{\delta,
  m}{\sum} | a_{12} (\delta, m) |^2 \rho^{(\delta + m) / 2}
  \bar{\rho}^{(\delta - m) / 2},  \nn\\
  \tilde{g}_{\bar{4} \bar{3} 34} (\rho, \bar{\rho}) & = & \underset{\delta,
  m}{\sum} | a_{\bar{4} \bar{3}} (\delta, m) |^2 \rho^{(\delta + m) / 2}
  \bar{\rho}^{(\delta - m) / 2} .  \label{gtildeexp:34}
\end{eqnarray}
Noticing that $| m | \leqslant \delta$, we estimate (\ref{gtildeexp}) by
absolute value and apply Cauchy-Schwarz inequality:
\begin{equation}
  | \tilde{g}_{1234} (\rho, \bar{\rho} ; \delta_{\ast}) | \leqslant
  \underset{\delta, m}{\sum} | a_{12} (\delta, m) | | a_{\bar{4} \bar{3}}
  (\delta, m) | r^{\delta} \leqslant [\tilde{g}_{12 \bar{2} \bar{1}} (r, r)
  \tilde{g}_{\bar{4} \bar{3} 34} (r, r)]^{1 / 2} . \label{gtilde:CS}
\end{equation}
The functions $\tilde{g}_{12 \bar{2} \bar{1}} (r, r)$ and $\tilde{g}_{\bar{4}
\bar{3} 34} (r, r)$ correspond to the 4-point functions at the Euclidean
configurations with $\rho = \bar{\rho} = r < 1$, hence their series expansions
(\ref{gtildeexp:34}) are convergent by the Euclidean OPE
axiom. Therefore, (\ref{gtildeexp}) is absolutely convergent when $| \rho |, |
\bar{\rho} | < 1$. This finishes the proof of part (a).

Using the t-channel OPE, we can show that for $0 \leqslant r < 1$,
\begin{eqnarray}
  \tilde{g}_{12 \bar{2} \bar{1}} (r, r) & \leqslant & C (1 - r)^{- 2 \Delta_1
  - 2 \Delta_2}, \nn
\\
  \tilde{g}_{\bar{4} \bar{3} 34} (r, r) & \leqslant & C (1 - r)^{- 2 \Delta_3
  - 2 \Delta_4},  \label{gtildebound:34}
\end{eqnarray}
with some $C > 0$. Combining 
(\ref{gtildebound:34})
with (\ref{gtilde:CS}) we get (\ref{gtildegeneral:bound}). This finishes the
proof of part (b).

\section{Optimal powerlaw bound from Cauchy-Schwarz $\rho, \bar{\rho}$
inequality}\label{secondpass}

In Sec.\ \ref{power4-point} we provided a powerlaw bound for the 4-point function,
based on the inequality {\eqref{rhobound}} for $\max (| \rho (c) |, |
\bar{\rho} (c) |)$. That did the job of allowing us to apply Theorem
\ref{ThVlad} and prove that the Minkowski 4-point function is a distribution, but
the actual bound {\eqref{rhobound}} is not optimal. It is interesting to get a
better bound on $| \rho (c) |, | \bar{\rho} (c) |$, because this will
translate into a better powerlaw bound for the 4-point function, allowing us to
get a better idea about the regularity of the Minkowski 4-point function as a
distribution, i.e.\ how many derivatives test functions must have. In the proof
of Theorem \ref{ThVlad}, parameters $A_n$ and $B_n$ of the powerlaw bound
enter into Eq.\ {\eqref{GMreg}} which provides an upper bound on the
regularity.

In this section we will provide such an optimal bound on $| \rho (c) |, |
\bar{\rho} (c) |$. The main idea of the bound and of its proof is inspired by
Sec.\ \ref{CScompl}. Let us denote by $\mathcal{D}^{(0)}_4$ the subset of
configurations $c \in \mathcal{D}_4$ satisfying the condition $\tmop{Re} x_1^0
> \tmop{Re} x_2^0 > 0 > \tmop{Re} x_3^0 > \tmop{Re} x_4^0$. We showed that the
4-point functions for complexified times satisfy the Cauchy-Schwarz inequality
{\eqref{CS4-point}} for $c = (x_1, x_2, x_3, x_4) \in \mathcal{D}^{(0)}_4$. For a
general configuration $c \in \mathcal{D}^{(0)}_4$ we define two configurations
\begin{equation}
  \label{def:configC12C34} c_{12} = (x_1, x_2, x_2^{\theta}, x_1^{\theta}),
  \quad c_{34} = (x_4^{\theta}, x_3^{\theta}, x_3, x_4),
\end{equation}
where $\theta$ is the operation in {\eqref{reflCompl}} which generalizes the
OS reflection to complexified times. We will call such configurations, for
obvious reasons, reflection-symmetric. It is clear that both $c_{12}, c_{34}
\in \mathcal{D}^{(0)}_4$. Eq.\ {\eqref{CS4-point}} can now be written as
\begin{equation}
  | G_4 (c) |^2 \leqslant G_4 (c_{12}) G_4 (c_{34}) \qquad (c \in
  \mathcal{D}^{(0)}_4) . \label{Gc12c34}
\end{equation}
Since we know that $G_4$ can be written as a convergent power series in
$\rho$, $\bar{\rho}$, Eq.\ {\eqref{Gc12c34}} suggests that there should be a
corresponding bound for the $\rho$, $\bar{\rho}$ coordinates. This is indeed
the case, as we have the following couple of results:

\begin{lemma}
  \label{boundLemma}Any reflection-symmetric configuration $c \in
  \mathcal{D}^{(0)}_4$ has $\rho (c), \bar{\rho} (c) \in (0, 1)$.
\end{lemma}

\begin{theorem}[Cauchy-Schwarz inequality for $\rho, \bar{\rho}$]
  \label{boundThm}For any configuration $c \in \mathcal{D}^{(0)}_4$ we have
  the inequality:
  \begin{equation}
    \label{maxrhoineq} \max \{ | \rho (c) |, | \bar{\rho} (c) | \}^2
    \leqslant \max \{ \rho (c_{12}), \bar{\rho} (c_{12}) \} \times \max
    \{ \rho (c_{34}), \bar{\rho} (c_{34}) \} .
  \end{equation}
\end{theorem}

We will next prove Lemma \ref{boundLemma}. We will then show how, combined
with Theorem \ref{boundThm}, this implies an optimal bound on $\rho,
\bar{\rho}$. Finally we will present a proof of Theorem \ref{boundThm}, which
is surprisingly subtle.

\subsection{Proof of Lemma \ref{boundLemma}}

To prove the lemma, consider a reflection-symmetric configuration $c$ as in
{\eqref{def:configC12C34}} with:
\begin{eqnarray}
  &  & x_1 = (\epsilon_1 + it_1, \mathbf{x}_1), \quad x_2 = (\epsilon_2 +
  it_2, \mathbf{x}_2), \quad \epsilon_1 > \epsilon_2 > 0, 
  \label{config:reflectionsym1}\\
  &  & x_3 = x_2^{\theta} = (- \epsilon_2 + it_2, \mathbf{x}_2), \quad x_4 =
  x_1^{\theta} = (- \epsilon_1 + it_1, \mathbf{x}_1) . \nonumber
\end{eqnarray}
We will compute $z (c)$, $\bar{z} (c)$ explicitly. We can use translations in
the $\mathbf{x}$ direction, as well as spatial rotations to simplify these
computations. All these transformations do not change the conformal class of
configuration, hence preserve $u, v$ and $z, \bar{z}$. They also commute with time
reflection, and so map reflection-symmetric configurations to
reflection-symmetric ones. By using this freedom, we get an equivalent
configuration $c'$ with the same $z, \bar{z}$:
\begin{equation}
  \label{config:reflectionsym2} x_1' = (\epsilon_1 + i t_1, \tmmathbf{0}),
  \quad x_2' = (\epsilon_2 + it_2, | \mathbf{x}_2 - \mathbf{x}_1 |, 0, \ldots,
  0), \quad x_3' = (x_2')^{\theta}, \quad x_4' = (x_1')^{\theta} .
\end{equation}
This is an effectively two-dimensional configuration. The $z, \bar{z}$
variables of a two-dimensional 4-point configuration $x_k = (x_k^0, x^1_k)$ are
given by Eq.\ {\eqref{zzbarglobal}}, which we copy here
\begin{equation}
  z = \dfrac{(z_1 - z_2)  (z_3 - z_4)}{(z_1 - z_3)  (z_2 - z_4)}, \quad
  \bar{z} = \dfrac{(\bar{z}_1 - \bar{z}_1)  (\bar{z}_3 -
  \bar{z}_4)}{(\bar{z}_1 - \bar{z}_3)  (\bar{z}_2 - \bar{z}_4)}, \quad z_k =
  x_k^0 + i x^1_k, \quad \bar{z}_k = x_k^0 - i x^1_k .
\end{equation}
Applying this to the configuration $c'$, we get $z, \bar{z}$ for $c'$ (which
are the same as for $c$). It's easy to see that $z_3 - z_4 = (z_1 -
z_2)^{\ast}$, $z_1 - z_3 = (z_2 - z_4)^{\ast}$ as a consequence of
reflection symmetry, and similarly for $\bar{z}$'s. So we get $z (c),
\bar{z} (c)$ both real and positive. Explicit expressions come out to be
\begin{eqnarray}
  z (c) = & \dfrac{(\epsilon_1 - \epsilon_2)^2 + (t_1 - t_2 - |
  \mathbf{x}_1 - \mathbf{x}_2 |)^2}{(\epsilon_1 + \epsilon_2)^2 + (t_1 - t_2 -
  | \mathbf{x}_1 - \mathbf{x}_2 |)^2},  \label{z12zbar12}\\
  \bar{z} (c) = & \dfrac{(\epsilon_1 - \epsilon_2)^2 + (t_1 - t_2 + |
  \mathbf{x}_1 - \mathbf{x}_2 |)^2}{(\epsilon_1 + \epsilon_2)^2 + (t_1 - t_2 +
  | \mathbf{x}_1 - \mathbf{x}_2 |)^2} . \nonumber
\end{eqnarray}
In particular we see that $0 < z (c), \bar{z} (c) < 1$. The function $f
(\zeta)$ in the definition of $\rho$ variables maps the interval $(0, 1)$ to
itself. Hence also $0 < \rho (c), \bar{\rho} (c) < 1$, and the lemma and
proved.

\subsection{Optimal bound for $\rho, \bar{\rho}$}\label{rhsbound}

We wish to derive a powerlaw bound on $\frac{1}{1 - r}$, $r = \max (| \rho |,
| \bar{\rho} |)$, since by the arguments in Sec.\ \ref{power4-point} this
implies a powerlaw bound for the 4-point function. Our aim here is to improve on
{\eqref{rhobound}}, {\eqref{Scbound}}.

Consider first a configuration $c \in \mathcal{D}^{(0)}_4$. For such a
configuration, by Theorem \ref{boundThm}, we have
\begin{equation}
  r (c) \leqslant \sqrt{r (c_{12}) r (c_{34})} \leqslant \max (r (c_{12}),
  r (c_{34})),
\end{equation}
and hence
\begin{equation}
  \frac{1}{1 - r (c)} \leqslant \max \left( \frac{1}{1 - r (c_{12})},
  \frac{1}{1 - r (c_{34})} \right) . \label{ineq:1-r}
\end{equation}
We are thus reduced to study $r (c)$ for reflection-symmetric configurations,
like in {\eqref{config:reflectionsym1}}. By definition (\ref{def:rho}) of
$\rho$ variables, we have
\begin{equation}
  \dfrac{1}{1 - \rho} = \dfrac{1 + \sqrt{1 - z}}{2 \sqrt{1 - z}} \leqslant
  \dfrac{1}{\sqrt{1 - z}}, \quad z \in [0, 1), \label{rhozest}
\end{equation}
so it suffices to study $1 / (1 - z)$ and $1 / (1 - \bar{z})$. Using $z,
\bar{z}$ for reflection-symmetric configurations computed in Eqs.\ {\eqref{z12zbar12}} we have
\begin{equation}
  \frac{1}{1 - z (c_{12})} = \dfrac{(\epsilon_1 + \epsilon_2)^2 + (t_1 - t_2 -
  | \mathbf{x}_1 - \mathbf{x}_2 |)^2}{4 \epsilon_1 \epsilon_2},
\end{equation}
and an analogous relation for $\frac{1}{1 - \bar{z} (c_{12})}$. From these
equations, using $\epsilon_2 < \epsilon_1$, and estimating $\epsilon_1 -
\epsilon_2, | t_1 - t_2 |, \left| \mathbf{x}_1 - \mathbf{x}_2 \right|$ from
above by $| x_1 - x_2 |$ (see {\eqref{absdef}}), we easily get
\begin{equation}
  \frac{1}{1 - z (c_{12})}, \frac{1}{1 - \bar{z} (c_{12})} \leqslant \left( 1
  + \frac{1}{\epsilon_2} \right)^2 (1 + | x_1 - x_2 |)^2,
\end{equation}
Putting together this relation, an analogous relation for $z (c_{34})$,
$\bar{z} (c_{34})$, Eqs.\ {\eqref{ineq:1-r}} and {\eqref{rhozest}}, we get
\begin{equation}
  \frac{1}{1 - r (c)} \leqslant \max \left\{ \left( 1 + \frac{1}{\epsilon_2}
  \right) (1 + | x_1 - x_2 |), \left( 1 + \frac{1}{| \epsilon_3 |}
  \right) (1 + | x_3 - x_4 |) \right\} \qquad (c \in \mathcal{D}_4^{(0)})\,.
  \label{1-r0}
\end{equation}
This was for $c \in \mathcal{D}_4^{(0)}$. For a general configuration $c \in
\mathcal{D}_4$, we will shift the coordinates by a translation in time
direction (which of course does not change $\rho, \bar{\rho}$), arranging so
that the shifted configurations $c'$ has $\epsilon_2 > 0 > \epsilon_3$, i.e.\ $c' \in \mathcal{D}_4^{(0)}$. Specifically we will choose
\begin{equation}
  \epsilon_2 (c') = \frac{1}{2} (\epsilon_2 (c) - \epsilon_3 (c)), \quad
  \epsilon_3 (c') = - \frac{1}{2} (\epsilon_2 (c) - \epsilon_3 (c)) . \quad
\end{equation}
Then, using {\eqref{1-r0}} for $c'$, we obtain a bound on $\frac{1}{1 - r
(c)}$ which for example can be expressed as
\begin{equation}
  \frac{1}{1 - r (c)} \leqslant 2 \left( 1 + \frac{1}{\epsilon_2 - \epsilon_3}
  \right) (1 + \max \{ | x_1 - x_2 |, | x_3 - x_4 | \}) \qquad (c \in
  \mathcal{D}_4) . \label{1-rbnd}
\end{equation}
This is a powerlaw bound of the type we were looking for. {By considering reflection-symmetric configurations, it's easy to see that the exponents in this bound cannot be improved. Eq.~\eqref{1-rbnd} is much stronger than our
previous suboptimal bound {\eqref{rhobound}}; in fact it implies a bound of the same form as \eqref{rhobound} with the power exponent 12 replaced by 2.} 


\subsection{Proof of Theorem \ref{boundThm}}\label{rhorhobarProof}

Although {\eqref{maxrhoineq}} looks like a simple-enough geometric inequality,
we do not know an elementary proof of this fact. We essentially guessed this
inequality, checked it numerically, and then looked for a proof. Our guess
started in the Euclidean region, where $\bar{\rho} = \rho^{\ast}$, and
{\eqref{maxrhoineq}} takes the form
\begin{equation}
  | \rho (c) |^2 \leqslant \rho (c_{12}) \rho (c_{34}) \qquad ( c \in
  \mathcal{D}^{(0)}_4 \text{\quad Euclidean} ) . \label{CSrrbarEucl}
\end{equation}
Even in this case we don't know an elementary proof. We guessed that this must
hold, because otherwise it was hard to imagine that the 4-point function itself
would satisfy a Cauchy-Schwarz inequality. Indeed {\eqref{CSrrbarEucl}}
implies the Euclidean version of {\eqref{CS4-point}}. We then guessed
{\eqref{maxrhoineq}} as a generalization of {\eqref{CSrrbarEucl}} for
complexified times.

Our proof of {\eqref{maxrhoineq}} reverses this logic, by deriving it from
{\eqref{CS4-point}}. There exist many explicit CFT 4-point functions, e.g.\ mean field
theories (MFT). One could imagine that by considering {\eqref{CS4-point}} for a
family of such 4-point functions, and passing to some limit (e.g.\ of scaling
dimension of the mean field going to infinity), one could recover
{\eqref{maxrhoineq}}. We haven't managed to make this work using MFTs, but a
closely related strategy does work. Namely we will apply this sort of argument
not to the full 4-point function, but to a single conformal block, since the
latter also satisfy {\eqref{CS4-point}} (as we will explain).

Now that we explained the main idea, let us supply the details. By applying a
translation, we may set $\mathbf{x}_1 = 0$. The remaining spacial component
vectors $\mathbf{x}_2, \mathbf{x}_3, \mathbf{x}_4$ span at most
three-dimensional subspace of $\mathbb{R}^{d - 1}$. This shows that it is
enough to prove the inequality {\eqref{maxrhoineq}} in the case $d = 4$. We
assume that the readers are familiar with the conformal blocks, which encode
contributions of a primary into a 4-point function. In the considered case of 4
identical Hermitean scalar, the relevant OPE has the form (simplifying the
general case considered in Sec.\ \ref{OSfromCFT})
\begin{equation}
  \varphi (x_1) \varphi (x_2) = f_{\varphi \varphi \mathcal{O}} C_{(\lambda)}
  (x_1, x_2, x_0, \mathcal{D}_0) \mathcal{O}^{(\lambda)} (x_0)
\end{equation}
where $\mathcal{O}^{(\lambda)}$ is a dimension $\Delta$, spin $\ell$ symmetric
traceless primary. The conformal block then can be computed by
\begin{equation}
  g_{\Delta, \ell} (c) = C_{(\lambda)} (x_1, x_2, x_0, \mathcal{D}_0)
  C_{(\mu)} (x^{\theta}_3, x^{\theta}_4, {x^\theta_0}, \mathcal{D}^{\theta}_0)
  \langle \mathcal{O}^{(\lambda)} (x_0) \mathcal{O}^{\dagger (\mu)}
  (x^{\theta}_0) \rangle . \label{CBrepr}
\end{equation}
The 4d Euclidean conformal blocks are known explicitly
{\cite{Dolan:2000ut,Dolan:2003hv}}:
\begin{equation}
  \label{cb:dolanosborn} g_{\Delta, \ell} (c) = \frac{z \bar{z}}{z -
  \bar{z}} [k_h (z) k_{\bar{h} - 1} (\bar{z}) - k_h (\bar{z}) k_{\bar{h} - 1}
  (z)],
\end{equation}
where $h, \bar{h} = (\Delta \pm \ell) / 2$, and $k_{\beta} (z) = z^{\beta}{}_2 F_1 (\beta, \beta, 2 \beta ; z)$. (We only cite the result for external operators with equal dimensions.) We will assume that the exchanged operator $\mathcal{O}$ satisfies the
4d unitarity bound $\Delta \geqslant \ell + 2$. As Eq.
{\eqref{cb:dolanosborn}} shows, Euclidean conformal blocks are real-analytic
functions whenever $| z | < 1$. We can also use this formula to analytically
continue them to the forward tube. We wish to show that this analytic
continuation satisfies some properties. This is best shown not from the
explicit formula, but by adapting the robust 4-point function arguments from
Sec.\ \ref{sec:4-point}. Indeed, conformal blocks allow an expansion of the same
form as {\eqref{g:rhoexpansion}}, with non-negative coefficients which are
fixed by conformal invariance. This can be shown by arguments similar to those
in Sec.\ \ref{Eucl4-point}. The existence of the representation {\eqref{CBrepr}}
guarantees Hilbert space unitarity. Then, by the arguments of Sec.\ \ref{anal4-point}, conformal blocks admit an analytic extension to the forward
tube (which is of course the same as the one following from the explicit
formula {\eqref{cb:dolanosborn}}). The point of the current construction is
that it shows that the analytic extension satisfies an inequality analogous to
{\eqref{ggE}}:
\begin{equation}
  | g_{\Delta, \ell} (c) | \leqslant g_{\Delta, \ell} (c_{\ast}) \label{ggECB}
\end{equation}
Then, by the arguments in Sec.\ \ref{power4-point}, conformal blocks satisfy the
powerlaw bound in the forward tube. (As is easy to see from
{\eqref{cb:dolanosborn}}, 4d Euclidean conformal blocks grow as $1 / (1 - z)$
as $z \rightarrow 1^-$ along the real axis, which replaces Eq.
{\eqref{g4cstar}}.) Finally, by the arguments analogous to Sec.\ \ref{CScompl} we conclude that the analytically continued conformal blocks
satisfy Cauchy-Schwarz inequality:
\begin{equation}
  | g_{\Delta, l} (c) |^2 \leqslant g_{\Delta, l} (c_{12}) g_{\Delta, l}
  (c_{34}) \quad \text{for any } c \in \mathcal{D}^{(0)}_4 . \label{CB-CS}
\end{equation}
(Euclidean reflection positivity of conformal blocks follows from the
representation {\eqref{CBrepr}}, which we assume to be valid in the Euclidean
region.)

In the rest of the argument we will only need two facts, the Cauchy-Schwarz
inequality {\eqref{CB-CS}} and the explicit Dolan-Osborn formula
{\eqref{cb:dolanosborn}}. We will apply {\eqref{CB-CS}} to the blocks of spin
$\ell \geqslant 1$ at the unitarity bound, i.e.\ with $\bar{h} = 1, h = \ell +
1$. The Cauchy-Schwarz inequality for $\rho, \bar{\rho}$ will follow by
extracting the asymptotics in the limit $h \rightarrow + \infty$. The
asymptotic behavior of $k_h$ is given by the following lemma:

\begin{lemma}
  \label{lemma:khasymp}For any fixed $z \in \mathbb{C} \setminus [1, +
  \infty)$, the function $k_h (z)$ has the following asymptotic behavior in
  terms of the $\rho$ variable defined in (\ref{def:rho}):
  \begin{equation}
    \label{kh:asymptotic} k_h (z) = (4 \rho)^h \left[ \dfrac{1}{\sqrt{1 -
    \rho^2}} + o (1) \right], \quad h \rightarrow + \infty .
  \end{equation}
\end{lemma}

\begin{proof}
  We have the following identity for $k_h (z)$ {\cite{Hogervorst:2013sma}}:
  \begin{equation}
    k_h (z) = (4 \rho)^h _2 F_1 (1 / 2, h ; h + 1 / 2 ; \rho^2) .
  \end{equation}
  The region $z \notin [1, + \infty)$ corresponds to $| \rho | < 1$, where the
  hypergeometric function $_2 F_1$ has the power series representation
  \begin{equation}
    \label{f21:series}_2 F_1 (1 / 2, h ; h + 1 / 2 ; \rho^2) = \sum_{n =
    0}^{\infty} \dfrac{(1 / 2)_n (h)_n}{n! (h + 1 / 2)_n} \rho^{2 n} .
  \end{equation}
  When $h \rightarrow + \infty$, each coefficient of the series increases
  monotonically, and tends to the coefficients of the convergent in the disk
  $| \rho | < 1$ series
  \[ \sum_{n = 0}^{\infty} \dfrac{(1 / 2)_n}{n!} \rho^{2 n} =
     \dfrac{1}{\sqrt{1 - \rho^2}} . \]
  This implies the statement of the lemma.
\end{proof}

Consider now inequality {\eqref{CB-CS}} for the blocks with $\bar{h} = 1$, $h
= \ell + 1$. Since $k_0 \equiv 1$, it reads:
\begin{equation}
  | w \cdot [k_h (z) - k_h (\bar{z})] |^2 \leqslant w_{12} w_{34} \cdot [k_h
  (z_{12}) - k_h (\bar{z}_{12})] [k_h (z_{34}) - k_h (\bar{z}_{34})],
  \label{CBineq}
\end{equation}
where we denoted $w = \frac{z \bar{z}}{z - \bar{z}}$, and similarly $w_{12},
w_{34}$. Let us assume that the configuration $c \in \mathcal{D}^{(0)}_4$ is
such that
\begin{equation}
  \label{condition:rho} | \rho | \neq | \bar{\rho} |, \quad \rho_{12} \neq
  \bar{\rho}_{12}, \quad \rho_{34} \neq \bar{\rho}_{34} .
\end{equation}
Then, using Lemma \ref{lemma:khasymp}, for large $h$ inequality
{\eqref{CBineq}} becomes:
\begin{equation}
  (A + o (1)) \max \{ | \rho |, | \bar{\rho} | \}^{2 h} \leqslant (B + o
  (1)) \max \{ \rho_{12}, \bar{\rho}_{12} \}^h \max \{ \rho_{34},
  \bar{\rho}_{34} \}^h,
\end{equation}
where $A, B$ are some positive $h$-independent quantities. Now, taking the
limit $h \rightarrow + \infty$, we obtain inequality {\eqref{maxrhoineq}}.

It's easy to see that configurations which violate the condition
{\eqref{condition:rho}} are non-generic. They can be approached by
configurations which do satisfy {\eqref{condition:rho}}. Therefore, by
continuity inequality {\eqref{maxrhoineq}} is valid also for such exceptional
configurations.

\section{OPE convergence in the forward tube and in Minkowski
space}\label{OPEconvMink}

We have several OPE convergence statements scattered throughout the paper. The
Euclidean CFT axioms assume convergence of the OPE series for $\mathcal{O}_1
(x_1) \mathcal{O}_2 (x_2)$ whenever the two points $x_1, x_2$ are closer to
the OPE center than any other point. Then we established OPE convergence in
the Hilbert space sense (Sec.\ \ref{Hilbert}) in the Euclidean region for
states generated by two operators in the half-space. Then in Sec.\ \ref{Eucl4-point} we used Hilbert space language to derive the power series
representation {\eqref{g:rhoexpansion}} for the 4-point function, whose
convergence is thus morally equivalent to OPE convergence (for the 4-point
functions). We then used this power series representation to analytically
continue the 4-point function to the forward tube, and then define the Minkowski
4-point function as a boundary value in the sense of distributions. Finally, in
Sec.\ \ref{sec:Wpos} we showed, by arguments not using conformal invariance,
that the OS states $| \nobracket \mathcal{O}_1 (x_1) \mathcal{O}_2 (x_2)
\rangle$ can be extended to the forward tube and (when integrated against test
functions) to the Minkowski region, and that the so obtained states can be
arbitrarily well approximated by (integrated) OS states. Therefore, OPE
convergence holds for these states, as for the OS states.

In this section we will give a more explicit discussion of the OPE
convergence for the Minkowski 4-point function and for the 2-operator states in
the forward tube and Minkowski space. We will also explain how our approach
and results compare to the classic paper by Mack {\cite{Mack:1976pa}}.

\subsection{Convergence of conformal block decomposition for 4-point functions}

Let us consider the 4-point function of identical scalars
{\eqref{def:Euclidean4-point}}:
\begin{equation}
  G (x_1, x_2, x_3, x_4) \equiv G (c) = \langle \mathcal{O} (x_1) \mathcal{O}
  (x_2) \mathcal{O} (x_3) \mathcal{O} (x_4) \rangle = \frac{g (\rho,
  \bar{\rho})}{(x_{12}^2 x_{34}^2)^{\Delta_{\mathcal{O}}}} .
\end{equation}
The discussion below can be easily extended to non-identical scalars using the
same ideas as in Sec.\ \ref{nonId}.

We know that in the Euclidean region the function $g (\rho, \bar{\rho})$ has a
convergent conformal block decomposition
\begin{equation}
  g (\rho, \bar{\rho}) = \underset{\Delta, l}{\sum} C_{\Delta, l}^2 g_{\Delta,
  l} (\rho, \bar{\rho}) .
\end{equation}
As in Sec.\ \ref{rhorhobarProof}, we will assume that the reader is familiar
with conformal blocks. The main point is that the conformal block
decomposition is obtained by separating the series {\eqref{g:rhoexpansion}}
into parts corresponding to the conformal multiplets of primary operators
$\mathcal{O}_{\Delta, l}$ occurring in the $\mathcal{O} \times \mathcal{O}$ OPE
with coefficients $C_{\Delta, l}$. Conformal blocks in the Euclidean region by
themselves have power series expansions like {\eqref{g:rhoexpansion}} with
positive coefficients (fixed by conformal symmetry). As in Sec.\ \ref{anal4-point}, we can use this expansion to analytically continue conformal
blocks to the forward tube. By an analogue of the bound {\eqref{maj}} we know
that conformal block expansion remains convergent everywhere in the forward
tube, since $| \rho |, | \bar{\rho} | < 1$ there. Since individual conformal
blocks are smaller than the 4-point function in the Euclidean region, by the
arguments in Sec.\ \ref{power4-point} we know that they satisfy a powerlaw
bound, and hence they become tempered distributions when going to the
Minkowski region.\footnote{This argument shows that any conformal block which
occurs in a reflection-positive CFT 4-point function satisfies a powerlaw bound.
E.g.\ conformal blocks for $l \geqslant 0$, $\Delta \geqslant l + d - 2$ occur
in a 4-point function $\langle \varphi_1 \varphi_2 \varphi_1 \varphi_2 \rangle$
where $\varphi_1, \varphi_2$ are two GFFs of appropriately chosen equal
dimension. It should be also possible to show that conformal blocks satisfy a
powerlaw bound without relying on a fiducial 4-point function. E.g.\ for $d = 4$
conformal blocks this follows from their explicit Dolan-Osborn expressions.
For general $d$, powerlaw bound on the diagonal $z = \bar{z}$ can be shown
using the differential equation found in {\cite{Hogervorst:2013kva}} and
extended to $z \neq \bar{z}$ by the usual arguments.}${}^{,}$\footnote{It should be
noted that away from light cones conformal blocks are better than
distributions: they are real-analytic there (although this fact won't play a
role for us). In even $d$ this is obvious from their explicit expressions in
terms of hypergeometric functions. For general $d$ this follows from a
first-order matrix ODE satisfied by a finite-length vector including the
conformal block and its low-order derivatives. Such an ODE exists for a
length-8 vector and can be built using the quadratic and quartic Casimir
equations {\cite{KravchukCB}}.}${}^{,}$\footnote{\label{noteMarc2}Also ``conformal
partial waves'' $\frac{g_{\Delta, l} (\rho, \bar{\rho})}{(x_{12}^2
x_{34}^2)^{\Delta_{\mathcal{O}}}}$ are tempered distributions in the
Minkowski space. Therefore their Fourier transforms are well defined. Explicit
expressions for these Fourier transforms were found recently in
{\cite{Gillioz:2020wgw}}. }

By the arguments like in Sec.\ \ref{power4-point}, $g (\rho, \bar{\rho})$,
the partial sums of the conformal block decomposition $g (\rho, \bar{\rho} ;
\Delta_{\ast}) = \underset{\Delta \leqslant \Delta_{\ast}, l}{\sum} C_{\Delta,
l}^2 g_{\Delta, l} (\rho, \bar{\rho})$, and the corresponding remainders satisfy in the forward
tube a uniform bound:
\begin{equation}
  | g (\rho, \bar{\rho} ; \Delta_{\ast}) |, | g (\rho, \bar{\rho}) - g (\rho,
  \bar{\rho} ; \Delta_{\ast}) | \leqslant C (1 - r)^{- 4 \Delta}, \qquad r =
  \max \{ | \rho |, | \bar{\rho} | \} . \label{powerlawbound:block}
\end{equation}
Consider the 4-point partial sums including the prefactor $G (c ; \Delta_{\ast}) =
\frac{1}{(x_{12}^2 x_{34}^2)^{\Delta_{\mathcal{O}}}} g (\rho, \bar{\rho} ;
\Delta_{\ast})$. By the powerlaw bound of $(1 - r (c))^{- 1}$, we have the
powerlaw bounds
\begin{eqnarray}
  | G (c ; \Delta_{\ast}) |, | G (c) - G (c ; \Delta_{\ast}) | & \leqslant &
  C_{}  \left( 1 + \max_k  \dfrac{1}{\epsilon_k - \epsilon_{k + 1}}
  \right)^{A}  (1 + \max_i  | x_i - x_{i + 1} |)^B  \label{powerlawbound:G}
\end{eqnarray}
for all $c \in \mathcal{D}_4$ and $\Delta_{\ast} \geqslant 0$. \ Consider the
boundary value of $G (c ; \Delta_{\ast})$, call it $G^M (x_1, x_2, x_3, x_4 ;
\Delta_{\ast})$, where $x_i \in \mathbb{R}^{1, d - 1}$; it is a tempered
distribution by Theorem \ref{ThVlad}. The following theorem establishes
distributional convergence of conformal block decomposition.

\begin{theorem}
  \label{theorem:districonverge}We have $G^M (x ; \Delta_{\ast}) \rightarrow
  G^M (x)$ in the sense of tempered distributions. 
\end{theorem}

\begin{proof}
  Denote $H (c ; \Delta_{\ast}) = G (c) - G (c ; \Delta_{\ast})$. We have to
  show that, as $\Delta_{\ast}$ goes to infinity, the boundary value of $H (c
  ; \Delta_{\ast})$ converges to 0 in the sense of tempered distributions,
  i.e, for any Schwartz test function $f \in \mathcal{S} (\mathbb{R}^{4 d})$
  \begin{equation}
    \underset{\Delta_{\ast} \rightarrow \infty}{\lim} \lim_{\lambda
    \rightarrow 0^+}  \int H (\lambda \epsilon + i t, \mathbf{x};
    \Delta_{\ast}) f (t, \mathbf{x})\, d t\, d \mathbf{x}= 0, \qquad
    (\epsilon_1 > \epsilon_2 > \epsilon_3 > \epsilon_4), \label{Hf}
  \end{equation}
  where we write for brevity $t$ instead of $t_1, t_2, t_3, t_4$ etc. The
  proof is the same as in our paper {\cite{paper1}}, Theorem 3.1. We will
  retrace here the main steps for completeness and because we will need it to
  establish a stronger result below. Define $L_f (\lambda ; \Delta_{\ast})
  \assign \int H (\lambda \epsilon + i t, \mathbf{x}; \Delta_{\ast}) f (x)\,
  d x$ with $x = (t, \mathbf{x})$. Since $H$ is holomorphic in $\tau = \lambda
  \epsilon + i t$ we have
  \begin{equation}
  \begin{split}
      L_f^{(n)} (\lambda ; \Delta_{\ast}) =& \int \left( \left( \epsilon \cdummy
      \frac{\partial}{i \partial t} \right)^n H (\lambda \epsilon + i t,
      \mathbf{x}; \Delta_{\ast}) \right) f (x)\, d x \\
       =& \int H (\lambda \epsilon
      + i t, \mathbf{x}; \Delta_{\ast}) \left( \left( i \epsilon \cdummy
      \frac{\partial}{\partial t} \right)^n f (x) \right)\, d x,
  \end{split}    
  \end{equation}
  which by the powerlaw bound {\eqref{powerlawbound:G}} implies
  \begin{equation}
    L_f^{(n)} (\lambda ; \Delta_{\ast}) \leqslant \frac{C_n}{\lambda^A} | f
    |_{p_n}, \qquad \lambda \in (0, 1], \qquad p_n = \max \{ n, \lceil B
    \rceil + 4 d + 1 \} . \label{Lfn:bound1}
  \end{equation}
  These bounds blow up in the $\lambda \rightarrow 0$ limit, but by using the
  Newton-Leibniz repeatedly one can get bounds which do not blow up:
  \begin{equation}
    L_f^{(n)} (\lambda ; \Delta_{\ast}) \leqslant D_n | f |_{p_{n + [A] + 1}},
    \qquad \lambda \in (0, 1] . \label{Lfn:bound2}
  \end{equation}
  Using this for $n = 1$ one proves that the limit $L_f (0 ; \Delta_{\ast})
  = \lim_{\lambda \rightarrow 0^+} L_f (\lambda ; \Delta_{\ast})$ exists
  and
  \begin{equation}
    | L_f (0 ; \Delta_{\ast}) - L_f (\lambda ; \Delta_{\ast}) | \leqslant D_1
    \lambda | f |_{\max \{ [A] + 2, \lceil B \rceil + 4 d + 1 \}} .
  \end{equation}
  By Lebesgue's dominated convergence theorem, for any fixed $\lambda$ in $(0,
  1]$, $L_f (\lambda ; \Delta_{\ast})$ tends to zero as $\Delta_{\ast}
  \rightarrow + \infty$. Thus the previous bound implies
  \begin{eqnarray}
    \underset{\Delta_{\ast} \rightarrow \infty}{\overline{\lim}} | L_f (0,
    \Delta_{\ast}) | & \leqslant & \underset{\Delta_{\ast} \rightarrow
    \infty}{\overline{\lim}} | L_f (0, \Delta_{\ast}) - L_f (\lambda ;
    \Delta_{\ast}) | + \underset{\Delta_{\ast} \rightarrow
    \infty}{\overline{\lim}} | L_f (\lambda ; \Delta_{\ast}) | \nonumber\\
    & \leqslant & D_1 \lambda | f |_{\max \{ [A] + 2, \lceil B \rceil + 4 d +
    1 \}} .  \label{Lf0estimate}
  \end{eqnarray}
  Since $\lambda$ can be arbitrarily small, we get $\underset{\Delta_{\ast}
  \rightarrow \infty}{\overline{\lim}} | L_f (0, \Delta_{\ast}) | = 0$. This
  finishes the proof.
\end{proof}

\subsubsection{Convergence rate for compactly supported test
functions}\label{subsection:rateconvergence}

Because of the use of Lebesgue's theorem on dominated convergence, Theorem
\ref{theorem:districonverge} does not give the rate of convergence. We will
now give the rate in an important special case of compactly supported test
functions. This provides an explicit example for the remark in
{\cite{paper1}}, the last paragraph of Sec.\ 3.3.

The idea is that not only $H (c, \Delta_{\ast}) \rightarrow 0$ pointwise but
it does so exponentially fast. We will first derive the exponential
convergence bound, upgrading the Euclidean argument from
{\cite{Pappadopulo:2012jk}}, to the forward tube, and then use it. Let $F (t)$
be the Laplace transform of a positive measure $\mu (E)$ on $E \geqslant 0$:
\begin{equation}
  F (t) = \int_0^{\infty} \mu (E) e^{- E t}\, d E, \qquad \mu (E) \geqslant 0.
  \label{def:ft}
\end{equation}
We assume that the integral is convergent for $t > 0$ and $F (t) \sim t^{-
\alpha}$ as $t \rightarrow 0^+$, and $\alpha > 0$. Then by the
Hardy-Littlewood tauberian theorem we know that
\begin{equation}
  M (E) = \int_0^E \mu (E')\, d E' \sim \frac{E^{\alpha}}{\Gamma (\alpha + 1)}
  \tmop{as} E \rightarrow + \infty .
\end{equation}
We can now estimate the remainder $F_{E_{\ast}} (t) \assign
\int_{E_{\ast}}^{\infty} \mu (E) e^{- E t}\, d E,$ via
\begin{equation}
  F_{E_{\ast}} (t) = \int_{E_{\ast}}^{\infty} e^{- E t}\, d M (E) = - M
  (E_{\ast}) e^{- E_{\ast} t} + t \int_{E_{\ast}}^{\infty} M (E) e^{- E t}\, d
  E,
\end{equation}
which gives $\begin{array}{lll}
  | F_{E_{\ast}} (t) | & \leqslant & C_1 e^{- E_{\ast} t} E_{\ast}^{\alpha} +
  C_2 e^{- E_{\ast} t} \left( \frac{1 + E_{\ast} t}{t} \right)^{\alpha},
\end{array}$ and finally
\begin{equation}
  | F_{E_{\ast}} (t) | \leqslant \tmop{const} \times e^{- E_{\ast} t} (t^{- 1}
  + E_{\ast})^{\alpha} . \label{ft:remainderestimate}
\end{equation}
for any $E_{\ast} \geqslant 1$ (say).

Now let's go back to the 4-point function of four identical scalars $G (c) =
\langle \mathcal{O} (x_1) \mathcal{O} (x_2) \mathcal{O} (x_3)
\mathcal{O} (x_4) \rangle$, $c \in \mathcal{D}_4$. The remainder $H (c ;
\Delta_{\ast})$ can clearly be bounded by replacing $\rho, \bar{\rho}$ with $r
= \max \{ | \rho |, | \bar{\rho} | \}$:
\begin{equation}
  | H (c ; \Delta_{\ast}) | \leqslant \frac{1}{| x_{12}^2
  |^{\Delta_{\mathcal{O}}} | x_{34}^2 |^{\Delta_{\mathcal{O}}}} [g (r, r) - g
  (r, r ; \Delta_{\ast})] \label{CS:GN},
\end{equation}
By setting $t = \log (1 / r)$, $g (r, r)$ with its representation
{\eqref{g:rhoexpansion}} is in the same form as (\ref{def:ft}) and by Eq.
{\eqref{g4cstar}} we have that the corresponding $F (t) \sim t^{- \alpha}$
with $\alpha = 4 \Delta_{\mathcal{O}}$. Bounding the remainder $g (r, r) - g
(r, r ; \Delta_{\ast})$ by (\ref{ft:remainderestimate}), and using \
{\eqref{CS:GN}}, we get a bound on the remainder $| H (c ; \Delta_{\ast})
|$ for any $\Delta_{\ast} \geqslant 1$ (say):
\begin{equation}
  | H (c ; \Delta_{\ast}) | \leqslant \frac{\tmop{const}}{| x_{12}^2
  |^{\Delta_{\mathcal{O}}} | x_{34}^2 |^{\Delta_{\mathcal{O}}}} \times r
  (c)^{\Delta_{\ast}} \left( \frac{1}{\log (r (c)^{- 1})} + \Delta_{\ast}
  \right)^{4 \Delta_{\mathcal{O}}} . \label{g1221:bound}
\end{equation}
Let us now convert this into an explicit estimate for the distributional
convergence rate, improving on Theorem \ref{theorem:districonverge} for
compactly supported test functions. Recall that we have an upper bound on $r
(c) = \max \{ |\rho(c)|, |\bar{\rho}(c)| \}$ ($c = (x_1, x_2, x_3,
x_4)$), Eq.\ {\eqref{1-rbnd}}, which we copy here:
\begin{equation}
  \frac{1}{1 - r (c)} \leqslant 2 \left( 1 + \frac{1}{\epsilon_2 - \epsilon_3}
  \right) (1 + \max \{ | x_1 - x_2 |, | x_3 - x_4 | \}) . \label{copybnd}
\end{equation}
This bound tells us how much $r (c)$ is separated from 1. In turn, by
{\eqref{g1221:bound}} this translates into an explicit bound on $| H (c ;
\Delta_{\ast}) |$. Let us retrace the proof of Theorem
\ref{theorem:districonverge}, replacing Eq.\ {\eqref{Lf0estimate}} by
\begin{eqnarray}
  | L_f (0, \Delta_{\ast}) | & \leqslant & | L_f (0, \Delta_{\ast}) - L_f
  (\lambda ; \Delta_{\ast}) | + | L_f (\lambda ; \Delta_{\ast}) | \nonumber\\
  & \leqslant & D_1 \lambda | f |_{\max \{ [A] + 2, \lceil B \rceil + 4 d + 1
  \}} + | L_f (\lambda ; \Delta_{\ast}) | .  \label{Lf0estimate1}
\end{eqnarray}
We will now choose $\lambda$ small, as a function of $\Delta_{\ast}$, so that
the second term in the r.h.s.\ is smaller than the first one. Let us choose and
fix $\epsilon_1 > \epsilon_2 > \epsilon_3 > \epsilon_4$. For $x_k^{\lambda} =
(\lambda \epsilon_k + i t_k, \mathbf{x}_k)$, $x_k^M = (i t_k, \mathbf{x}_k)$,
bound {\eqref{copybnd}} gives
\begin{equation}
  \frac{1}{1 - r (c^{\lambda})} \leqslant \frac{C_{\epsilon}}{\lambda} (1 +
  \max \{ | x^M_1 - x^M_2 |, | x^M_3 - x^M_4 | \}) \qquad (0 < \lambda
  \leqslant 1), \label{rlambda:bound}
\end{equation}
where $C_{\epsilon}$ is a constant which depends only on $\epsilon_i$ but not
on $\lambda$ or $x_k^M$. If $f \in C^{\infty}_0 (\mathbb{R}^{4 d})$, a
compactly supported test function, then (\ref{rlambda:bound}) implies
\begin{equation}
  \frac{1}{1 - r (c^{\lambda})} \leqslant \frac{A_f}{\lambda}
  \label{rlambda:fbound}  \qquad (0 < \lambda \leqslant 1, (x_k^M) \in
  \tmop{supp} (f)),
\end{equation}
where $A_f = 2 C_{\epsilon} \underset{(x_k^M) \in \tmop{supp} (f)}{\sup} (1 +
\max \{ | x^M_1 - x^M_2 |, | x^M_3 - x^M_4 | \})$.

Now by {\eqref{g1221:bound}}, {\eqref{rlambda:fbound}}, and Lemma
\ref{x2bnd}(b) we have a bound for a compactly supported test function $f$:
\begin{equation}
  | L_f (\lambda ; \Delta_{\ast}) | \leqslant \frac{\tmop{const}}{\lambda^{4
  \Delta_{\mathcal{O}}}} e^{- \frac{\lambda}{A_f} \Delta_{\ast}} \left(
  \frac{A_f}{\lambda} + \Delta_{\ast} \right)^{4 \Delta_{\mathcal{O}}} \int |
  f |\, d x.
\end{equation}
By this and {\eqref{Lf0estimate1}} we have
\begin{equation}
  | L_f (0 ; \Delta_{\ast}) | \leqslant A_1 \lambda + \frac{A_2}{\lambda^{4
  \Delta_{\mathcal{O}}}} e^{- \frac{\lambda}{A_f} \Delta_{\ast}} \left(
  \frac{A_f}{\lambda} + \Delta_{\ast} \right)^{4 \Delta_{\mathcal{O}}},
\end{equation}
where all constants $A_1, A_2, A_f$ are $f$-dependent. If we choose $\lambda =
1 / \Delta_{\ast}^{\gamma}$ then for any $\gamma \in (0, 1)$ the first term
dominates (the second term decays exponentially fast for large
$\Delta_{\ast}$). It is easy to see that the first term still dominates for
$\lambda = \kappa \frac{\log \Delta_{\ast}}{\Delta_{\ast}}$ with
sufficiently large $\kappa$. We therefore obtain the following promised
strengthening of Theorem \ref{theorem:districonverge}.

\begin{theorem}
  \label{theorem:districonverge1}For any compactly supported $C^{\infty}$ test
  function $f$, we have $(G^M (\Delta_{\ast}), f) - (G^M, f) \rightarrow 0$ as
  $O \left( \left. \frac{\log \Delta_{\ast}}{\Delta_{\ast}} \right) \right.$.
\end{theorem}

It is somewhat surprising that conformal block decomposition converges so
slowly in the Minkowski region, while it converges exponentially fast in the
Euclidean region.

\subsection{OPE convergence in the sense of
$\mathcal{H}^{\tmop{CFT}}$}\label{OPEHCFT}

We will now rephrase the question of OPE convergence from the point of view of
states generated by two operators (rather than 4-point functions which represent
inner products of such states). We already discussed these questions to some
extent in Sec.\ \ref{Hilbert} in the Euclidean region, and Sec.\ \ref{sec:Wpos} in the forward tube and in the Minkowski region. We will now
update that discussion.

Recall that in Sec.\ \ref{sec:Wpos} we defined states (see Eq.
{\eqref{Ocompl}} and Remark \ref{genrefl})
\begin{equation}
  \psi (x_1, x_2) = \left| \mathcal{O} (x_1) \mathcal{O} (x_2) \rangle, \quad
  x_k = (\epsilon_k + i t_k, \mathbf{x}_k + i\mathbf{y}_k) \right., \quad 0
  \succ (\epsilon_1, \mathbf{y}_2) \succ (\epsilon_1, \mathbf{y}_2),
  \label{f.t.states}
\end{equation}
as elements of a vector space with inner products {\eqref{complinner}}:
\begin{equation}
  \langle \mathcal{O} (x_1) \mathcal{O} (x_2) | \nobracket \mathcal{O} (x_3)
  \mathcal{O} (x_4) \rangle = G_4 (x^{\theta}_2, x^{\theta}_1, x_3, x_4),
\end{equation}
where $x^{\theta} = (- \epsilon + i t, \mathbf{x}- i\mathbf{y})$ for $x =
(\epsilon + i t, \mathbf{x}+ i\mathbf{y})$. We then proved that this inner
product was positive definite, and that these states could be approximated in
norm by Euclidean states, and so belong to the same Hilbert space. The
arguments of Sec.\ \ref{sec:Wpos} did not use conformal symmetry.

So, by arguments of Sec.\ \ref{sec:Wpos} we have a map $\psi (x_1, x_2)$
from $x_1, x_2$ as in {\eqref{f.t.states}} into $\mathcal{H}^{\tmop{CFT}}$. {We claim that this map is holomorphic.
To begin with, this map is continuous with respect to
the $\mathcal{H}^{\tmop{CFT}}$ norm, and in particular bounded on compact subsets. This follows easily from the continuity of $G_4$. To show holomorphicity, we will use Morera's theorem and Osgood's lemma (which
remain valid for Hilbert-space-valued functions of complex
variables). Morera's theorem says that a locally continuous
function of one complex variable is holomorphic if its integral over any small contour is zero. Let $C$ be a small 1d contour in the region of $\xi =
(x_1, x_2)$ where $\psi$ is defined (we assume that one of the components of $\xi$ goes around the contour while the others stay fix). We have
\begin{equation}
  \left\| \int_C d \xi\, \psi (\xi) \right\|^2 = \int_{\xi' \in C} d
  \overline{\xi'}\, \int_{\xi \in C} d \xi\, G_4 (\xi^{\prime \theta}, \xi) = 0\,,
  \label{Morera}
\end{equation}
the last integral being zero because $G_4$ is holomorphic in $\xi$.
Hence $\int_C d \xi\, \psi (\xi)=0$ and by Morera's theorem $\psi(x_1,x_2)$ is holomorphic in each component separately.
Finally, by Osgood's lemma \cite{Osgood} $\psi(x_1,x_2)$ is holomorphic in all variables jointly.}

Let us connect this discussion with the OPE. In the Euclidean region, OPE says
\begin{eqnarray}
  \psi (x_1, x_2) = | \mathcal{O} (x_1) \mathcal{O} (x_2) \rangle & = &
  \underset{k, \lambda}{\sum} f_{\mathcal{O}\mathcal{O}k} C_{k, \lambda} (x_1,
  x_2, x_S, \mathcal{D}) | (\mathcal{O}^{\dagger}_k)^{(\lambda)} (x_S)
  \rangle,  \label{OPE:Euclidean}\\
  C_{k, \lambda} (x_1, x_2, x_S, \mathcal{D}) & = & \sum_{\alpha}
  f_{\mathcal{O}\mathcal{O}k} C_{k, \lambda, \alpha} (x_1, x_2, x_S)
  \mathcal{D}^{\alpha}, \nonumber
\end{eqnarray}
where $x_1, x_2$ are two Euclidean points in the lower halfspace, $0 > x_1^0,
x_2^0$, $x_S = (- 1, 0, \ldots)$ is the south pole, and $\mathcal{D}
=\mathcal{D}_{x_S}$ is the image of $\partial / \partial x |_{x = 0}
\nobracket$ under a conformal transformation which maps $0, \infty$ to $x_S,
x_N$. We proved in Sec.\ \ref{Hilbert}, from the Euclidean CFT axioms, that
the series in the r.h.s.\ converges in $\mathcal{H}^{\tmop{CFT}}$. As discussed
in Sec.\ \ref{Hilbert}, convergence holds provided that the series is summed
in a certain manner: for each $\Lambda$ we define a partial sum
$\psi_{\Lambda} (x_1, x_2)$ over all terms with $\Delta_k + | \alpha | <
\Lambda$, and then tend $\Lambda \rightarrow \infty$. This procedure is needed
because, although the states $\mathcal{D}^{\alpha} |
(\mathcal{O}^{\dagger}_k)^{(\lambda)} (x_S) \rangle$ are orthogonal for
different $| \alpha |$ (because they are eigenvectors of $\frac{K^0 - P^0}{2}$
with different eigenvalues), they are not orthonormal. This can be corrected
as follows. For each $k$, let us orthonormalize the infinite multiplet of
states $\mathcal{D}^{\alpha} | (\mathcal{O}^{\dagger}_k)^{(\lambda)} (x_S)
\rangle$. Let $e_{k, n}$, $n \in \mathbb{Z}_{\geqslant 0}$, be the
corresponding orthonormal basis of states (there is obviously a lot of
arbitrariness in this basis). We then can write
\begin{equation}
  \psi (x_1, x_2) = | \mathcal{O} (x_1) \mathcal{O} (x_2) \rangle =
  \sum_k f_{\mathcal{O}\mathcal{O}k} \underset{n}{\sum} \tilde{C}_{k, n} (x_1,
  x_2) e_{k, n}, \label{psialtexp}
\end{equation}
where $\tilde{C}_{k, n}$'s are some finite linear combinations of $C_{k,
\lambda, \alpha} (x_1, x_2, x_S)$. Since $e_{k, n}$'s with different $k$ are
also orthogonal, this equation is an expansion of the state $\psi (x_1, x_2)$
in an orthonormal basis. Convergence of the series is now equivalent to the
finiteness of the norm of $\psi (x_1, x_2)$:
\begin{equation}
  \| \psi (x_1, x_2) \|_{\mathcal{H}^{\tmop{CFT}}} = \sum_{k, n} |
  f_{\mathcal{O}\mathcal{O}k} \tilde{C}_{k, n} (x_1, x_2) |^2 < \infty .
  \label{psinorm}
\end{equation}
Moreover, by definition, this norm is nothing but the 4-point function $\langle
\mathcal{O} (x_2^{\theta}) \mathcal{O} (x_1^{\theta}) \mathcal{O} (x_1)
\mathcal{O} (x_2) \rangle$.

Eqs.\ {\eqref{OPE:Euclidean}}-{\eqref{psinorm}} were all in the Euclidean
region, but we claim that they continue to make sense in the forward tube. The
argument is as follows. We know by the arguments around {\eqref{Morera}} that
$\psi (x_1, x_2)$ have analytic continuation to the region
{\eqref{f.t.states}}. The inner product $\langle e_{k, n} | \psi (x_1, x_2)
\rangle$ is then the analytic continuation of $f_{\mathcal{O}\mathcal{O}k}
\tilde{C}_{k, n} (x_1, x_2)$ from the Euclidean to the same region. (This
inner product is a finite linear combination of $x_N$-derivatives of the 3-point
function $\langle (\mathcal{O}_k)^{(\lambda)} (x_N) \mathcal{O} (x_1)
\mathcal{O} (x_2) \rangle$, hence holomorphic in the forward tube.) We thus
obtain the following fact:

\begin{theorem}
  Expansion {\eqref{psialtexp}}, analytically continued from the Euclidean
  region to the forward tube term by term, converges in the sense of
  $\mathcal{H}^{\tmop{CFT}}$ to the same states $\psi (x_1, x_2)$ in the
  region {\eqref{f.t.states}} that we defined in Sec.\ \ref{sec:Wpos}.
\end{theorem}

For the subsequent discussion, we also define the states
\begin{equation}
  \psi_k (x_1, x_2) = f_{\mathcal{O}\mathcal{O}k} \underset{n}{\sum}
  \tilde{C}_{k, n} (x_1, x_2) e_{k, n} .
\end{equation}
The norms of these states is given by conformal blocks (up to prefactor
$f_{\mathcal{O}\mathcal{O}k}^2 / (x_{12}^2)^{4 \Delta_{\mathcal{O}}}$). Just
as the state $\psi (x_1, x_2)$, each state $\psi_k (x_1, x_2)$ is an
$\mathcal{H}^{\tmop{CFT}}$-valued holomorphic function in the region
{\eqref{f.t.states}}, moreover in this region we have
\begin{equation}
  \psi (x_1, x_2) = \sum_k \psi_k (x_1, x_2), \label{psipsik}
\end{equation}
the series convergent in the sense of $\mathcal{H}^{\tmop{CFT}}$. The norm of
the tail of this series is given by the function $H (c, \Delta_{\ast})$ from
the proof of Theorem \ref{theorem:districonverge}:
\begin{equation}
  \Biggl\| \sum_{\Delta_k > \Delta_{\ast}} \psi_k (x_1, x_2)
  \Biggr\| = H (c, \Delta_{\ast}), \qquad c = (x_2^{\theta}, x_1^{\theta},
  x_1, x_2) . \label{tailbound1}
\end{equation}

Vladimirov's Theorem \ref{ThVlad} remains true for Hilbert-space-valued
holomorphic functions, whose norm satisfies a powerlaw bound in the forward tube.
Applying such a version of Theorem \ref{ThVlad}, as well as arguments from the
proof of Theorem \ref{theorem:districonverge} and from Sec.\ \ref{subsection:rateconvergence} (in particular using the bound
{\eqref{g1221:bound}}), it's easy to obtain the following result (we omit the
proof).

\begin{theorem}
  \label{th:distropeconvergence}(a) The boundary value $\tmop{bv} (\psi) =
  \lim_{\epsilon_i \rightarrow 0} \psi (x_1, x_2)$ exists as
  $\mathcal{H}^{\tmop{CFT}}$-valued tempered distributions, and similarly for
  each $\tmop{bv} (\psi_k)$: $\tmop{bv} (\psi), \tmop{bv} (\psi_k) \in
  \mathcal{S}' (\mathbb{R}^{2 d}) \otimes \mathcal{H}^{\tmop{CFT}}$.
  Explicitly, the limit
  \begin{equation}
    (\tmop{bv} (\psi), f) = \underset{\epsilon \rightarrow 0}{\lim} \int
    \psi (\epsilon + i t, \mathbf{x}) f (t, \mathbf{x})\, d^2 t\, d^{2 (d - 1)}
    \mathbf{x}
  \end{equation}
  exists as a vector in $\mathcal{H}^{\tmop{CFT}}$ for any Schwartz function
  $f \in \mathcal{S} (\mathbb{R}^{2 d})$, and is a continuous linear operator
  from $\mathcal{S} (\mathbb{R}^{2 d})$ to $\mathcal{H}^{\tmop{CFT}}$, and
  analogously for $\tmop{bv} (\psi_k)$;
  
  (b) (Distributional OPE convergence in $\mathcal{H}^{\tmop{CFT}}$) For each
  Schwartz test function $f \in \mathcal{S} (\mathbb{R}^{2 d})$, the series
  $\underset{k}{\sum} (\tmop{bv} (\psi_k), f)$ converges in
  $\mathcal{H}^{\tmop{CFT}}$ norm to $(\tmop{bv} (\psi), f)$;
  
  (c) (Convergence rate) For compactly supported test functions, the series in
  (b) summed over $\Delta_k \leqslant \Delta_{\ast}$ converges with rate $O
  \left( \left. \sqrt{\frac{\log \Delta_{\ast}}{\Delta_{\ast}}} \right)
  \right.$.
\end{theorem}

\begin{remark}
  The following more finegrained version of Theorem
  \ref{th:distropeconvergence}(b) also holds. Denote
  \begin{equation}
    E_{k, n} (x_{1,} x_2) = f_{\mathcal{O}\mathcal{O}k} \underset{}{}
    \tilde{C}_{k, n} (x_1, x_2) e_{k, n} .
  \end{equation}
  As explained above, $\tilde{C}_{k, n} (x_1, x_2)$ is a finite sum of terms
  like $(\mathcal{D}^{\theta})^{\alpha} \langle (\mathcal{O}_k)^{(\lambda)}
  (x_N) \mathcal{O} (x_1) \mathcal{O} (x_2) \rangle$ (all having the same $|
  \alpha |$), so $\tmop{bv} (E_{k, n})$ exists by Vladimirov's theorem. Then
  \tmtextit{for each Schwartz test function $f \in \mathcal{S} (\mathbb{R}^{2
  d})$, the series $\underset{k, n}{\sum} (\tmop{bv} (E_{k, n}), f)$ converges
  in $\mathcal{H}^{\tmop{CFT}}$ norm to $(\tmop{bv} (\psi), f)$.}
  
  It would be interesting to prove a version Theorem
  \ref{th:distropeconvergence}(c), truncating the series $\underset{k,
  n}{\sum} (\tmop{bv} (E_{k, n}), f)$ to $k, n$ such as the corresponding
  $\Delta_k + | \alpha | \leqslant \Delta_{\ast}$. To do so one would have to
  find an analogue of the bounds {\eqref{tailbound1}}, {\eqref{g1221:bound}}
  valid for such a truncation. This is not straightforward because the partial
  sums of the series {\eqref{psinorm}} truncated to $\Delta_k + | \alpha |
  \leqslant \Delta_{\ast}$ do not correspond to a simple truncation of the
  $\rho, \bar{\rho}$ series of the full 4-point function (basically because the
  transformation which maps $x_1, x_2, x_1^{\theta}, x_2^{\theta}$ to their
  $\rho, \bar{\rho}$ conformal frame does not necessarily map $x_S, x_N$ to
  $0, \infty$).
\end{remark}

\subsection{Comparison to Mack's work on OPE convergence}\label{MackComp}

To assume OPE convergence as an axiom in Euclidean CFT, and to derive
Minkowski physics from it, as we did in this paper, seems natural from the
modern perspective. On the contrary, in the early days of CFTs it was
considered natural to assume standard Minkowski physics (such as Wightman
axioms). The OPE convergence was not assumed at the time, but was something to
be derived.

This was the underlying philosophy of the works by L\"uscher and Mack
{\cite{Luscher:1974ez}} and of Mack {\cite{Mack:1975je,Mack:1976pa}}. Written
45 years ago, these papers remain widely cited, but not everyone is familiar
with what precisely has been derived there. Here we will present a short
review for the benefit of the modern audience.

These works make two main assumptions. \tmtextbf{First,} that we have a
unitary quantum field theory in the Minkowski signature which satisfies
Wightman axioms (in particular has a Hilbert space $\mathcal{H}$ on which the
Poincar\'e group acts unitarily and quantum fields are operator-valued
distributions). Correlators then have the usual analyticity properties of
Wightman functions, in particular they are real-analytic in the Euclidean. The
\tmtextbf{second} main assumption is that these Euclidean correlators are
invariant under the action of the Euclidean conformal group.

Using these two assumptions, L\"uscher and Mack {\cite{Luscher:1974ez}} proved
that the Hilbert space $\mathcal{H}$ carries a unitary representation not just
of the Poincar\'e but of the group $G^{\ast} =$universal cover of the Lorentzian
conformal group $\tmop{SO} (d, 2)$.\footnote{\label{LMcomplain}One also often
quotes L\"uscher and Mack {\cite{Luscher:1974ez}} for proving that CFT
correlation functions may be continued to the Minkowski cylinder $S^{d - 1}
\times \mathbb{R}$. This is a misquotation as they did not prove this, but
posed it as a conjecture. What they did prove was that CFT correlation
functions can be analytically continued to a domain of which $S^{d - 1} \times
\mathbb{R}$ is a real boundary. One still needs to establish a powerlaw bound
to take the boundary limit and obtain a distribution. We plan to derive this
fact in our future work {\cite{paper3}}, for 4-point functions, from the
Euclidean CFT axioms, using the $\rho, \bar{\rho}$ expansion.} \ Mack
{\cite{Mack:1975je}} then classified all unitary positive energy
representations of $G^{\ast}$. It should be mentioned that Refs.\ {\cite{Luscher:1974ez,Mack:1975je,Mack:1976pa}} only consider $d = 4$
spacetime dimensions. Many explicit group theoretic calculations are done only
for this value of $d$. The results should however generalize to arbitrary $d$
with appropriate modifications.

Continuing this program, Mack {\cite{Mack:1976pa}} studied distributional OPE
convergence in Minkowski CFT. Since we also have results of this kind (Sec.\ \ref{OPEHCFT}), it will be particularly interesting to compare with Mack's
discussion. So let us review his argument. Compared to
{\cite{Luscher:1974ez,Mack:1975je}}, Ref.\ {\cite{Mack:1976pa}} includes one
extra assumption: that the OPE $\varphi_i (x_1) \varphi_j (x_2)$ of two fields
acting on the Minkowski vacuum is valid in an asymptotic sense. Namely that
for some dense set of states $\psi$ we have\footnote{Mack assumes $x_2 = -
x_1$ but here for simplicity we will assume that this is valid for any $x_1,
x_2$.}
\begin{equation}
  \langle \psi | \nobracket \varphi_i (x_1) \varphi_j (x_2) \rangle \sim
  \sum_k C_{i j k} (x_1, x_2) \langle \psi | \nobracket \varphi_k (0) \rangle,
  \label{asOPE}
\end{equation}
as an asymptotic expansion for rescaling $x_1, x_2 \rightarrow \lambda x_1,
\lambda x_2$, where $C_{i j k} (x_1, x_2)$ are some $\psi$-independent
homogeneous functions: $C_{i j k} (\lambda x_1, \lambda x_2) =
\lambda^{\Delta_k - \Delta_i - \Delta_j} C_{i j k} (x_1, x_2)$.\footnote{In fact  $C_{i j k}$ is a distribution so homogeneity should be understood in the sense of pairing with a rescaled test function.} Asymptotic
means that if we truncate the expansion at $\Delta_k = \Delta_{\ast}$ and take
$\lambda \rightarrow 0$ limit for any fixed $x_1, x_2$ then the error is $o
(\lambda^{\Delta_{\ast} - \Delta_i - \Delta_j})$. Note that there are both
primaries and descendants among $\varphi_k$'s. The main result of
{\cite{Mack:1976pa}} is to convert this asymptotic expansion to an expansion
convergent in the Hilbert space sense.

The first step is to use a general result that any Hilbert space carrying a
unitary representation of a semisimple Lie group can be decomposed as a direct
integral of unitary irreducible representations {\cite{Mackey1}}. Since, by
the above-mentioned result of {\cite{Luscher:1974ez}}, $\mathcal{H}$ carries a
unitary representation of $G^{\ast}$, we thus have
\begin{equation}
  \mathcal{H}= \int^{\oplus} d \mu_{\chi}\, d \tilde{\mu}_{\nu}\,
  \mathcal{H}^{\chi \nu}, \label{directint}
\end{equation}
where $\chi$ labels different unitary irreps of $G^{\ast}$, $\chi =
(\Delta_{\chi}, \rho_{\chi})$ with $\Delta_{\chi}$ the scaling dimension and
$\rho_{\chi}$ a Lorentz group irrep, and $\nu$ labels different copies of the
same irrep. Only positive energy irreps may occur, since all states of
Wightman theory have positive energy. By definition, Eq.\ {\eqref{directint}}
identifies every vector $\psi \in \mathcal{H}$ with a Hilbert-space-valued
function $(\chi, \nu) \mapsto \psi_{\chi \nu} \in \mathcal{H}^{\chi}$, some
standard realization of the irrep $\chi$. It is assumed that $\langle \psi |
\psi \rangle < \infty$, inner products being given by the following integral:
\begin{equation}
  \langle \psi | \psi' \rangle \nobracket = \int d \mu_{\chi}\, d
  \tilde{\mu}_{\nu}\, \langle \psi_{\chi \nu} | \psi'_{\chi \nu}
  \rangle_{\mathcal{H}^{\chi}}, \label{psipsiprime}
\end{equation}
Also $G^{\ast}$ acts on $\psi$ by acting on each $\psi_{\chi \nu}$.
Furthermore, Ref.\ {\cite{Mack:1975je}} realized $\mathcal{H}^{\chi}$ as a
space of distributions $\psi (x)$ on Minkowski space\footnote{More properly
$\psi (x)$ is a distribution on the Lorentzian cylinder on which the group
$G^{\ast}$ acts naturally, but due to a periodicity condition it may be
reconstructed from its values on the Poincar\'e patch.} taking values in the
representation space of $\rho_{\chi}$, with Fourier transform supported in the
forward light cone, and for which the following inner product (defined
initially on a dense subset of smooth $\psi, \psi'$) is finite:
\begin{equation}
  \langle \psi | \psi' \rangle_{\mathcal{H}^{\chi}} = \int d x\, d y\,
  \overline{\psi (x)} I^{\chi} (x - y) \psi' (y),
  \label{Hchi}
\end{equation}
where $I^{\chi}$ is an ``intertwining kernel''. Physically, $I^{\chi}$
is the Minkowski CFT 2-point function of the primary in the ``shadow irrep'' of
$\chi$.

The above integration measure $d \mu_{\chi} d \tilde{\mu}_{\nu}$ depends on
the theory; from the abstract arguments alone it may be continuous or
discrete. Ref.\ {\cite{Mack:1976pa}} then proceeds to show that (a) this
measure is actually discrete (a sum of delta-functions), so that the direct
integral is a direct sum; (b) that the state $| \nobracket \mathcal{O}_i (x_1)
\mathcal{O}_j (x_2) \rangle$ produced by two Minkowski primary operators
acting on the Minkowski vacuum can be written as
\begin{equation}
  | \nobracket \mathcal{O}_i (x_1) \mathcal{O}_j (x_2) \rangle = \sum_{k, a}
  f^a_{i j k} \int d x\, B^a_{k, i j} (x, x_1, x_2) | \nobracket \mathcal{O}_k
  (x) \rangle, \label{OMack}
\end{equation}
where $B_{i j k}^a (x, x_1, x_2)$, $a = 1 \ldots N_{i j k},$ are some
kinematically determined distributions, the convergence is in the Hilbert
space sense after integrating with an arbitrary test functions $f (x_1, x_2)$,
and the local primary operators $\mathcal{O}_k$ occurring in this sum have
quantum numbers in the discrete set where the integration measure $d
\mu_{\chi} d \tilde{\mu}_{\nu}$ is supported.

To show how this comes about, let us focus on the case of scalar identical
$\mathcal{O}_i =\mathcal{O}_j =\mathcal{O}$ for simplicity. In this case
expansion {\eqref{OMack}} will end up being precisely our expansion
{\eqref{psipsik}} (although derived under very different assumptions), with
$B (x, x_1, x_2)$ related to the OPE kernel in {\eqref{OPE:Euclidean}}.

{Applying \eqref{psipsiprime} with 
$|\psi'\rangle =| \nobracket \mathcal{O} (x_1) \mathcal{O}(x_2) \rangle$ gives} (Eq.\ (2.6) in {\cite{Mack:1976pa}}):
\begin{equation}
  \langle \psi | \mathcal{O} (x_1) \mathcal{O} (x_2) \rangle \nobracket =
  \int d \mu_{\chi}\, d \tilde{\mu}_{\nu}\, c_{\chi \nu} \int d x\,
  \overline{\psi_{\chi \nu} (x)} B_{\chi} (x, x_1, x_2), \label{psiO}
\end{equation}
where we denoted $\int d y\, I^{\chi} (x-y) \psi_{\chi \nu}' (y) =
c_{\chi \nu} B_{\chi} (x, x_1, x_2)$ where $c_{\chi \nu}$ is a proportionality
factor, and $B_{\chi} (x, x_1, x_2)$ is a kinematically determined
distribution (it is basically the amputated Minkowski 3-point function $\langle
\mathcal{O}_{\chi} (x) \mathcal{O} (x_1) \mathcal{O} (x_2) \rangle$). Mack then undertakes a meticulous study of $B_{\chi} (x, x_1, x_2)$ and of its
Fourier transform with respect to the first argument $\hat{B}_{\chi} (p, x_1,
x_2)$. This actually takes most of his paper, and involves many explicit
nontrivial calculations (e.g.\ it involves the first ever explicit
characterization of the most general 3-point function of CFT primaries in
arbitrary irreps). One of the main results is that $\hat{B}_{\chi} (p, x_1,
x_2)$ are entire functions of $p$:
\begin{equation}
  \hat{B}_{\chi} (p, x_1, x_2) = \sum_{| \alpha | \geqslant 0}
  b^{\alpha}_{\chi} (x_1, x_2) (- i p)_{\alpha}, \label{Bhatexp}
\end{equation}
where $b^{\alpha}_{\chi} (x_1, x_2) = (x^2_{12})^{- \Delta_{\mathcal{O}} +
\Delta_{\chi} / 2}$ times a polynomial in $x_1, x_2$ of degree $| \alpha |$,
in particular $b^{\alpha}_{\chi} (\lambda x_1, \lambda x_2) =
\lambda^{\Delta_{\chi} + | \alpha | - 2 \Delta_{\mathcal{O}}}
b^{\alpha}_{\chi} (x_1, x_2)$.

Let us now specialize to states $\psi$ for which the function $\psi_{\chi \nu}
(x)$ has Fourier transform of compact support (one can show that such
states are dense in $\mathcal{H}$). Then the previous equations imply the
following convergent expansion for the integrand in {\eqref{psiO}}:
\begin{equation}
  \int d x\, \overline{\psi_{\chi \nu} (x)} B_{\chi} (x, x_1, x_2) =
  \sum_{\alpha} b^{\alpha}_{\chi} (x_1, x_2)  \overline{\partial^{\alpha}
  \psi_{\chi} (0)}, \label{Bhatint:exp}
\end{equation}
Mack then claims (before Eq.\ (2.11$'$)) that if, for each $\chi$, this
convergent expansion is truncated at $\Delta_{\chi} + | \alpha | =
\Delta_{\ast}$, and inserted back into {\eqref{psiO}}, this results in an
asymptotic expansion for the l.h.s.\ of {\eqref{psiO}}. I.e.\ for any
$\Delta_{\ast}$ (Mack does not write this equation explicitly):
\begin{equation}
  \langle \psi | \mathcal{O} (x_1) \mathcal{O} (x_2) \rangle \nobracket =
  \left\{ \int d \mu_{\chi}\, d \tilde{\mu}_{\nu}\, c_{\chi \nu}
  \sum_{\Delta_{\chi} + | \alpha | \leqslant \Delta_{\ast}} b^{\alpha}_{\chi}
  (x_1, x_2)  \overline{\partial^{\alpha} \psi_{\chi} (0)} \right\} + E
  (x_1, x_2 ; \Delta_{\ast}), \label{gap?}
\end{equation}
where the error term $E (\lambda x_1, \lambda x_2 ; \Delta_{\ast}) = O
(\lambda^{\Delta_{\ast} - 2 \Delta_{\mathcal{O}}})$ as $\lambda \rightarrow 0$
for any fixed $x_1$. Unfortunately, Mack does not give any justification of
this claim, which to us does not appear self-evident. The difficulty is that
although for every $\chi, \nu$ the truncated series has error $O
(\lambda^{\Delta_{\ast} - 2 \Delta_{\mathcal{O}}})$, the constant will
certainly depend on $\chi, \nu$. How do we know that the error estimate
survives after the integration in $\chi, \nu$? It might be possible to close
this omission in Mack's reasoning using normalizability of $| \psi \rangle
\nobracket$, but this needs extra arguments compared to what is given in his
paper, and we have not investigated this in detail.\footnote{We also tried, but unfortunately we did not manage, to get feedback from Prof. Gerhard Mack concerning this matter.}
Researchers relying on Mack's result should keep this caveat in mind.

Assuming that {\eqref{gap?}} is true, the argument is completed as follows. We
now have two asymptotic expansions for the l.h.s.\ of $\langle \psi |
\mathcal{O} (x_1) \mathcal{O} (x_2) \rangle \nobracket$, one coming from
{\eqref{gap?}}, and another from {\eqref{asOPE}}. The second one is discrete
(by assumption), so the first one also must be discrete. This establishes that
the measure $d \mu_{\chi} d \tilde{\mu}_{\nu}$ is discrete, a sum of delta
functions, hence we can write {\eqref{psiO}} with the r.h.s.\ as a sum, not an
integral:
\begin{equation}
  \langle \psi | \mathcal{O} (x_1) \mathcal{O} (x_2) \rangle \nobracket =
  \sum_n c_{\chi_n} \int d x\, \overline{\psi_{\chi_n} (x)} B_{\chi_n} (x,
  x_1, x_2) . \label{psiO1}
\end{equation}
A more detailed comparison of this equation with {\eqref{asOPE}} leads us to
conclude that (a)
\begin{equation}
  c_{\chi_n}  \overline{\psi_{\chi_n} (x)} = f_n \langle \psi |
  \mathcal{O}_{\chi_n} (x) \rangle, \label{psiO2}
\end{equation}
where $\mathcal{O}_{\chi_n}$ are primary operators related by rescaling to a
subset of the local operators $\varphi_k$, we choose them unit-normalized
(hence a coefficient $f_n$); (b) that all the other operators $\varphi_k$ are
the descendants $\mathcal{O}_{\chi_n}$'s; and (c) that all coefficients $C_k
(x_1, x_2)$ are basically the expansion coefficients $b^{\alpha}_{\chi_n}
(x_1, x_2)$ in {\eqref{Bhatexp}}. From {\eqref{psiO1}} and {\eqref{psiO2}}, we
have
\begin{equation}
  \langle \psi | \mathcal{O} (x_1) \mathcal{O} (x_2) \rangle \nobracket =
  \sum_n f_n \int d x\, \langle \psi | \mathcal{O}_{\chi_n} (x) \rangle
  B_{\chi_n} (x, x_1, x_2), \label{psiO3}
\end{equation}
for a dense set of states $\psi$. Because of the orthogonality of different $|
\nobracket \mathcal{O}_{\chi_n} (x) \rangle$'s, this implies that
\begin{equation}
   | \mathcal{O} (x_1) \mathcal{O} (x_2) \rangle \nobracket = \sum_n f_n
  \int d x\, B_{\chi_n} (x, x_1, x_2) | \nobracket \mathcal{O}_{\chi_n} (x)
  \rangle, \label{psiO4}
\end{equation}
the sum convergent in the Hilbert space sense after integrating out with any
test function $f (x_1, x_2)$. This is Eq.\ {\eqref{OMack}} in the considered
case $\mathcal{O}_i =\mathcal{O}_j =\mathcal{O}$.

\subsubsection{Relating Mack's kernel $B$ to the Euclidean OPE kernel $C$}

Now we would like to relate Mack's OPE kernel $B$ to our OPE kernel $C_{a,
(\lambda)}^{(\mu) (\nu)} (x_1, x_2, x_0, \partial_0)$ defined by Eqs.~(\ref{OPEgeneral}) and (\ref{Cexp}). We only consider the OPE kernel for the
scalar external operators for simplicity, i.e.\ $C_{\chi} (x_1, x_2, x_0,
\partial_0)$, $\chi = (\Delta, \ell = 0)$; similar remarks apply in the
general case. We first give the conclusion:
\begin{equation}
  C_{\chi} (x_1, x_2, x_0, \partial_0) = \underset{\mu}{\sum} b_{\chi}^{\mu}
  (x_1 - x_0, x_2 - x_0) \partial_0^{\mu}, \label{ope:relation}
\end{equation}
where the coefficient functions $b_{\chi}^{\mu}$ are the same as in
(\ref{Bhatexp}). One could ``derive'' this by using (\ref{OMack}) and formally
manipulating the integral in the momentum space:\footnote{Since here we are
only interested in the OPE kernels $B_{\chi}$ and $C_{\chi}$, we set $f_{\chi}
= 1$ (the overall coefficient) for convenience.}
\begin{eqnarray}
  | \mathcal{O} (x_1) \mathcal{O} (x_2) \rangle_{\chi} & = & \int d x\, B_{\chi}
  (x, x_1, x_2) | \nobracket \mathcal{O}_{\chi} (x) \rangle \nonumber\\
  & = & \int d p\, \hat{B}_{\chi} (p, x_1, x_2) | \nobracket
  \hat{\mathcal{O}}_{\chi} (p) \rangle = \underset{\mu}{\sum} b_{\chi}^{\mu}
  (x_1, x_2) | \partial^{\mu} \mathcal{O}_{\chi} (0) \rangle, 
\end{eqnarray}
which shows (\ref{ope:relation}) in the case when $x_0 = 0$. The general $x_0$
case follows by translation invariance. This derivation is not rigorous for
various reasons: (a) we did not clarify the meaning of
$\hat{\mathcal{O}}_{\chi} (p)$; (b) why can we swap the order of summation and
integration? (c) the above derivation is done in Minkowski region, how do we
match the coefficients $b_{\chi}^{\mu} (x_1, x_2)$ with the Euclidean
coefficients in (\ref{Cexp})?

Below we will give a rigorous justification of (\ref{ope:relation}), using
only the two- and 3-point functions which are kinematically determined by
conformal invariance. Recall that on the Euclidean side, the formal power
series of $C_{\chi}$ (the scalar version of (\ref{Cexp}))
\begin{equation}
  C_{\chi} (x_1, x_2, x_0, \partial_0) =
  \frac{1}{(x_{12}^2)^{\Delta_{\mathcal{O}} - \Delta_{\chi} / 2}}
  \underset{\mu}{\sum} c_{\chi}^{\mu} (x_{10}, x_{20}) \partial_0^{\mu}
  \label{Cexp:scalar}
\end{equation}
is determined by the Euclidean two- and 3-point functions:
\[ \langle \mathcal{O}_{\chi}^{\dag} (y) \mathcal{O} (x_1) \mathcal{O} (x_2)
   \rangle_E = \frac{1}{(x_{12}^2)^{\Delta_{\mathcal{O}} - \Delta_{\chi} / 2}}
   \underset{\alpha}{\sum} c_{\chi, \alpha} (x_{10}, x_{20}) \langle
   \mathcal{O}_{\chi}^{\dag} (y) \partial^{\alpha} \mathcal{O}_{\chi} (x_0)
   \rangle_E . \label{c:Euclidean} \]
Here we already used translation invariance, which implies $C_{a, (\lambda),
\alpha}^{(\mu) (\nu)} (x_1, x_2, x_0) = c_{a, (\lambda), \alpha}^{(\mu) (\nu)}
(x_{10}, x_{20})$ on the r.h.s.~of (\ref{Cexp}). One can match the coefficients
$c_{\chi}^{\mu} (x_{10}, x_{20})$ with the Taylor expansion of $\langle
\mathcal{O}_{\chi}^{\dag} (y) \mathcal{O} (x_1) \mathcal{O} (x_2) \rangle_E$
around $x_1 = x_2 = x_0$. In Euclidean one can always find a proper region for
the matching: let $y$ be sufficiently far from the $(x_0, x_1, x_2)$ cluster,
so that the Taylor expansions of $[(y - x_1)^2]^{\#}$ around $x_1 = x_0$ and
$[(y - x_2)^2]^{\#}$ around $x_2 = x_0$ are convergent.

On the Minkowski side, the OPE kernel $B_{\chi}$ is kinematically determined
by the equality
\begin{equation}
  \hat{G}_{\mathcal{O}_{\chi} \mathcal{O}\mathcal{O}} (p, x_1, x_2) =
  \hat{B}_{\chi} (p, x_1, x_2) \hat{G}_{\chi} (p), \label{Mackope:3-pointFT}
\end{equation}
where $\hat{G}_{\mathcal{O}_{\chi} \mathcal{O}\mathcal{O}} (p, x_1, x_2) =
\int d y\, \langle \mathcal{O}_{\chi}^{\dag} (y) \nobracket \mathcal{O} (x_1)
\mathcal{O} (x_2) \rangle_M e^{- i p \cdummy y}$ and $\hat{G}_{\chi} (p) =
\int d y\, \langle \mathcal{O}_{\chi}^{\dag} (y) \nobracket
\mathcal{O}_{\chi} (0) \rangle_M e^{- i p \cdummy y}$ (see
{\cite{Mack:1976pa}}, Eq.\ (8.2)).\footnote{In the unitary CFTs,
$\hat{G}_{\bar{n} \mathcal{O}\mathcal{O}} (p, x_1, x_2)$ and $\hat{G}_{\chi_n}
(p)$ vanish unless $p \in \overline{V_+}$, so the behavior of
$\hat{B}_{\chi_n} (p, x_1, x_2)$ outside the forward light cone is not
important.} All Fourier transforms here are in the sense of distributions. To
get an equation valid in the sense of functions we pick a test function
$\varphi$ with compactly supported Fourier transform, and integrate
{\eqref{Mackope:3-pointFT}} against $\hat{\varphi}$, which gives:
\begin{equation}
  \int \langle \mathcal{O}_{\chi}^{\dag} (x) \nobracket \mathcal{O} (x_1)
  \mathcal{O} (x_2) \rangle_M \varphi (x)\, d x = \int d p\, \hat{\varphi} (p)
  \hat{B}_{\chi} (p, x_1, x_2) \hat{G}_{\chi} (p) . \label{3-point:exp10}
\end{equation}
The variable $x$ ranges over the Minkowski space, while we will pick $x_1,
x_2$ complex, in the forward tube region
\begin{equation}
  \tmop{Im} (x_1), \tmop{Im} (x_2) \prec 0. \label{region:matching}
\end{equation}
Then the 3-point function $\langle \mathcal{O}_{\chi}^{\dag} (x) \nobracket
\mathcal{O} (x_1) \mathcal{O} (x_2) \rangle_M$ is nonsingular as a function of
$x$ and the l.h.s.\ of {\eqref{3-point:exp10}} is a finite number. To transform the
r.h.s.\ of {\eqref{3-point:exp10}} we will use the fact that $\hat{B}_{\chi}$ has
the following form (a more detailed version than (\ref{Bhatexp})):
\begin{equation}
  \hat{B}_{\chi} (p, x_1, x_2) = \frac{e^{- i p \cdummy
  x_1}}{(x_{12}^2)^{\Delta_{\mathcal{O}} - \Delta_{\chi} / 2}} E_{\chi}
  (x_{12} \cdummy p, x_{12}^2 p^2), \label{Bhat:entire}
\end{equation}
where $E_{\chi} (z_1, z_2)$ is some entire function on $\mathbb{C}^2$. Hence
as long as $x_{12}^2 \neq 0$ (not necessarily real), $\hat{B}_{\chi} (p, x_1,
x_2)$ has the following convergent expansion:
\begin{equation}
  \hat{B}_{\chi} (p, x_1, x_2) = \frac{e^{- i p \cdummy
  x_1}}{(x_{12}^2)^{\Delta_{\mathcal{O}} - \Delta_{\chi} / 2}}
  \underset{\alpha}{\sum} a_{\chi}^{\alpha} (x_{21}) (- i p)_{\alpha},
  \label{Bhatexp2}
\end{equation}

where $a_{\chi}^{\alpha} (x)$ is some $\tmop{SO} (1, d - 1)$-covariant,
homogeneous, symmetric polynomial of degree $| \alpha |$. Plugging this into
{\eqref{3-point:exp10}}, using that the expansion {\eqref{Bhatexp2}} converges
uniformly on the support of $\hat{\varphi}$ (assumed compact), and the fact
that $\hat{G}_{\chi}$ is a tempered measure,\footnote{This is a consequence of
the Bochner-Schwartz theorem: any positive tempered distribution is the
Fourier transform of some positive tempered measure (see {\cite{Vladimirov2}},
Sec.\ 8.2).} we obtain:
\begin{equation}
  \int \langle \mathcal{O}_{\chi}^{\dag} (x) \nobracket \mathcal{O} (x_1)
  \mathcal{O} (x_2) \rangle_M \varphi (x)\, d x =
  \frac{1}{(x_{12}^2)^{\Delta_{\mathcal{O}} - \Delta_{\chi} / 2}}
  \underset{\alpha}{\sum} a_{\chi}^{\alpha} (x_{21}) \int \langle
  \mathcal{O}_{\chi}^{\dag} (x) \partial_{\alpha} \mathcal{O}_{\chi} (x_1)
  \rangle_M \varphi (x)\, d x. \label{3-point:exp1}
\end{equation}
At this stage we have established that for any $x_1, x_2$ as in
{\eqref{region:matching}}, and for any test $\varphi$ with compact
$\hat{\varphi}$, the series in the r.h.s.\ converges to the l.h.s.

Now the key point is that the r.h.s.\ of {\eqref{3-point:exp1}} is a convergent
power series in $x_{21}$, while the l.h.s.\ can be expanded in such a
convergent power series. Indeed we know the explicit form of $\langle
\mathcal{O}_{\chi}^{\dag} (x) \nobracket \mathcal{O} (x_1) \mathcal{O} (x_2)
\rangle_M$:
\begin{equation}
  \langle \mathcal{O}_{\chi}^{\dag} (x) \nobracket \mathcal{O} (x_1)
  \mathcal{O} (x_2) \rangle_M = \frac{1}{(x_{12}^2)^{\Delta_{\mathcal{O}} -
  \Delta_{\chi} / 2}}  \frac{1}{[(x - x_1)^2 (x - x_2)^2]^{\Delta_{\chi} / 2}}
  . \label{3-point:explicit}
\end{equation}
For all $x$ in the Minkowski space, the function $[(x - x_1)^2 (x - x_2)^2]^{-
\Delta_{\chi} / 2}$ is holomorphic in $x_1, x_2$ as long as $\tmop{Im} (x_1),
\tmop{Im} (x_2) \prec 0$. It's easy to show that this remains true after
integration in $\varphi$. At this point we can match the expansions for the
two sides of {\eqref{3-point:exp1}}, and get
\begin{equation}
  \sum_{| \alpha | = n} \int d x\, \varphi (x) \left\{ \frac{
  (x_{21})_{\alpha}}{\alpha !} \partial_{x_2}^{\alpha} [(x - x_1)^2 (x -
  x_2)^2]^{- \Delta_{\chi} / 2}  | \nobracket_{x_2 = x_1} -
  a_{\chi}^{\alpha} (x_{21}) \langle \mathcal{O}_{\chi}^{\dag} (x)
  \partial_{\alpha} \mathcal{O}_{\chi} (x_1) \rangle_M \right\} = 0 .
  \label{phiint}
\end{equation}
Up to this point it was crucial to keep the function $\varphi$ in the game to
keep convergence issues under control, but now we can get rid of it. Indeed
$\underset{| \alpha | = n}{\sum}$ is a finite sum, also $\langle
\mathcal{O}_{\chi}^{\dag} (x) \partial_{\alpha} \mathcal{O}_{\chi} (x_1)
\rangle_M$ is a holomorphic function in the forward tube $\mathcal{T}_2$. For
any Minkowski point $x_0$, we choose a sequence of test functions $\varphi_k$
of compact support $\hat{\varphi}_k$ such that $\varphi_k$ tends to $\delta (x
- x_0)$. Passing to the limit, {\eqref{phiint}} implies the same equality for
the integrand. I.e.\ for any fixed $n \in \mathbb{N}$, and any Minkowski $x$,
\begin{equation}
  \frac{ (x_{21})_{\alpha}}{\alpha !} \partial_{x_2}^{\alpha} [(x - x_1)^2
  (x - x_2)^2]^{- \Delta_{\chi} / 2}  | \nobracket_{x_2 = x_1} -
  a_{\chi}^{\alpha} (x_{21}) \langle \mathcal{O}_{\chi}^{\dag} (x)
  \partial_{\alpha} \mathcal{O}_{\chi} (x_1) \rangle_M = 0. \label{ope:match}
\end{equation}
Now as promised we are reduced to an equation which only involves 2-point
and 3-point functions which are holomorphic. E.g.\ we can take $x = 0$ and
$x_1, x_2$ in Euclidean. Then this equation is the same one as the equation
which determines the Euclidean OPE kernel for $x_0 = x_1$, i.e.\ $C_{\chi}
(x_1, x_2, x_1, \partial)$. For convenience in this discussion we use
Minkowski coordinates for Euclidean correlators (i.e.\ we write the Euclidean
correlators as $\langle \nobracket \langle \mathcal{O} (- i \tau, \mathbf{x})
\ldots \rangle_M$). Under this convention we have
\begin{equation}
  C_{\chi} (x_1, x_2, x_1, \partial) =
  \frac{1}{(x_{12}^2)^{\Delta_{\mathcal{O}} - \Delta_{\chi} / 2}}
  \underset{(\mu)}{\sum} a_{\chi}^{\mu} (x_{21}) \partial_{\mu} =
  \underset{(\mu)}{\sum} b_{\chi}^{\alpha} (0, x_{21}) \partial_{\alpha} .
  \label{ope:kernelmatch1}
\end{equation}

This establishes (\ref{ope:relation}) for $x_0 = x_1$. The general case
reduces to this one by noticing that $c_{\chi}^{\alpha}$ satisfies the
relation:
\begin{equation}
  c_{\chi}^{\alpha} (x_{10}, x_{20}) = \underset{}{\underset{\beta
  \leqslant \alpha}{\sum}} \frac{1}{\beta !} c_{\chi}^{\alpha - \beta} (0,
  x_{21}) x_{10}^{\beta},
\end{equation}
where $\beta \leqslant \alpha$ means $\beta_i \leqslant \alpha_i$ for all $i$;
and $b_{\chi}$ satisfies the identical relation with $c_{\chi} \rightarrow
b_{\chi}$. For $c_{\chi}^{\alpha}$ this follows by translation invariance and
analyticity of CFT two- and 3-point functions, and for $b_{\chi}^{\alpha}$
from $\hat{B}_{\chi} (p, x_{10}, x_{20}) = e^{- i p \cdummy x_{10}}
\hat{B}_{\chi} (p, 0, x_{21})$.

\section{Review of Osterwalder-Schrader theorem}\label{OS}

In this section we review the results of
{\cite{osterwalder1973,osterwalder1975}} and, in particular, discuss the
linear growth condition and why it was necessary for establishing Wightman
axioms in {\cite{osterwalder1975}}.

In {\cite{osterwalder1973}} Osterwalder and Schrader formulated an equivalence
theorem which stated that a set of axioms for Euclidean correlation functions
(a version of the Osterwalder-Schrader axioms described in Sec.\ \ref{OSaxioms}) is
equivalent to Wightman axioms for Euclidean correlation functions.
Unfortunately, later a technical error was discovered in their proof, and in
{\cite{osterwalder1975}} Osterwalder and Schrader gave two new results.

The first result of {\cite{osterwalder1975}} is a revised equivalence
theorem, which shows that a stronger version of Euclidean axioms is in fact
equivalent to Wightman axioms. The proof of this theorem is rather simple.
However, as we will review, this is at the expense of the new version of
Euclidean axioms being rather hard to verify.

The second result of {\cite{osterwalder1975}} shows that the original OS
axioms, plus a ``linear growth condition,'' imply Wightman axioms and a growth
condition on Wightman distributions. A partial result in the reverse direction
is also valid. It assumes a stronger growth condition on the Wightman
distributions than follows from the direct result, and it yields a growth
condition on Euclidean correlators which is weaker than the linear growth
condition. Therefore, these latter results do not establish an equivalence of
any two systems of axioms, but they do allow to establish Wightman axioms from
OS axioms in some situations.

In what follows we will review the general structure of the arguments of
{\cite{osterwalder1973,osterwalder1975}}. For our purposes it will suffice to
ignore the space coordinates and focus only on the time arguments of the
fields. We will not completely reproduce all arguments of
{\cite{osterwalder1973,osterwalder1975}}, and in some of the omitted steps the
space arguments and Lorentz symmetry are important. We will also work with
correlation functions involving a single hermitian scalar field $\phi$,
similarly to {\cite{osterwalder1973,osterwalder1975}}. In CFT applications we
are interested in correlation functions of all local operators. It should be
relatively straightforward to adapt the discussion of
{\cite{osterwalder1973,osterwalder1975}} to this more general setup.

Our main goal is to construct an analytic continuation of the Euclidean
correlation functions
\begin{equation}
  G^E_n (t_1, \ldots, t_n) \equiv \langle \phi (t_1) \ldots \phi (t_n) \rangle
\end{equation}
from real to complex $t_k$, and to establish that the Wightman functions
recovered in the limit of pure imaginary $t_k$ (real Lorentzian times) are
tempered distributions. This is the most non-trivial part of the argument.
Other Wightman axioms such as positivity, spectrum condition, etc., follow
relatively easily and have been reviewed in Sec.\ \ref{sec:4-point}.

\subsection{The argument of {\cite{osterwalder1973}}}

Physically, the analyticity of position-space correlation functions is due to
positivity of energy. More concretely, the Euclidean evolution operator $e^{-
H t}$ is well-defined and holomorphic in $t$ for $\tmop{Re} t > 0$ due to the
spectrum of $H$ being non-negative. The first step to establishing analyticity
is thus to construct the operator $H$, and for this we first need to construct
a Hilbert space on which it acts.

The Hilbert space $\mathcal{H}^{\tmop{OS}}$ is constructed, as we discussed in
Sec.\ \ref{OSaxioms}, by considering the vector space
$\mathcal{H}^{\tmop{OS}}_0$ of formal linear combinations of
states\footnote{In Sec.\ \ref{OSaxioms} the states are introduced as
integrals of these quantities. This is also what is done in
{\cite{osterwalder1973,osterwalder1975}}, since they assume only that the
Euclidean correlators are distributions. Here, for simplicity of discussion,
we use the knowledge that correlators are functions and use states evaluated
at points. The arguments easily generalize to distributions and smeared
states, but become more technical.}
\begin{equation}
  | \phi (t_1) \phi (t_2) \ldots \phi (t_n) \rangle \label{basisstate}
\end{equation}
with $0 > t_1 > t_2 > \cdots > t_n .$ A Hermitian inner product is introduced
on $\mathcal{H}^{\tmop{OS}}_0$ by
\begin{equation}
  \langle \phi (s_1) \phi (s_2) \ldots \phi (s_m) \nobracket | \phi (t_1) \phi
  (t_2) \ldots \phi (t_n) \rangle \equiv G_n^E (- s_m, \ldots, - s_1, t_1,
  \ldots, t_n) .
\end{equation}
By OS reflection positivity-axiom this inner product is positive-semidefinite.
The Hilbert space $\mathcal{H}^{\tmop{OS}}$ is obtained from
$\mathcal{H}^{\tmop{OS}}_0$ by modding out null states and completing the
resulting quotient space with respect to the above inner product. We can
naturally think of $| \phi (t_1) \phi (t_2) \ldots \phi (t_n) \rangle$ as
states in $\mathcal{H}^{\tmop{OS}}$.

Physically, to construct the Hamiltonian $H$, we first define it by its action
on {\eqref{basisstate}}. Then we note that $H$ has to be positive, otherwise
the correlation functions would grow exponentially at large distances.
Formally, one first defines for $t > 0$ an operator $U_t$ on
$\mathcal{H}^{\tmop{OS}}_0$ by
\begin{equation}
  U_t | \phi (t_1) \phi (t_2) \ldots \phi (t_n) \rangle \equiv
  {| \phi (t_1 - t) \phi (t_2 - t) \ldots \phi (t_n - t)
  \rangle} .
\end{equation}
The usual care must be taken to ensure that this defines an operator on
$\mathcal{H}^{\tmop{OS}}$. For this one notes that for any $\Psi \in
\mathcal{H}^{\tmop{OS}}_0$ we have $| \nobracket \langle \Psi | U_t | \Psi
\rangle | \nobracket \leqslant P (t)$ for some polynomial $P (t)$ since the
Euclidean correlation functions are assumed to be powerlaw-bounded when groups
of points are separated to infinity. Then a simple estimate gives
\begin{equation}
  | \nobracket \langle \Psi | U_t | \Psi \rangle | \nobracket \leqslant \|
  \Psi \| \| U_t \Psi \| = \| \Psi \| | \nobracket \langle \Psi | U_{2 t} |
  \Psi \rangle | \nobracket^{1 / 2} \leqslant \cdots \leqslant \| \Psi \|^{1 +
  1 / 2 + \cdots + 1 / 2^{n - 1}} | \nobracket \langle \Psi | U_{2^n t} | \Psi
  \rangle | \nobracket^{1 / 2^n} .
\end{equation}
Using $| \langle \Psi | U_t | \Psi \rangle | \leqslant P (t)$ we get in the
limit $n \rightarrow \infty$
\begin{equation}
  | \nobracket \langle \Psi | U_t | \Psi \rangle | \nobracket \leqslant \|
  \Psi \|^{1 + 1 / 2 + \cdots + 1 / 2^{n - 1}} (P (2^n t))^{1 / 2^n}
  \rightarrow \| \Psi \|^2 . \label{contractive}
\end{equation}
This shows that $U_t$ maps null states to null states and thus is defined on
(a dense subset of) $\mathcal{H}^{\tmop{OS}}$. By the above, it is also a
bounded operator, so it extends in a unique way to all of
$\mathcal{H}^{\tmop{OS}}$. Furthermore, noting that it is symmetric, of norm
at most 1, and we have the semigroup law $U_t U_s = U_{t + s}$, we find that
$U_t = e^{- H t}$ for a non-negative self-adjoint Hamiltonian $H$ (see, e.g.,
{\cite{funcan}} Sec.\ 141).

Since the domain in which we need to construct the analytic continuation of
$G_n^E$ is awkward to define in $t_k$ variables, we introduce the difference
variables $y_k \equiv t_k - t_{k + 1} .$ Due to translation invariance, $G_n^E
(t_1 \ldots t_n)$ can be rewritten as
\begin{equation}
  G_n^E (t_1 \ldots t_n) = S_{n - 1} (y_1 \ldots y_{n - 1})
\end{equation}
for some functions $S_n$. Similarly, we will use the following notation for
states in terms of $y_k$ variables,
\begin{equation}
  | \Psi_n (- t_1 ; y_1 \ldots y_{n - 1}) \rangle \equiv | \phi (t_1) \phi
  (t_2) \ldots \phi (t_n) \rangle .
\end{equation}
Note that
\begin{equation}
  \langle \nobracket \Psi_m (t' ; y'_1 \ldots y'_{m - 1}) | \Psi_n (t ;
  y_1 \ldots y_{n - 1}) \rangle = S_{m + n - 1} (y_{m - 1}', \ldots, y_1', t'
  + t, y_1, \ldots, y_{n - 1}) .
\end{equation}
In terms of $S_{n - 1} (y_1 \ldots y_{n - 1})$, our goal is to construct an
analytic continuation to $\tmop{Re} y_k > 0$ and show that the limit of all
$\tmop{Re} y_k \rightarrow 0^+$ exists in the sense of tempered distributions.

With a positive $H$ now constructed, we can define a holomorphic family of
bounded operators $U_{\tau} = e^{- H \tau}$ for $\tmop{Re} \tau > 0$, which
will be our main tool for analytically continuing the correlation functions
$S_n$. In particular, we can now consider the matrix elements
\begin{equation}
  \langle \nobracket \Psi_m (t' ; y'_1 \ldots y'_{m - 1}) | U_{\tau} |
  \nobracket \Psi_n (t ; y_1 \ldots y_{n - 1}) \rangle = S_{m + n - 1}
  (y_{m - 1}', \ldots, y_1', t' + t + \tau, y_1, \ldots, y_{n - 1}),
\end{equation}
which are analytic for $\tmop{Re} \tau>0$. This establishes the desired
analyticity of $S_{n - 1} (y_1 \ldots y_{n - 1})$ in each variable $y_k$
separately. In {\cite{osterwalder1973}} they additionally establish some
growth conditions on these individual holomorphic functions which then imply that
for fixed $y_k, y'_k$ and $\tmop{Re} \tau > 0$ the above function can be
represented as the Fourier-Laplace transform
\begin{equation}
  S_{m + n - 1} (y_{m - 1}', \ldots, y_1', \tau, y_1, \ldots, y_{n - 1}) =
  \int d \alpha\, e^{- \alpha \tau}  \check{S} (\alpha)
\end{equation}
for some tempered distribution $\check{S} (\alpha)$. In other words, $S_{m + n
- 1}$ can be extended to a holomorphic function in the right-half plane in each
variable separately, and each such holomorphic function can be represented as a
Fourier-Laplace transform of a tempered distribution. The erroneous Lemma 8.8
of {\cite{osterwalder1973}} states that under these conditions, the full
function $S_{m + n - 1}$ is a simultaneous Fourier-Laplace transform in all
its variables of a tempered distribution,
\begin{equation}
  S_{m + n - 1} (\tau_1 \ldots \tau_{m + n - 1}) = \int d \alpha\, e^{- \alpha_1
  \tau_1 - \cdots - \alpha_{m + n - 1} \tau_{m + n - 1}}  \check{S}_{m + n
  - 1} (\alpha_1 \ldots \alpha_{m + n - 1}) . \label{FLtransform}
\end{equation}
From this, the tempered Wightman distributions are obtained immediately by
setting $\tmop{Re} \tau_k \rightarrow 0^+$ in which case the Fourier-Laplace
transform above becomes a Fourier transform of a tempered distribution.
Fourier transform of a tempered distribution is, of course, itself tempered.

\subsection{The argument of {\cite{osterwalder1975}}}

\subsubsection{Fixing the equivalence theorem}

Unfortunately, Lemma 8.8 of {\cite{osterwalder1973}} is wrong. As explained in
{\cite{osterwalder1975}}, the function $S_2 (y_1, y_2) = e^{- y_1 y_2}$ gives
a simple counter-example. For fixed $y_2 > 0,$we find that $S_2 (y_1, y_2)$ is
holomorphic for $\tmop{Re} y_1 > 0$ and is there the Fourier-Laplace transform of
the tempered distribution $\delta (\alpha - y_2)$. The same statements hold
with $y_1$ and $y_2$ exchanged. However, $S_2 (y_{1,} y_2)$ is not a
Fourier-Laplace transform of a tempered distribution in both variables
simultaneously. For if this were the case, the corresponding Wightman function
$S_2 (i x_1, i x_2) = e^{x_1 x_2}$ would be a tempered distribution, which it
is not since it grows faster than any power in some directions.

The first result of {\cite{osterwalder1975}} (see also the review in \cite{simon1974}) 
rescues Lemma 8.8 by making a
stronger assumption about $S_n (y_1 \ldots y_n)$ which they denote by
$\check{E 0}$. Concretely, let $\mathbb{R}^n_+$ be the set of points $(y_1,
\ldots, y_n)$ with $y_k > 0.$ Let $\mathcal{S} (\mathbb{R}_+^n)$ be the
subspace of the space of Schwartz functions, consisting of functions supported
on $\mathbb{R}^n_+$ with the induced topology. The functions $S_n (y_1 \ldots
.y_n)$ can be viewed as distributions in the continuous dual space
$\mathcal{S}' (\mathbb{R}_+^n)$ defined by, for $f \in \mathcal{S}
(\mathbb{R}_+^n)$
\begin{equation}
  S_n (f) \equiv \int d y_1 \ldots d y_n\, S_n (y_1 \ldots y_n) f (y_1 \ldots
  y_n) .
\end{equation}
Note that smoothness of $f$ together with its support properties ensures that
$f (y_1 \ldots y_n)$ vanishes with all derivatives whenever $y_k = y_j$ for $k
\neq j$. The assumption that $S_n$ has at most powerlaw singularities at
coincident points and at infinity means that $S_n (f)$ is continuous in $f$ in
the topology of $\mathcal{S} (\mathbb{R}_+^n)$. The additional assumption
$\check{E 0}$ is that it is also continuous in $f$ in a weaker topology. This
weaker topology is defined by the usual Schwartz norms on
$\overline{\mathbb{R}_+^n}$
\begin{equation}
  | g |_{p, +} = \sup_{x \in \overline{\mathbb{R}_+^n}, | \alpha | \leqslant
  p} | (1 + x^2)^{p / 2} \partial^{(\alpha)} g (x) |
\end{equation}
but applied not to $f$ and instead to its Fourier-Laplace\footnote{As written,
this is a Laplace transform. It is a Fourier transform in the spatial
variables which we are omitting.} transform $\check{f}$
\begin{equation}
  \check{f} (q_1 \ldots q_n) \equiv \int d y_1 \ldots d y_n\, e^{- q_1 y_1 -
  \cdots - q_n y_n} f (y_1 \ldots y_n) .
\end{equation}
One establishes that $\check{f} = 0$ iff $f = 0$ (injectivity) and that the
set of all images $\check{f}$ is dense in an appropriate space of Schwartz
functions (denseness). The proof of {\eqref{FLtransform}} then becomes
straightforward: one first defines $\check{S}_n$ by $\check{S}_n (\check{f}) =
S_n (f)$. This definition makes sense due to the injectivity property. The
assumption $\check{E 0}$ ensures that $\check{S}_n$ is continuous. The
denseness property just mentioned then allows to extend $\check{S}_n$ to an
appropriate space of Schwartz functions by continuity, establishing
temperedness of $\check{S}_n$ and allowing one to define tempered Wightman
distributions as Fourier transforms of $\check{S}_n$. It is similarly not
difficult to show that Wightman axioms imply $\check{E 0} .$

\subsubsection{Wightman axioms from linear growth condition}

As we can see, the axiom $\check{E 0}$ is not very different from directly
assuming temperedness of Wightman distributions, even though it is formulated
for Euclidean correlators. It is also unclear how to verify this axiom in
practice.\footnote{We would also like to mention related work by Zinoviev
{\cite{Zinoviev1995}}. Zinoviev replaces axiom $\check{E 0}$ by an axiom E5
which imposes that certain limits exist which allow to compute the inverse
Laplace transform of $S_n$. While E5 may look like a more constructive version
of $\check{E 0}$, in practice its verification appears just as hard as
assuming outright that $S_n$ is a Laplace transform (which is what $\check{E
0}$ essentially does). We are grateful to David Brydges for an enlightening
explanation of Zinoviev's construction, and in particular for pointing out
that it represents a generalization of Post's Laplace transform inversion
formula {\cite{Post}} to the case of distributions.} For this reason,
{\cite{osterwalder1975}} introduced an alternative ``linear growth condition''
on the correlation functions $S_n$ which is easier to verify and has been
established in some models (see below), yet is also sufficient to establish
temperedness of Wightman functions (though this condition is not known to follow
from Wightman axioms). The construction of the analytic continuation of the
functions $S_n$ as well the proof of the temperedness of the resulting
Wightman distributions is much more complicated than using $\check{E 0}$.
Therefore, our review of these arguments will be even more schematic than the
above, and we will only try to illustrate the key ideas and explain why and
how the linear growth condition is used. We are not aware of any previous attempt to review this part of \cite{osterwalder1975}.

First of all, let us state the linear growth condition of
{\cite{osterwalder1975}}. Note that the correlation functions $G_n^E$ can be
viewed as distributions in $({}^0 \mathcal{S})' (\mathbb{R}^{d \cdummy n})$,
where ${}^0 \mathcal{S} (\mathbb{R}^{d \cdummy n})$ is the space of Schwartz
functions of $n$ arguments in $\mathbb{R}^d$ which vanish with all derivatives
at coincident points, by
\begin{equation}
  G_n^E (f) = \int d^d x_1 \ldots d^d x_n\, f (x_1, \ldots, x_n) G_n^E (x_{1,}
  \ldots, x_n) .
\end{equation}
Here we have temporarily restored the spatial coordinates. In fact,
{\cite{osterwalder1973,osterwalder1975}} do not assume that $G_n^E$ are
functions, and only that they are distributions in $({}^0 \mathcal{S})'
(\mathbb{R}^{d \cdummy n})$. It follows, however, from the OS axioms (without
the linear growth condition) that $G_n^E$ are real-analytic functions, as
shown in {\cite{Glaser1974,osterwalder1975}}.

Note that the assumption that $G_n^E \in ({}^0 \mathcal{S})' (\mathbb{R}^{d
\cdummy n})$ means $G_n^E$ is sufficiently continuous as a linear functional or, equivalently, is sufficiently bounded. That is,
\begin{equation}
  | G_n^E (f) | \leqslant \sigma_n | f |_{q_n} \label{J0distr}
\end{equation}
for all $f \in {}^0 \mathcal{S} (\mathbb{R}^{d \cdummy n})$ and some $\sigma_n
> 0$ and $q_n \in \mathbb{Z}_{\geqslant 0}$, where $| f |_p$ denotes the
Schwartz norms on ${}^0 \mathcal{S} (\mathbb{R}^{d \cdummy n})$. The linear
growth condition requires $q_n$ to grow at most linearly, and $\sigma_n$ at
most as a power of a factorial. In other words, the linear growth condition is
the statement that there exists a positive integer $s$ and a sequence
$\sigma_n$ such that
\begin{equation}
  | G_n^E (f) | \leqslant \sigma_n | f |_{n \cdummy s} \label{lineargrowth}
\end{equation}
for any $n$ and $f \in {}^0 \mathcal{S} (\mathbb{R}^{d \cdummy n})$, and
$\sigma_n \leqslant \alpha (n!)^{\beta}$ for some constants $\alpha \comma
\beta$.

The unusual feature of the linear growth condition is that this is a condition
on all $n$-point correlation functions $G_n^E .$ It has to hold for all $n$ in
order for the result of {\cite{osterwalder1975}} to imply, say, even just the
temperedness of $3$-point Wightman distribution. In order to understand why
this is required, below we will review the basic strategy behind the proof of
{\cite{osterwalder1975}}. There are two steps in the argument. In the first
step, one establishes analyticity of $S_n (y_1, \ldots ., y_n)$ in the region
$\tmop{Re} y_k > 0.$ This does not require the linear growth condition
{\cite{Glaser1974,osterwalder1975}}. In the second step, which does use the
linear growth condition, one proves a bound on $S_n$ in this region, which
allows the application of Vladimirov's theorem and thus the construction of
tempered Wightman distributions.

We conclude this section with additional comments about the linear growth
condition. First of all, Appendix of {\cite{osterwalder1975}} shows that the
linear growth condition follows from requiring that $G_n^E \in \mathcal{S}'
(\mathbb{R}^{d n})$ and imposing
\begin{equation}
  | G_n^E (f_1 \otimes \ldots \otimes f_n) | \leqslant \sigma_n \prod_{i =
  1}^n | f_i |_r, \label{E0pp}
\end{equation}
for any $n$, where $f_i \in \mathcal{S} (\mathbb{R}^d)$, $| \cdot |_r$ is some
fixed Schwartz space norm, and $\sigma_n \leqslant \alpha (n!)^{\beta}$. In
other words, while in {\eqref{lineargrowth}} the $n$-point function variables
are smeared jointly, here each variable is smeared separately. Note that the
total smearing $f_1 \otimes \ldots \otimes f_n$ does not necessarily exclude
coincident points, that's why we need to assume $G_n^E \in \mathcal{S}'
(\mathbb{R}^{d n})$ and not $G_n^E \in ({}^0 \mathcal{S})' (\mathbb{R}^{d
\cdummy n})$ as above.

Although {\eqref{E0pp}} is stronger than {\eqref{lineargrowth}}, it is easier
to verify in particular models. E.g.\ it holds for any gaussian scalar field
$\mathcal{O}$ with a two point function $G_2$ having a powerlaw asymptotics in
the UV.\footnote{Then $G_n^E (f_1 \otimes \ldots \otimes f_n)$ is a sum of $(n - 1) !!$
terms, products of Wick contractions $G_2  (f_i \otimes f_j)$, which can be
bounded by $A | f_i |_r | f_j |_r$ where $r$ depends on the UV dimension of
$\mathcal{O}$. We thus get {\eqref{E0pp}} with $\sigma_n = (n - 1) !!A^{n/2}$.}
It has been also established in some non-gaussian models.\footnote{See e.g.~\cite{GJS_os_axiom},Theorem 1.1.8, which establishes Eq.~\eqref{E0pp} for Schwinger functions of arbitrarily high normal-ordered powers $:\!\phi^n\!:$ of the fundamental field  $\phi$ in weakly coupled $P(\phi)_2$ theories.}

More generally, bound {\eqref{E0pp}} is natural for field theories realizable
as random distributions.\footnote{See~\cite{Glimm:1981xz}, Sec.\ 6. This book introduced axioms for random distributions, numbered OS0-OS5. This chosen name is a bit unfortunate because these axioms are quite different in spirit from the original Osterwalder-Schrader axioms  described in Sec.~\ref{OSaxioms}, and appear much stronger. E.g.~they make the recovery of Wightman axioms a relatively trivial task. We don't know how to derive the axioms of \cite{Glimm:1981xz} from the Euclidean CFT axioms.} Imagine that there
is a measure $d \mu$ in the space of distributions $\phi \in \mathcal{S}'
(\mathbb{R}^d)$ such that for every test function $f \in \mathcal{S}
(\mathbb{R}^d)$ the following expectation value is finite:
\begin{equation}
  S (f) = \int d \mu \, e^{\phi (f)}\,. \label{genf}
\end{equation}
Such measures make rigorous sense of the Feynman path integral. Eq.
{\eqref{genf}} is a rigorous version of generating functional, and
differentiating with respect to $f$ one defines correlation functions $\langle
\phi (x_1) \ldots \phi (x_n) \rangle$ which are in this framework
automatically distributions in $\mathcal{S}' (\mathbb{R}^{d n})$. Bound
{\eqref{E0pp}} in this case can be reduced to an estimate on the growth of $S
(f)$. The Osterwalder-Schrader and Wightman axioms then follow.

{The field $\phi$ in \eqref{genf} is naturally a ``fundamental field'' of some model, such as $P(\phi)_2$ \cite{Glimm:1981xz} or $(\phi^4)_3$ (see \cite{Glimm:1981xz}, Sec. 23.1 for references). Sometimes this framework can be extended to generating functionals $\int d \mu \, e^{\phi' (f)}$ where $\phi'$ is a composite operator. E.g. $\phi'=\,:\!\phi^n\!:$, $n<\deg P$, in $P(\phi)_2$ is treated in \cite{Glimm:1981xz}. See also 
\cite{Abdesselam:2016npc} for the general problem to construct $:\!\phi^2\!:$ as a random distribution given $\phi$.}

\subsubsection{Analytic continuation}

There are three tricks used together to construct the analytic continuation of $S_n$.
The first trick was already used above: it is the observation that if the
states $| \Psi_n (t ; y_{1,} \ldots, y_n) \rangle$ and $| \Psi_m (t' ; y'_{1,}
\ldots, y'_m) \rangle$ are defined for some values of $t, y_k$ and $t'$,
$y'_k$, then we can compute the matrix elements
\begin{equation}
  \langle \nobracket \Psi_m (t' ; y'_1 \ldots y'_{m - 1}) | U_{\tau} |
  \nobracket \Psi_n (t ; y_1 \ldots y_{n - 1}) \rangle = S_{m + n - 1}
  (\overline{y_{m - 1}'}, \ldots, \overline{y_1'}, t' + t + \tau, y_1, \ldots,
  y_{n - 1}) \label{trick1}
\end{equation}
with $\mathrm{Re}\,\tau>0$, thus potentially extending the domain of analyticity of $S_{m + n - 1}$.

The second trick, intuitively, says that we can write
\begin{equation}
  \langle \Psi_n (t ; y_1 \ldots y_{n - 1}) | \nobracket \Psi_n (t ; y_1
  \ldots y_{n - 1}) \rangle = S_{2 n - 1} (\overline{y_{n - 1}}, \ldots,
  \overline{y_1}, 2 t, y_1, \ldots, y_{n - 1}), \label{trick2}
\end{equation}
and so the state $| \Psi_n (t ; y_1 \ldots y_{n - 1}) \rangle$, whose norm
appears in the left-hand side, should be well-defined as long as the
correlation function in the right-hand side is well-defined. That is, while we
start with the states $| \Psi_n (t ; y_1 \ldots y_{n - 1}) \rangle$ defined
for positive real $y_k$, we should be able to analytically continue them in
$y_k$ if we manage to analytically continue the correlators $S_{2 n - 1} .$ Of
course, this is not a proof that $| \Psi_n (t ; y_1 \ldots y_{n - 1})
\rangle$ is well-defined. We will give the proof below, after we get more
information about the domain in which we wish to construct it.

The final trick is the idea of analytic completion for functions of several
complex variables. Recall that for $n > 1$ not every domain in $\mathbb{C}^n$
is the domain of holomorphy of some holomorphic function: there exist domains
$\mathcal{D} \subset \mathbb{C}^n$ such that any $f$ holomorphic in
$\mathcal{D}$ can be extended to a function holomorphic in a strictly larger
domain $\mathcal{D}' \supset \mathcal{D}$. For our applications the relevant
theorem is Bochner's tube theorem, which states that any holomorphic function
in a tube domain of the form $\mathcal{D}=\mathbb{R}^n + i X$, where $X$ is a
connected open subset of $\mathbb{R}^n$, can be extended to a holomorphic
function on $\mathcal{D}' = \tmop{ch} (\mathcal{D}) =\mathbb{R}^n + i
\tmop{ch} (X),$ where $\tmop{ch}$ denotes the convex hull. Note that since
$\mathcal{D}'$ is a convex set, it is a domain of holomorphy\footnote{To see
this, it suffices to show that for any point $z$ on the boundary of
$\mathcal{D}'$ there exists a function holomorphic in $\mathcal{D}'$ but
singular at $z$. In general, such functions might not exist since the set of
singularities of a holomorphic function cannot be arbitrary. However, it is
easy to construct a function singular on any complex codimension-1 hyperplane
in $\mathbb{C}^n$ (take the reciprocal of an affine-linear function). For a
convex $\mathcal{D}'$ one can always find such a hyperplane passing through a
given boundary point but staying away from the interior of $\mathcal{D}'$.}
and so $f$ cannot be extended any further by analytic completion alone. The
requirement that $X$ is open is a bit too restrictive and we'll need also a
degenerate case of this theorem, as described below.

These three tricks are applied one by one infinitely many times in order to
construct the full analytic continuation of $S_n .$ Instead of setting up the
procedure in its full glory, we will only follow the first steps to see how it
works in principle. The full analysis is performed in
{\cite{osterwalder1975}}.

First, it helps to introduce new variables $w_i$ by
\begin{equation}
  e^{w_i} = y_i .
\end{equation}
Our domains of analyticity in terms of $w_i$ will always be tubes of the form
$(w_1 \ldots w_n) \in \mathcal{D} (X) \equiv \mathbb{R}^n + i X$ for various
$X \subset \mathbb{R}^n$, and so we'll often just describe $X$. For example,
we start with $S_n$ and $\Psi_n$ defined for real positive $y_i$, which
corresponds to the domain $\mathcal{D} (\{ 0 \})$ in $w_i$.

Consider the 2-point function $S_1 (y_1)$. We start with the domain
$\mathcal{D} (\{ 0 \}) =\mathbb{R}$ in $w_1$, corresponding to real positive
$y_1 .$ Next, we apply the first trick. Specifically, we write
\begin{equation}
  \langle \nobracket \Psi_1 (t') | U_{\tau} | \nobracket \Psi_1 (t)
  \rangle = S_1 (t' + t + \tau),
\end{equation}
and since we are free to choose $t > 0$ and $t' > 0$ arbitrarily small, while
$U_{\tau}$, as discussed above, is a well-defined bounded operator for
$\tmop{Re} \tau > 0$, we obtain an analytic continuation of $S_1 (y_1)$ to the
right half-plane.

We are now done with the analytic continuation of $S_1$, since our goal was to
continue all $y_k$ to the right-half plane. In terms of $w_i$, this
corresponds to the domain $w_1 \in \mathcal{D} \left( \left( - \frac{\pi}{2},
+ \frac{\pi}{2} \right) \right)$, i.e.\ a strip. For higher-point functions, in
terms of $w_i$, we should stop when our domain of analyticity is $\mathcal{D}
\left( \left( - \frac{\pi}{2}, + \frac{\pi}{2} \right) \times \cdots \times
\left( - \frac{\pi}{2}, + \frac{\pi}{2} \right) \right)$.

Consider now the 3-point function $S_2 (y_1, y_2)$. We can again use the
first trick and define it on $(w_1, w_2) \in \mathcal{D} (X_2)$, where $X_2 =
\{ 0 \} \times \left( - \frac{\pi}{2}, + \frac{\pi}{2} \right) \cup \left( -
\frac{\pi}{2}, + \frac{\pi}{2} \right) \times \{ 0 \}$ (see Fig.\ \ref{extend1}, left). In more detail, we write the following two equations for
$S_2 (y_1, y_2)$, representing it as an inner product in two ways, and
inserting a $U_{\tau},$
\begin{equation}
  \langle \Psi_2 (t' ; y'_1) | U_{\tau} | \Psi_1 (t) \rangle = S_2
  (y_1', t' + t + \tau), \label{s2trick11}
\end{equation}
\begin{equation}
  \langle \Psi_1 (t') | U_{\tau} | \Psi_2 (t, y_1) \rangle = S_2 (t' + t
  + \tau, y_1), \label{s2trick12}
\end{equation}
where the left-hand sides are well-defined (at this point) for $t, t', y_1,
y_1'>0$ and $\tmop{Re} \tau > 0$. We see that the first equation defines $S_2
(y_1, y_2)$ for real $y_1 > 0$ and $\tmop{Re} y_2 > 0$ as a holomorphic function
of $y_2$. The second equation does the same, but with $y_1$ and $y_2$
exchanged. In terms of $(w_1, w_2)$ this corresponds to the ``analyticity
domain'' $\mathcal{D} (X_2)$ described above. We write ``analyticity domain''
in quotes because $\mathcal{D} (X_2)$ is not open (and has empty interior),
and thus is not a domain. Correspondingly, we cannot say that $S_2$ is an
holomorphic function of two variables on $\mathcal{D} (X_2)$. We will deal with
this problem momentarily.

\begin{figure}[h]\centering
  \raisebox{-0.5\height}{\includegraphics[width=4.27990948445494cm,height=4.16486291486291cm]{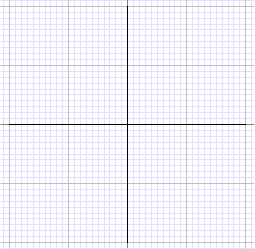}}
  $\Longrightarrow$
  \raisebox{-0.5\height}{\includegraphics[width=4.27990948445494cm,height=4.16486291486291cm]{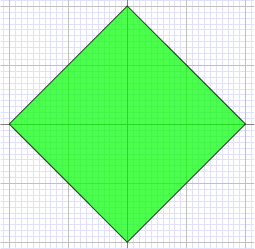}}
  \caption{\label{extend1}Left: set $X_2$. Right: domain $X_2'$ which defines
  the envelope of holomorphy $\mathcal{D} (X_2')$ of $\mathcal{D} (X_2)$.}
\end{figure}

To proceed with the analytic continuation of $S_2 (y_1, y_2)$, we want to use
the third trick, the tube theorem, to extend the analyticity domain from
$\mathcal{D} (X_2)$ to $\mathcal{D} (X_2')$, with $X_2' \equiv \tmop{ch}
(X_2)$ (Fig.\ \ref{extend1}, right).

The problem with this is that $X_2$ is not open, as mentioned above, so the
tube theorem does not apply. Instead, for this step one has to use
Malgrange-Zerner theorem {\cite{Epstein:1966yea}}, which allows $X_2$ to be a
union of intervals, with $S_2 (y_1, y_2)$ separately holomorphic in one variable
on each of these intervals, as is the case in our setup. The conclusion is
still that $S_2 (y_1, y_2)$ can be analytically continued to $\mathcal{D}
(X_2')$.

Note that the domain $\mathcal{D} (X_2')$ is not yet the full analyticity
domain $\mathcal{D} \left( \left( - \frac{\pi}{2}, + \frac{\pi}{2} \right)
\times \left( - \frac{\pi}{2}, + \frac{\pi}{2} \right) \right)$ that we are
aiming for. In particular, $X_2'$ is a proper subset of the square $\left( -
\frac{\pi}{2}, + \frac{\pi}{2} \right) \times \left( - \frac{\pi}{2}, +
\frac{\pi}{2} \right)$, see the right panel of Fig.\ \ref{extend1}.
Importantly, it doesn't approach the corners $\left( \pm \frac{\pi}{2}, \pm
\frac{\pi}{2} \right)$, which correspond to pure imaginary $y_1, y_2$. Pure
imaginary $y_1, y_2$ is, in turn, where we want to recover the Wightman
distributions.

To extend the domain of analyticity of $S_2 (y_1, y_2)$ even further, we need
to first extend the domain of $\Psi_2 (t, y_1)$, which can be done by the
second trick above --- via the equality
\begin{equation}
  \langle \Psi_2 (t, y_1) | \Psi_2 (t, y_1) \rangle = S_3 (\overline{y_1}, 2
  t, y_1) .
\end{equation}
Note that we are not interested in the analytic continuation in $t$ here ---
it is automatic when we act on $\Psi_2$ with $e^{- H t}$ --- so we can assume
$t$ is real. For $S_3$ we can run the same argument as we just did for $S_2$
and conclude that it is holomorphic in $\mathcal{D} (X_3')$, where $X_3'$ is the
convex hull of three intersecting intervals on coordinate axes (an
octahedron). As discussed above, we expect that $\Psi_2 (t, y_1)$ is defined
whenever $t$ and $y_1$ are such that the arguments of $S_3$ above are in its
analyticity domain. This happens whenever
\[ (\overline{w_1}, \log 2 t, w_1) \in \mathcal{D} (X_3'), \]
which is equivalent to
\begin{equation}
  (\tmop{Im} \overline{w_1}, \tmop{Im} \log 2 t, \tmop{Im} w_1) \in X_3' .
\end{equation}
Since we take $t$ to be real and positive, we have $\tmop{Im}
\log 2 t = 0$ and so $t$ is otherwise unconstrained. By construction of $X_3'$
and $X_2'$, $(\tmop{Im} \overline{w_1}, 0, \tmop{Im} w_1) \in X_3'$ is
equivalent to $(\tmop{Im} \overline{w_1}, \tmop{Im} w_1) \in X_2'$. Using
$\tmop{Im} \overline{w_1} = - \tmop{Im} w_1$, we conclude that $w_1$ is
constrained by
\begin{equation}
  (- \tmop{Im} w_1, \tmop{Im} w_1) \in X_2' .
\end{equation}
This is equivalent to $| \tmop{Im} w_1 | < \frac{\pi}{4}$, which is the same
as $w_1 \in \mathcal{D} \left( \left( - \frac{\pi}{4}, + \frac{\pi}{4} \right)
\right) .$ To conclude, we expect $\Psi_2 (t, y_1)$ to be defined and holomorphic
in $y_1$ for $t > 0$ and $w_1 \in \mathcal{D} \left( \left( - \frac{\pi}{4}, +
\frac{\pi}{4} \right) \right)$.

We can now apply the first trick to $S_2 (y_1, y_2)$ again, writing it as
inner product of $\Psi_1$ and $\Psi_2$ in the two ways {\eqref{s2trick11}} and
{\eqref{s2trick12}}. However, this time we can use $\Psi_2 (t, y_1)$ in a
wider domain of $y_1$, as computed above, equivalent to $w_1 \in \mathcal{D}
\left( \left( - \frac{\pi}{4}, + \frac{\pi}{4} \right) \right)$. From
{\eqref{s2trick11}} we conclude that $S_2 (y_1, y_2)$ is analytic for
\begin{equation}
  (w_1, w_2) \in \mathcal{D} \left( \left( - \frac{\pi}{4}, + \frac{\pi}{4}
  \right) \times \left( - \frac{\pi}{2}, + \frac{\pi}{2} \right) \right),
\end{equation}
where the domain of analyticity in $w_1$ comes from that of $\Psi_2 (t, y_1)$,
and in $w_2$ from $e^{- H \tau}$. Similarly, {\eqref{s2trick12}} now implies
analyticity in the domain
\begin{equation}
  (w_1, w_2) \in \mathcal{D} \left( \left( - \frac{\pi}{2}, + \frac{\pi}{2}
  \right) \times \left( - \frac{\pi}{4}, + \frac{\pi}{4} \right) \right) .
\end{equation}
Combining the two together, we find that $S_2 (y_1, y_2)$ is analytic for
$(w_1, w_2) \in \mathcal{D} (X_2'')$, where
\begin{equation}
  X_2'' \equiv \left( - \frac{\pi}{4}, + \frac{\pi}{4} \right) \times \left( -
  \frac{\pi}{2}, + \frac{\pi}{2} \right) \cup \left( - \frac{\pi}{2}, +
  \frac{\pi}{2} \right) \times \left( - \frac{\pi}{4}, + \frac{\pi}{4}
  \right),
\end{equation}
see the left panel of Fig.\ \ref{extend2}.

\begin{figure}[h]\centering
	\centering
  \raisebox{-0.5\height}{\includegraphics[width=4.27990948445494cm,height=4.16486291486291cm]{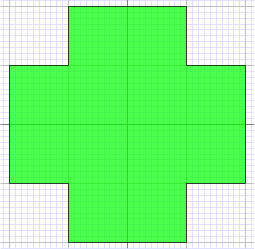}}
  $\Longrightarrow$
  \raisebox{-0.5\height}{\includegraphics[width=4.27990948445494cm,height=4.16486291486291cm]{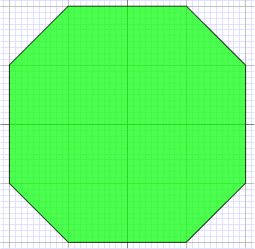}}
  \caption{Left: set $X_2''$. Right: domain $X_2'''$ which defines the
  envelope of holomorphy $\mathcal{D} (X_2''')$ of $\mathcal{D}
  (X_2'')$.\label{extend2}}
\end{figure}

Using the tube theorem, we can now extend the analyticity domain from
$\mathcal{D} (X_2'')$ further to $\mathcal{D} (X_2''')$, where $X_2''' \equiv
\tmop{ch} (X_2'')$ is the convex hull of $X_2''$ shown in the right panel of
Fig.\ \ref{extend2}.

We see that in order to analytically continue the 3-point function $S_2$,
it was useful to split it into an inner product of one-operator and
two-operator states $\Psi_1$ and $\Psi_2$, and use the information about the
latter that is provided by its norm, the 4-point function $S_3 .$ Still, we
have not yet managed to analytically continue $S_2$ to the entire region of
interest (we still have the corners missing in the right panel of Fig.\ \ref{extend2}). The only way to fix this is to extend the region of
analyticity of $S_3 .$ For that, we have to split it into a product of two
states, and extend the region of analyticity of these states. It is useless to
split it as a product of two $\Psi_2$ states, since their norm is computed by
$S_3$ itself and we won't learn anything new in this way. Instead, we have to
split it as a product of $\Psi_1$ and $\Psi_3$. This will lead us to consider
the norm of $\Psi_3$, which is computed by the six-point function $S_5$.
Following this logic, eventually, we will be forced to consider $S_n$ with
arbitrarily high $n$ just in order to construct the analytic continuation of
$S_2$. Fortunately, it can be shown that this procedure converges to the
desired domain $\mathcal{D} \left( \left( - \frac{\pi}{2}, + \frac{\pi}{2}
\right) \times \cdots \times \left( - \frac{\pi}{2}, + \frac{\pi}{2} \right)
\right)$ for all $S_n$, see {\cite{osterwalder1975}} for details.

To finish the discussion of the analytic continuation of $S_n$, let us justify
the second trick, which constructs the states $\Psi_n$ based on analyticity of
their norm $S_{2 n - 1} .$ Let $C$ be the domain of analyticity of $S_{2 n -
1} (y_1 \ldots y_{2 n - 1})$ known to us, expressed in terms of $w_i$, and let
$D$ be the domain of the arguments $t, w_1 \ldots w_{n - 1}$ of $\Psi_n (t ;
y_1 \ldots y_{n - 1})$ for which the arguments of $S_{2 n - 1}$ in the
right-hand side of
\begin{equation}
  \langle \Psi_n (t ; y_1 \ldots y_{n - 1}) | \nobracket \Psi_n (t ; y_1
  \ldots y_{n - 1}) \rangle = S_{2 n - 1} (\overline{y_{n - 1}}, \ldots,
  \overline{y_1}, 2 t, y_1, \ldots, y_{n - 1}),
\end{equation}
belong to $C$. As is clear from the above discussion, $C$ (expressed in terms
of $w_i$) is always of the form $C =\mathcal{D} (X)$ for some $X$. We
similarly have $D =\mathcal{D} (Y)$ for some $Y$. By the tube theorem (or
Malgrange-Zerner theorem), we can assume that $X$ (and thus also $Y$)
is open, non-empty, and convex. Furthermore, it is easy to convince oneself
that $X$, and thus $Y$, is invariant under reflections along any of the
coordinate real axes (i.e.\ sending $w_i$ to $\overline{w_i}$ for some $i$).

Suppose now that we have a point $\tmmathbf{} (t ; w_1^0 \ldots w_{n - 1}^0)
\in D$. Then by definition of $D$ we have
\begin{equation}
  p \equiv \tmmathbf{} (\overline{w_{n - 1}^0} \ldots \overline{w_1^0}, \log 2
  t, w_1^0 \ldots w_{n - 1}^0) \in C.
\end{equation}
The above properties imply that there are $r_i > 0$ such that the polydisk
\begin{equation}
  P = \{ (w_1 \ldots w_{2 n - 1})  |  | w_i | < r_i \} \nobracket + \tmop{Re}
  p
\end{equation}
is contained in $C$, $P \subset C$, and moreover $p \in P.$ Indeed, since $C
=\mathcal{D} (X),$ this will be true if $\tmop{Im} P \subset X$ and $\tmop{Im}
p \in \tmop{Im} P$.\footnote{The latter is because $\tmop{Im}$ of the section
of $P$ by $\tmop{Re} x = \tmop{Re} p$ is $\tmop{Im} P$.} By construction,
$\tmop{Im} P$ is a box with sides $2 r_i$ centered at 0. On the other hand,
the properties of $X$ imply that together with any point $x$, $X$ contains
such a box with $x$ being one of its vertices. We can then find an
$\varepsilon > 0$ such that $(1 + \varepsilon) \tmop{Im} p \in X$, and take
$\tmop{Im} P$ to be the box defined by the vertex $x = (1 + \varepsilon)
\tmop{Im} p$. See Fig.\ \ref{CPfig} for an intuitive picture.

\begin{figure}[h]\centering
  \raisebox{-0.5\height}{\includegraphics[width=8.94149285058376cm,height=5.35478158205431cm]{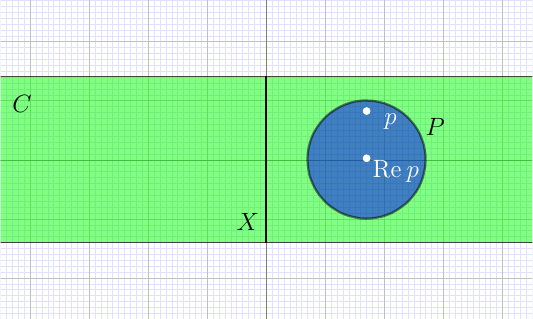}}
  \caption{Schematic picture of the tube $C$ and polydisk $P$.\label{CPfig}}
\end{figure}

Writing temporarily the state $\Psi_n$ as a function of $w_k$ instead of
$y_k$, we define it at $w_k$ by the Taylor series
\begin{equation}
  | \Psi_n (t ; w_1 \ldots w_{n - 1}) \rangle \equiv \sum_{\alpha} \frac{(w -
  \tmop{Re} w^0)^{\alpha}}{\alpha !} \partial^{\alpha} | \Psi_n (t ; \tmop{Re}
  w^0_1 \ldots \tmop{Re} w^0_{n - 1}) \rangle \label{stateext}
\end{equation}
($\alpha$ is a multiindex so $w^{\alpha} = w_1^{\alpha_1} w_2^{\alpha_2}
\ldots$ etc.). Note that the state in the right-hand side is well defined since
the corresponding $y_k = e^{\tmop{Re} w^0_k} > 0$. To check whether this
Taylor series converges, we look at its remainder
\begin{equation}
  \sum_{| \alpha | > N} \frac{(w - \tmop{Re} w^0)^{\alpha}}{\alpha !}
  \partial^{\alpha} | \Psi_n (t ; \tmop{Re} w^0_1 \ldots \tmop{Re} w^0_{n -
  1}) \rangle,
\end{equation}
whose norm squared is
\begin{equation}
  \sum_{| \alpha | > N} \sum_{| \beta | > N} \frac{(w - \tmop{Re}
  w^0)^{\alpha}}{\alpha !} \frac{(\bar{w} - \tmop{Re} w^0)^{\beta}}{\beta !}
  \partial^{\alpha} \partial^{\beta} S_{2 n - 1} (\tmop{Re} w^0_{n - 1} \ldots
  \tmop{Re} w^0_1, \log 2 t, \tmop{Re} w^0_1 \ldots \tmop{Re} w^0_{n - 1}),
\end{equation}
where $\beta$-derivatives act on the first $n - 1$ arguments of $S_{2 n - 1}$,
while $\alpha$-derivatives act on the last $n - 1$ arguments. Here we also
temporarily write $S_{2 n - 1}$ as function of $w_k$. This norm is clearly
just the tail of the Taylor series of $S_{2 n - 1}$ expanded around the point
$\tmop{Re} p$, and evaluated at $(\overline{w_{n - 1}}, \ldots,
\overline{w_1}, \log 2 t, w_1, \ldots, w_{n - 1})$. (We are not expanding in
$t$.) Since $S_{2 n - 1}$ is holomorphic in the polydisk $P$ centered at
$\tmop{Re} p$, this Taylor series converges in $P$ and thus this remainder
tends to $0$ there.

Since $p = \tmmathbf{} (\overline{w_{n - 1}^0} \ldots \overline{w_1^0}, \log
2 t, w_1^0 \ldots w_{n - 1}^0) \in P$, the remainder tends to 0 at $p$, and
thus {\eqref{stateext}} converges at $\tmmathbf{} (t ; w_1^0 \ldots w_{n -
1}^0)$. Furthermore, since $P$ is open, it follows that {\eqref{stateext}}
converges in some neighborhood of $(t ; w_1^0 \ldots w_{n - 1}^0)$, defining
$| \Psi_n (t ; w_1 \ldots w_{n - 1}) \rangle$ as a holomorphic
$\mathcal{H}^{\tmop{OS}}$-valued function in that neighborhood. Since the
choice of $\tmmathbf{} (t ; w_1^0 \ldots w_{n - 1}^0) \in D$ was arbitrary, we
have defined $| \Psi_n (t ; w_1 \ldots w_{n - 1}) \rangle$ as a holomorphic
function of $w_i$ for all points in $D$.

\subsubsection{Temperedness bound}

Now that the correlation functions $S_n (y_1 \ldots .y_n)$ have been
analytically continued from $y_k > 0$ to $\tmop{Re} y_k > 0$, we only need to
establish a bound on their growth as $\tmop{Re} y_i \rightarrow 0$ in order to
construct tempered Wightman distributions by an application of Vladimirov's
theorem. The logic proceeds by establishing a bound on $S_n (y_1 \ldots y_n)$
for real $y_k$, and then repeating the analytic continuation described above,
while keeping track of this bound. We will only sketch this rather technical
argument in very general terms.

The final temperedness bound that we want to establish is
\begin{equation}
  | S_n (y_1 \ldots y_n) | \leqslant c_n \left( \left( 1 + \sum_k | y_k |
  \right) \left( 1 + \sum_k (\tmop{Re} y_k)^{- 1} \right) \right)^{p_n},
  \label{tempbound}
\end{equation}
for some sequences $c_n$ and $p_n$.\footnote{Here, for simplicity, we again
ignore spatial arguments of the correlation functions, although they need to
be taken care of at this step in order to establish ``temperedness in spatial
directions.'' Furthermore, note that Osterwalder and Schrader establish
additional bounds on $c_k$, etc., which are not important for the application
of Vladimirov's theorem.} We would like {\eqref{tempbound}} to hold for all
$y_k,$ $\tmop{Re} y_k > 0.$ For real positive $y_k$ (i.e.\ in the Euclidean)
this holds as a consequence of {\eqref{OSmod}}. As discussed in Remark
\ref{OSnewVSold}, the original OS axioms did not include {\eqref{OSmod}}, so
their first step was to derive {\eqref{tempbound}} for $y_k > 0$ using
{\eqref{J0distr}}.

In principle at fixed $n$, {\eqref{tempbound}} looks reasonable given
{\eqref{J0distr}}: both say, intuitively, that the correlation functions
cannot be too singular at coincident points or grow too fast at infinity.
However, {\eqref{tempbound}} imposes this in a much more direct way. It turns
out that in general one cannot derive direct bounds such as
{\eqref{tempbound}} from averaged statements such as {\eqref{J0distr}}, even
if we know that $S_n$ is real analytic.

Consider the real-analytic function $\sin (e^x)$, $x \in \mathbb{R}$. It is a
bounded function, hence a tempered distribution. Thus its first derivative
$e^x \cos (e^x)$ is also a tempered distribution. This is an example of a
real-analytic tempered distribution which is not polynomially bounded. So some
further assumptions are needed beyond real analyticity.\footnote{Incidentally,
our example shows that the Corollary of Lemma 1 in {\cite{Glaser1974}} is
wrong.} \

In our case, the functions $S_n (y_1 \ldots y_n)$ are real-analytic and
satisfy {\eqref{J0distr}}. In addition, they satisfy OS positivity. We already
used OS positivity to show real-analyticity, and we will now have to invoke it
again to prove {\eqref{tempbound}} for $y_k > 0$. The full argument is rather
technical; we will explain the main idea on the example of $S_1 (y)$. Since
we know that $S_1$ is holomorphic, in particular harmonic, by the mean value
theorem for harmonic functions we can write it as a radially symmetric average
\begin{eqnarray}
  S_1 (y) & = & \int d x\, d t\, S_1 (y + x + i t) k_{\rho}
  (x, t) \nonumber\\
  & = & \int_{| t |, | t' | < \rho} d t\, d t'\, T (t | \nobracket g_{\rho}
  (\cdot, t + t'), g_{\rho} (\cdot, t')),  \label{SfromT}\\
  T (t | \nobracket \varphi_1, \varphi_2) & \assign & \int d x\, d x'\, S_1 (y
  + x + i t) \varphi_1 (x + x') \varphi_2 (x'), \nonumber
\end{eqnarray}
where $k_{\rho}$ is a $C_0^{\infty}$ radial function supported in a ball of
radius $\rho$ and of integral 1, and we choose $\rho$ sufficiently small so
that all points under the integral sign are where $S_2$ is analytic. We also
chose
\begin{equation}
  k_{\rho} (x, t) = \int d x'\, d t'\, g_{\rho} (x + x', t + t') g_{\rho} (x',
  t'),
\end{equation}
a convolution of another radial $C_0^{\infty}$ function with itself (and hence
a radial function). The point of this construction is that, for generic
$\varphi_1, \varphi_2$, $T (0 | \nobracket \varphi_1, \varphi_2)$ is an inner
product $\langle \Psi_1 | \Psi_2 \rangle$ of two OS states:
\begin{equation}
  \langle \Psi_1 | \nobracket = \int d x\, \mathcal{O} (y / 2 + x) \varphi_1
  (x), \qquad | \nobracket \Psi_2 \rangle = \int d x\, \mathcal{O} (- y / 2 + x)
  \varphi_2 (x) .
\end{equation}
The norm of these states, and hence their inner product, can be bounded using
{\eqref{J0distr}}. Furthermore $T (t | \nobracket \varphi_1, \varphi_2) =
\langle \Psi_1 | e^{- i H t} | \nobracket \Psi_2 \rangle$ satisfies the same
bound. Using this bound for $\varphi_1 = g_{\rho} (\cdot, t + t'), \varphi_2 =
g_{\rho} (\cdot, t')$, Eq.\ {\eqref{SfromT}} gives a bound on $S_1 (y)$. The
same idea works for higher point functions. We first have to estimate the norm
of some states using {\eqref{J0distr}}.\footnote{Note that the linear growth
condition is not needed at this point: Eq.\ {\eqref{J0distr}} with some
$\sigma_n$ and $q_n$ suffices to establish {\eqref{tempbound}} with some $c_n$
and $p_n$. The linear growth condition gives in addition $c_n$ of factorial
growth and $p_n$ growing at most linearly. This turns out important later in
the proof, see below.} We then analytically continue separately in each time,
and then use Malgrange-Zerner theorem to extend the bound on $T$ to an open
set. A single use of Malgrange-Zerner theorem suffices here, like in Fig.\ \ref{extend1}. We refer the reader to Sec.\ VI.1 of {\cite{osterwalder1975}}
for full details.

Once {\eqref{tempbound}} is established for $y_k > 0$, one repeats the
analytic continuation procedure that we described above, keeping track of the
implications of {\eqref{tempbound}}. The analytic continuation used three
tricks: (1) analytically continuing $S_n$ by representing it in the form
{\eqref{trick1}} (as $e^{- H \tau}$ inserted between two states), (2)
expressing the norms of these states in terms of higher-point $S_n$ as in
{\eqref{trick2}}, and (3) analytic completion.

The bound {\eqref{tempbound}} propagates through the tricks (1) and (2) by the
use of Cauchy-Schwarz inequality, as well as by using the fact that the norm
of $e^{- H \tau}$ is bounded from above by 1 (i.e.\ Eq.
{\eqref{contractive}}).

To propagate the bound through trick (3), the following simple idea is used.
Suppose we have domains $\mathcal{D}' \supset \mathcal{D}$ such that any
holomorphic function $f$ on $\mathcal{D}$ can be extended to a holomorphic
function on $\mathcal{D}'$. Then we have the equality of images
\begin{equation}
  f (\mathcal{D}') = f (\mathcal{D}),
\end{equation}
and in particular
\begin{equation}
  \sup_{z \in \mathcal{D}'} | f (z) | = \sup_{z \in \mathcal{D}} | f (z) | .
\end{equation}
To see this, suppose $a \in \mathbb{C}$ is a value which $f$ assumes in
$\mathcal{D}'$ but not in $\mathcal{D}$. Then the function $(f (z) - a)^{- 1}$
is holomorphic in $\mathcal{D}$ but has a singularity in $\mathcal{D}'$, which
is a contradiction. This shows that if we have a bound on $f$ in
$\mathcal{D}$, it is also valid in $\mathcal{D}'$.

Finally, recall that in order to construct the analytic continuation of
$S_{n_0}$ for some fixed $n_0$, we had to use $S_n$ with arbitrarily high $n$
in the process. This means that in order to establish the bound
{\eqref{tempbound}} on $S_{n_0}$ for all $\tmop{Re} y_k$, we have to use
{\eqref{tempbound}} for $y_k > 0$ for $S_n$ with arbitrarily high $n$. These
bounds need to combine in a way that is strong enough to establish
{\eqref{tempbound}} for $S_{n_0}$. For this, it is important that $c_n$ is of
factorial growth and $p_n$ grows at most linearly. This requires the same of
the sequence $\sigma_n$ and the index of the seminorm in {\eqref{J0distr}},
explaining the need for the linear growth condition.

\section{Conclusions}\label{conclusions}

In this paper we studied the relationship between the modern Euclidean CFT
axioms (which we formulated in Sec.\ \ref{ECFTax}) and the more traditional
Osterwalder-Schrader and Wightman axioms. We showed that at least for $(n
\leqslant 4)$-point functions, both OS and Wightman axioms follow from the
Euclidean CFT axioms. Our Euclidean CFT axioms are quite modest. In
particular, beyond the minimal assumptions of regularity of correlators and
the standard constraints of unitarity, we assumed only a very weak form of the
convergent OPE.

Our derivation of Wightman axioms is of particular importance: it shows that
the conformal Wightman 4-point functions are well-defined tempered
distributions for arbitrary configurations of the 4 points, even when no OPE
channel is convergent in the sense of functions. We have furthermore shown
that these tempered distributions can always be computed by a conformal block
expansion which is convergent in the sense of distributions, generalizing
our previous results in cross-ratio space {\cite{paper1}}, and giving a
derivation of Mack's results {\cite{Mack:1976pa}} from Euclidean CFT axioms.

\begin{figure}[h]\centering
  \raisebox{-0.481498738926148\height}{\includegraphics[width=7.92247146792601cm,height=5.45475862521317cm]{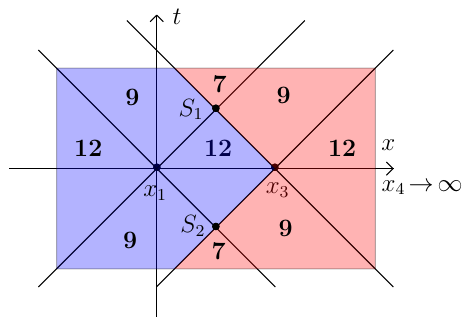}}
  \caption{\label{134}Minkowski configurations with $x_1 = 0, x_3 = \hat{e}_1,
  x_4 = \infty \hat{e}_1$ and $x_2 = t \hat{e}_0 + x \hat{e}_1$. Blue:
  configurations where $| \rho |, | \bar{\rho} | < 1$ and the 4-point
  functions is analytic. Red: configurations where $| \rho |$ and/or $|
  \bar{\rho} | =1$. Boldface numbers $X =\tmmathbf{7}, \tmmathbf{9},
  \tmmathbf{1}\tmmathbf{2}$ denote the causal type of the configuration
  according to {\cite{Qiao:2020bcs}} (excluding lightlike separations). $S_{1,
  2}$ are double light-cone singularities.}
\end{figure}

For example, consider the configuration in Fig.\ \ref{134}, where the
operators in a 4-point function are inserted at $x_1 = 0, x_3 = \hat{e}_1,
x_4 = \infty \hat{e}_1$, while $x_2 = t \hat{e}_0 + x \hat{e}_1$ is allowed to
move in a plane parametrized by $(t, x) .$ The cross-ratios for this
configuration are $z, \bar{z} = x \pm t$. It is then easy to see that for
$x_2$ in the blue region of Fig.\ \ref{134} $| \rho |, | \bar{\rho} | < 1$
and the s-channel OPE converges in the sense of functions. Our results imply
that the s-channel OPE also converges in the red region where $| \rho |$ and/or
$| \bar{\rho} | = 1$, but now the convergence is in the sense of
distributions. In particular, the 4-point function is at least a
distribution for all values of $x_2$. Of course, in some regions of Fig.\ \ref{134} this was obviously true -- for example, in the red part of the
regions $\tmmathbf{9}, \tmmathbf{1}\tmmathbf{2}$ (labeling according to the
classification in {\cite{Qiao:2020bcs}}), one can show that the 4-point
function is real-analytic using the convergent t-channel OPE. One may hope
to establish real-analyticity also in the region \tmtextbf{7} using
u-channel OPE. {This would indeed be the case for the ordering
$\langle \mathcal{O} (x_2) \mathcal{O} (x_1) \mathcal{O} (x_3) \mathcal{O}
(x_4) \rangle$. However, for the ordering \ $\langle \mathcal{O} (x_1)
\mathcal{O} (x_2) \mathcal{O} (x_3) \mathcal{O} (x_4) \rangle$ that we are
discussing here}, it turns out that no OPE channel converges in region
\tmtextbf{7} in the sense of functions.\footnote{For a reader comfortable with
cuts in $z, \bar{z}$ plane the intuitive argument is simple: we have $z < 0$
(on what we'll call s-channel cut), $\bar{z} > 1$ (on t-channel cut).
Furthermore according to the operator ordering, when $z$ crosses $0$ we need
to make $\tmop{Im} t$ slightly negative (and thus $\tmop{Im} z$ slightly
negative), because $x_2$ at this point crosses the null cone of $x_1$, and
when $\bar{z}$ crosses 1 we need to make $\tmop{Im} t$ slightly positive (and
thus $\tmop{Im} \bar{z}$ slightly negative) because it corresponds to $x_2$
crossing the null cone of $x_3$. Thus both $z$ and $\bar{z}$ end on lower
sides of their respective cuts, and so one of them must have crossed the
$u$-channel cut at $(0, 1)$ when analytically continuing from a Euclidean
configuration. We conclude that s- and t- channel OPEs are only
distributionally convergent, while $u$-channel is badly divergent.} Therefore,
before our work it was not at all clear whether this correlator makes any
sense in region \tmtextbf{7} if we assume only the Euclidean CFT axioms.

While we have shown that the correlator is \tmtextit{at least }distributional
in region \tmtextbf{7}, we have not excluded the possibility of it being
real-analytic there. For example, in 2 dimensions Virasoro symmetry implies
that the 4-point function is analytic everywhere away from light-cone
singularities {\cite{Maldacena:2015iua}}. This is perhaps too much to expect
in higher dimensions, but one can still ask whether analyticity can be
established in a larger domain. One approach is to ask for the envelope of
holomorphy of the known domain of analyticity. {Since the 4-point function is
essentially only a function of two cross-ratios, this might be a tractable
question \cite{PetrEnvelope}. We leave working out the full consequences of this idea for future work.\footnote{Another approach could be via alternative representations of the 4-point function having an extended region of analyticity, e.g.~\cite{Caron-Huot:2020nem}.}}


{In an upcoming paper {\cite{paper2a}}, we will generalize our
results to external operators with spins.} In addition, there are many other
fundamental open questions which we believe are important to understand. First
of all, this paper is concerned with properties of CFT Wightman functions in
Minkowski space. However, it is expected that Lorentzian CFTs should be
naturally defined on Minkowski cylinder {\cite{Luscher:1974ez}}, which is the
smallest physically-sensible space on which finite conformal transformations
can act. Yet, it is not known whether CFT Wightman functions can be defined as
tempered distributions on Minkowski cylinder (see note \ref{LMcomplain}).
Answering this question in the positive for CFT $(n \leqslant 4)$-point
functions is the main goal of our forthcoming paper {\cite{paper3}}.

An important problem is to extend our results to $(n > 4)$-point functions. As
we discuss in App.\ \ref{OShigher}, even deriving the OS
axioms might require some strengthening of the OPE axiom. Another interesting
possibility is to formulate Euclidean CFT axioms as OS axioms supplemented
with a very weak form of the OPE (for example, asymptotic OPE in Euclidean
space). This is perhaps less attractive, since it is desirable to formulate
CFT axioms directly in terms of the CFT data (scaling dimensions and OPE
coefficients). However, it will still be interesting to establish an
equivalence between OS+(weak OPE axiom) and (possibly a stronger version of)
our Euclidean CFT axioms, perhaps using arguments similar to those of
{\cite{Mack:1976pa}}. Once OS axioms are established, it is likely that a
strategy similar to that of the present paper can be pursued to establish
Wightman axioms, using a comb-like OPE channel.

In this paper we only considered Wightman functions, but in
practice one often needs time-ordered Minkowski correlators. Textbook definition
of time-ordered correlators involves multiplying Wightman functions by
$\theta$-functions implementing time ordering. Since Wightman functions are in
general distributions, this definition does not make rigorous sense at
coincident points. As a matter of fact, time-ordered correlators have not
been rigorously defined just from Wightman axioms alone (see e.g.\ {\cite{bogolubov2012general}}, p.505) in a general QFT. In a general QFT
setting, it is known that defining time-ordered Minkowski correlators is
closely related to defining Euclidean correlators at coincident points
{\cite{Eckmann:1979vq}}. In the future, it would be interesting to define
time-ordered CFT Minkowski correlators as distributions just from Euclidean
CFT axioms.\footnote{Time-ordered Minkowski correlators are expected to be to some extent ambiguous at coincident points, just as the Euclidean correlators. For external CFT operators of non-integer external dimension one may choose to fix this ambiguity imposing scale invariance. Ambiguity will remain for external fields of integer dimensions, like the stress tensor.}

A more ambitious goal is to understand the relationship of CFT axioms to
Haag-Kastler axioms. This appears to be considerably harder since these axioms
deal with operator algebras rather than local correlation functions, and some
qualitatively new ideas seem to be required.

\section*{Acknowledgements}

Some of our results were first presented in a talk at the Simons Foundation
{\cite{talkSimons}}, in lectures at the IPHT Saclay {\cite{lecturesSaclay}},
and in an online seminar {\cite{seminarAntti}}. SR thanks Riccardo Guida and
Antti Kupiainen for organizing the last two events.

SR thanks Gian Michele Graf for mentioning Zinoviev's work {\cite{Zinoviev1995}} and to David Brydges for the explanation of this work. We thank Marc Gillioz and Tom Hartman for communications concerning their work. We are grateful to Malek Abdesselam, Simon Caron-Huot, Tom Hartman and David Simmons-Duffin for comments on the draft.

PK is supported by DOE grant DE-SC0009988 and by the Adler Family Fund at the
Institute for Advanced Study. The work of SR and JQ is supported by the Simons
Foundation grant 488655 and 733758 (Simons Collaboration on the
Nonperturbative Bootstrap). SR is supported by Mitsubishi Heavy Industries as
an ENS-MHI Chair holder.

\appendix\section{Lorentzian CFT literature}\label{literature}

Recent years have seen an explosion of the uses of Lorentzian CFT, motivated
in particular by the conformal bootstrap applications. In this appendix we
will mention some of these works, and comment on their underlying assumptions.
{See also \cite{DSDLorentzian} for a modern pedagogical introduction to Lorentzian CFT.}

\tmtextbf{Conformal collider bounds.} One of the first ``modern'' Lorentzian
CFT results was obtained in {\cite{Hofman:2008ar}}. This work considered a
thought experiment, creating a CFT state via a (smeared) local operator and
measuring energy coming out at null infinity in a particular direction,
integrated over time. On physical grounds, one expects $\<\Psi| \int d x^-\, T_{- -} | \Psi \> \geqslant 0$ for any
state (``averaged null energy condition'' - ANEC). One interesting case is of
3-point functions $\langle \mathcal{O}^{\dagger} T_{\mu \nu} \mathcal{O} \rangle$
where $\mathcal{O}$ has nontrivial spin, when there are several independent
OPE coefficients multiplying different tensor structures allowed by conformal
symmetry. In this case ANEC implies that certain linear combinations of these
OPE coefficients must be non-negative (``conformal collider bounds''). Interference effects can be used to strengthen conformal collider bounds to provide
explicit lower bounds~\cite{Cordova:2017zej}, while combining conformal collider bounds
 with stress-tensor Ward identities leads to 
constraints on operator dimensions which are sometimes stronger than standard unitarity bounds~\cite{Cordova:2017dhq}. See below for work aiming to justify ANEC, or to derive conformal conformal bounds
directly without using ANEC.

\tmtextbf{Light-cone bootstrap.} Refs.\ {\cite{Fitzpatrick:2012yx,Komargodski:2012ek}} were the first to notice that
some bootstrap constraints become more visible in the Lorentzian signature.
These references pioneered the ``analytic light-cone bootstrap'' which studies
conformal four point functions in the regime of $0 < z, \bar{z} < 1$ real,
i.e.\ in the kinematics of Fig.\ \ref{134} when the point $x_2$ is spacelike
with respect to $x_1, x_3$. By studying the light cone limit $z \rightarrow 0$
at fixed $\bar{z}$ of one OPE channel and requiring that it should be
reproduced by the crossed channel, they argued that, in any CFT for $d > 2$,
the OPE should contain a series of operators of arbitrarily large spin and
twist asymptoting to a particular value. The original argument has some
caveats (see the discussion in {\cite{Qiao:2017xif}}, App.\ F) and a
mathematically rigorous proof is lacking. It would be nice to provide such a
proof, given the extreme importance of the light-cone bootstrap in the modern
bootstrap literature. There is little doubt that the light-cone bootstrap
results are correct. Numerical bootstrap studies of the critical 3d Ising and
the $O (2)$ models {\cite{Simmons-Duffin:2016wlq,Liu:2020tpf}} have found the
series of operators predicted by the light-cone bootstrap {see also \cite{Caron-Huot:2020ouj})}.
Ref.\ {\cite{Hofman:2016awc}} used the light-cone bootstrap to derive the conformal
collider bounds of {\cite{Hofman:2008ar}} without using ANEC.

\tmtextbf{Causality constraints.} Refs.\ {\cite{Hartman:2015lfa,Hartman:2016dxc,Hartman:2016lgu}} pioneered the study
of causality constraints for CFT 4-point functions. In particular Ref.\ {\cite{Hartman:2015lfa}} pointed out that the $z, \bar{z}$ and $\rho,
\bar{\rho}$ expansions are sufficient to construct Lorentzian 4-point functions
for many kinematic configurations and show local commutativity (i.e.\ that
spacelike-separated operators commute). See also note \ref{noteShock}. These
techniques led to a proof of ANEC {\cite{Hartman:2016lgu}}. As mentioned in
footnote \ref{caveats}, some steps in these papers are not completely
rigorous. See App.\ \ref{Tom} below for a more detailed review of
{\cite{Hartman:2015lfa}}.

\tmtextbf{Bulk point singularity.} Ref.\ {\cite{Maldacena:2015iua}} studied the
CFT 4-point function on the Lorentzian cylinder focusing on ``bulk-point''
configurations which correspond to scattering events in AdS/CFT
{\cite{Polchinski:1999yd,Gary:2009ae,Heemskerk:2009pn,Penedones:2010ue,Okuda:2010ym}}.
Using a local AdS dual description, one may suspect that the 4-point
function should be singular at such configurations. However, on the boundary
CFT side, one does not see this singularity in perturbation theory in $d = 2$
and $d = 3$ dimensions {\cite{Maldacena:2015iua}}. In $d = 2$, Ref.\ {\cite{Maldacena:2015iua}} showed non-perturbatively that the CFT 4-point
function is analytic everywhere away from light cones (in particular regular at
bulk-point configurations). This assumes Virasoro symmetry and unitarity and
uses Zamolodchikov's $q$-variables {\cite{zamolodchikov1987conformal}}. What
happens non-perturbatively in $d \geqslant 3$ (or in $d = 2$ in the absence of
the local stress tensor) is still an open problem. Note that at bulk-point
configurations, the $\rho$-expansion of the CFT 4-point function does not
absolutely converge in s-channel (as $| \rho | = | \bar{\rho} | = 1$ there)
and diverges in t-,u-channels {\cite{Qiao:2020bcs}}. In this paper we only
considered the CFT 4-point functions in flat space, but by the same
strategy we will show in {\cite{paper3}} that the Wightman axioms also hold
for CFT 4-point functions on Lorentzian cylinder. In particular, this will
show that the CFT 4-point functions are well defined at bulk-point
configurations in the sense of tempered distributions (but it will not settle
the question of their analyticity there).

\tmtextbf{Lorentzian inversion formula.} Ref.\ {\cite{Caron-Huot:2017vep}}
introduced an analogue of Froissart-Gribov formula in the context of conformal
field theory, which is now known as the Lorentzian inversion formula (LIF).
This formula computes the OPE data of a scalar 4-point function in terms of
a Lorentzian integral of this 4-point function. The OPE data is extracted
in the form of a function $C (\Delta, \ell)$. For integer $\ell$, the function
$C (\Delta, \ell)$ encodes the scaling dimensions of exchanged primary
operators of spin $\ell$ in the positions of poles in $\Delta$, and the
corresponding OPE coefficients are encoded in residues. LIF has many
interesting properties, such as analyticity in $\ell$, and suppression of
double-twist operators when a cross-channel conformal block expansion is used
under the integral. The original derivation of {\cite{Caron-Huot:2017vep}} was
done in cross-ratio space. The formula was re-derived in position space in
{\cite{Simmons-Duffin:2017nub}}. The derivation was further simplified and
generalized in {\cite{Kravchuk:2018htv}}.

Among other applications, LIF has been used to systematize and extend many of
the results of light-cone bootstrap (see, e.g.,
{\cite{Liu:2018jhs,Cardona:2018dov,Albayrak:2019gnz,Cardona:2018qrt,Liu:2020tpf,Iliesiu:2018zlz,Iliesiu:2018fao,Albayrak:2020rxh}}).
Similarly to light-cone bootstrap, this application is not completely rigorous
simply due to the fact that LIF expresses $C (\Delta, \ell)$ in terms of an
integral, and the local operators correspond to singularities of $C (\Delta,
\ell)$. In other words, the integral has no chance of converging near the
values of $\Delta, \ell$ relevant to local operators, except perhaps for
leading-twist operators {(see~\cite{Caron-Huot:2020ouj} for steps in this direction)}. This necessarily makes any conclusions about
anomalous dimensions of local operators reliant on additional assumptions.
These are easy to justify in some perturbative expansions, but in
non-perturbative setting do not appear to have been solidly understood.

\tmtextbf{Light-ray operators.} Ref.\ {\cite{Kravchuk:2018htv}} generalized LIF
to external operators with spin and uncovered an interesting relation to
Knapp-Stein intertwining operators, especially to what they called the light
transform. They interpreted the analyticity of LIF in $\ell$ in terms of
families of non-local non-integer-spin operators, the light-ray operators. These
operators are defined for generic complex $\ell$ and reduce to
light-transforms (null integrals) of local operators for integer spins. More
recently, light-ray operators have been used to understand an OPE for
event-shape observables such as energy-energy correlators in CFT
{\cite{Kologlu:2019mfz,Kologlu:2019bco,Chang:2020qpj}} (see also
{\cite{Dixon:2019uzg,Korchemsky:2019nzm}}). The light-ray operators correspond
to poles in $\Delta$ of $C (\Delta, \ell),$ and the issues with convergence of
LIF described above prevent a simple rigorous proof of their non-perturbative
existence. (E.g., for generic $\ell$, $C (\Delta, \ell)$ could have cuts or a
natural boundary of analyticity in $\Delta$.) It would be interesting to find
such a proof. In addition to clarifying the nature of light-ray operators, it
would probably also have a bearing on the light-cone bootstrap results
discussed above.

\tmtextbf{Conformal Regge theory} provides a way to understand Minkowski
correlators in Regge limit, and was developed in Refs.\ {\cite{Brower:2006ea,Cornalba:2007fs,Cornalba:2008qf,Costa:2012cb}}. Regge
limit in CFT is a limit of a 4-point function in Lorentzian signature in
which $\mathcal{O}_2$ approaches the ``image of $\mathcal{O}_3$ in the next
Poincar\'e patch,'' in 4-point function with the ordering
\begin{equation}
  \langle \mathcal{O}_4 \mathcal{O}_3 \mathcal{O}_2 \mathcal{O}_1 \rangle .
\end{equation}
The operators $\mathcal{O}_1$ and $\mathcal{O}_4$ are kept spacelike
separated, with $\mathcal{O}_1$ in past of $\mathcal{O}_2$ and $\mathcal{O}_4$
in the future of $\mathcal{O}_3$.\footnote{In a symmetric version of the
limit, which is related to the one described here by a conformal
transformation, the operators $\mathcal{O}_1$ and $\mathcal{O}_4$ approach
each other's images in the same way as $\mathcal{O}_2$ and $\mathcal{O}_3$
do.} The image of $\mathcal{O}_3$ in the next Poincar\'e patch is the first
point on Minkowski cylinder where all future-directed null geodesics from
$\mathcal{O}_3$ meet. A lot of interest in Regge limit comes from its
interpretation as bulk high-energy scattering through AdS/CFT. Kinematically,
this limit is somewhat similar to the $\mathcal{O}_2 \rightarrow
\mathcal{O}_3$ limit because $\mathcal{O}_3$ and its image in the next
Poincar\'e patch transform in the same way under conformal group. For example,
the cross-ratios $z_t, \bar{z}_t \rightarrow 0$ in Regge limit. (Here by $z_t,
\bar{z}_t$ we mean the cross-ratios for t-channel $\mathcal{O}_2 \times
\mathcal{O}_3$.) However, they do so after $\bar{z}_t$ crosses the cut $[1,
\infty)$, and so in terms of $\rho_t, \bar{\rho}_t$ we have $\rho_t
\rightarrow 0$ and $\bar{\rho}_t \rightarrow \infty$. Therefore, the
$\mathcal{O}_2 \times \mathcal{O}_3$ OPE is divergent. Conformal Regge theory
gives a way to resum the $\mathcal{O}_2 \times \mathcal{O}_3$ OPE in a way
that exhibits a dominant contribution from a ``Reggeon'' exchange, which is an
example of a light-ray operator. Justification for this resummation, which
involves analytic continuation of OPE data in spin, comes from LIF (which
historically was understood after Conformal Regge theory was established). In
the context of our paper, it would be interesting to understand whether such
resummations can be made rigorous enough (in axiomatic sense) and used to
prove that Minkowski correlators are functions in regions where so far only
temperedness has been proven.\footnote{In the classic Regge limit there is a
channel in which the OPE converges regularly, but it is possible that some
causal orderings can be relaxed while keeping the resummation procedure
valid.} For this it might not be necessary to understand the Reggeon or more
general light-ray operators, since the resummation procedure can be stopped at
a point where the correlator is expressed as an integral of $C (\Delta, \ell)$
over a region where LIF converges. See~\cite{Caron-Huot:2020nem} for progress on these questions.

\tmtextbf{Works of Gillioz, Luty et al.} Papers by this group of authors are
characterized by the systematic use of momentum space in Lorentzian CFT. So,
Refs.\ {\cite{Gillioz:2019lgs,Gillioz:2020wgw}} computed Lorentzian momentum
space 3-point functions (3 scalars and scalar-scalar-spin $\ell$) by solving the
conformal Ward identities. In momentum space, it's also possible to form
conformal blocks by gluing 3-point functions {\cite{Gillioz:2020wgw}}. See also
notes \ref{noteMarc1}, \ref{noteMarc2}.

Ref.\ {\cite{Gillioz:2019iye}} carried out this program quite explicitly in 2d
CFT, with an eye towards eventual conformal bootstrap applications. They
stressed that the momentum conformal block expansion generally converges only
in the sense of distributions---one of the first mentions of distributional
convergence in the modern CFT literature. For some momenta configurations,
they argued that the momentum conformal blocks can be pointwise bounded by the
position conformal blocks with an appropriately chosen real $z \in (0, 1)$.
For such configurations the momentum expansion converges in the ordinary sense
of functions. The same work also proposes a bootstrap equation in the momentum
space, obtained by transforming the local commutativity constraint multiplied
by a test function selecting configurations with a spacelike pair of points
(however, examples of test functions chosen in {\cite{Gillioz:2019iye}} may be
too singular).

Refs.\ {\cite{Gillioz:2016jnn,Gillioz:2018kwh,Gillioz:2020mdd} studied the Fourier transform of the time-ordered Minkowski 4-point function in
relation to various interesting physics questions. Note that, as mentioned in
the conclusions, time-ordered Minkowski CFT 4-point functions have not yet been
rigorously defined as a distributions. The Fourier transform depends on 4
momenta $p_i$, and to reduce functional complexity it is interesting to take
some or all of these momenta lightlike, $p_i^2 \rightarrow 0$. \

So, Ref.\ {\cite{Gillioz:2016jnn}} considered the Fourier transform of the
connected time-ordered 4-point function $\langle \mathcal{T} \{
\mathcal{O}_1 \mathcal{O}_2 \mathcal{O}_3 \mathcal{O}_4 \} \rangle_c$. Here
they worked with operators of scaling dimension $\Delta_i > d / 2$ for which
the Fourier transform is expected to have a finite limit as $p_i^2 \rightarrow
0$.\footnote{We thank Marc Gillioz for explanations of his work and in
particular of the distinction between the high dimension case discussed here
and the low dimension case below.} Ref.\ {\cite{Gillioz:2016jnn}} proposed a
Lorentzian CFT analogue of the optical theorem:
\begin{equation}
  \tmop{Im} \mathcal{M}_{1234} (s, t) = \underset{\mathcal{O} \neq 1}{\sum}
  f_{12\mathcal{O}} f_{\bar{4} \bar{3} \mathcal{O}}^{\ast}
  \mathcal{N}_{\mathcal{O}} (q) \langle \bar{\mathcal{T}} \{
  \hat{\mathcal{O}}_1 (p_1) \hat{\mathcal{O}}_2 (p_2) \} \mathcal{O}^{\dag}
  (0) \rangle \langle \mathcal{O} (0) \mathcal{T} \{ \hat{\mathcal{O}}_3 (p_3)
  \hat{\mathcal{O}}_4 (p_4) \} \rangle, \label{CFToptical}
\end{equation}
where $\mathcal{M}_{1234}$ is proportional to the Fourier transform of
$\langle \mathcal{T} \{ \mathcal{O}_1 \mathcal{O}_2 \mathcal{O}_3
\mathcal{O}_4 \} \rangle_c$, $f_{i j k}$ is the same as in (\ref{3-pointGeneral}),
$\mathcal{N}_{\mathcal{O}} (q)$ is some normalization factor at $q = p_1 + p_2
(= - p_3 - p_4)$, and $\hat{\mathcal{O}}$ denotes the Fourier transform of
$\mathcal{O}$. Eq.\ (\ref{CFToptical}) is supposed to hold in the following
kinematic region in the momentum space:
\[ p_i^2 = 0, \qquad s = (p_1 + p_2)^2 > 0, \qquad t = (p_1 + p_3)^2 \leqslant
   0, \]
and was derived from a combinatorial operator identity
\begin{equation}
  \underset{k = 0}{\overset{n}{\sum}} (- 1)^k \underset{\sigma \in S_n}{\sum}
  \frac{1}{k! (n - k) !} \bar{\mathcal{T}} \{ \mathcal{O}_{\sigma_1}
  (x_{\sigma_1}) \ldots \mathcal{O}_{\sigma_k} (x_{\sigma_k}) \} \mathcal{T}
  \{ \mathcal{O}_{\sigma_{k + 1}} (x_{\sigma_{k + 1}}) \ldots
  \mathcal{O}_{\sigma_n} (x_{\sigma_n}) \} = 0, \label{identity:comb}
\end{equation}
summing over all permutations, with $\mathcal{T}$($\bar{\mathcal{T}}$) time
ordering (anti-time ordering). Note that the use of this identity may be not
fully safe in the distributional context, as it arises from a non-smooth
partition of unity.

The CFT optical theorem (\ref{CFToptical}) was used in Ref.\ {\cite{Gillioz:2016jnn}} to study the scale anomalies that appear in a
specific class of CFT correlation functions. In fact, unlike Wightman
functions which are conformally invariant distributions, time-ordered
correlator distributions may, for certain scaling dimensions, contain pieces
which violate scale invariance. Thus the scale anomaly describes the violation
of dilatation Ward identities, and in position space it is an ultralocal term,
located at coincident points. In Fourier space, scale anomaly translates into
a nonzero imaginary part of $\mathcal{M}_{1234} (s, t)$ at $t = 0$. Eq.
(\ref{CFToptical}) then computes the scale anomaly coefficient through a
positive definite sum rule (in particular predicts that it is positive). These
scale anomalies also appear in the Euclidean signature, and a similar sum rule
for anomaly coefficients can also be found in Ref.\ {\cite{Gillioz:2016jnn}}.
However, the Euclidean sum rule is not positive definite unlike the Lorentzian
case.

Ref.\ {\cite{Gillioz:2016jnn}} tested the above ideas for the scalar 4-point
function of external dimensions $\Delta = 3 d / 4$. Ref.\ {\cite{Gillioz:2018kwh}} then studied the more interesting case of the stress
tensor 4-point function $\langle \mathcal{T} \{ T_{\mu_1 \nu_1} T_{\mu_2
\nu_2} T_{\mu_3 \nu_3} T_{\mu_4 \nu_4} \} \rangle$ whose scale anomaly is
proportional to the stress-tensor 2-point function coefficient $c_T$.
Assuming that the $t \rightarrow 0$ limit is finite, the CFT optical theorem
expresses $c_T$ as a sum of positive contributions of all operators in the $T
\times T$ OPE apart from the identity (the stress tensor contribution is
known, proportional to $c_T$, and can be moved to the l.h.s.). The
contributions from the scalars and the spin-2 operators are computed
explicitly in {\cite{Gillioz:2018kwh}}.

The more recent Ref.\ {\cite{Gillioz:2020mdd}} studied instead the Fourier
transform of {{the time-ordered 4-point function}} (or Euclidean
4-point function) in the opposite case of the low external dimensions
$\Delta_{\phi} < d / 2$. Unlike in {\cite{Gillioz:2016jnn,Gillioz:2018kwh}},
in this low dimension case the Fourier transform is singular as $p_i^2
\rightarrow 0$, and one obtains a finite quantity multiplying it by
$(p_i^2)^{d / 2 - \Delta_{\phi}}$ before taking the limit, a CFT analogue of
LSZ reduction. Doing so, they defined a ``CFT scattering amplitude'' $A (s, t,
u)$ ($p_i^2 \rightarrow 0$ for all $i$) and a closely related ``form factor''
$F (s, t, u)$ where $p_i^2 \rightarrow 0$ only for $i = 1, 2, 3$. Because the
limit $p_i^2 \rightarrow 0$ has to be taken one momentum at a time, crossing
symmetry is not obvious. Ref.\ {\cite{Gillioz:2020mdd}} also gave an
alternative derivation, starting from the Mellin representation of the CFT
4-point function, where crossing symmetry of $F (s, t, u)$ and $A (s, t,
u)$ follows from the crossing symmetry of the Mellin amplitude. In the future,
crossing symmetric quantities $A (s, t, u)$ and $F (s, t, u)$ may turn out
useful in a bootstrap analysis. It should be stressed that the results of
{\cite{Gillioz:2020mdd}} in no way contradict the usual lore that there are no
S-matrices in interacting CFTs. In spite of the name adopted in
{\cite{Gillioz:2020mdd}}, the existence of the quantity $A (s, t, u)$ does not
imply that we can set up a wave-packet scattering experiment in a CFT.
Wave-packets would quickly diffuse before reaching the interaction region, the
singularity of $(p_i^2)^{\Delta_{\phi} - d / 2}$ being a cut rather than a
pole.}

\subsection{Review of Hartman et al {\cite{Hartman:2015lfa}}}\label{Tom}

\tmtextbf{Relating different orderings via analytic continuations.} Here we
will comment on some of the results of {\cite{Hartman:2015lfa}} in more
details. The first part of this paper considers the Lorentzian CFT 4-point
functions with operators $\mathcal{O}_1$, $\mathcal{O}_3$, $\mathcal{O}_4$ fixed
at zero time and the spatial positions 0, $\hat{e}_1$ and $\infty$, while the operator $\mathcal{O}_2$ is inserted at Minkowski position $t_2\hat e_0+y_2\hat e_1$. 
They consider
four different operator orderings
\begin{equation}
  \langle \mathcal{O}_2 \mathcal{O}_1 \mathcal{O}_3 \mathcal{O}_4 \rangle,
  \langle \mathcal{O}_1 \mathcal{O}_2 \mathcal{O}_3 \mathcal{O}_4 \rangle,
  \langle \mathcal{O}_3 \mathcal{O}_2 \mathcal{O}_1 \mathcal{O}_4 \rangle,
  \langle \mathcal{O}_1 \mathcal{O}_3 \mathcal{O}_2 \mathcal{O}_4 \rangle
\end{equation}
in the region of $0 < y_2 < 1 / 2$ and $t_2$ positive. As
$t_2$ is increased from zero, the operator $\mathcal{O}_2$, initially
spacelike with respect to all other insertions, crosses the light cone first of
$\mathcal{O}_1$ and then of $\mathcal{O}_3$. With $z, \bar{z} = y_2 \pm t_2$,
denoting $G (z, \bar{z}) = \langle \mathcal{O}_2 \mathcal{O}_1 \mathcal{O}_3
\mathcal{O}_4 \rangle$, they give the following prescription to compute the
correlators for the other orderings (Ref.~\cite{Hartman:2015lfa}, Eq.\ (3.22)):
\begin{eqnarray}
  \langle \mathcal{O}_1 \mathcal{O}_2 \mathcal{O}_3 \mathcal{O}_4 \rangle & =
  & G (z, \bar{z})_{z \rightarrow e^{- 2 \pi i} z}, \label{Tomrecipe} \\
  \langle \mathcal{O}_3 \mathcal{O}_2 \mathcal{O}_1 \mathcal{O}_4 \rangle & =
  & G (z, \bar{z})_{(\bar{z} - 1) \rightarrow e^{- 2 \pi i} (\bar{z} - 1)},
  \nonumber\\
  \langle \mathcal{O}_1 \mathcal{O}_3 \mathcal{O}_2 \mathcal{O}_4 \rangle & =
  & G (z, \bar{z})_{z \rightarrow e^{- 2 \pi i} z, (\bar{z} - \bar{z}_0)
  \rightarrow e^{- 2 \pi i} (\bar{z} - \bar{z}_0)} . \nonumber
\end{eqnarray}
Their justification of this prescription relied on some presumed analyticity
properties of $G (z, \bar{z})$ which, to our knowledge, have never been shown
in a general QFT context. Nevertheless we will see below that for CFTs Eq.
{\eqref{Tomrecipe}} turns out to be true (with $\bar{z}_0 = 1$).

The real parameter $\bar{z}_0$ was introduced in {\cite{Hartman:2015lfa}} as
the position of the first $\bar{z}$ singularity of $G (z, \bar{z})_{z
\rightarrow e^{- 2 \pi i} z}$. Their goal was to show that $\bar{z}_0
\geqslant 1$, which using {\eqref{Tomrecipe}} then implies local commutativity
$\langle \mathcal{O}_1 [\mathcal{O}_2, \mathcal{O}_3] \mathcal{O}_4 \rangle =
0$ for $\bar{z} < 1$ i.e.\ when $\mathcal{O}_2$ is spacelike to
$\mathcal{O}_3$. In our paper (Sec.\ \ref{local-comm}) we presented a
different way to understand and derive local commutativity which is closer to
the classic literature: as we have reviewed there, it is a robust consequence
of the existence of the analytic continuation to the forward tube, which we
constructed.

Let's see what it would take to justify {\eqref{Tomrecipe}}. We define the
function $G (z, \bar{z})$ by the $\mathcal{O}_2 \times \mathcal{O}_1$ OPE
expansions, which converges for $| \rho |, | \bar{\rho} | < 1$, i.e.\ as long
as $z, \bar{z}$ stay away from $z, \bar{z} \in (+ 1, \infty)$. The points $z =
0$ and $\bar{z} = 0$ are branch points singularities with cuts which we put
along the negative real axis. Note that the contours in {\eqref{Tomrecipe}}
are all within the analyticity domain of $G (z, \bar{z})$. So the prescription
{\eqref{Tomrecipe}} is meaningful. We still have to see if it agrees with the
rigorous definition which computes the Minkowski 4-point function by analytically
continuing from the Euclidean region staying in the forward tube corresponding
to the chosen operator ordering. We will see that it will indeed agree, but
showing it for the last ordering will be subtle.

For definiteness we will focus on the region $t_2 > 1 - y_2$ \ i.e.\ $\bar{z}
> 1$, where $\mathcal{O}_2$ crossed both light cones. The end point of the
analytic continuation contour is always the same while the initial point
depends on the operator ordering. E.g.\ for the ordering $\langle \mathcal{O}_2
\mathcal{O}_1 \mathcal{O}_3 \mathcal{O}_4 \rangle$ we have to pick initial
Euclidean times $\varepsilon_2 > \varepsilon_1 > \varepsilon_3$, while for
$\langle \mathcal{O}_1 \mathcal{O}_2 \mathcal{O}_3 \mathcal{O}_4 \rangle$ we
have $\varepsilon_1 > \varepsilon_2 > \varepsilon_3$ etc. For any of these
orderings, we denote
\begin{equation}
  z_1, \bar{z}_1 = \pm i \varepsilon_1, \quad z_2, \bar{z}_2 = y_2 \pm i
  (\varepsilon_2 + i t_2), \quad z_3, \bar{z}_3 = 1 \pm i \varepsilon_3
\end{equation}
and compute (in the limit $z_4, \bar{z}_4 = \infty$)
\begin{equation}
  z = \frac{z_1 - z_2}{z_1 - z_3}, \quad \bar{z} = \frac{\bar{z}_1 -
  \bar{z}_2}{\bar{z}_1 - \bar{z}_3} .
\end{equation}
We are interested in the curves which $z, \bar{z}$ trace as the Euclidean
times are scaled to zero and the Minkowski time $t_2$ from 0 to its final
value. For the first three orderings the resulting curves are shown in Fig.\ \ref{first3orderings}.

\begin{figure}[h]\centering
  \raisebox{-0.437612502733711\height}{\includegraphics[width=16.2870097074643cm,height=4.87365538501902cm]{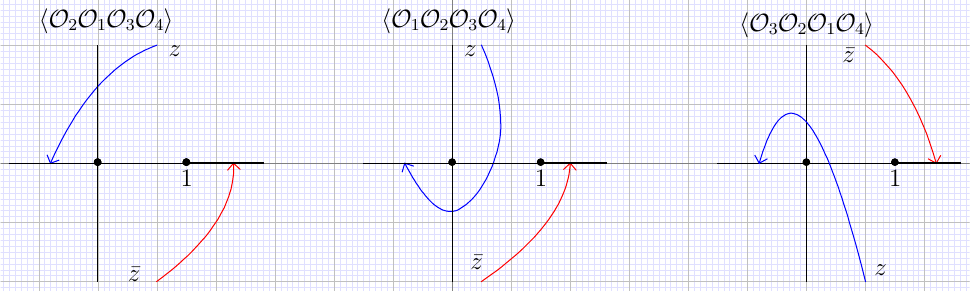}}
  \caption{\label{first3orderings}$z, \bar{z}$ curves for the analytic
  continuation from the Euclidean; the first 3 orderings.}
\end{figure}

We see that in all these cases, the curves lie in the analyticity domain of $G
(z, \bar{z})$ i.e.\ they don't cross $(1, + \infty)$. For the first two cases
this was guaranteed by our results that $| \rho |, | \bar{\rho} | < 1$ for the
s-channel OPE expansion. For the third case it was not guaranteed but it also
turns out to be true, by inspection. We also see that in all these 3 cases,
the curves go around $z = 0$ and $\bar{z} = 1$ in agreement with
{\eqref{Tomrecipe}}. \footnote{To see this more clearly in the third case, deform the
curves continuously moving the initial $z$ into the upper half plane and the
initial $\bar{z}$ into the lower half plane.}

For the fourth ordering $\langle \mathcal{O}_1 \mathcal{O}_3 \mathcal{O}_2
\mathcal{O}_4 \rangle$ when we have to assign $\varepsilon_1 > \varepsilon_3 >
\varepsilon_2$, the analytic continuation inside the forward tube gives the
$z, \bar{z}$ curves shown in Fig.\ \ref{last0}, while prescription
{\eqref{Tomrecipe}} would correspond to Fig.\ \ref{last1}.

\begin{figure}[h]\centering
  \raisebox{-0.4738733350422\height}{\includegraphics[width=5.97432113341204cm,height=3.45444706808343cm]{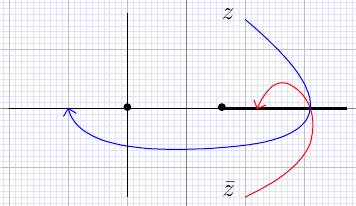}}
  \caption{\label{last0}$z, \bar{z}$ curves for the analytic continuation from
  the Euclidean for the $\langle \mathcal{O}_1 \mathcal{O}_3 \mathcal{O}_2
  \mathcal{O}_4 \rangle$ ordering.}
\end{figure}

\

\begin{figure}[h]\centering
  \raisebox{-0.4738733350422\height}{\includegraphics[width=5.97432113341204cm,height=3.45444706808343cm]{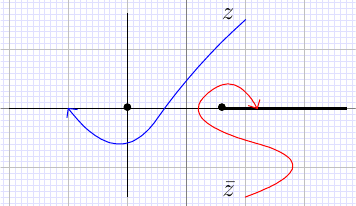}}
  \caption{\label{last1}$z, \bar{z}$ curves for computing the $\langle
  \mathcal{O}_1 \mathcal{O}_3 \mathcal{O}_2 \mathcal{O}_4 \rangle$ ordering
  via {\eqref{Tomrecipe}}.}
\end{figure}

The two figures are clearly not the same. Moreover the curves in the first
figure cross $(+ 1, \infty)$ where the definition of $G (z, \bar{z})$ via the
$\mathcal{O}_1 \times \mathcal{O}_2$ channel OPE expansion stops converging.
Can we show that the analytic continuation in Fig.\ \ref{last0} exists and that
it agrees with the one in Fig.\ \ref{last1}?

For this, let us bring in the $\mathcal{O}_2 \times \mathcal{O}_3$ OPE
expansion, which correspond to expanding in $z_t = 1 - z$, $\bar{z}_t = 1 -
\bar{z}$ or in the corresponding $\rho_t$, $\bar{\rho}_t$. In the Euclidean
region the two expansions agree. The $\mathcal{O}_2 \times \mathcal{O}_3$
expansion converges away from $z_t, \bar{z}_t \in (+ 1, \infty)$ i.e.\ $z,
\bar{z} \in (- \infty, 0)$. Thus the curves in both Figs.\ \ref{last0},
\ref{last1} lie within the range of analyticity of the $\mathcal{O}_2 \times
\mathcal{O}_3$ expansion, so we can compare the analytic continuations. Since
only integer spins occur in the expansion, the analytic continuation does not
change under $\rho_t \rightarrow e^{2 \pi i} \rho_t$, $\bar{\rho}_t
\rightarrow e^{- 2 \pi i} \bar{\rho}_t$ (such arguments were systematically
exploited in {\cite{Qiao:2020bcs}}).\footnote{Sometimes this property is called ``Euclidean single-valuedness''.} So let us add extra loops to the blue and
the red curves in the opposite directions around $1$, see Fig.\ \ref{last2}.
Adding the loops and deforming the curves continuously (the first step is
shown in Fig.\ \ref{last2}) we can bring them to those in Fig.\ \ref{last1}.
This finishes the proof that the prescription {\eqref{Tomrecipe}} is correct
also for the fourth ordering.

\begin{figure}[h]\centering
  \raisebox{-0.4738733350422\height}{\includegraphics[width=5.97432113341204cm,height=3.45444706808343cm]{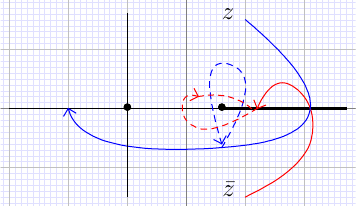}}$\Longrightarrow$\raisebox{-0.4738733350422\height}{\includegraphics[width=5.97432113341204cm,height=3.45444706808343cm]{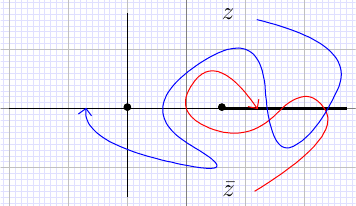}}
  \caption{\label{last2}Deforming the $z, \bar{z}$ curves for the $\langle
  \mathcal{O}_1 \mathcal{O}_3 \mathcal{O}_2 \mathcal{O}_4 \rangle$ ordering.}
\end{figure}

\tmtextbf{Positivity constraints.} We wish to comment on another result of
{\cite{Hartman:2015lfa}}: an argument for positivity of certain conformal
block expansion coefficients. We present the argument exchanging the role of s
an t channels w.r.t. {\cite{Hartman:2015lfa}}. Let $G (z, \bar{z}) = 1 +
\ldots$ be the holomorphic function defined by the s-channel OPE expansion (i.e.\ $(z \bar{z})^{\Delta_1 + \Delta_2}$ times the $G (z, \bar{z})$ discussed
above). We will define a certain analytic continuation of the function $G (z,
\bar{z})$. Let us start with $\bar{z}$ and $z$ close to zero, $\bar{z}$ in the
upper half plane and $z$ in the lower half plane. In this range all three
channels s,t,u converge. We wish to analytically continue $G (z, \bar{z})$ by
taking $z$ through $(1, + \infty)$ and bring it back close to zero, in the
upper half plane (see Fig.\ \ref{Gtilde}), while we don't touch $\bar{z}$. This
analytic continuation can be performed using the t-channel or u-channel
expansions, with the same result (but not the s-channel since it stops
converging on $(1, + \infty)$). We denote the result of this analytic
continuation by $\hat{G} (z, \bar{z})$, with $\tmop{Im} z, \tmop{Im} \bar{z} >
0$.

\begin{figure}[h]\centering
  \raisebox{-0.485281603306646\height}{\includegraphics[width=5.02620359438541cm,height=2.83259543486816cm]{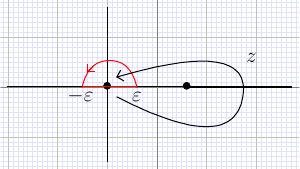}}
  \caption{\label{Gtilde}Definition of $\hat{G} (z, \bar{z})$. Red:
  integration contour in Eq.\ {\eqref{contour}}.}
\end{figure}

Although $\hat{G} (z, \bar{z})$ is so defined with both $z, \bar{z}$ in the
upper half plane, it has continuous limits when they both approach positive
real axis, or both approach negative real axis, since in the first case the
t-channel and in the second case the u-channel remains convergent. We will be
interested in the situation when \ $\mathcal{O}_1 =\mathcal{O}_2$,
$\mathcal{O}_3 =\mathcal{O}_4$. In this case all expansion coefficients in the
t and u channels are positive. This implies that the Euclidean correlator $G_E
(z, \bar{z})$ will be positive for real $z, \bar{z} > 0$ (using t-channel) and
for real $z, \bar{z} < 0$ (using u-channel). The difference between $G_E (z,
\bar{z})$ and $\hat{G} (z, \bar{z})$, for real $z, \bar{z} > 0$ or $z, \bar{z}
< 0$ is that in the first case $z, \bar{z}$ approach the real axis from the
opposite sides while in the second case from the same one. When we take $z$
through $(1, + \infty)$ cut, this only brings in some phases in the t and
u-channel expansion of $\hat{G} (z, \bar{z})$ with respect to $G_E (z,
\bar{z})$. This implies that we have a bound for real $z, \bar{z} > 0$ or $z,
\bar{z} < 0$:
\begin{equation}
  | \hat{G} (z, \bar{z}) | \leqslant G_E (z, \bar{z}) \label{GtildeG}
\end{equation}
In what follows $\hat{G} (z, \bar{z})$ will be used as a holomorphic function
with $z, \bar{z}$ in the upper half plane satisfying the bound
{\eqref{GtildeG}} on its boundary, while $G_E (z, \bar{z})$ will be used only
with real $z, \bar{z}$.

In particular, since $G_E \approx 1$ for $z, \bar{z}$ near zero up to small
corrections, Eq.\ {\eqref{GtildeG}} says that $\hat{G}$ is bounded, for small
real $z, \bar{z} > 0$ or $z, \bar{z} < 0$, by 1 up to small corrections. This
argument can be generalized to show that $\hat{G} (z, \eta z)$ for $\eta > 0$
real and $z$ near zero in the upper half plane is bounded by a
constant.\footnote{Let $z = r e^{i \varphi}$, $r \ll 1$. We consider $0
\leqslant \varphi \leqslant \pi / 2$, when the argument uses the t-channel,
the case $\pi / 2 \leqslant \varphi \leqslant \pi$ is analogous using the
u-channel. The key point is that the $\rho$ variable in the t-channel $\rho_t
\approx 1 - \sqrt{r} e^{i \varphi / 2}$, $| \rho_t | \approx 1 - \sqrt{r} \cos
(\varphi / 2)$. This allows to compare the function $\hat{G} (z, \eta z)$ to
$G_E (z', \eta z')$ with real $z' = r \cos^2 (\varphi / 2)$, times a factor
$\sim (z \bar{z})^{\Delta_1} / (z' \bar{z}')^{\Delta_1}$ from the crossing
kernel, which is bounded by a constant.}

We now pass to the non-rigorous part of the argument. Although the s-channel
stops converging when crossing $(1, + \infty)$, Ref.\ {\cite{Hartman:2015lfa}}
proposed that, in the regime $| \bar{z} | \ll | z | \ll 1$, the behavior
$\hat{G} (z, \bar{z})$ can nevertheless be predicted from the s-channel
expansion, by organizing it in $\bar{z}^{\tau / 2}$ where $\tau = \frac{1}{2}
(\Delta - \ell)$ is the twist. The typical term is
\begin{equation}
  \bar{z}^{\tau / 2} k_{\frac{1}{2} (\Delta + \ell)} (z),
\end{equation}
where $k_h (x) =_2 F_1 (h, h, 2 h, x)$ is the collinear conformal block. This
is the same expansion as used in the light-cone bootstrap, which has its own
problems of rigor, but here it is proposed to use it after $z - 1 \rightarrow
e^{2 \pi i} (z - 1)$. Under this continuation the collinear conformal block,
which has a $\log (1 - z)$ behavior near $z = 1$, picks up an imaginary piece
which, for $z$ small, behaves as $\sim i z^{1 - \frac{1}{2} (\Delta + \ell)}$
(see {\cite{Hartman:2015lfa}}, Eq.\ (4.28)). Considering $\bar{z} = \eta z$,
$\eta \ll 1$, $| z | \ll 1$, we then have, according to the proposal of Ref.\ {\cite{Hartman:2015lfa}},
\begin{equation}
  \hat{G} (z, \eta z) \approx 1 - B (\Delta, \ell) p_{\Delta, \ell} \times i
  \frac{\eta^{\tau / 2}}{z^{\ell - 1}}, \label{pred4}
\end{equation}
where $\Delta, \ell$ are the dimension and spin of the leading twist operator
(which may e.g.\ be the stress tensor), $p_{\Delta, \ell}$ its conformal block
coefficient, and $B (\Delta, \ell) \geqslant 0$ some explicitly known
constant. The spin $\ell$ is even since we are assuming $\mathcal{O}_1
=\mathcal{O}_2$. Eq.\ {\eqref{pred4}} assumes that the limit $\eta \rightarrow
0$ is taken before $z \rightarrow 0$.\footnote{In fact, in the opposite limit
$z \rightarrow 0$ for fixed $\eta$, Eq.\ {\eqref{pred4}} would violate the
discussed above rigorous bound that $\hat{G} (z, \eta z)$ is bounded by a
constant. There is no paradox because that's not the limit we are supposed to
be taking.}

Now, let us consider the holomorphic function $f (z) = 1 - \hat{G} (z, \eta z)$,
and integrate $z^{\ell - 2} f (z)$ along the contour shown in Fig.\ \ref{Gtilde}.\footnote{We can also take an intermediate step of adding a small
semicircle of radius $\varepsilon'$ around zero, but since $\hat{G}$ is
bounded for small $z$, the limit $\varepsilon' \rightarrow 0$ is not
problematic.} We have
\begin{equation}
  \int_{\tmop{arc}} z^{\ell - 2} f (z)\, d z + \int_{-
  \varepsilon}^{\varepsilon} x^{\ell - 2} f (x)\, d x = 0 . \label{contour}
\end{equation}
Using {\eqref{pred4}} and that the integral over the arc of $1 / z$ is $\pi
i$, we get in particular:\footnote{Note that
	the quantity $\tmop{Re} [G_E (z, \eta z) - \hat{G} (z, \eta z)]$ appearing in~\eqref{result4} is essentially
	the double discontinuity considered in~\cite{Caron-Huot:2017vep}. Similarly, Eq.~\eqref{result4} can
	be formally obtained from the Lorentzian inversion formula of~\cite{Caron-Huot:2017vep} by expanding the integrand
	in a light-cone limit. We thank Tom Hartman for pointing this out.
}
\begin{equation}
  \pi B (\Delta, \ell) p_{\Delta, \ell} \approx \eta^{- \tau / 2} \int_{-
  \varepsilon}^{\varepsilon} x^{\ell - 2} \tmop{Re} [1 - \hat{G} (z, \eta z)]\,
  d x \approx \eta^{- \tau / 2} \int_{- \varepsilon}^{\varepsilon} x^{\ell
  - 2} \tmop{Re} [G_E (z, \eta z) - \hat{G} (z, \eta z)]\, d x, \label{result4}
\end{equation}
where in the final step we replaced 1 by $G_E (z, \eta z)$. Since $G_E (z,
\eta z)$ has a rigorously convergent expansion for small $z$, it satisfies the
bound:
\begin{equation}
  G_E (z, \bar{z}) = 1 + O (\eta^{\tau / 2}) .
\end{equation}
So the last replacement was legitimate if e.g.\ $\ell \geqslant 2$. By
{\eqref{GtildeG}}, the r.h.s.\ of {\eqref{result4}} is a positive quantity. \
This equation then implies that $p_{\Delta, \ell}$ must be positive as well.

As already mentioned, the weak point of this argument is that the s-channel
expansion stops converging when we cross $(1, + \infty)$. It is therefore not
at all obvious that analytic continuations of the individual conformal block
expansion terms have anything to do with the asymptotics of $\hat{G} (z,
\bar{z})$. Ref.\ {\cite{Hartman:2015lfa}} was of course aware of this, and
provided some arguments, inspired by the light-cone bootstrap, why nevertheless
the asymptotics from the leading twist terms can be trusted. We don't know how
to make those arguments rigorous. It would be interesting to understand if
asymptotics {\eqref{pred4}} can be justified using just Euclidean CFT axioms
of Sec.\ \ref{ECFTax} or requires additional assumptions. The same question
also looms over the proofs of ANEC {\cite{Hartman:2016lgu}} {and ANEC commutativity \cite{Kologlu:2019bco} which involved similar ``light-cone limit
on the second sheet'' considerations.}

\section{OS axioms for higher-point functions}\label{OShigher}

In this appendix we discuss the modifications necessary to derive from the
Euclidean CFT axioms the OS axioms (positivity and cluster property) for
$n$-point functions with $n > 4$, compared to the $n \leqslant 4$ case
considered in Sec.\ \ref{OSfromCFT}. As we explain below, it appears that
there is no simple proof of OS positivity for $n > 4$ from the Euclidean CFT
axioms of Sec.\ \ref{ECFTax}. Since the reason for this is rather technical,
let us first discuss the conceptual implications of this.

Ideally, one would like to have a set of Euclidean CFT axioms that would imply
Wightman axioms (and therefore also OS axioms) and also be powerful enough to
derive all the usual CFT lore such as OPEs, radial quantization,
operator-state correspondence, crossing symmetry, etc. These statements, as
we saw in the main text, make sense and can be non-trivial even when we
restrict our attention to $n$-point functions with bounded $n$.

In particular, we have found that the axioms we formulated in Sec.\ \ref{ECFTax} achieve the above goal for $n \leqslant 4$. Extending our results
to $n > 4$ using the same strategy would require a solution to two problems:
first, we need to derive OS axioms (specifically, positivity and cluster
property) for $n > 4$, and, second, we need to prove that OS axioms together
with the OPE imply Wightman axioms.

Conceptually, it seems plausible that OS axioms + OPE imply Wightman axioms
for $n > 4$ because we expect that for $n > 4$ there is again an OPE channel
which is convergent in the entire forward tube (i.e.\ the one given by taking
the OPE in the same order as the operators appear in the Wightman ordering).
This question clearly merits further study but is beyond the scope of this
paper.

However, it is less clear to us how to even attempt a derivation of OS
positivity for $n > 4$ from Euclidean CFT axioms of Sec.\ \ref{ECFTax}. Let
us first explain why this is the case, and then we will discuss the possible
modifications to these CFT axioms.

Suppose we want to prove the positivity
\begin{equation}
  \langle \Psi | \nobracket \Psi \rangle \geqslant 0,
\end{equation}
where $\Psi$ is a state created by a product of three local operators, $| \Psi
\rangle = | \varphi_1 (x_1) \varphi_2 (x_2) \varphi_3 (x_3) \rangle$.

To prove this positivity the natural idea would be to use the OPE expansion
repeatedly for the two copies of $\Psi$ and then use the positivity of the
2-point function. However, for this we need our OPE approximation for
$\langle \nobracket \Psi | \nobracket$ to be conjugate to our approximation
for $| \Psi \rangle$. This is non-trivial to achieve because we have to
perform the OPEs one at a time. For example, we can first construct an
approximation of $| \Psi \rangle$ in terms of a state $| \Psi' \rangle$,
created by single operator insertions, such that
\begin{equation}
  | \langle \Psi | \nobracket \Psi \rangle - \langle \Psi | \nobracket \Psi'
  \rangle | < \varepsilon .
\end{equation}
Similarly, we can construct a state $\langle \Psi'' |$ such that
\begin{equation}
  | \langle \Psi'' | \nobracket \Psi' \rangle - \langle \Psi | \nobracket
  \Psi' \rangle | < \varepsilon
\end{equation}
and thus
\begin{equation}
  | \langle \Psi'' | \nobracket \Psi' \rangle - \langle \Psi | \nobracket \Psi
  \rangle | < 2 \varepsilon .
\end{equation}
These approximations are possible by the repeated use of the OPE
{\eqref{OPEplanar}}. Note, however, that since the OPE axiom is formulated for
correlation functions, the number of terms we have to include in the OPE for a
given $\varepsilon$ depends on the correlation function in which the OPE is
performed. It then follows that the state $\langle \Psi'' |$ depends on $|
\Psi' \rangle$ (because in order to construct it we use the OPE in the
correlation function $\langle \Psi | \nobracket \Psi' \rangle$) and is in
general different from it. It is therefore not obvious that $\langle \Psi'' |
\nobracket \Psi' \rangle \geqslant 0$, which is what we would like to use in
order to prove $\langle \Psi | \nobracket \Psi \rangle \geqslant 0$ with the
help of the above inequalities.

In the case when $n = 4$ and $| \Psi \rangle$ is created by 2 operators we
were able to solve this difficulty. This was because in this case the only
difference between $| \Psi' \rangle$ and $\langle \Psi'' |$ can be in the
number of OPE terms included in the approximation, and we were able to use an
orthogonality property of the 2-point function to show $\langle \Psi'' |
\nobracket \Psi' \rangle = \langle \Psi' | \nobracket \Psi' \rangle$ by
arguing that we can assume that $\langle \Psi'' |$ contains more terms than $|
\Psi' \rangle$ and that those terms which are in $\langle \Psi'' |$ but not in
$| \Psi' \rangle$ do not contribute to the product $\langle \Psi'' |
\nobracket \Psi' \rangle$.

This argument does not work in the case at hand, $| \Psi \rangle = | \varphi_1
(x_1) \varphi_2 (x_2) \varphi_3 (x_3) \rangle$. The reason for this is that in
order to construct $| \Psi' \rangle$ or $\langle \Psi'' |$ we need to perform
two OPE's in each case. For example, the first one can be $\varphi_1 \times
\varphi_2 = \sum_k \mathcal{O}_k$ and the second one can be $\varphi_3 \times
\mathcal{O}_k$. Both OPE's have to be truncated at some point, and while the
truncation of the second OPE affects only the set of terms that are present in
$| \Psi' \rangle$ or $\langle \Psi'' |$, where we truncate the first
$\varphi_1 \times \varphi_2$ OPE affects the \tmtextit{coefficients} of these
terms. Since now $| \Psi' \rangle$ and $\langle \Psi'' |$ contain terms with
differing coefficients, we cannot use orthogonality to argue $\langle \Psi'' |
 \Psi' \rangle = \langle \Psi' | \Psi' \rangle $anymore.
There is no way to ensure that $\varphi_1 \times \varphi_2$ OPEs are truncated
in the same way in the construction of both states because the truncation in
$\langle \Psi'' |$ depends, through our OPE axiom, on $| \Psi' \rangle,$ and
thus might happen to be always at a higher order than the truncation used for
$| \Psi' \rangle .$

This all is to say that due to a rather technical reason it appears that there
is no \tmtextit{simple} proof of OS positivity of higher-point functions from
the Euclidean CFT axioms as stated in Sec.\ \ref{ECFTax}. Importantly, this
doesn't mean that there is no proof at all. Indeed, the Euclidean CFT axioms
are sufficient to derive the standard crossing-symmetry equations for
4-point functions. It could happen that in all solutions to these
crossing-symmetry equations the OPE coefficients have such asymptotics that a
stronger form of the OPE axiom holds and allows us to prove the OS positivity
for $n > 4$. However, it is not clear how to implement this line of reasoning
in practice.

It is therefore interesting to look for a stronger version of Euclidean CFT
axioms. We discuss below some simple modifications of the OPE axiom which
avoid the above problem and allow to prove OS positivity for higher-point
functions.

Morally, we want some kind of statement of uniformity for the convergence rate of the OPE: 
it should make $|\Psi''\>$ above independent of the truncation made in $|\Psi'\>$, as long as this truncation
is done at a sufficiently high order. This would allow us to make both truncations at a high order and
ensure $\langle \Psi'' |\Psi' \rangle = \langle \Psi' | \Psi' \rangle \geq 0$.

One option is to assume a stronger form of the OPE, which allows us to perform
two OPE's simultaneously Specifically, we can assume that the double sum
\begin{equation}
  \langle \mathcal{O}_1 \mathcal{O}_2 \mathcal{O}_3 \mathcal{O}_4 \ldots
  \rangle = \sum_{k, l} \langle \mathcal{O}_k \mathcal{O}_l \ldots \rangle,
\end{equation}
is convergent, where we wrote the two OPEs schematically as $\mathcal{O}_1
\mathcal{O}_2 = \sum_k \mathcal{O}_k$ and $\mathcal{O}_3 \mathcal{O}_4 =
\sum_l \mathcal{O}_l$. Convergence of the double sum means that
\begin{equation}
  \left| \langle \mathcal{O}_1 \mathcal{O}_2 \mathcal{O}_3 \mathcal{O}_4
  \ldots \rangle - \sum_{k, l} \langle \mathcal{O}_k \mathcal{O}_l \ldots
  \rangle \right| < \varepsilon,
\end{equation}
when the sums are truncated in a way that includes some
$\varepsilon$-dependent finite set of terms, but is otherwise arbitrary. In
particular, both sums can be truncated in the same way, and this solves the
problem that we encountered above. A disadvantage of this approach is that it
is unclear how to derive this axiom from OS axioms and the usual single OPE axiom (however,
a heuristic argument based on cutting the Euclidean path integral can be
made). This is somewhat subtle and is related to the question of whether the path integral over a spherical layer $(r_1<r<r_2)$ with operator insertions in the interior represents a bounded operator. We can't say with confidence whether or not this is the case.

Another option is to assume resumed repeated OPE, i.e.\ that the following sum
converges, schematically,
\begin{equation}
  \langle \mathcal{O}_1 \ldots \mathcal{O}_m \mathcal{O}_{m + 1} \ldots
  \mathcal{O}_n \rangle = \sum_k c_k \langle \mathcal{O}_k \mathcal{O}_{m + 1}
  \ldots \mathcal{O}_n \rangle,
\end{equation}
where the coefficients $c_k$ are chosen so that
\begin{equation}
  \langle \mathcal{O}_1 \ldots \mathcal{O}_m \mathcal{O}_k^{\theta} \rangle =
  c_k \langle \mathcal{O}_k \mathcal{O}_k^{\theta} \rangle,
\end{equation}
assuming $\langle \mathcal{O}_k \mathcal{O}_l^{\theta} \rangle \propto
\delta_{k l} .$ This version of the axiom is essentially the statement that
one-operator states form a basis of the CFT Hilbert space, formulated without
explicitly introducing the Hilbert space. In other words, above we are
approximating the state $\langle \mathcal{O}_1 \ldots \mathcal{O}_m |$ in
terms of an orthonormal basis of states $\langle \mathcal{O}_k |$, and the
coefficients are computed by inner products. This form of the axiom is easy to
derive from OS + convergent OPE, and also easily allows us to solve our
problem by using the same strategy as in the case $n = 4$. However, it does
appear to be an overly strong assumption, making our axioms not very different
from assuming OS + convergent OPE outright.

Finally, an interesting prospect might be, instead of formulating an entirely
new set of axioms, to add an {\tmem{asymptotic}} OPE axiom (and conformal
invariance) to OS axioms. It is likely that using logic very similar to that
of Mack {\cite{Mack:1976pa}}, which we reviewed in Sec.\ \ref{MackComp}, one
can prove that (OS axioms)+(asymptotic OPE)+(conformal invariance) imply
convergent OPE.

\section{Details on Vladimirov's theorem}\label{Vlad}

\subsection{Limit in the sense of distributions}

Let us start with a reminder of what the limit in the sense of tempered
distributions means. Let $f (u)$, $u = (t_k, \mathbf{x}_k) \equiv (t_1,
\mathbf{x}_1, \ldots, t_n, \mathbf{x}_n) \in \mathbb{R}^{n d}$, be a Schwartz
test function, i.e.\ an infinitely differentiable function decreasing at
infinity faster than any power together with all its derivatives. This can be
also stated as finiteness of all Schwartz norms:
\begin{equation}
  | f |_N = \sup_{u \in \mathbb{R}^{n d}, | \alpha | \leqslant N} (1 + u^2)^{N
  / 2} | \partial^{\alpha}_u f | < \infty \qquad \forall N \geqslant 0 .
  \label{semin}
\end{equation}
That the limit {\eqref{limit}} exists in the sense of distributions means two
requirements. First, that the r.h.s.\ of {\eqref{limit}} has a finite limit
integrated against any $f$ as above:
\begin{equation}
  (G_n^M, f) \assign \lim_{\epsilon_k \rightarrow 0} \int d t\, d\mathbf{x}\, G_n
  (\epsilon_k + i t_k, \mathbf{x}_k) f (t_k, \mathbf{x}_k) \quad \text{exists
  for any Schwartz} f. \label{intGn}
\end{equation}
The $G_n^M$ defined by this equation is a linear functional on the Schwartz
space. The second requirement is that this functional should be continuous
(and thus is itself a tempered distribution). Continuity means that it should
be bounded by one of the norms {\eqref{semin}} with a sufficiently large $N$,
i.e.:
\begin{equation}
  | (G_n^M, f) | \leqslant C | f |_{N_{\ast}}, \label{GMcont}
\end{equation}
with $f$-independent $C$ and $N_{\ast}$.

Note that by Eq.\ {\eqref{GMcont}}, $G_n^M$ can be extended from the
Schwartz space to a larger space of test functions, which are required to be
differentiable only $N_{\ast}$ times and have a finite $| f
|_{N_{\ast}}$. Parameter $N_{\ast}$ thus characterizes regularity of the
distribution $G_n^M$. The proof of Theorem \ref{ThVlad} will determine
$N_{\ast}$ in terms of $A_n$ and $B_n$, see Eq.\ {\eqref{GMreg}}.

\subsection{Proof of Theorem \ref{ThVlad}}\label{Proof1}

Unfortunately, we do not know a reference where Theorem \ref{ThVlad} is stated
and proved succinctly in the form we need. Such results are considered
standard in the theory of several complex variables. For similar statements
see {\cite{Vladimirov}}, Chapter 5, and {\cite{Streater:1989vi}}, Theorem
2-10. For the convenience of the reader, we present here a proof based on
these sources.

The usefulness of Vladimirov's theorems for establishing distributional
properties of CFT correlators was recognized in our recent work
{\cite{paper1}}. There, we considered expansions of the CFT 4-point function $g
(\rho, \bar{\rho})$ in terms of conformally invariant cross-ratios $\rho$,
$\bar{\rho}$. It is well known that such expansions converge in the interior
of the unit disk $| \rho |, | \bar{\rho} | < 1$. Using Vladimirov's theorems,
we showed in \ {\cite{paper1}} that they also converge on the boundary of this
disk, in the sense of distributions. In this paper we are interested in CFT
correlators as functions of positions $x_k$, not of cross-ratios, but the
basic principle is the same as in {\cite{paper1}}: a powerlaw bound on an
holomorphic function near a boundary implies temperedness of the limiting
distribution.

By translation invariance it's enough to study the function $G_n$ expressed
in terms of the differences $y_k = x_k - x_{k + 1}$ which we denote by
$\mathcal{G} (y)$, $y = (y_1, \ldots, y_{n - 1})$. We also denote $y_k =
(y_k^0, \mathbf{y}_k),$ $y_k^0 = \varepsilon_k + i s_k$, $\varepsilon_k > 0$,
$\mathbf{y}_k \in \mathbb{R}^{d - 1}$.

Consider first the case when all $\varepsilon_k$ go to zero together along a
fixed direction: $\varepsilon_k = r v_k$ where $r \rightarrow 0$ and $v =
(v_k)$ is a vector with positive components. Later on we will show that the
limit continues to exist if $\varepsilon_k \rightarrow 0$ independently (as
well as the more general statement about the limit from inside the forward
tube).

So, let us prove that $\mathcal{G} (y)$ has a limit as $r \rightarrow 0$ which
is a tempered distribution in variables $s_k$, $\mathbf{y}_k$. As in
{\eqref{intGn}}, we fix a Schwartz test function $f$ and consider the integral
(we will omit index $k$ on $\varepsilon, s, v, \mathbf{y}$ if no confusion may
arise)
\begin{equation}
  h (r) = \int d s\, d\mathbf{y}\,\mathcal{G} (r v + i s, \mathbf{y}) f (s,
  \mathbf{y}) .
\end{equation}
Te problem is analogous to theorems used in {\cite{paper1}}, so we will be
brief. As in {\cite{paper1}}, Sec.\ 3.3 and App.\ C, using analyticity in
$y^0$, integration by parts, and the powerlaw bound one can show that
derivatives of $h$ in $r$ satisfy the bound:
\begin{equation}
  | \partial_r^j h (r) | \leqslant \frac{C}{r^{A_n}}  | f |_N,
  \label{jbound}
\end{equation}
where $| f |_N$ is a Schwartz norm {\eqref{semin}} of a sufficiently large
order $N$ depending on $j$ and $B_n$. The constant $A_n$ is the same as in
{\eqref{powerlawbound}}, in particular the same $A_n$ works for all $j$. In
what follows we only need this equation for finitely many $j$ (up to $[A_n] +
1$). Using the Newton-Leibniz formula in the $r$ direction several times, one
then proves that the same bound as {\eqref{jbound}} holds in fact without $1 /
r^{A_n}$ singularity in the r.h.s. It then follows that, first of all,
$\lim_{r \rightarrow 0} h (r)$ exists, and second, it is a continuous linear
functional of $f$, that is, a distribution. The limit holds uniformly when the
components $v_k$ vary on any fixed compact interval contained in $(0, +
\infty)$. Its $v$-independence is shown exactly as in {\cite{paper1}}, Eq.
(C.7). Let us denote the limiting distribution $\mathcal{G} (i s, \mathbf{y})
\equiv \mathcal{G}^M (s, \mathbf{y})$.

It is of some interest to know the precise regularity of the distribution
$\mathcal{G}^M$ (i.e.\ how many derivatives the test function must have to be
pairable with $\mathcal{G}^M$) and the rate of its growth at infinity.
Following the above argument in detail, one can show the following bound which
contains this information:
\begin{equation}
  | (\mathcal{G}^M, f) | \leqslant \tmop{Const} . \int d s\, d\mathbf{y}\, (1 +
  | s | + | \mathbf{y} |)^{B_n} \max_{| \alpha | \leqslant [A_n] + 1} |
  \partial^{\alpha}_s f (s, \mathbf{y}) | . \label{GMreg}
\end{equation}
This in particular implies {\eqref{GMcont}} with $N_{\ast} = \max ([A_n] + 1,
B_n + n d + 1)$.

Parts 2,3 of Theorem \ref{ThVlad} are new compared to {\cite{paper1}}, since
such questions do not arise in the cross-ratio space.

Lorentz invariance is easy to show, as follows. Rotation invariance of
$G^E_n$ implies that $\mathcal{G} (y)$ satisfies for real $y$ the
differential equations
\begin{equation}
  \{ y^a \partial_{y^b} - y^b \partial_{y^a} \} \mathcal{G} (y) = 0,
  \qquad a, b \in \{ 0, 1, \ldots, d - 1 \} \label{DErot}
\end{equation}
(as usual $y = (y_k)$, summation in $k$ understood). By the uniqueness of
analytic continuation, these equations continue to hold for complex $y_0$.
That's the only place where we use real-analyticity in the spatial
direction.\footnote{With some extra tricks, it's possible to replace it by the
assumption of mere continuity in $\mathbf{y}$, as in {\cite{osterwalder1975}},
Theorem 4.3. In the CFT applications we have in mind, real analyticity appears
a more natural assumption.} By taking the limit $\varepsilon \rightarrow 0$ in
{\eqref{DErot}}, we recover precisely the differential equations expressing
the Lorentz invariance of $\mathcal{G}^M$. Let us explain in more detail how
the limit is taken and why it exists. Consider for definiteness $a = 0$, $b =
1$, other cases being similar. Eq.\ {\eqref{DErot}} then says \ $\{
(\varepsilon + i t) \partial_{y^1} + i y^1 \partial_t \} \mathcal{G}
(\varepsilon + i t, \mathbf{y}) = 0$, in the sense of functions, and hence
integrating by parts in the sense of distributions acting on test functions
$\varphi (t, \mathbf{y})$:
\begin{equation}
  (\mathcal{G}_{\varepsilon}, \{ (\varepsilon + i t) \partial_{y^1} + i y^1
  \partial_t \} \varphi) = 0, \label{Gepsdistr}
\end{equation}
where we denoted $\mathcal{G}_{\varepsilon} (t, \mathbf{y}) =\mathcal{G}
(\varepsilon + i t, \mathbf{y})$. Now we take the limit $\varepsilon
\rightarrow 0$. We know that (a) $\mathcal{G}_{\varepsilon} \rightarrow
\mathcal{G}^M$ in the sense of distributions, and also that (b) $|
(\mathcal{G}_{\varepsilon}, \varphi) |$ is uniformly bounded as $\varepsilon
\rightarrow 0$ by some Schwartz norm of $\varphi$. By (b) the term
$(\mathcal{G}_{\varepsilon}, \varepsilon \varphi)$ in {\eqref{Gepsdistr}}
drops out when $\varepsilon \rightarrow 0$, and by (a) the rest tends to
$(\mathcal{G}^M, \{ i t \partial_{y^1} + i y^1 \partial_t \} \varphi)$.
So we conclude that \ $(\mathcal{G}^M, \{ i t \partial_{y^1} + i y^1
\partial_t \} \varphi) = 0$ which expresses invariance of $\mathcal{G}^M$
under the $01$ Lorentz transformation.

Let us proceed to show the rest of Parts 2,3. It will be crucial that
$\mathcal{G}$ can be written as a ``Fourier-Laplace transform'':
\begin{equation}
  \mathcal{G} (\varepsilon + i s, \tmmathbf{\mathbf{y}}) = \int \frac{d E\,
  d\mathbf{p}}{(2 \pi)^{d (n - 1)}}\, g (E, \mathbf{p}) e^{- (\varepsilon + i s)
  E - i\mathbf{p}\mathbf{y}} \label{FL},
\end{equation}
where $g (E, \mathbf{p}), E \in \mathbb{R}^{n - 1}, \mathbf{p} \in
(\mathbb{R}^{d - 1})^{n - 1}$ is a tempered distribution, called ``spectral
function'', supported at $E \geqslant 0$ (by which we mean all $E_k \geqslant
0$) [later this will be improved to $E \geqslant | \mathbf{p} |$]. We are
omitting the indices, thus $\varepsilon E = \sum_k \varepsilon_k E_k$, etc.
The equality in {\eqref{FL}} is understood in the sense of distributions, with
the r.h.s.\ being the inverse Fourier transform of the tempered distribution $g
(E, \mathbf{p}) e^{- \varepsilon E}$. In other words, what this means is that
\begin{equation}
  \int d s\, d\mathbf{y} \, \mathcal{G} (\varepsilon + i s,
  \tmmathbf{\mathbf{y}}) f (- s, -\mathbf{y}) = \int \frac{d E\, d\mathbf{p}}{(2
  \pi)^{d (n - 1)}}  \, g (E, \mathbf{p}) e^{- \varepsilon E} \hat{f} (E,
  \mathbf{p}), \label{FTmeans}
\end{equation}
for any Schwartz test function $f$, and $\hat{f}$ its Fourier transform.

Let us show {\eqref{FL}}. Notice first that for every $\varepsilon > 0$ we can
write
\begin{equation}
  \mathcal{G} (\varepsilon + i s, \tmmathbf{\mathbf{y}}) = \int \frac{d E\,
  d\mathbf{p}}{(2 \pi)^{d (n - 1)}}\, g_{\varepsilon} (E, \mathbf{p}) e^{- i s E
  - i\mathbf{p}\mathbf{y}} \label{gx},
\end{equation}
where $g_{\varepsilon}$ is the Fourier transform of $\mathcal{G} (\varepsilon
+ i s, \tmmathbf{\mathbf{y}})$ with respect to $s, \mathbf{y}$. This Fourier
transform exists as a tempered distribution, since $\mathcal{G} (\varepsilon +
i s, \tmmathbf{\mathbf{y}})$ is itself a tempered distribution in $s,
\mathbf{y}$ (being a real-analytic function, bounded by a power at infinity).
In addition, $\mathcal{G} (\varepsilon + i s, \tmmathbf{\mathbf{y}})$ is
differentiable in $\varepsilon$ and $s$ and satisfies the Cauchy-Riemann
equations. From here it's easy to show that $g_{\varepsilon}$ as a
distribution is differentiable in $\varepsilon$ and satisfies the differential
equations:
\begin{equation}
  \frac{\partial g_{\varepsilon}}{\partial \varepsilon_k} + E_k
  g_{\varepsilon} = 0 \qquad (k = 1, \ldots, d - 1) .
\end{equation}
From here we conclude that
\begin{equation}
  g (E, \mathbf{p}) : = g_{\varepsilon} (E, \mathbf{p}) e^{\varepsilon E}
  \label{gge}
\end{equation}
is an $\varepsilon$-independent distribution. Substituting $g_{\varepsilon}
(E, \mathbf{p}) = g (E, \mathbf{p}) e^{- \varepsilon E}$ into {\eqref{gx}},
we obtain {\eqref{FL}}. Note that since $g$ and $g_{\varepsilon}$ are related
by an exponential factor, we can so far only claim that $g$ is defined as a
distribution on test functions of compact support. Let us show next that it is
in fact tempered (i.e.\ extends to Schwartz test functions).

To this end, consider the inverse of {\eqref{gx}}:
\begin{equation}
  g_{\varepsilon} (E, \mathbf{p}) = g (E, \mathbf{p}) e^{- \varepsilon E} =
  \int d s\, d\mathbf{y}\, \mathcal{G} (\varepsilon + i s, \tmmathbf{\mathbf{y}})
  e^{i s E + i\mathbf{p}\mathbf{y}},
\end{equation}
and integrate it against a compactly supported test function $\varphi (E,
\mathbf{p})$. We get (compare {\eqref{FTmeans}}):
\begin{equation}
  \int d E\, d\mathbf{p}\, g (E, \mathbf{p}) e^{- \varepsilon E} \varphi (E,
  \mathbf{p}) = \int d s\, d\mathbf{y}\, \mathcal{G} (\varepsilon + i s,
  \tmmathbf{\mathbf{y}}) \hat{\varphi} (- s, -\mathbf{y}) .
\end{equation}
As $\varepsilon \rightarrow 0$, the l.h.s.\ tends to the pairing $(g,
\varphi)$. Using Part 1 of the theorem, the r.h.s.\ tends in the same limit to
$\int d s\, d\mathbf{y}\, \mathcal{G}^M (s, \tmmathbf{\mathbf{y}}) \hat{\varphi} (-
s, -\mathbf{y})$ which exists in the sense of tempered distributions and so is
bounded by some Schwartz-space norm $| \hat{\varphi} |_N$. We get
\begin{equation}
  | (g, \varphi) | \leqslant \tmop{const} . | \hat{\varphi} |_N \leqslant
  \tmop{const} . | \varphi |_{N'},
\end{equation}
where in the second inequality we used that the Fourier transform is
continuous in the Schwartz space. This inequality, valid for any compactly
supported $\varphi$, means that $g$ extends to a tempered distribution on the
whole Schwartz space. The representation {\eqref{FL}} is thus established.

Next let us show that $g$ is supported at $E \geqslant 0$. For this we will
pass to the large $\varepsilon$ limit in {\eqref{gge}}. Supposing that $E_k <
0$ for some $k$, the factor $e^{E \varepsilon}$ in {\eqref{gge}} decreases
exponentially as the corresponding $\varepsilon_k \rightarrow + \infty$. On
the other hand $g_{\varepsilon} (E, \mathbf{p})$ is bounded in the same limit
by a power of $\varepsilon$, because it's the Fourier transform of
$\mathcal{G} (\varepsilon + i s, \tmmathbf{\mathbf{y}})$ which satisfies a
powerlaw bound.\footnote{This is the only place where we use the powerlaw
bound on $\mathcal{G} (\varepsilon + i s, \tmmathbf{\mathbf{y}})$ for large
rather than small $\varepsilon$.} This implies that $g (E, \mathbf{p}) = 0$
unless $E \geqslant 0$.\footnote{If unhappy with this intuitive reasoning, the
argument may be made more rigorous in its integrated version: show that $g$
vanishes on test functions supported in the complement of $E \geqslant 0$.}

Consider then the following lemma, proven analogously to, and easier than,
Lemma \ref{tubeLemma} below.

\begin{lemma}
  \label{gEp}Let $g (E, \mathbf{p})$ be a tempered distribution supported at
  $E \geqslant 0$, and consider the distribution $g (E, \mathbf{p}) e^{-
  \varepsilon E}$ ($\varepsilon > 0$). This distribution, being initially
  defined by this formula on compactly supported test functions, extends to a
  tempered distribution, and moreover $g (E, \mathbf{p}) e^{- \varepsilon E}
  \rightarrow g (E, \mathbf{p})$ as $\e\to 0$, in the sense of tempered distributions.
\end{lemma}

Let us now take the $\varepsilon \rightarrow 0$ limit on both sides of
{\eqref{FL}} (or, which is the same, {\eqref{FTmeans}}). The l.h.s.\ has a
limit by Part 1, while the r.h.s.\ has a limit by Lemma \ref{gEp}. We obtain
that $\mathcal{G}^M (s, \tmmathbf{\mathbf{y}})$ and $g (E,
\mathbf{p})$ are related by the Fourier transform:
\begin{equation}
  \mathcal{G}^M (s, \tmmathbf{\mathbf{y}}) = \int \frac{d E\, d\mathbf{p}}{(2
  \pi)^{d (n - 1)}}\, g (E, \mathbf{p}) e^{- i s E - i\mathbf{p}\mathbf{y}} .
  \label{gxM}
\end{equation}
We can now complete the proof of Part 2, namely to show the spectral
condition. Above we proved that $\mathcal{G}^M (s, \tmmathbf{\mathbf{y}})$ is
Lorentz invariant. Since $g (E, \mathbf{p})$ is its Fourier transform, it is
also Lorentz invariant, and in particular its support must be a
Lorentz-invariant set. We also know that $\tmop{supp} g \subset \{ E \geqslant
0 \}$. These two facts together imply that $\tmop{supp} g$ must be contained
in the product of the forward null cones, i.e.\ $g (E, \mathbf{p}) = 0$ unless
each $E_k \geqslant | \mathbf{p}_k |$, which is the spectral condition.

Part 3 follows by standard Wightman theory arguments. Namely, let us write $(i
y_k^0, \mathbf{y}_k) = \xi_k + i \eta_k $ where $\xi_k, \eta_k \in
\mathbb{R}^{1, d - 1}$ and $\eta_k = (\tmop{Re} y_k^0, \tmop{Im}
\mathbf{y}_k) \succ 0$. The extension to the forward tube is given by the
equation (with $p=(E,\mathbf{p})$)
\begin{equation}
  \int d p\, g (p) e^{i (p, \xi)} e^{(p, \eta)}, \label{FText}
\end{equation}
which reduces to {\eqref{FL}} for real $\mathbf{y}_k$. It is holomorphic by Part
(c) of the following lemma, while Parts (a,b) imply that this extension has
the same limit as {\eqref{FL}}.

\begin{lemma}
  \label{tubeLemma}Let $g (p)$ be a tempered supported at $p \succeq 0$
  (closed forward light cone). Consider the distribution $g_{\eta} (p) = g (p)
  e^{- (p, \eta)}$, initially defined by this formula on compactly supported
  test functions. Then
  
  (a) $g_{\eta}$ for $\eta \succ 0$ extends to a tempered distribution;
  
  (b) $g_{\eta} \rightarrow g$ as $\eta \rightarrow 0$ from inside the forward
  light cone $\eta \succ 0$, in the sense of tempered distributions;
  
  (c) The Fourier transform $\widehat{g_{\eta}} (\xi)$ of the distribution
  $g_{\eta} (p)$ is a holomorphic function of $\xi + i \eta$ {for $\eta\succ 0$}.
\end{lemma}

\begin{proof}
  Let $\omega (p)$ be a $C^{\infty}$ function which is identically 1 on the
  forward light cone $\overline{V_+}$, and zero as soon as $\tmop{dist} (p,
  \overline{V_+}) \geqslant 1$ where $\tmop{dist}$ is the Euclidean distance.
  We can choose this function so that all its derivatives are uniformly
  bounded by a constant depending only on the derivative order: $|
  \omega^{(\alpha)} (p) | \leqslant C_{\alpha}$ for any $p$.
  
  Consider the family of $C^{\infty}$ functions parametrized by $\xi, \eta \in
  \mathbb{R}^{1, d - 1}$:
  \begin{equation}
    \Omega_{\xi, \eta} (p) = e^{i (p, \xi)} e^{(p, \eta)} \omega (p) .
  \end{equation}
  It is not hard to check that $\Omega_{\xi, \eta}$ is a Schwartz function for
  $\eta \succ 0$ and any $\xi$.
  
  Let us define $g_{\eta}$ paired with a Schwartz function $\varphi (p)$ via
  \begin{equation}
    (g_{\eta}, \varphi) = (g, \Omega_{0, \eta} \varphi) .
  \end{equation}
  We know that $\Omega_{0, \eta} \varphi$ is a Schwartz function for $\eta
  \succ 0$, so this definition makes sense. Furthermore it is not hard to
  check that $\Omega_{0, \eta} \varphi \rightarrow \omega \varphi$ in the
  Schwartz space topology as $\eta \rightarrow 0, \eta \succ 0$. This proves
  Parts (a),(b).
  
  Next, let us define
  \begin{equation}
    F (\xi, \eta) = (g, \Omega_{\xi, \eta}), \qquad \xi, \eta \in
    \mathbb{R}^{1, d - 1} .
  \end{equation}
  We know that $\Omega_{\xi, \eta}$ is a Schwartz function for $\eta \succ 0$,
  so $F (\xi, \eta)$ is a function. Moreover it is not hard to show that the
  family $\Omega_{\xi, \eta}$ is continuous and continuously differentiable in
  the Schwartz space topology. It also obviously satisfies the Cauchy-Riemann
  equations: $(\partial_{\xi} + i \partial_{\eta}) \Omega_{\xi, \eta} = 0$.
  This implies that $F (\xi, \eta)$ is a holomorphic function in $\xi + i \eta$.
  It remains to show that $F (\xi, \eta) = \widehat{g_{\eta}} (\xi)$. It's
  enough to check this integrated against a compactly supported test function
  $\chi (\xi)$:
  \begin{equation}
    \int F (\xi, \eta) \chi (\xi)\, d \xi = \int (g, \Omega_{\xi, \eta}) \chi
    (\xi)\, d \xi = \left( g, \int d \xi\, \chi (\xi) \Omega_{\xi, \eta} \right) =
    (g, \Omega_{0, \eta} \hat{\chi} ) = (g_{\eta}, \hat{\chi} ) =
    (\widehat{g_{\eta}}, \chi) .
  \end{equation}
  The proof is complete.
\end{proof}

\section{Intuition about Lemma \ref{lemma:fcheckdense}}\label{IntLem1}

The proof of Lemma \ref{lemma:fcheckdense} in Sec.\ \ref{MinkFromEucl} was
by contradiction. To help intuition, we will give here a constructive argument
of a special case of Lemma \ref{lemma:fcheckdense}, namely $d = 1$ and $n =
2$. I.e.\ we will show how any Schwartz function $f \in \mathcal{S}
(\mathbb{R})$ can be approximated by Schwartz functions $g$ which for $E
\geqslant 0$ agree with Laplace transform:
\begin{equation}
  \mathcal{L} (\varphi) (E) = \int_0^{\infty} d t\, \varphi (t) e^{- E t},
\end{equation}
$\varphi \in C_0^{\infty} (\mathbb{R}_+)$ (compactly supported with support
strictly inside $(0, + \infty)$), while for $E < 0$, $g (E)$ is extended
arbitrarily. Recall that the Schwartz space topology is given by the family of
norms
\begin{equation}
  | f |_n = \sup_{E \in \mathbb{R}, m \leqslant n} (1 + E^2)^{n / 2} | f^{(m)}
  (E) |, \label{fnnorm}
\end{equation}
and we need to find a sequence $\{ g_r \}_{r = 1}^{\infty}$ such that $| f -
g_r |_n \rightarrow 0$ as $r \rightarrow \infty$ for any $n$ (we stress that
one sequence $g_r$ should work for any $n$).

We will also consider the Schwartz space $\mathcal{S}
(\overline{\mathbb{R}_+})$, consisting of $C^{\infty}$ functions on $E
\geqslant 0$ (not necessarily vanishing at $E = 0$) with topology given by the
family of norms $| f |_{n, +}$ defined by the same equations as
{\eqref{fnnorm}} but with sup taken over $E \geqslant 0$. It will be
sufficient to arrange that for any $n$
\begin{equation}
  | f -\mathcal{L} (\varphi_r) |_{n, +} \rightarrow 0 \qquad (r \rightarrow
  \infty) \label{suffSRp} .
\end{equation}
This is because there exists an extension operator which takes a function $h
\in \mathcal{S} (\overline{\mathbb{R}_+})$ and provides a function
$\mathcal{E} (h) \in \mathcal{S} (\mathbb{R})$ such that $\mathcal{E}
(h) = h$ for $E \geqslant 0$ (which is why it called an extension operator),
and in addition
\begin{equation}
  | \mathcal{E} (h) |_n \leqslant C_n | h |_{n, +} \label{Seeley}
\end{equation}
for all $n$ with some finite constants $C_n$ independent of $h$. E.g.,
Seeley's linear extension operator
{\cite{seeley_1964,wiki:Whitney_extension_theorem}} has this property. Then,
given {\eqref{suffSRp}}, we put
\begin{equation}
  g_r = f +\mathcal{E} (\mathcal{L} (\varphi_r) - f),
\end{equation}
which, on the one hand satisfies $g_r (E) =\mathcal{L} (\varphi_r) (E)$ for $E
\geqslant 0$ and on the other hand by {\eqref{suffSRp}} and {\eqref{Seeley}}
has $| g_r - f |_n \leqslant C_n | \mathcal{L} (\varphi_r) - f |_{n, +}
\rightarrow 0$ which is what we need.

So let us focus on satisfying {\eqref{suffSRp}}. By a map $x = \frac{1}{1 +
E}$ the half-line $[0, + \infty)$ is mapped to the interval $(0, 1]$ and the
function $f (E)$ is mapped to a function $F (x) = f \left( \frac{1}{x} - 1
\right)$ which is a $C^{\infty}$ function vanishing at $x = 0$ faster than any
power of $x$. For any $\varepsilon$ and any $N$ we can find, by the
Weierstrass theorem, a polynomial $Q (x)$ such that
\begin{equation}
  | F^{(N)} (x) - Q (x) | \leqslant \varepsilon \qquad (0 \leqslant x
  \leqslant 1) .
\end{equation}
Let $P (x)$ be the polynomial such that $P^{(N)} (x) = Q (x)$ and $P (0) =
\cdots = P^{(N - 1)} (0) = 0$. Then $P (x) = O (x^N)$ and it is not hard to
see that
\begin{equation}
  | F^{(n)} (x) - P^{(n)} (x) | \leqslant \varepsilon x^{N - n} \qquad (0
  \leqslant x \leqslant 1) . \label{FPbound}
\end{equation}
We also put $p (E) = P \left( \frac{1}{1 + E} \right)$. From \ $f (E) = F
\left( \frac{1}{1 + E} \right)$ we know that
\begin{equation}
  | f^{(n)} (E) | \leqslant B_n \max_{m \leqslant n} \left| F^{(m)} \left(
  \frac{1}{1 + E} \right) \right| .
\end{equation}
So combining this with Eq.\ {\eqref{FPbound}}, and going up to $n = N / 2$ we
may conclude that
\begin{equation}
  | f - p |_{N / 2, +} \leqslant B'_N \varepsilon . \label{f-p}
\end{equation}
Now, by construction $p$ has the form
\begin{equation}
  p (E) = \underset{N \leqslant n \leqslant M}{\sum} a_n \frac{1}{(1 + E)^n} .
  \label{def:g}
\end{equation}
Since $\frac{1}{(1 + E)^n} = \frac{1}{(n - 1) !} \int_0^{\infty} t^{n - 1}
e^{- (1 + E) t}\, d t$, we see that $p (E)$ is the Laplace transform of a
function $\psi (t)$:
\begin{equation}
  p =\mathcal{L} (\psi), \quad \psi (t) = \underset{N \leqslant n \leqslant
  M}{\sum} \frac{a_n}{(n - 1) !} t^{n - 1} e^{- t} .
\end{equation}

Now we can finish the argument as follows. For $r = 1, 2, 3, \ldots$ we apply
the above argument with $N = 2 r$ and $\varepsilon = 1 / (B_N' r)$ to find
$\psi_r$ such that, by {\eqref{f-p}},
\begin{equation}
  | f -\mathcal{L} (\psi_r) |_{r, +} \leqslant 1 / r . \label{fpsir}
\end{equation}
The function $\psi_r$ is not in $C_0^{\infty} (0, \infty)$ although
$\psi_r^{(k)} = 0$ for $k = 0 \ldots 2 r - 2$, and it vanishes at $\infty$
exponentially. We can therefore approximate $\psi_r$ by a $C_0^{\infty} (0,
\infty)$ function $\varphi_r$ so that $| \psi_r - \varphi_r |_{2 r - 2, +}$ is
arbitrarily small, where the order $2 r - 2$ of the norm is related to the
order of the vanishing of $\psi_r$ at $t = 0$. Furthermore we have the
following lemma:

\begin{lemma}
  Let $\chi$ be a $C^{\infty}$ function on $[0, + \infty)$ which exponentially
  vanishes at infinity and
  \begin{equation}
    \chi^{(k)} (0) = 0, \quad k = 0 \ldots n - 1 \label{chik} .
  \end{equation}
  Then, with some constant $D_n$ independent of $\chi$,
  \begin{equation}
    | \mathcal{L} (\chi) |_{n, +} \leqslant D_n | \chi |_{n + 2, +} .
    \label{further}
  \end{equation}
\end{lemma}

\begin{proof}
  We use the following elementary properties of Laplace transform:
  \begin{eqnarray}
    & \left( \frac{d}{d E} \right)^m \mathcal{L} (\chi) (E) =\mathcal{L}
    [\chi (t) (- t)^m] (E), &  \nonumber\\
    & E^n \mathcal{L} (\chi) (E) =\mathcal{L} [\chi^{(n)} (t)] (E), & 
  \end{eqnarray}
  where the second equation is derived by integration by parts and is valid
  under {\eqref{chik}} and exponential decay. So we have (where $\lesssim$
  denotes $\leqslant$ with some $n$-dependent but function-independent
  constant)
  \begin{equation}
    | \mathcal{L} (\chi) |_{n, +} \lesssim \sum_{m = 0}^n \sup_{E \geqslant 0}
    (1 + E^n) | \mathcal{L} (\chi)^{(m)} (E) | \leqslant \sum_{m = 0}^n
    \sup_{E \geqslant 0} | \mathcal{L} [\chi (t) t^m] (E) | + | \mathcal{L}
    [(\chi (t) t^m)^{(n)}] (E) |,
  \end{equation}
  Using further the elementary bound $| \mathcal{L} (f) (E) | \lesssim \sup_{t
  \geqslant 0} | (1 + t^2) f (t) |$ we deduce {\eqref{further}}.
\end{proof}

We use this lemma with $n = r$ and $\chi = \psi_r - \varphi_r$, which
satisfies $\chi^{(k)} = 0$ up to $k = 2 r - 2 \geqslant r - 1$, so
{\eqref{chik}} is satisfied. By {\eqref{further}}, we have
\begin{equation}
  | \mathcal{L} (\psi_r) -\mathcal{L} (\varphi_r) |_{r, +} \leqslant D_r |
  \psi_r - \varphi_r |_{r + 2, +} \leqslant D_r | \psi_r - \varphi_r |_{2 r -
  2, +}
\end{equation}
as long as $r \geqslant 4$ so that $2 r - 2 \geqslant r + 2$. As mentioned
above $| \psi_r - \varphi_r |_{2 r - 2, +}$ can be made arbitrarily small.
Combining with {\eqref{fpsir}}, we can arrange so that $| f -\mathcal{L}
(\varphi_r) |_{r, +} \leqslant 2 / r \rightarrow 0$ as $r \rightarrow \infty$,
which in particular implies {\eqref{suffSRp}}.

\small
\bibliography{lorentz}
\bibliographystyle{utphys}

\end{document}